\documentclass[12pt]{article}

\usepackage{amssymb}
\usepackage{amsmath}
\usepackage{amsthm}
\usepackage{graphicx}
\usepackage{caption}
\usepackage{subcaption} 
\usepackage{booktabs}
\usepackage{multirow}
\usepackage{xcolor}

\newtheorem{algorithm}{Algorithm}
\newtheorem{assumption}{Assumption}
\newtheorem{lemma}{Lemma}
\newtheorem{theorem}{Theorem}
\newtheorem{corollary}{Corollary}
\newtheorem{remark}{Remark}

\numberwithin{equation}{section}
\numberwithin{algorithm}{section}
\numberwithin{assumption}{section}
\numberwithin{lemma}{section}
\numberwithin{theorem}{section}
\numberwithin{corollary}{section}
\numberwithin{remark}{section}
\numberwithin{figure}{section}
\numberwithin{table}{section}

\graphicspath{{./figures/}} 

\title{\textbf{A Simple Robust Procedure in Instrumental Variables Regression\thanks{
I am grateful to Brendan Beare, Vanessa Berenguer-Rico, Debopam Bhattacharya, Steve Bond, David Cox, James Duffy, Jianqing Fan, Christophe Galliac, Sukjin Han, David Hendry, Ayden Higgins, Tetsuya Kaji, Jens Klooster, Anders Kock, Jonas Kurle, Michal Koles{\'a}r, Sophocles Mavroeidis, Bent Nielsen, Claudia Noack, Martin Weidner, Frank Windmeijer, and Qiwei Yao for helpful discussions and suggestions. I also thank seminar \& conference participants at Oxford, Warwick, Bristol, Essex, Peking, RES, and World Congress \& European \& Asian \& China Meetings of Econometric Society for their critical comments. 
}}}

\author{
Xiyu Jiao\thanks{E-mail: xiyu.jiao@economics.gu.se.} \\
Department of Economics, University of Gothenburg \\
Nuffield College, University of Oxford \\ 
}

\date{}

\begin{document}
\maketitle


\begin{abstract}
A common concern in empirical modelling centres around whether estimated regression coefficients are affected by a small set of outlying observations. To conduct outlier robustness checks in practical applications of instrumental variables regressions, the common practice is to run ordinary two stage least squares (2SLS) and remove observations with standardised residuals beyond a chosen cut-off value. Subsequently, the trimmed 2SLS is computed and compared to the original full-sample 2SLS. This paper aims to understand and improve the above heuristic procedure by establishing an asymptotic theory. Specifically, there are three main contributions of the paper. First, the trimmed 2SLS has a positive probability of removing observations even under the null hypothesis where the model contains no outliers. Under this situation, we derive a limiting Normal distribution of the trimmed 2SLS with the asymptotic variance as the ordinary one multiplied by a relative efficiency inflator. Furthermore, a bias correction factor is introduced for the variance estimator of structural errors, which otherwise would be downward biased. Second, a Hausman-type test is constructed to formalize the heuristic procedure of comparing between the two 2SLS estimators. Third, the trimmed 2SLS is a two-step procedure, which can be iterated until a fixed point is reached. The fixed point is shown to have the same first order asymptotics as the Huber-skip M-estimator. Our analysis involves a new class of empirical processes, whose theory would be of independent interest in applied probability. Simulation studies lend support to the asymptotic theory. An empirical illustration to Acemoglu et al. (2019) shows the utility of the proposed method.

\bigskip

\noindent
\textbf{JEL Classification:} C26, C36. \\
\noindent
\textbf{Keywords:} Outlier robustness; Instrumental variables; Trimmed 2SLS; Fixed point; Huber-skip estimator; Empirical processes.
\end{abstract}


\newpage
\section{Introduction}
A frequent concern in applied research centres around whether key empirical findings are driven by a tiny set of outliers. In practice, the most common approach to conduct outlier analysis is a trimming method based on the size of regression residuals, see for example Auerbach et al. (1994), De Long and Summers (1994), Fabrizio et al. (2007), Albouy (2012), Acemoglu et al. (2012), Toda and Walsh (2015), and Acemoglu et al. (2019). Those papers in instrumental variables regressions identify observations whose standardised structural (second stage) residual is above or below a chosen cut-off value (like $\pm 1.96$) and re-estimate the 2SLS model without these observations. The trimmed 2SLS is then compared to the original full-sample estimate to heuristically assess outlier sensitivity of estimated regression coefficients. This paper aims to understand and improve the above described heuristic procedure by establishing an asymptotic theory.

In the empirical literature, there is a recurrent critique that questions whether crucial results are robust after removing outlying observations. De Long and Summers (1991) found a positive effect of equipment investment on economic growth, which was followed by discussion of whether outliers drove such a positive connection, see Auerbach et al. (1994) and De Long and Summers (1994). When estimating an institutional impact on economic growth, Acemoglu et al. (2001) attracted extensive discussion of whether outliers undermined the validity of mortality rates as instruments, see Albouy (2012) and Acemoglu et al. (2012). Guthrie et al. (2012) criticised that a result in Chhaochharia and Grinstein (2009) was invalidated by the exclusion of a few observations. Herndon et al. (2014) also commented on the outlier sensitivity of the findings in Reinhart and Rogoff (2010). Toda and Walsh (2015) reported that trimming a small sub-sample of observations drastically changes the estimate of the risk aversion coefficient, making otherwise insignificant inference now significant. After conducting a comprehensive survey study of re-examining around 1400 IVs regressions in 32 papers published in journals of the American Economic Association, Young (2022) reported that most of the empirical findings obtained from the 2SLS models are sensitive to outliers. Broderick et al. (2020) developed a metric related to influence functions named ``Approximate Maximum Inference Perturbation" to find small subsets of the data that strongly affect regression estimates. When re-exploring several applied papers using the robustness metric, they reported that empirical results can often be reversed by removing less than one percent of the sample.

All valid empirical results should therefore share one essential feature that the analysis are not affected by a small portion of outliers to a non-negligible degree, see Young (2018). In practical applications of IVs regressions, the common practice is to conduct outlier robustness checks by re-running 2SLS on clean data and comparing the results from the original full-sample ones. The clean data is obtained by identifying and removing observations whose standardised structural (second stage) residuals are beyond a chosen cut-off value. However, such a heuristic procedure in the empirical literature lacks formal justification. The leading purpose of this paper is thus to develop its asymptotic theory under the null hypothesis of no outliers. Specifically, there are three main contributions of the paper.

First, the trimmed 2SLS has a positive probability of removing observations even under the null hypothesis where the model contains no outliers. Under this situation, we derive a limiting Normal distribution of the trimmed 2SLS estimator with the asymptotic variance as the ordinary one multiplied by a relative efficiency inflator. The trimmed estimator is less efficient under the null of no outliers in the sense that its asymptotic variance is larger than the full-sample estimator such that the relative efficiency inflator is strictly greater than one. Furthermore, a bias correction factor is introduced for the variance estimator of structural errors, which otherwise would be inconsistent and downward biased. Simulation studies lend strong support to our established asymptotic theory. Subsequently, an adjusted inference on structural parameters is suggested by taking into account of both the relative efficiency and bias correction factors.

Second, our established asymptotic theory of the trimmed 2SLS estimator is necessary to formalize the outlier robustness checks as statistical tests. A Hausman-type test is thus constructed to assess the outlier sensitivity by comparing between the ordinary and trimmed 2SLS estimators. With our results, we can test whether the parameter of interest changes its value significantly before and after outlier removal, enabling formal statistical investigation of outlier robustness analysis. When outliers are present, we should rely on the trimmed 2SLS estimator that is more robust. When no outliers are present, the full-sample 2SLS is more efficient. Unfortunately, in practice the presence of outliers and any resulting distortion in coefficients is unknown. The proposed test therefore also helps assess whether the gain in robustness offsets the loss of efficiency when using the robust trimmed estimator.

Third, the trimmed 2SLS is actually a two-step procedure, and most empirical studies stop after the initial iteration. In contrast, we suggest iterating the two-step procedure until a fixed point is reached. The paper demonstrates that the fixed point of the iterated procedure has the same first order asymptotics as the Huber (1964) skip regression. Thus, repeated application of the procedure would provide an estimate that is more resistant to outliers. In practice, empirical researchers commonly choose the full-sample 2SLS as a starting point to identify and remove observations. We could use a more robust initial estimator, since the full-sample estimator is not robust to outliers itself especially in presence of high leverage points. Iterating the procedure would also produce a more stable estimate, which is less influenced by the choice of the initial estimator.

An empirical illustration to Acemoglu et al. (2019) shows the utility of the proposed method. The paper tries to estimate the causal effect of democracy on economic growth measured by log GDP per capita. Due to the endogeneity problem, they apply IVs regressions by constructing a spatial instrument named regional waves of democratization. To confirm the validity of their positive estimates on democracy, they choose the cut-off value $1.96$ and perform (and compare to) the trimmed 2SLS. This paper re-explores the outlier robustness analysis conducted in Acemoglu et al. (2019) by using the improved robust procedure that is guided by our asymptotic theory. We find that the Hausman-type tests are rejected at the $5 \%$ significance level for most coefficients. This indicates that the outlier removal estimator is distinct from the full-sample one at least in magnitude. The adjusted inference based on the robust trimmed estimator cannot reject the hypothesis that the coefficient on democracy equals zero at the $5 \%$ significance level. Thus, the result of outlier analysis does not robustly support that democracy has a strictly positive effect causally on economic growth, though there is little evidence for the view that democracy is a constraint for economic development.

Analysis of robust methods in the iterated one-step framework is by no means new. One-step estimators have been investigated by Bickel (1975), Ruppert and Carroll (1980), and Welsh and Ronchetti (2002). The idea of iterating one-step estimators can be found in Dollinger and Staudte (1991), who applied an influence function argument to demonstrate convergence of iteratively re-weighted least squares with smooth weights. Notwithstanding we are interested in binary weights, the spirit of their argument still provides a guidance for our work. Other related work includes for example Cavaliere and Georgiev (2013), who explored the first order autoregression with infinite variance. Huber-skip M and L-estimators have been recently studied by Johansen and Nielsen (2009, 2013, 2016a, 2016b, 2018). Berenguer-Rico and Nielsen (2018) and Berenguer-Rico and Wilms (2018) explored and proposed diagnostic tests on residuals for normality and heteroskedasticity after outlier removal. However, all these works are restricted in ordinary regressions.

There is a closely related paper that tries to solve exactly the same problem in IVs regressions. Kaji (2018) studied the trimmed 2SLS estimator by developing a theory of empirical processes and a functional delta method for quantile processes. His machinery can be used to investigate L-statistics represented outlier-robust estimators (integrals of empirical quantile functions with respect to corresponding random sample selection measures). Compared to Kaji (2018), we characterise the trimmed 2SLS statistics as a different class of empirical processes, for which we also need to develop a theory. Our argument is constructed specifically for the trimmed 2SLS, which thus allow us to explore the procedure in depth. For example, we can study the asymptotics of the variance estimator, iteration of the algorithms, variations of the robustification parameter (cut-off), and the procedure starting with a different initial estimator, like the least trimmed squares. In addition, Kaji (2018) applied the non-parametric bootstrap to implement the test of comparing the full-sample 2SLS to the trimmed one. Although the bootstrapping version of the test is valid under the null, it has the lower power under alternatives where there are significant outliers in the data. This is because the bootstrapping distribution is distorted by those bootstrap re-samples which include outliers, and thus it is distinct from what it should be under the null of no outliers.


Our analysis involves a new class of empirical processes. We derive first order asymptotic expansions for the required new class of weighted and marked empirical processes of residuals using a martingale decomposition argument. This is proceeded in a classical way by adding and subtracting a compensator term to construct a martingale. There are three key proofs to establish our empirical process theory. First, we linearise the compensator by using the first order Taylor expansion. Second, a chaining argument and concentration inequalities are applied to prove that the constructed martingale converges to zero in probability. Third, an asymptotic equi-continuity is demonstrated for the empirical processes by performing a dyadic argument. Technical difficulty here is to show convergence uniformly in all dimensions. The established theory of empirical processes would be of independent interest in applied probability.

We define a filtration $\mathcal{F}_{i - 1} = \sigma (z_{1}, \ldots, z_{i}, r_{1}, \ldots, r_{i - 1}, u_{1}, \ldots, u_{i - 1})$ to construct a martingale. Note that $\{ (z_{i}, r_{i}, u_{i}) \}_{i = 1}^{n}$ is the set of instruments, projection (first stage) errors, structural (second stage) errors. The empirical processes can then be defined from the generalized empirical distribution function of structural residuals
\begin{equation*}
\widehat{\mathsf{F}}_{u, n}^{w, p}(a, b, c) = \frac{1}{n} \sum_{i=1}^{n} w_{in} u_{i}^{p} 1_{(u_{i} \le \sigma c +n^{-1/2}a c + z_{in}^{\prime} \Pi b + n^{-1/2}r_{i}^{\prime}b)},
\end{equation*}
where the weights $w_{in}$ are combinations of the normalized $\mathcal{F}_{i - 1}$ measurable instruments $z_{in}$ and $u_{i}^{p}$ are the $\mathcal{F}_{i}$ adapted marks, while $a$, $b$ represent the normalized estimation errors for $\sigma$ (variance of structural errors), $\beta$ (structural parameters). Also note that $\Pi$ is the first stage regression coefficient matrix. We derive the theory that are uniform in $a$, $b$, $c$ and allow for a near $n^{1/4}$ inefficiency in the estimation uncertainties $a$, $b$.

The empirical process literature dates back to Kolmogorov (1933) and Smirnov (1939) who proposed a type of goodness of fit tests that check whether the distribution under the null is well specified by comparing it with the empirical distribution function. To build asymptotics of Kolmogorov-Smirnov type test statistics, Doob (1949) first demonstrated weak convergence of empirical distribution function, and then Donsker (1952) established the empirical process central limit theorem to close a gap in Doob's proof. Because there exists the measurability issue in $D[0,1]$ space, the classical method applies the Skorokhod metric instead of uniform topology to avoid non-measurability\footnote{see details in Billingsley (1968).}. Meanwhile, many researchers still equip $D[0,1]$ space with uniform metric but use the Hoffmann-J\o rgensen\footnote{a sequence of random elements $X_{n}$ weakly converges to the limiting random element $X$ if and only if $\mathsf{E}^{\ast} f(X_{n}) \to \mathsf{E} f(X)$ as $n \to \infty$ for any continuous and bounded function $f$ where $\mathsf{E}^{\ast}$ denotes the outer expectation.} definition of weak convergence instead. With the new concept of weak convergence, they extend the classical theory to the empirical process indexed by VC class of sets or functions using the entropy and bracketing argument\footnote{see the summary in van der Vaart and Wellner (1996).}.  However, their work has not yet been extended to deal with the weights $w_{in}$ and marks $u_{i}^{p}$ appearing in our context.

Following the classical idea, Koul and Ossiander (1994) used the entropy argument to first build up the theory for the weighted empirical process in the autoregressive model\footnote{also see Koul (2002) and Koul and Ling (2006).}. Recently, empirical processes including both weights and marks have been analyzed by a series of papers, see Engler and Nielsen (2009), Johansen and Nielsen (2009, 2016a), Jiao and Nielsen (2017), Jiao (2017), Berenguer-Rico and Nielsen (2018), Berenguer-Rico, Johansen, and Nielsen (2019a, b), and Nielsen and Qian (2022). However, all these works are restricted in ordinary regressions. 

The paper proceeds as follows: \S \ref{model and a robust procedure} introduces the models and proposes the improved procedure robust to outliers. \S \ref{the main results} presents the main results followed by \S \ref{simulation studies} simulation studies. \S \ref{empirical application to Acemoglu et al. (2019)} re-explores outlier robustness analysis in Acemoglu et al. (2019). \S \ref{weighted and marked empirical process} establishes the theory for a new class of empirical processes with proofs in Appendix \S \ref{a metric on R and some inequalities}, \ref{proofs of the empirical process results}. Lastly, proofs of main theorems in \S \ref{the main results} are shown in Appendix \S \ref{proofs of the main results}.

\section{Model and a robust procedure} \label{model and a robust procedure}
The instrumental variables (IVs) regressions with some notations are described first. It is widely known that two stage least squares estimator is sensitive to outliers. We review an outlier robust procedure in the instrumental variables regressions.

\subsection{Model} \label{model}
Suppose in the cross-sectional settings we have independently and identically distributed data $\{(y_{i}, x_{i}, z_{i})\}_{i = 1}^{n}$ across individuals, where $y_{i}$ is univariate and $x_{i}$, $z_{i}$ are multivariate with dimension $d_{x}$, $d_{z}$.
Assume the data $\{(y_{i}, x_{i})\}_{i = 1}^{n}$ satisfies the structural equation
\begin{equation} \label{structural equation}
y_{i} = x_{i}^{\prime} \beta + u_{i}, \quad i = 1, 2, \ldots, n.
\end{equation}
Structural errors $\{u_{i}\}_{i = 1}^{n}$ are univariate i.i.d. random variables with scale $\sigma$ so that $u_{i}/\sigma$ has the known density $\mathsf{f}_{u}(y)$ and distribution function $\mathsf{F}_{u}(y) = \mathsf{P}(u_{i}/\sigma \le y)$ for $y \in \mathbb{R}$ with mean $0$ and variance $1$. In the structural model some elements of $x_{i}$ are endogenous, i.e. $\mathsf{E} x_{i} u_{i} \neq 0_{d_{x}}$, instruments $z_{i}$ are required for estimating the parameter $\beta \in \mathbb{R}^{d_{x}}$ consistently. 
The first stage regression holds for $\{(x_{i}, z_{i})\}_{i = 1}^{n}$
\begin{equation} \label{first stage regression}
x_{i}^{\prime} = z_{i}^{\prime} \Pi + r_{i}^{\prime}, \quad i = 1, 2, \ldots, n.
\end{equation}
Instruments $z_{i}$ are assumed to have the zero mean so $\mathsf{E} z_{i} = 0_{d_{z}}$ and to be orthogonal to errors $r_{i}$ in the first stage equation so $\mathsf{E} z_{i} r_{i}^{\prime} = 0_{d_{z} \times d_{x}}$. Innovations $\{r_{i}\}_{i = 1}^{n}$ are $d_{x}$-variate i.i.d. random vector with symmetric and positive definite dispersion matrix $\Sigma \in \mathbb{R}^{d_{x} \times d_{x}}$ so $\Sigma^{-1/2} r_{i}$ follows the density $\mathsf{f}_{r}(x)$ and distribution function $\mathsf{F}_{r}(x)$ for $x \in \mathbb{R}^{d_{x}}$ with mean $0_{d_{x}}$ and identity variance-covariance matrix $I_{d_{x}}$. The parameter $\Pi$ lies in $\mathbb{R}^{d_{z} \times d_{x}}$ and we suppose $d_{x} \le d_{z}$ meaning the number of regressors is less than or equal to the number of instruments. Furthermore, to identify structural parameters $\beta$, instruments $z_{i}$ need to be valid, i.e. $\mathsf{E} z_{i} u_{i} = 0_{d_{z}}$, and informative, i.e. rank condition for identification ($\operatorname{rank} \Pi = d_{x}$ and $\operatorname{rank} \mathsf{E} z_{i} z_{i}^{\prime} = d_{z}$).

To apply the martingale argument, we need to construct the filtration
\begin{equation} \label{filtration for martingale argument}
\mathcal{F}_{i - 1} = \sigma (z_{1}, \ldots, z_{i}, r_{1}, \ldots, r_{i - 1}, u_{1}, \ldots, u_{i - 1})
\end{equation}
such that $u_{i}$, $r_{i}$ are $\mathcal{F}_{i}$ measurable and independent of $\mathcal{F}_{i - 1}$ while $z_{i}$ is adapted to $\mathcal{F}_{i - 1}$. Notice in the IVs setting $u_{i}$ is correlated with $r_{i}$. Assume the scaled errors $(u_{i}/\sigma, \Sigma^{-1/2} r_{i})$ have the joint density $\mathsf{f}_{u, r}(y, x)$ and distribution function $\mathsf{F}_{u, r}(y, x)$ for $y \in \mathbb{R}$, $x \in \mathbb{R}^{d_{x}}$. Note the joint distribution $\mathsf{f}_{u, r}, \mathsf{F}_{u, r}$ does also depend on the covariance $\Omega = \mathsf{Cov}(u_{i}/\sigma, \Sigma^{-1/2} r_{i}) = \mathsf{E} (\Sigma^{-1/2} r_{i} u_{i} / \sigma)$, however for simplicity $\Omega \in \mathbb{R}^{d_{x}}$ is suppressed in the notation of joint density. Further suppose the joint density can be decomposed into the conditional and marginal ones so $\mathsf{f}_{u, r}(y, x) = \mathsf{f}_{u|r}(y|x) \mathsf{f}_{r}(x) = \mathsf{f}_{r|u}(x|y) \mathsf{f}_{u}(y)$. Then denote the conditional expected value
\begin{equation} \label{conditional expectation of r given u = y}
\xi_{y} = \mathsf{E} (\Sigma^{-1/2} r_{i} | u_{i} / \sigma = y) = \int_{\mathbb{R}^{d_{x}}} x \mathsf{f}_{r|u}(x|y) (dx).
\end{equation}
Notice $\xi_{y} \in \mathbb{R}^{d_{x}}$ is related to the covariance $\Omega$ between $u_{i} / \sigma$ and $\Sigma^{-1/2} r_{i}$. In practice $\mathsf{f}_{u, r}$, $\mathsf{F}_{u, r}$ will often be assumed to be $(1 + d_{x})$-variate normal, so the scaled error vector $(u_{i} / \sigma, \Sigma^{-1/2} r_{i})$ follows
\begin{equation} \label{normality for scaled error vector}
\begin{pmatrix}
u_{i} / \sigma \\
\Sigma^{-1/2} r_{i}
\end{pmatrix}
\overset{\mathsf{D}}{=}
\mathsf{N}
\begin{Bmatrix}
\begin{pmatrix}
0 \\
0_{d_{x}}
\end{pmatrix}
,
\begin{pmatrix}
1 & \Omega^{\prime} \\
\Omega & I_{d_{x}}
\end{pmatrix}
\end{Bmatrix}.
\end{equation}
As the marginal and conditional distribution of multivariate normal are still normally distributed, then $\Sigma^{-1/2} r_{i} | u_{i} / \sigma \sim \mathsf{f}_{r|u}$ follows normal $\mathsf{N} (\Omega u_{i} / \sigma, I_{d_{x}} - \Omega \Omega^{\prime})$. Thus the conditional expectation in (\ref{conditional expectation of r given u = y}) has the form $\xi_{y} = \Omega y$. Denote
\begin{equation} \label{plus and minus conditional mean}
\zeta_{y}^{+} = \xi_{y} + \xi_{-y}, \qquad \zeta_{y}^{-} = \xi_{y} - \xi_{-y},
\end{equation}
and introduce the truncated covariance
\begin{equation} \label{truncated covariance}
\omega_{c} = \mathsf{E} (\Sigma^{-1/2} r_{i} \frac{u_{i}}{\sigma} 1_{(|u_{i}| \le \sigma c)}) = \iint_{y \in \mathbb{R}, x \in \mathbb{R}^{d_{x}}} x y 1_{(|y| \le c)} \mathsf{f}_{u, r}(y, x) dy (dx).
\end{equation}
Further derive $\omega_{c}$ to attain
\begin{equation} \label{another expression for truncated covariance}
\omega_{c} = \int_{-c}^{c} \{ \int_{x \in \mathbb{R}^{d_{x}}} x \mathsf{f}_{r|u}(x|y) (dx) \} y \mathsf{f}_{u}(y) dy = \int_{-c}^{c} \xi_{y} y \mathsf{f}_{u}(y) dy.
\end{equation}
Under normality in (\ref{normality for scaled error vector}), $\omega_{c} = \tau_{2}^{c} \Omega$ and $\xi_{-y} = - \Omega y = - \xi_{y}$ such that $\zeta_{y}^{+} = 0_{d_{x}}$ and $\zeta_{y}^{-} = 2 \xi_{y} = 2 \Omega y$.

The robust procedure we analyzed detects outliers through checking if absolute standardised residuals in the structural equation (\ref{structural equation}) are beyond the chosen cut-off value $c$ and then calculating the robust two stage least squares estimator from the non-outlying sample. This implicitly assumes symmetry of $u_{i}/\sigma$, while non-symmetry leads to specific forms of bias. We assume symmetry for $\mathsf{f}_{u}$ when analyzing the robust algorithm in \S \ref{the main results}, but not for the general empirical process results in \S \ref{weighted and marked empirical process}. Let the absolute error $|u_{i}|/\sigma$ in (\ref{structural equation}) have a density $\mathsf{g}_{u}(y)$ and distribution function $\mathsf{G}_{u}(y) = \mathsf{P}(|u_{i}|/\sigma \le y)$ for $y > 0$. With a symmetry assumption, $\mathsf{G}_{u}(y) = 2\mathsf{F}_{u}(y) - 1$ and $\mathsf{g}_{u}(y) = 2\mathsf{f}_{u}(y)$ for $y > 0$. Define $\psi = \mathsf{G}_{u}(c)$ so the probability of exceeding the cut-off $c$ is $\gamma = 1 - \psi$. Suppose the $k$-th moment of the density $\mathsf{f}_{u}$ exists, then introduce
\begin{equation} \label{moments and truncated moments}
\tau_{k} = \int_{- \infty}^{\infty} y^{k} \mathsf{f}_{u}(y) dy, \qquad \tau_{k}^{c} = \int_{-c}^{c} y^{k} \mathsf{f}_{u} (y) dy.
\end{equation}
Thus $\tau_{0}^{c} = \psi$, $\tau_{2} = 1$ while $\tau_{k} = \tau_{k}^{c} = 0$ for odd $k$ when assuming symmetry for $\mathsf{f}_{u}$. Define the conditional variance of $u_{i}/\sigma$ given $( |u_{i}|/\sigma \le c )$ as
\begin{equation} \label{bias correction factor}
\varsigma_{c}^{2} = \frac{\tau_{2}^{c}}{\psi} = \frac{\int_{-c}^{c} y^{2} \mathsf{f}_{u} (y) dy}{\mathsf{P}(|u_{i}| \le \sigma c)}.
\end{equation}
This will be used as a bias correction factor for the variance estimate computed from the selected non-outlying sample. With normality assumption (\ref{normality for scaled error vector}), we have $u_{i} / \sigma \sim \mathsf{f}_{u}$ follows standard normal $\mathsf{N}(0, 1)$ then $\tau_{2}^{c} = \psi - 2c\mathsf{f}_{u}(c)$, $\tau_{4} = 3$, $\tau_{4}^{c} = 3 \psi - 2c(c^{2} + 3)\mathsf{f}_{u}(c)$.

For matrices $M$, choose the spectral norm $|M| = \max \{ \mathrm{eigen} (M^{\prime} M) \}^{1/2}$ so that for vectors $x$ then $|x|$ is the Euclidean norm. The spectral norm is compatible with respect to the Euclidean norm so $|Mx| \le |M| |x|$. Notice $(dM)$ represents the exterior product of all elements in the matrix $M$ to describe the multivariate integral, see the chapter 2 of Muirhead (1982). For instance consider a vector $x \in \mathbb{R}^{d_{x}}$, then $(dx)$ is the exterior product of some measures on $\mathbb{R}^{d_{x}}$.

\subsection{Procedures robust to outliers} \label{robust statistical algorithms to detect outliers}
We first define an iterated version of robust two stage least squares estimators. Two specific examples with different initial estimators are then described.

To carry out outlier analysis, Acemoglu et al. (2019) first compute the full sample 2SLS and choose it as the initial estimator for the parameter $\beta$, then calculate all the residuals in the structural equation and detect outliers if the absolute value of standardised residuals in (\ref{structural equation}) are beyond the selected cut-off value $1.96$. Running two stage least squares based on the non-outlying observations gives them the updated robust estimator for $\beta$. Finally, they compare two $\beta$ estimates for outlier robustness checks. In principle, a different initial estimators or cut-off values can be chosen and this robust procedure can also be iterated. We define a formal version of the algorithm with the corrected variance estimator as follows.

\begin{algorithm} \label{iterated version of robust 2sls}
(\textbf{Iterated 2SLS}). Choose a cut-off $c > 0$. \\
1. Select initial estimators $\widehat{\beta}_{c}^{(0)}$, $(\widehat{\sigma}_{c}^{(0)})^{2}$ and let $m = 0$. \\
2. Define indicator variables for selecting non-outlying observations
\begin{equation} \label{2sls indicator for non-outlying observations}
v_{i, c}^{(m)} = 1_{(|y_{i} - x_{i}^{\prime}\widehat{\beta}_{c}^{(m)}| \le \widehat{\sigma}_{c}^{(m)} c)}.
\end{equation}
3. Calculate least squares estimators for $\Pi$
\begin{equation} \label{location estimator for the first stage regression}
\widehat{\Pi}_{c}^{(m + 1)} = (\sum_{i = 1}^{n} z_{i}z_{i}^{\prime} v_{i, c}^{(m)})^{-1}(\sum_{i = 1}^{n} z_{i}x_{i}^{\prime} v_{i, c}^{(m)}).
\end{equation}
4. Compute two stage least squares estimators for $\beta$ and $\sigma^{2}$
\begin{align}
\widehat{\beta}_{c}^{(m+1)} & = (\widehat{\Pi}_{c}^{(m + 1) \prime} \sum_{i=1}^{n} z_{i}z_{i}^{\prime} v_{i, c}^{(m)} \widehat{\Pi}_{c}^{(m + 1)})^{-1} (\widehat{\Pi}_{c}^{(m + 1) \prime} \sum_{i=1}^{n} z_{i}y_{i}  v_{i, c}^{(m)}), \label{updated 2sls location} \\
(\widehat{\sigma}_{c}^{(m+1)})^{2} & = \varsigma_{c}^{-2} (\sum_{i=1}^{n} v_{i, c}^{(m)})^{-1} \{ \sum_{i=1}^{n} (y_{i} - x_{i}^{\prime} \widehat{\beta}_{c}^{(m+1)})^{2} v_{i, c}^{(m)} \}. \label{updated 2sls variance}
\end{align}
5. Let $m = m + 1$ and repeat 2, 3, and 4.
\end{algorithm}

Algorithm \ref{iterated version of robust 2sls} is a similar procedure as the iterated one-step Huber-skip M-estimator\footnote{one-step updated estimator mimics the Huber (1964) skip estimator, which has criterion function $\rho (y)= \min(y^2, c^2)/2$ as opposed to the Huber estimator with criterion function $\rho (y)= y^{2}/2$ for $|y| \le c$ and $\rho(y)=c|y|-c^2/2$ otherwise, see Hampel et al. (1986, p. 104).} in the instrumental variables regressions. Note that the bias correction factor $\varsigma_{c}^{2}$ is required in (\ref{updated 2sls variance}) in order to have a consistent estimator of $\sigma^{2}$, otherwise $\sigma^{2}$-estimator is downward biased with the probability limit $\varsigma_{c}^{2} \sigma^{2}$. 

In \S \ref{the main results}, we establish asymptotic theory of this algorithm, which can hold uniformly in the cut-off values $c$. The algorithm could start with a robust estimator, while as the procedure applied in Acemoglu et al. (2019) the so called Robustified 2SLS is initiated using the full sample 2SLS. The following algorithm formally defines Robustified 2SLS.

\begin{algorithm} \label{robustified two stage least squares}
(\textbf{Robustified 2SLS}). Choose a cut-off $c > 0$. \\
1.1. Calculate least squares estimator for $\Pi$ based upon the whole sample
\begin{equation} \label{r2sls initial location estimator in the first stage}
\widehat{\Pi}_{c}^{(0)} = (\sum_{i = 1}^{n} z_{i}z_{i}^{\prime})^{-1}(\sum_{i = 1}^{n} z_{i}x_{i}^{\prime}).
\end{equation}
1.2. Compute two stage least squares estimator for $\beta$ and $\sigma^{2}$ for the whole sample
\begin{equation} \label{r2sls initial 2sls estimator}
\widehat{\beta}_{c}^{(0)} = (\widehat{\Pi}_{c}^{(0) \prime} \sum_{i=1}^{n} z_{i}z_{i}^{\prime} \widehat{\Pi}_{c}^{(0)})^{-1} (\widehat{\Pi}_{c}^{(0) \prime} \sum_{i=1}^{n} z_{i}y_{i}), \qquad (\widehat{\sigma}_{c}^{(0)})^{2} = \frac{1}{n} \sum_{i=1}^{n} (y_{i} - x_{i}^{\prime} \widehat{\beta}_{c}^{(0)})^{2}.
\end{equation}
1.3. Choose $\widehat{\beta}_{c}^{(0)}$, $(\widehat{\sigma}_{c}^{(0)})^{2}$ through $(\ref{r2sls initial 2sls estimator})$, and let $m = 0$. \\
2. Follow the step 2, 3, 4, and 5 in Algorithm \ref{iterated version of robust 2sls}.
\end{algorithm}

Algorithm \ref{robustified two stage least squares} is not very robust with respect to high leverage points when we have i.i.d. cross-sectional data. Thus, to deal with the high leverage points, another typical example of Algorithm \ref{iterated version of robust 2sls} starting with the split sample 2SLS was first proposed but in the classical regression by Hendry (1999) in his empirical work. This saturated idea is to divide the full sample into two sub-samples and use 2SLS calculated from each sub-sample to detect outliers in the other sub-sample. The following algorithm formally defines Saturated 2SLS.

\begin{algorithm} \label{split sample version}
(\textbf{Saturated 2SLS}). Choose a cut-off $c > 0$. \\
1.1. Split full sample into two sets $\mathcal{I}_{j}$, $j = 1, 2$ of $n_{j}$ observations where $\sum_{j=1}^{2} n_{j} = n$. \\
1.2. Calculate least squares estimators for $\Pi$ based upon each sub-sample $\mathcal{I}_{j}$ for $j = 1, 2$
\begin{equation} \label{split version initial location estimator in the first stage}
\widehat{\Pi}_{j} = (\sum_{i \in \mathcal{I}_{j}} z_{i}z_{i}^{\prime})^{-1} (\sum_{i \in \mathcal{I}_{j}} z_{i}x_{i}^{\prime}).
\end{equation}
1.3. Compute two stage least squares estimators for $\beta$ and $\sigma^{2}$ for sub-sample $\mathcal{I}_{j}$, $j = 1, 2$
\begin{equation} \label{split version initial 2sls estimator}
\widehat{\beta}_{j} = (\widehat{\Pi}_{j}^{\prime} \sum_{i \in \mathcal{I}_{j}} z_{i}z_{i}^{\prime} \widehat{\Pi}_{j})^{-1} (\widehat{\Pi}_{j}^{\prime} \sum_{i \in \mathcal{I}_{j}} z_{i}y_{i}), \qquad \widehat{\sigma}_{j}^{2} = \frac{1}{n_{j}} \sum_{i \in \mathcal{I}_{j}} (y_{i} - x_{i}^{\prime} \widehat{\beta}_{j})^{2}.
\end{equation}
1.4. Define the initial indicator variables for selecting non-outlying observations
\begin{equation} \label{initial indicator for split sample}
v_{i, c}^{(0)} = 1_{(i \in \mathcal{I}_{1})} 1_{(|y_{i} - x_{i}^{\prime}\widehat{\beta}_{2}| \le \widehat{\sigma}_{2}c)} + 1_{(i \in \mathcal{I}_{2})} 1_{(|y_{i} - x_{i}^{\prime}\widehat{\beta}_{1}| \le \widehat{\sigma}_{1}c)}.
\end{equation}
1.5. Compute $\widehat{\Pi}_{c}^{(1)}$, $\widehat{\beta}_{c}^{(1)}$, $(\widehat{\sigma}_{c}^{(1)})^{2}$ using $(\ref{location estimator for the first stage regression})$, $(\ref{updated 2sls location})$, $(\ref{updated 2sls variance})$ with $m = 0$, and let $m = 1$. \\
2. Follow the step 2, 3, 4, and 5 in Algorithm \ref{iterated version of robust 2sls}.
\end{algorithm}

Algorithm \ref{split sample version} is possibly more robust than the Robustified Two Stage Least Squares when we have prior knowledge that outliers are located in a particular subset of the whole sample. Since the location of contaminated observations is unknown in most practical situations, the choice of the initial sets $\mathcal{I}_{1}$ and $\mathcal{I}_{2}$ should be iterated.

\section{The main results} \label{the main results}
We start by listing assumptions, then asymptotics for Iterated 2SLS and weak convergence for Robustified 2SLS and Saturated 2SLS. It is followed by comparing efficiency between robust and non-robust 2SLS under the null of no outliers. Finally, I provide the valid inference procedure on $\beta$ when implementing these trimming methods and propose a new Hausman type test to check the presence of influential outliers. 

\subsection{Assumptions}
We list the sufficient assumptions for asymptotic theory of Algorithm \ref{iterated version of robust 2sls}, \ref{robustified two stage least squares}, \ref{split sample version}. To simplify analysis, this section focuses on somewhat stronger conditions than they need to be. In section \S \ref{weighted and marked empirical process} on the one-sided empirical process, we will introduce some weaker assumptions. For example, we impose the symmetricity assumption on $\mathsf{f}_{u}$ in this section but not in section \S \ref{weighted and marked empirical process}, though the idea of reasoning is the same and the results in \S \ref{the main results} can be extended to the asymmetric case using the argument in Johansen and Nielsen (2009).

Since $\{ (y_{i}, x_{i}, z_{i}) \}_{i = 1}^{n}$ are i.i.d., instruments $z_{i}$ can be normalized by the rate $n^{-1/2}$ so denote $z_{in} = n^{-1/2} z_{i}$, then we have $M_{z z, n} = \sum_{i = 1}^{n} z_{in} z_{in}^{\prime} = n^{-1} \sum_{i = 1}^{n} z_{i} z_{i}^{\prime} \overset{\mathsf{P}}{\to} \mathsf{E} z_{i} z_{i}^{\prime} = M_{z z}$ by Law of Large Numbers assuming $\mathsf{E} |z_{i}|^{2} < \infty$. Denote the infeasible ideal fitted value $\tilde{x}_{i} = \Pi^{\prime} z_{i}$ while the feasible counterpart is attained by estimating $\Pi$ consistently. The normalized $\tilde{x}_{i}$ is defined as $\tilde{x}_{in} = n^{-1/2} \tilde{x}_{i} = \Pi^{\prime} z_{in}$ such that
\begin{equation*}
M_{\tilde{x} \tilde{x}, n} = \sum_{i = 1}^{n} \tilde{x}_{in} \tilde{x}_{in}^{\prime} = \Pi^{\prime} M_{z z, n} \Pi \overset{\mathsf{P}}{\to} \Pi^{\prime} M_{z z} \Pi = \mathsf{E} \tilde{x}_{i} \tilde{x}_{i}^{\prime}  = M_{\tilde{x} \tilde{x}}.
\end{equation*}
If $d_{x} \le d_{z}$ and rank condition for identification hold so that $\operatorname{rank} \Pi = d_{x}$, $M_{z z} > 0$, then we have $M_{\tilde{x} \tilde{x}} > 0$ as well.

To carry out asymptotic analysis, the scaled error vector $(u_{i} / \sigma, \Sigma^{-1/2} r_{i})$, instruments $z_{i}$, and the initial estimator $(\widetilde{\beta}, \widetilde{\sigma}^{2})$ must satisfy the following conditions.

\begin{assumption} \label{sufficient assumptions}
Let $\mathcal{F}_{i - 1} = \sigma(z_{1}, \ldots, z_{i}, r_{1}, \ldots, r_{i - 1}, u_{1}, \ldots, u_{i - 1})$ be an increasing sequence of $\sigma$-fields so $u_{i-1}$, $r_{i - 1}$, $z_{i}$ are $\mathcal{F}_{i-1}$ measurable while $r_{i}$, $u_{i}$ are independent of $\mathcal{F}_{i - 1}$. Suppose $(u_{i}/\sigma, \Sigma^{-1/2} r_{i})$ have continuously differentiable joint, conditional, and marginal densities $\mathsf{f}_{u, r}(y, x) = \mathsf{f}_{u|r}(y|x) \mathsf{f}_{r}(x) = \mathsf{f}_{r|u}(x|y) \mathsf{f}_{u}(y)$ which are positive on $y \in \mathbb{R}$, $x \in \mathbb{R}^{d_{x}}$. Assume $d_{x} \le d_{z}$ and $\operatorname{rank} \Pi = d_{x}$. For $0 \le \kappa < \eta \le 1/4$, choose an integer $s \ge 2$ such that $2^{s - 1} > 1 + (1/4 - \eta)(1 + d_{x})$. Let $e = 1 + 2^{s + 1}$. Suppose \\
$(i)$ the marginal density $\mathsf{f}_{u}(y)$ is symmetric and satisfies for $y \in \mathbb{R}$ \\
\indent $(a)$ tail monotonicity: $y^{e} \mathsf{f}_{u}(y), |y^{e + 1} \dot{\mathsf{f}}_{u}(y)|$ are decreasing for large $y$; \\
$(ii)$ the marginal density $\mathsf{f}_{r}(x)$ satisfies for $x \in \mathbb{R}^{d_{x}}$ \\
\indent $(a)$ moments: $\int_{x \in \mathbb{R}^{d_{x}}} |x|^{4} \mathsf{f}_{r}(x) (dx) < \infty$; \\
$(iii)$ the joint and conditional densities $\mathsf{f}_{u, r}(y, x), \mathsf{f}_{u|r}(y|x)$ satisfy for $y \in \mathbb{R}, x \in \mathbb{R}^{d_{x}}$ \\
\indent $(a)$ boundedness: $\sup_{y \in \mathbb{R}, x \in \mathbb{R}^{d_{x}}} |(1 + y) y^{e - 2} \mathsf{f}_{u|r}(y|x) + y^{e - 1} \dot{\mathsf{f}}_{u|r, y}(y|x)| < \infty$; \\
$(iv)$ the instruments $z_{i}$ satisfy \\
\indent $(a)$ $M_{z z, n} = \sum_{i = 1}^{n} z_{in} z_{in}^{\prime} \overset{\mathsf{P}}{\to} M_{z z} > 0$; \\
\indent $(b)$ $\max_{1 \le i \le n} |n^{1/2 - \kappa} z_{in}| = \mathrm{O}_{\mathsf{P}}(1)$; \\
\indent $(c)$ $n^{-1} \mathsf{E} \sum_{i = 1}^{n} |n^{1/2} z_{in}|^{e} = \mathsf{E} |n^{1/2} z_{in}|^{e} = \mathrm{O}(1)$; \\
$(v)$ the initial estimators $(\widetilde{\beta}, \widetilde{\sigma}^{2})$ satisfy \\
\indent $(a)$ $n^{1/2} (\widetilde{\beta} - \beta) = \mathrm{O}_{\mathsf{P}}(n^{1/4 - \eta})$; \\
\indent $(b)$ $n^{1/2} (\widetilde{\sigma}^{2} - \sigma^{2}) = \mathrm{O}_{\mathsf{P}}(n^{1/4 - \eta})$.
\end{assumption}

The conditions $(ia, iia, iiia)$ are satisfied in a range of situations. In particular $(ia)$ is satisfied by the normal and t distribution; see Johansen and Nielsen (2016a, Example 3.1), while $(iia, iiia)$ holds for the normal distribution. Condition $(iv)$ is standard for stationary instruments, see Johansen and Nielsen (2016a, Example 3.2). Condition $(v)$ allows the standardised estimation errors to diverge at a rate of $n^{1/4 - \eta}$ rather than being bounded in probability. In particular, $\eta = 1/4$ can be chosen for estimators with standard convergence rates.

Notice $\kappa$ is not present in the moment condition $2^{s - 1} > 1 + (1/4 - \eta)(1 + d_{x})$, thus there is a trade-off between the convergence rate $\eta$ of initial estimators and the required number $s$ of moments for the density $\mathsf{f}_{u, r}$. In standard situations the normalized estimator is bounded so $\eta = 1/4$, then the moment condition to $s$ reduces to $s = 2$. Otherwise when the initial estimator diverges at the rate $\eta$, then the number of moments $s$ grows linearly with the dimension of the regressors $d_{x}$\footnote{also see discussion in Berenguer-Rico, Johansen, and Nielsen (2019a, Remark 3.1, 3.2).}.

\subsection{Asymptotics of Algorithm \ref{iterated version of robust 2sls}} 
This subsection shows asymptotic theory for Algorithms \ref{iterated version of robust 2sls}. Before establishing asymptotic properties of the iterated estimators for $\beta$, $\sigma^{2}$, we first demonstrate that the updated estimator for $\Pi$ is consistent uniformly in $c \in [c_{+}, \infty)$ for a finite number $c_{+} > 0$, given tightness of previous estimators for $\beta$, $\sigma^{2}$.

\begin{theorem} \label{consistency of location estimator in the first stage}
Consider Algorithm \ref{iterated version of robust 2sls}. Suppose Assumption \ref{sufficient assumptions}$(ia, iia, iiia, iv)$ holds, and that $n^{1/2} (\widehat{\beta}_{c}^{(m)} - \beta)$, $n^{1/2} (\widehat{\sigma}_{c}^{(m)} - \sigma$) are $\mathrm{O}_{\mathsf{P}}(1)$ for any $m \in [0, \infty)$. Then as $n \to \infty$
\begin{equation*}
\sup_{c_{+} \le c < \infty} |\widehat{\Pi}_{c}^{(m + 1)} - \Pi| = \mathrm{o}_{\mathsf{P}}(1).
\end{equation*}
\end{theorem}

The proof of the above theorem uses empirical process theory studied in \S \ref{weighted and marked empirical process}. Then with uniform consistency of location estimator for $\Pi$ in the first stage regression (\ref{first stage regression}) we build a one-step stochastic expansion of the updated estimators for structural parameters $\beta$, $\sigma^{2}$ in (\ref{structural equation}) in terms of its original estimators, kernels, and small remainder terms.

\begin{theorem} \label{1-step stochastic expansion allows varying cut-off c for the iterated 2sls}
Consider Algorithm \ref{iterated version of robust 2sls}. Suppose Assumption \ref{sufficient assumptions}$(ia, iia, iiia, iv)$ holds, and that $n^{1/2} (\widehat{\beta}_{c}^{(m)} - \beta)$, $n^{1/2} (\widehat{\sigma}_{c}^{(m)} - \sigma$) are $\mathrm{O}_{\mathsf{P}}(1)$ for any $m \in [0, \infty)$. Then as $n \to \infty$ and uniformly in $c \in [c_{+}, \infty)$
\begin{align*}
n^{1/2} (\widehat{\beta}_{c}^{(m + 1)} - \beta) & = \frac{2c \mathsf{f}_{u}(c)}{\psi} n^{1/2} (\widehat{\beta}_{c}^{(m)} - \beta) + (M_{\tilde{x} \tilde{x}, n} \psi)^{-1} \sum_{i = 1}^{n} \tilde{x}_{in} u_{i} 1_{(|u_{i}| \le \sigma c)} + \mathrm{o}_{\mathsf{P}}(1), \\
n^{1/2} (\widehat{\sigma}_{c}^{(m + 1)} - \sigma) & = \frac{c (c^{2} - \varsigma_{c}^{2}) \mathsf{f}_{u}(c)}{\tau_{2}^{c}} n^{1/2} (\widehat{\sigma}_{c}^{(m)} - \sigma) + \frac{\sigma}{2 \tau_{2}^{c}} n^{-1/2}  \sum_{i = 1}^{n} (\frac{u_{i}^{2}}{\sigma^{2}} - \varsigma_{c}^{2}) 1_{(|u_{i}| \le \sigma c)} \\
& + \frac{\mathsf{f}_{u}(c)}{\tau_{2}^{c}} ( \frac{c^{2} - \varsigma_{c}^{2}}{2} \zeta_{c}^{-} - \frac{2 c}{\psi} \omega_{c} )^{\prime} \Sigma^{1/2} n^{1/2} (\widehat{\beta}_{c}^{(m)} - \beta) \\
& - \frac{1}{\psi \tau_{2}^{c}} \omega_{c}^{\prime} \Sigma^{1/2} M_{\tilde{x} \tilde{x}, n}^{-1} \sum_{i = 1}^{n} \tilde{x}_{in} u_{i} 1_{(|u_{i}| \le \sigma c)} + \mathrm{o}_{\mathsf{P}}(1).
\end{align*}
\end{theorem}

\begin{remark} \label{compare to the iterated 1-step Huber-skip OLS case}
The above theorem explores expansion of Algorithm \ref{iterated version of robust 2sls} (Iterated 2SLS) in the IVs setting where $\mathsf{E} x_{i} u_{i} \neq 0_{d_{x}}$, which is the generalized result of iterated 1-step Huber-skip M-estimators in the classical regression; see Theorem 1 in Jiao and Nielsen (2017). To be specific, the expansion of $\beta$ estimator only depends on its own estimation error in the previous step; see Proof of Theorem \ref{1-step stochastic expansion allows varying cut-off c for the iterated 2sls} (Appendix \ref{results for iterated 2sls}) why $n^{1/2} (\widehat{\sigma}_{c}^{(m)} - \sigma)$ does not occur, and moreover if $\tilde{x}_{in}$ is replaced by $x_{in}$ then it is the same as that in the classical regression. While the expansion of $\sigma^{2}$ estimator does also depend on the estimation error $n^{1/2} (\widehat{\beta}_{c}^{(m)} - \beta)$, this term disappears in the classical regression where $\mathsf{E} x_{i} u_{i} = 0_{d_{x}}$ and $\mathsf{E} r_{i} u_{i} = 0_{d_{x}}$ such that $\xi_{c} = 0_{d_{x}}$, $\zeta_{c}^{+} = \zeta_{c}^{-} = 0_{d_{x}}$, and $\omega_{c} = 0_{d_{x}}$. Thus, the second expansion degenerates to that in Jiao and Nielsen (2017).
\end{remark}

\begin{remark} \label{coefficients in the 1-step stochastic expansion of iterated estimator}
In the one-step expansions, the autoregressive coefficients are $2c \mathsf{f}_{u}(c) / \psi$, $c (c^{2} - \varsigma_{c}^{2}) \mathsf{f}_{u}(c) / \tau_{2}^{c}$. By Johansen and Nielsen (2013, Theorem 3.6), also see Jiao and Nielsen (2017, Theorem 2), Assumption \ref{sufficient assumptions}$(ia)$ implies that they are strictly bounded by one such that
\begin{equation} \label{stationary condition for autoregressive coefficients}
\sup_{c_{+} \le c < \infty} \max \{ |\frac{2c \mathsf{f}_{u}(c)}{\psi}|, |\frac{c (c^{2} - \varsigma_{c}^{2}) \mathsf{f}_{u}(c)}{\tau_{2}^{c}}| \} < 1.
\end{equation}
This in fact shows the spectral radius of the autoregressive coefficient matrix is smaller than one in the iterative system in Theorem \ref{1-step stochastic expansion allows varying cut-off c for the iterated 2sls}, which is significant to establish the tightness and fixed point result, see the unit cycle boundedness condition (\ref{spectral radius of autoregressive coefficient is bounded by 1}) in the one-step stochastic expansion system (\ref{autoregressive equation for theta with 0 remainder term allows the varying cut-off c}), (\ref{estimation error theta and autoregressive coefficient with the same c}), (\ref{kernel with the same c}) in Proof of Theorem \ref{tightness allows varying cut-off c for the iterated 2sls} (Appendix \ref{results for iterated 2sls}).
\end{remark}

Assumption \ref{sufficient assumptions}$(v)$ with $\eta = 1/4$ corresponds to a standard convergence rate for the initial estimator. Theorem \ref{1-step stochastic expansion allows varying cut-off c for the iterated 2sls} provides an iterative equation between the updated and original estimators, while its autoregressive coefficient has a spectral radius strictly bounded by the unit cycle, see (\ref{stationary condition for autoregressive coefficients}) in Remark \ref{coefficients in the 1-step stochastic expansion of iterated estimator}. Thus a geometric argument and mathematical induction are then used to show $\widehat{\beta}_{c}^{(m)}$, $\widehat{\sigma}_{c}^{(m)}$ are tight in iteration $m \in [0, \infty)$ and in the cut-off value $c \in [c_{+}, \infty)$.

\begin{theorem} \label{tightness allows varying cut-off c for the iterated 2sls}
Consider Algorithm \ref{iterated version of robust 2sls}. Suppose Assumption \ref{sufficient assumptions} holds with $\eta = 1/4$. Then as $n \to \infty$
\begin{equation*}
\sup_{0 \le m < \infty} \sup_{c_{+} \le c < \infty} |n^{1/2} (\widehat{\beta}_{c}^{(m)} - \beta)| + |n^{1/2} (\widehat{\sigma}_{c}^{(m)} - \sigma)| = \mathrm{O}_{\mathsf{P}}(1).
\end{equation*}
\end{theorem}

The above tightness result is required for building the weak convergence theory of iterated estimators for $\beta$, $\sigma^{2}$. Further combined with Theorem \ref{consistency of location estimator in the first stage} a corollary follows, to demonstrate uniform consistency of $\widehat{\Pi}_{c}^{(m)}$ in $m$ and $c$.

\begin{corollary} \label{consistency of iterated location estimator in the first stage}
Consider Algorithm \ref{iterated version of robust 2sls}. Suppose Assumption \ref{sufficient assumptions} holds with $\eta = 1/4$. Then as $n \to \infty$
\begin{equation*}
\sup_{0 \le m < \infty} \sup_{c_{+} \le c < \infty} |\widehat{\Pi}_{c}^{(m)} - \Pi| = \mathrm{o}_{\mathsf{P}}(1).
\end{equation*}
\end{corollary}

Since tightness results for iterated estimators of $\beta$, $\sigma^{2}$ have been established by Theorem \ref{tightness allows varying cut-off c for the iterated 2sls}, we can apply the one-step expansion in Theorem \ref{1-step stochastic expansion allows varying cut-off c for the iterated 2sls} recursively. Then for any $m \in [0, \infty)$ the stochastic expansions of $m + 1$ step estimators are explored in terms of the initial estimators, kernels, and small remainder terms.

\begin{theorem} \label{stochastic expansion of the iterated 2sls in terms of initial estimators allows varying cut-off c}
Consider Algorithm \ref{iterated version of robust 2sls}. Suppose Assumption \ref{sufficient assumptions} holds with $\eta = 1/4$. Then as $n \to \infty$ and uniformly in $c \in [c_{+}, \infty)$ we have for any $m \in [0, \infty)$
\begin{align*}
n^{1/2} (\widehat{\beta}_{c}^{(m + 1)} - \beta) & = \varrho_{\beta \beta, c}^{(m + 1)} n^{1/2} (\widehat{\beta}_{c}^{(0)} - \beta) + \varrho_{\beta \tilde{x} u, c}^{(m + 1)} M_{\tilde{x} \tilde{x}, n}^{-1} \sum_{i = 1}^{n} \tilde{x}_{in} u_{i} 1_{(|u_{i}| \le \sigma c)} + \mathrm{o}_{\mathsf{P}}(1), \\
n^{1/2}(\widehat{\sigma}_{c}^{(m+1)}-\sigma) & = \varrho_{\sigma\sigma,c}^{(m+1)} n^{1/2}(\widehat{\sigma}_{c}^{(0)}-\sigma) + 
\frac{\sigma}{2} \varrho_{\sigma uu,c}^{(m+1)} n^{-1/2} \sum_{i=1}^{n} ( \frac{u_i^2}{\sigma^2}-\varsigma_c^2 ) 1_{(|u_i|\leq\sigma c)}
\\
& + \frac{\mathsf f_u(c)\varrho_{\sigma\beta,c}^{(m+1)}} {\tau_2^c} ( \frac{c^2-\varsigma_c^2}{2}\zeta_c^- - \frac{2c}{\psi}\omega_c )^{\prime} \Sigma^{1/2} n^{1/2}(\widehat{\beta}_{c}^{(0)}-\beta)
\\
& + \{
\frac{ \varrho_{\sigma\tilde{x}u,c}^{(m+1)} (c^2-\varsigma_c^2)}{2} \zeta_c^- - \frac{2c\varrho_{\sigma\tilde{x}u,c}^{(m+1)} +\varrho_{\sigma uu,c}^{(m+1)}}{\psi} \omega_c \}^{\prime}
\\
&\qquad\qquad\times
\Sigma^{1/2}
M_{\tilde{x}\tilde{x},n}^{-1}
\sum_{i=1}^{n}
\tilde{x}_{in}u_i
1_{(|u_i|\leq\sigma c)}
+
\mathrm{o}_{\mathsf P}(1),
\end{align*}
where coefficients have expressions
\begin{align*}
\varrho_{\beta \beta, c}^{(m + 1)} & = \{ \frac{2c \mathsf{f}_{u}(c)}{\psi} \}^{m + 1}, \qquad \varrho_{\beta \tilde{x} u, c}^{(m + 1)} = \frac{\psi^{m + 1} - \{ 2 c \mathsf{f}_{u}(c) \}^{m + 1}}{\psi^{m + 1} \{ \psi - 2 c \mathsf{f}_{u}(c) \}}, \\
\varrho_{\sigma \sigma, c}^{(m + 1)} & = \{ \frac{c (c^{2} - \varsigma_{c}^{2}) \mathsf{f}_{u}(c)}{\tau_{2}^{c}} \}^{m + 1}, \qquad \varrho_{\sigma u u, c}^{(m + 1)} = \frac{(\tau_{2}^{c})^{m + 1} - \{ c (c^{2} - \varsigma_{c}^{2}) \mathsf{f}_{u}(c) \}^{m + 1}}{ (\tau_{2}^{c})^{m + 1} \{ \tau_{2}^{c} - c (c^{2} - \varsigma_{c}^{2}) \mathsf{f}_{u}(c) \}}, \\
\varrho_{\sigma \beta, c}^{(m+1)} &= \sum_{l=0}^{m}
\{
\frac{2c\mathsf f_u(c)}{\psi}
\}^{m-l}
\{
\frac{
c (c^2-\varsigma_c^2) \mathsf{f}_u(c)
}{
\tau_2^c
}
\}^{l}, \\
\varrho_{\sigma\tilde{x}u,c}^{(m+1)}
&=
\frac{\mathsf f_u(c)}
{
\tau_2^c\{\psi-2c\mathsf f_u(c)\}
}
\Bigg[
\frac{
(\tau_2^c)^{m+1}
-
\{c (c^2-\varsigma_c^2) \mathsf f_u(c)\}^{m+1}
}{
(\tau_2^c)^m
\{\tau_2^c-c (c^2-\varsigma_c^2) \mathsf f_u(c)\}
}
\\
&\qquad\qquad
-
\sum_{l=0}^{m}
\{
\frac{2c\mathsf f_u(c)}{\psi}
\}^{m-l}
\{
\frac{
c (c^2-\varsigma_c^2) \mathsf f_u(c)
}{
\tau_2^c
}
\}^{l}
\Bigg].
\end{align*}
\end{theorem}

Initially, a tight estimator is assumed to be available. It is then iterated using the one-step stochastic expansion presented in Theorem \ref{1-step stochastic expansion allows varying cut-off c for the iterated 2sls}. The spectral radius of the corresponding autoregressive coefficient matrix is uniformly bounded away from one, as shown in (\ref{stationary condition for autoregressive coefficients}) of Remark \ref{coefficients in the 1-step stochastic expansion of iterated estimator}. Consequently, as the number of iterations becomes sufficiently large, the effects of the initial estimation errors vanish and the iterated estimator converges in probability to a fixed point. To characterize this fixed point, let \(m\to\infty\) in Theorem \ref{stochastic expansion of the iterated 2sls in terms of initial estimators allows varying cut-off c}. This gives
\begin{align*}
\varrho_{\beta\beta,c}^{(\infty)}
&=0,
&
\varrho_{\beta\tilde{x}u,c}^{(\infty)}
&=
\frac{1}{\psi-2c\mathsf f_u(c)},
\\
\varrho_{\sigma\sigma,c}^{(\infty)}
&=0,
&
\varrho_{\sigma uu,c}^{(\infty)}
&=
\frac{1}{
\tau_2^c-c (c^2-\varsigma_c^2) \mathsf f_u(c)
},
\nonumber
\\
\varrho_{\sigma\beta,c}^{(\infty)}
&=0,
&
\varrho_{\sigma\tilde{x}u,c}^{(\infty)}
&=
\frac{\mathsf f_u(c)}
{
\{\psi-2c\mathsf f_u(c)\}
\{\tau_2^c-c (c^2-\varsigma_c^2) \mathsf f_u(c)\}
},
\nonumber
\end{align*}
such that the fixed-point follows
\begin{equation*}
n^{1/2}(\widehat{\beta}_c^{(\ast)}-\beta) = n^{1/2}(\widehat{\beta}_c^{(\infty)}-\beta), \quad
n^{1/2}(\widehat{\sigma}_c^{(\ast)}-\sigma) = n^{1/2}(\widehat{\sigma}_c^{(\infty)}-\sigma).
\end{equation*}

\begin{theorem} \label{fixed point allows varying cut-off c for the iterated 2sls}
Consider Algorithm \ref{iterated version of robust 2sls}. Suppose Assumption \ref{sufficient assumptions} holds with $\eta = 1/4$. Then for all $\epsilon, \delta > 0$ a pair $n_{0} > 0, m_{0} > 0$ exists, so for $n > n_{0}$ and $m > m_{0}$
\begin{equation*}
\mathsf{P} \{ \sup_{c_{+} \le c < \infty} | n^{1/2} (\widehat{\beta}_{c}^{(m)} - \widehat{\beta}_{c}^{(\ast)}) | + | n^{1/2} (\widehat{\sigma}_{c}^{(m)} - \widehat{\sigma}_{c}^{(\ast)}) | > \delta \} < \epsilon,
\end{equation*}
where
\begin{align*}
n^{1/2} (\widehat{\beta}_{c}^{(\ast)} - \beta) & = \frac{1}{\psi - 2c \mathsf{f}_{u}(c)} M_{\tilde{x} \tilde{x}, n}^{-1} \sum_{i = 1}^{n} \tilde{x}_{in} u_{i} 1_{(|u_{i}| \le \sigma c)}, \\
n^{1/2} (\widehat{\sigma}_{c}^{(\ast)} - \sigma) & = \frac{\sigma}{2 \{ \tau_{2}^{c} - c(c^{2} - \varsigma_{c}^{2})\mathsf{f}_{u}(c) \}} n^{-1/2}  \sum_{i = 1}^{n} (\frac{u_{i}^{2}}{\sigma^{2}} - \varsigma_{c}^{2}) 1_{(|u_{i}| \le \sigma c)} \\
& + \frac{\{\frac{(c^{2} - \varsigma_{c}^{2}) \mathsf{f}_{u}(c)}{2} \zeta_{c}^{-} - \omega_{c} \}^{\prime} }{\{ \psi - 2 c \mathsf{f}_{u}(c) \} \{ \tau_{2}^{c} - c (c^{2} - \varsigma_{c}^{2}) \mathsf{f}_{u}(c) \}} \Sigma^{1/2} M_{\tilde{x} \tilde{x}, n}^{-1} \sum_{i = 1}^{n} \tilde{x}_{in} u_{i} 1_{(|u_{i}| \le \sigma c)}.
\end{align*}
\end{theorem}

Based on Theorem \ref{tightness allows varying cut-off c for the iterated 2sls}, if the initial estimator is bounded in a large compact set with large probability, then any iterated estimator takes values in the same compact set. The proof of Theorem \ref{fixed point allows varying cut-off c for the iterated 2sls} is to further argue that the deviation between the $m$-fold iterated estimator and the fixed point is the sum of two terms vanishing exponentially and in probability respectively when $m$ and $n$ tend to infinity.

Algorithm \ref{iterated version of robust 2sls} mimics the Huber (1964) skip estimator in the IVs context. By investigating Theorem \ref{fixed point allows varying cut-off c for the iterated 2sls} we reason that through infinite iterations the fixed point of the algorithm approximates the Huber-skip IVs estimator in the sense that they have the same asymptotic expansion in terms of kernels and thus follow the same limiting distribution (see Theorem \ref{weak convergence for the fixed point of beta estimator}).

Similar to Remark \ref{compare to the iterated 1-step Huber-skip OLS case}, Theorems \ref{stochastic expansion of the iterated 2sls in terms of initial estimators allows varying cut-off c} and \ref{fixed point allows varying cut-off c for the iterated 2sls} extend the $(m + 1)$-fold expansion and the fixed point result of the iterated 1-step Huber-skip M-estimator in classical regression to the IVs setting where $\mathsf{E} x_{i} u_{i} \neq 0_{d_{x}}$ is allowed\footnote{Replacing $\tilde{x}_{in}$ by $x_{in}$ and setting $\mathsf{E} x_{i} u_{i} = 0_{d_{x}}$ such that $\mathsf{E} r_{i} u_{i} = 0_{d_{x}}$, we then have $\xi_{c} = 0_{d_{x}}$, $\zeta_{c}^{+} = \zeta_{c}^{-} = 0_{d_{x}}$, and $\omega_{c} = 0_{d_{x}}$. Thus all results degenerate to those shown in Jiao and Nielsen (2017).}. Under normality of $(u_{i}/\sigma, \Sigma^{-1/2} r_{i})$, $\xi_{c} = \Omega c$, $\zeta_{c}^{+} = 0_{d_{x}}$, $\zeta_{c}^{-} = 2 \Omega c$, and $\omega_{c} = \tau_{2}^{c} \Omega$, so coefficients appearing in expansions in Theorems \ref{stochastic expansion of the iterated 2sls in terms of initial estimators allows varying cut-off c} and \ref{fixed point allows varying cut-off c for the iterated 2sls} can be further simplified.

\subsection{Weak convergence of Algorithm \ref{robustified two stage least squares} and \ref{split sample version}}
Algorithms \ref{robustified two stage least squares} and \ref{split sample version} are special versions of Algorithm \ref{iterated version of robust 2sls} with different starting points. Their initial estimators are either the full sample or split sample two stage least squares which do not depend on the cut-off, and so satisfy the tightness property. Therefore, theorems and corollaries in the previous subsection apply for these two algorithms as well. Moreover, since Algorithms \ref{robustified two stage least squares} and \ref{split sample version} start with two stage least squares estimators whose statistical properties are well known, the asymptotic distribution can then be established for these two robust procedures.

When researchers carry out empirical work using instrumental variables regression, their interest mainly lies in making inference on $\beta$ whereas they are only concerned about the consistency of estimators for $\sigma^{2}$, $\Pi$. Theorem \ref{tightness allows varying cut-off c for the iterated 2sls} and Corollary \ref{consistency of iterated location estimator in the first stage} have already shown that iterated estimators of $\sigma^{2}$, $\Pi$ are consistent, so in order to perform inference on structural parameters in (\ref{structural equation}) we next build distributional theory for the iterated estimator of $\beta$ in Algorithms \ref{robustified two stage least squares} and \ref{split sample version}.

Choosing the distinct cut-off value $c$ in the interval $[c_{+}, \infty)$, we then get a process $\mathbb{G}_{n}^{(m + 1)}(c) = n^{1/2} (\widehat{\beta}_{c}^{(m + 1)} - \beta)$ for any $m \in [0, \infty)$. A weak convergence theory for $\mathbb{G}_{n}^{(m + 1)}$ follows from a finite dimensional convergence derived by the expansion in Theorem \ref{stochastic expansion of the iterated 2sls in terms of initial estimators allows varying cut-off c} and tightness in Theorem \ref{tightness allows varying cut-off c for the iterated 2sls} (see Billingsley, 1968). The iterated estimator of $\beta$ as a process is asymptotically approximated by the Gaussian process\footnote{see the definition of Gaussian processes in Adler and Taylor (2009, p.\ 27).}.

We first analyze Algorithm \ref{robustified two stage least squares} by providing asymptotics for the full sample two stage least squares estimator $\widetilde{\beta}$ defined through its initial estimate $\widehat{\beta}_{c}^{(0)}$ in (\ref{r2sls initial 2sls estimator}).

\begin{lemma} \label{expansion of 2sls and its limiting distribution}
Consider the full sample two stage least squares
\begin{equation*}
\widetilde{\beta} = (\widetilde{\Pi}^{\prime} \sum_{i = 1}^{n} z_{i} z_{i}^{\prime} \widetilde{\Pi})^{-1} (\widetilde{\Pi}^{\prime} \sum_{i = 1}^{n} z_{i} y_{i}), \quad \textit{where} \quad \widetilde{\Pi} = (\sum_{i = 1}^{n} z_{i} z_{i}^{\prime})^{-1} (\sum_{i = 1}^{n} z_{i} x_{i}^{\prime}).
\end{equation*}
Suppose Assumption \ref{sufficient assumptions}$(ia, iia, iva, ivc)$ holds. Then as $n \to \infty$
\begin{equation*}
n^{1/2} (\widetilde{\beta} - \beta) = M_{\tilde{x} \tilde{x}, n}^{-1} \sum_{i = 1}^{n} \tilde{x}_{in} u_{i} + \mathrm{o}_{\mathsf{P}}(1).
\end{equation*}
Furthermore, we have
\begin{equation*}
n^{1/2} (\widetilde{\beta} - \beta) \overset{\mathsf{D}}{\to} \mathsf{N}(0_{d_{x}}, \sigma^{2} M_{\tilde{x} \tilde{x}}^{-1}).
\end{equation*}
\end{lemma}

The next step is to establish the weak convergence theory for the process $\mathbb{G}_{n}^{(m + 1)}$ of the iterated $\beta$ estimator where $m \in [0, \infty)$.

\begin{theorem} \label{weak convergence for the m+1 step beta estimator for robustified 2sls}
Consider Algorithm \ref{robustified two stage least squares}. Suppose Assumption \ref{sufficient assumptions}$(ia, iia, iiia, iv)$ holds. Denote the process $\mathbb{G}_{n}^{(m + 1)}(c) = n^{1/2} (\widehat{\beta}_{c}^{(m + 1)} - \beta)$ for $c \in [c_{+}, \infty)$ and $m \in [0, \infty)$. Then as $n \to \infty$ we have
\begin{equation*}
\mathbb{G}_{n}^{(m + 1)}(c) = \varrho_{\beta \beta, c}^{(m + 1)} M_{\tilde{x} \tilde{x}, n}^{-1} \sum_{i = 1}^{n} \tilde{x}_{in} u_{i} + \varrho_{\beta \tilde{x} u, c}^{(m + 1)} M_{\tilde{x} \tilde{x}, n}^{-1} \sum_{i = 1}^{n} \tilde{x}_{in} u_{i} 1_{(|u_{i}| \le \sigma c)} + \mathrm{o}_{\mathsf{P}}(1),
\end{equation*}
where $\varrho_{\beta \beta, c}^{(m + 1)}$, $\varrho_{\beta \tilde{x} u, c}^{(m + 1)}$ are defined in Theorem \ref{stochastic expansion of the iterated 2sls in terms of initial estimators allows varying cut-off c}. Furthermore, $\mathbb{G}_{n}^{(m + 1)}$ weakly converges to a zero mean Gaussian process $\mathbb{G}^{(m + 1)}$ with variance given as
\begin{equation*}
\mathsf{Var} \{ \mathbb{G}^{(m + 1)}(c) \} = \{ (\varrho_{\beta \beta, c}^{(m + 1)})^{2} + 2 \tau_{2}^{c} \varrho_{\beta \beta, c}^{(m + 1)} \varrho_{\beta \tilde{x} u, c}^{(m + 1)} + \tau_{2}^{c} (\varrho_{\beta \tilde{x} u, c}^{(m + 1)})^{2} \} \sigma^{2} M_{\tilde{x} \tilde{x}}^{-1}.
\end{equation*}
\end{theorem}

The one-step updated estimator from the full sample 2SLS is of particular interest in Algorithm \ref{robustified two stage least squares}. To explore its asymptotics let $m = 0$ in Theorem \ref{weak convergence for the m+1 step beta estimator for robustified 2sls} such that $\varrho_{\beta \beta, c}^{(1)} = 2 c \mathsf{f}_{u}(c) / \psi$ and $\varrho_{\beta \tilde{x} u, c}^{(1)} = \psi^{-1}$, then the corollary below follows.

\begin{corollary} \label{weak convergence for the first step beta estimator for robustified 2sls}
Consider Algorithm \ref{robustified two stage least squares}. Suppose Assumption \ref{sufficient assumptions}$(ia, iia, iiia, iv)$ holds. Denote the process $\mathbb{G}_{n}^{(1)}(c) = n^{1/2} (\widehat{\beta}_{c}^{(1)} - \beta)$ for $c \in [c_{+}, \infty)$. Then as $n \to \infty$ we have
\begin{equation*}
\mathbb{G}_{n}^{(1)}(c) = \frac{2 c \mathsf{f}_{u}(c)}{\psi} M_{\tilde{x} \tilde{x}, n}^{-1} \sum_{i = 1}^{n} \tilde{x}_{in} u_{i} + (M_{\tilde{x} \tilde{x}, n} \psi)^{-1} \sum_{i = 1}^{n} \tilde{x}_{in} u_{i} 1_{(|u_{i}| \le \sigma c)} + \mathrm{o}_{\mathsf{P}}(1).
\end{equation*}
Furthermore, $\mathbb{G}_{n}^{(1)}$ weakly converges to a zero mean Gaussian process $\mathbb{G}^{(1)}$ with variance given as
\begin{equation*}
\mathsf{Var} \{ \mathbb{G}^{(1)}(c) \} = \frac{4 c^{2} \mathsf{f}_{u}^{2}(c) + 4 \tau_{2}^{c} c \mathsf{f}_{u}(c) + \tau_{2}^{c}}{\psi^{2}} \sigma^{2} M_{\tilde{x} \tilde{x}}^{-1}.
\end{equation*}
\end{corollary}

Then we analyze Algorithm \ref{split sample version} and find asymptotically it is equivalent to Algorithm \ref{robustified two stage least squares} although they start with distinct initial estimates when implementing Algorithm \ref{iterated version of robust 2sls}.

\begin{theorem} \label{weak convergence for the m+1 step beta estimator for split sample iterated 2sls}
Consider Algorithm \ref{split sample version} where $n_{1} = n_{2} = n / 2$. Suppose Assumption \ref{sufficient assumptions}$(ia, iia, iiia , iv)$ holds for each sub-sample set $\mathcal{I}_{1}$, $\mathcal{I}_{2}$.
Then as $n \to \infty$ and for any $m \in [0, \infty)$, the process of the estimator $n^{1/2} (\widehat{\beta}_{c}^{(m + 1)} - \beta)$ for $c \in [c_{+}, \infty)$ has the same asymptotic expansion as for Algorithm \ref{robustified two stage least squares} and thus weakly converges to the identical Gaussian process reported in Theorem \ref{weak convergence for the m+1 step beta estimator for robustified 2sls}.
\end{theorem}

The proof involves checking whether the expansion of the first step updated estimator for Algorithm \ref{split sample version} is the same as Algorithm \ref{robustified two stage least squares} using arguments in Theorems \ref{consistency of location estimator in the first stage} and \ref{1-step stochastic expansion allows varying cut-off c for the iterated 2sls}. Since the two algorithms are iterated from the split half or full sample 2SLS respectively, they have the same asymptotics if their $\widehat{\beta}_{c}^{(1)}$ perform identically when $n \to \infty$.

Finally, when $n \to \infty$ and $m \to \infty$, the process $\mathbb{G}_{n}^{(m)}$ uniformly converges to the limiting process $\mathbb{G}^{(\infty)} = \mathbb{G}^{(\ast)}$ in probability. Since $\varrho_{\beta \beta, c}^{(\infty)} = 0$, $\varrho_{\beta \tilde{x} u, c}^{(\infty)} = 1 / \{ \psi - 2 c \mathsf{f}_{u}(c) \}$, then Theorems \ref{weak convergence for the m+1 step beta estimator for robustified 2sls} and \ref{weak convergence for the m+1 step beta estimator for split sample iterated 2sls} immediately show the weak convergence of the fixed point process $\mathbb{G}^{(\ast)}$ for Algorithms \ref{iterated version of robust 2sls}, \ref{robustified two stage least squares}, and \ref{split sample version}.

\begin{theorem} \label{weak convergence for the fixed point of beta estimator}
Consider Algorithm \ref{iterated version of robust 2sls} with tight initial estimators, such as Algorithm \ref{robustified two stage least squares}, \ref{split sample version}. Suppose Assumption \ref{sufficient assumptions}$(ia, iia, iiia, iv)$ holds. When $m \to \infty$ let the fixed point process $\mathbb{G}_{n}^{(\ast)}(c) = n^{1/2} (\widehat{\beta}_{c}^{(\ast)} - \beta)$ for $c \in [c_{+}, \infty)$ as in Theorem \ref{fixed point allows varying cut-off c for the iterated 2sls}. Then as $n \to \infty$ the process $\mathbb{G}_{n}^{(\ast)}$ has the asymptotic expansion shown in Theorem \ref{fixed point allows varying cut-off c for the iterated 2sls} and thus weakly converges to a zero mean Gaussian process $\mathbb{G}^{(\ast)}$ with variance given as
\begin{equation*}
\mathsf{Var} \{ \mathbb{G}^{(\ast)}(c) \} = \frac{\tau_{2}^{c}}{\{ \psi - 2 c \mathsf{f}_{u}(c) \}^{2}} \sigma^{2} M_{\tilde{x} \tilde{x}}^{-1}.
\end{equation*}
\end{theorem}

Consider Algorithms \ref{robustified two stage least squares} and \ref{split sample version}. Weak convergence in Theorems \ref{weak convergence for the m+1 step beta estimator for robustified 2sls} and \ref{weak convergence for the m+1 step beta estimator for split sample iterated 2sls} implies pointwise convergence, so for any $c \in [c_{+}, \infty)$, $m \in [0, \infty)$, and as $n \to \infty$ it holds
\begin{equation} \label{asymptotic distribution for general step beta estimator}
n^{1/2} (\widehat{\beta}_{c}^{(m + 1)} - \beta) \overset{\mathsf{D}}{\to} \mathsf{N} [0_{d_{x}}, \{ (\varrho_{\beta \beta, c}^{(m + 1)})^{2} + 2 \tau_{2}^{c} \varrho_{\beta \beta, c}^{(m + 1)} \varrho_{\beta \tilde{x} u, c}^{(m + 1)} + \tau_{2}^{c} (\varrho_{\beta \tilde{x} u, c}^{(m + 1)})^{2} \} \sigma^{2} M_{\tilde{x} \tilde{x}}^{-1}],
\end{equation}
where $\varrho_{\beta \beta, c}^{(m + 1)}$, $\varrho_{\beta \tilde{x} u, c}^{(m + 1)}$ are defined in Theorem \ref{stochastic expansion of the iterated 2sls in terms of initial estimators allows varying cut-off c}.

Johansen and Nielsen (2009, 2013) built up the asymptotic distributions for the one-step Huber-skip M-estimator and its infinite iteration either starting from full sample or split sample least squares in the classical setting where $\mathsf{E} x_{i} u_{i} = 0_{d_{x}}$\footnote{also see Johansen and Nielsen (2016b) for the asymptotic distribution of the $(m + 1)$-step estimator where $m \in [0, \infty)$.}. The distributional results for Algorithms \ref{robustified two stage least squares} and \ref{split sample version} are suprisingly the same as the iterated one-step Huber-skip M-estimators if $\tilde{x}_{in}$ is replaced by $x_{in}$ even though dependence is allowed between regressors and errors in the IVs regression. This is due to the fact that the stochastic expansion for the updated $\beta$ estimator only depends on its own previous step, not on the $\sigma$ estimator (see Theorem \ref{1-step stochastic expansion allows varying cut-off c for the iterated 2sls} and Remark \ref{compare to the iterated 1-step Huber-skip OLS case}). Furthermore, this result is because $\zeta_{c}^{+}$ defined in (\ref{plus and minus conditional mean}, \S \ref{model}) does not appear in the $\beta$ expansion since the data is demeaned such that $\mathsf{E} z_{i} = 0_{d_{z}}$\footnote{see the details in Proof of Theorem \ref{1-step stochastic expansion allows varying cut-off c for the iterated 2sls}, Appendix \ref{results for iterated 2sls}.}. 

\subsection{Efficiency comparison} \label{efficiency comparison}
Outlier analysis is trade-off between efficiency and robustness: Algorithms \ref{robustified two stage least squares} and \ref{split sample version} will lose efficiency under the null of no outliers relative to the non-robust ordinary 2SLS since it is with positive probability to wrongly detect outliers, whereas they are indeed more robust under the alternative when there is data contamination. This subsection mainly investigates efficiency loss of two algorithms in return for gaining robustness.

Distributions of $\widehat{\beta}_{c}^{(m + 1)}$ are given by (\ref{asymptotic distribution for general step beta estimator}) while Lemma \ref{expansion of 2sls and its limiting distribution} shows the limit of the ordinary two stage least squares estimator $\widetilde{\beta}$. Thus, we can now compare the efficiency of the robust estimator $\widehat{\beta}_{c}^{(m + 1)}$ with respect to the non-robust estimator $\widetilde{\beta}$ under the null of no outliers. Relative efficiency is defined by
\begin{align}
\vartheta_{c}^{(m + 1)} & = \mathrm{efficiency} (\widehat{\beta}_{c}^{(m + 1)}, \widetilde{\beta}) = \{ \mathsf{avar} (\widehat{\beta}_{c}^{(m + 1)}) \}^{-1} \{ \mathsf{avar} (\widetilde{\beta}) \} \nonumber \\
& = \{ (\varrho_{\beta \beta, c}^{(m + 1)})^{2} + 2 \tau_{2}^{c} \varrho_{\beta \beta, c}^{(m + 1)} \varrho_{\beta \tilde{x} u, c}^{(m + 1)} + \tau_{2}^{c} (\varrho_{\beta \tilde{x} u, c}^{(m + 1)})^{2} \}^{-1} I_{d_{x}}. \label{efficiency between robust estimators with respect to the non-robust under the null}
\end{align}
In the rest of the paper we treat $\vartheta_{c}^{(m + 1)}$ as a scalar factor though it is a scaled identity matrix. We are interested in efficiency comparison for two special cases: $m = 0$ and $m \to \infty$ such that
\begin{align}
\vartheta_{c}^{(1)} & = \mathrm{efficiency} (\widehat{\beta}_{c}^{(1)}, \widetilde{\beta}) = \frac{\psi^{2}}{4 c^{2} \mathsf{f}_{u}^{2}(c) + 4 \tau_{2}^{c} c \mathsf{f}_{u}(c) + \tau_{2}^{c}} I_{d_{x}}, \label{efficiency between first step estimator with respect to the non-robust under the null} \\
\vartheta_{c}^{(\ast)} & =\mathrm{efficiency} (\widehat{\beta}_{c}^{(\ast)}, \widetilde{\beta}) = \frac{\{ \psi - 2 c \mathsf{f}_{u}(c) \}^{2}}{\tau_{2}^{c}} I_{d_{x}}. \label{efficiency between fixed point estimator with respect to the non-robust under the null}
\end{align}


\begin{figure}[hp!]
    \centering
        \begin{subfigure}[b]{\textwidth}
        \includegraphics[width=\textwidth]{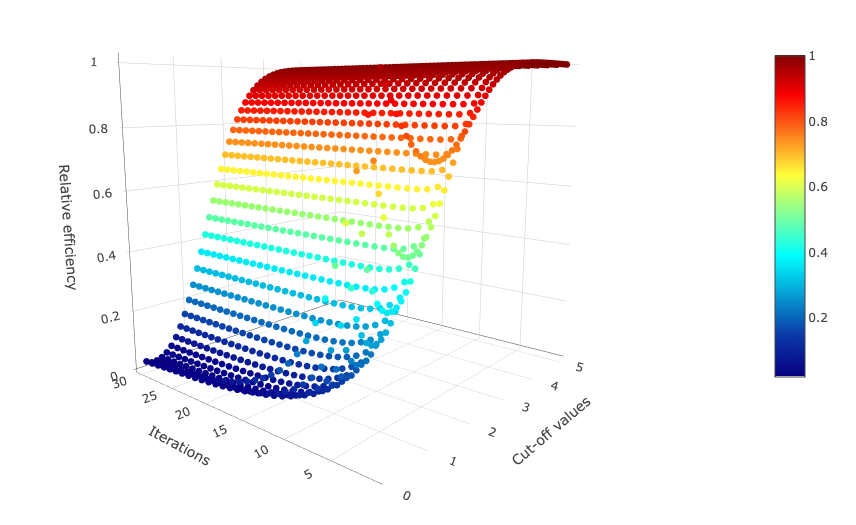}
        \caption{\label{efficiency loss against cut-offs and iterations} Efficiency against cut-offs and iterations}
        \end{subfigure}
    \newline
    \begin{subfigure}[b]{0.45\textwidth}
        \includegraphics[width=\textwidth]{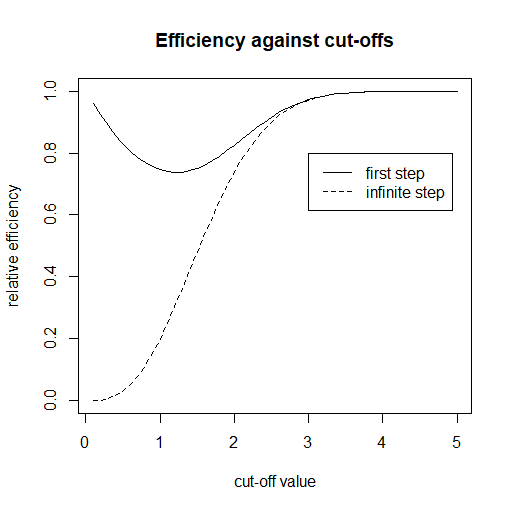}
        \caption{\label{efficiency loss against cut-offs of first and infinite step} Efficiency against cut-offs}
    \end{subfigure}
    \begin{subfigure}[b]{0.45\textwidth}
        \includegraphics[width=\textwidth]{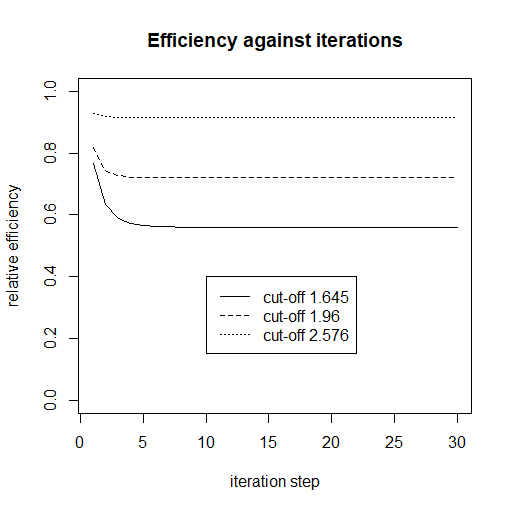}
        \caption{\label{efficiency loss against iterations for three cut-offs} Efficiency against iterations}
    \end{subfigure}
    \caption{The efficiency of $\widehat{\beta}_{c}^{(m + 1)}$ with $c \in [0.1, 5]$ and $m \in [0, 29)$ relative to the two stage least squares $\widetilde{\beta}$ for $\mathsf{f}_{u}$ equal to the standard normal density.}
\end{figure}

Using the standard normal density for $\mathsf{f}_{u}$, we plot the relative efficiency (\ref{efficiency between robust estimators with respect to the non-robust under the null}) against cut-off values $c \in [c_{+}, \infty)$ and iterations $m \in [0, \infty)$ in Figure \ref{efficiency loss against cut-offs and iterations}. To further illustrate the 3D Figure \ref{efficiency loss against cut-offs and iterations}, we also plot two 2D figures, Figure \ref{efficiency loss against cut-offs of first and infinite step}: The efficiency (\ref{efficiency between first step estimator with respect to the non-robust under the null}), (\ref{efficiency between fixed point estimator with respect to the non-robust under the null}) against cut-off values; Figure \ref{efficiency loss against iterations for three cut-offs}: The efficiency (\ref{efficiency between robust estimators with respect to the non-robust under the null}) against iteration steps for the cut-off values $c_{1} = 1.645 (10 \%)$, $c_{2} = 1.96 (5 \%)$, $c_{3} = 2.576 (1 \%)$.



Asymptotics in this paper are derived under the null hypothesis of no outliers where Algorithms \ref{robustified two stage least squares} and \ref{split sample version} have a probability to falsely trim the non-outlying observations given the cut-off $c$. Thus, robust estimators have a higher asymptotic variance relative to the ordinary 2SLS, and so the relative efficiency plots lie in the interval $(0, 1)$ as shown in the three relative efficiency figures.

When the cut-off values $c$ become larger, the chance of wrongly detecting outliers tends to be smaller. Subsequently, the robust estimators have the smaller variance, and the efficiency become larger and closer to one, and vice versa if $c$ decreases. First, observe the efficiency plot of the infinite step estimator $\widehat{\beta}_{c}^{(\ast)}$ in Figure \ref{efficiency loss against cut-offs of first and infinite step}, and find that it is consistent with what we discussed. Then check the solid line of the first step estimator $\widehat{\beta}_{c}^{(1)}$, and find that its efficiency tends to one when $c \to 0$ as well as when $c \to \infty$. It is intuitive that for large $c$ the efficiency increases as $c$ increases. However we also find for small $c$ the efficiency increases as $c$ decreases. This is due to the fact that the updated estimator $\widehat{\beta}_{c}^{(1)}$ is asymptotically equivalent to 2SLS when $c \to 0$ as stated in Remark \ref{updated estimator asymptotically equivalent to the initial as c goes to 0}. This fact is also supported by the efficiency plot \ref{efficiency loss against cut-offs and iterations} for other iterations.

\begin{remark} \label{updated estimator asymptotically equivalent to the initial as c goes to 0}
Let $m = 0$ and $n \to \infty$, Theorem \ref{stochastic expansion of the iterated 2sls in terms of initial estimators allows varying cut-off c} gives the equation
\begin{equation*}
n^{1/2} (\widehat{\beta}_{c}^{(1)} - \beta) = \frac{2 c \mathsf{f}_{u}(c)}{\psi} n^{1/2} (\widetilde{\beta} - \beta) + (M_{\tilde{x} \tilde{x}, n} \psi)^{-1} \sum_{i = 1}^{n} \tilde{x}_{in} u_{i} 1_{(|u_{i}| \le \sigma c)} + \mathrm{o}_{\mathsf{P}}(1),
\end{equation*}
uniformly in $c \in [c_{+}, \infty)$. Note that $\widetilde{\beta}$ is 2SLS so $\widehat{\beta}_{c}^{(1)}$ is its updated estimator. Suppose $\mathsf{f}_{u}$ follows the standard normal then $\psi = 2 \mathsf{F}_{u}(c) - 1$, $\tau_{2}^{c} = \psi - 2 c \mathsf{f}_{u}(c)$, and $\dot{\mathsf{f}}_{u}(c) = - c \mathsf{f}_{u}(c)$. Let $c_{+} > 0$ be sufficiently small so that we can investigate the situation where $c \to 0$. Lemma \ref{limiting distribution of kernels} (Appendix \ref{results for robustified and saturated 2sls}) shows the weak limit of the kernel term such that
\begin{equation*}
(M_{\tilde{x} \tilde{x}, n} \psi)^{-1} \sum_{i = 1}^{n} \tilde{x}_{in} u_{i} 1_{(|u_{i}| \le \sigma c)} \overset{\mathsf{D}}{\to} \mathsf{N}(0_{d_{x}}, \frac{\tau_{2}^{c}}{\psi^{2}} \sigma^{2} M_{\tilde{x} \tilde{x}}^{-1}).
\end{equation*}
By L'H{\^o}pital's rule, we have $\tau_{2}^{c} / \psi^{2} \to 0$ as $c \to 0$ so the kernel term vanishes. Again when $c \to 0$, L'H{\^o}pital's rule gives $2 c \mathsf{f}_{u}(c) / \psi \to 1$, thus $n^{1/2} (\widehat{\beta}_{c}^{(1)} - \beta)$ is asymptotically equivalent to $n^{1/2} (\widetilde{\beta} - \beta)$ indicating an efficiency of unity.

In fact, check the efficiency formula given by (\ref{efficiency between first step estimator with respect to the non-robust under the null}). When $c \to 0$, use $\tau_{2}^{c} / \psi^{2} \to 0$ and $2 c \mathsf{f}_{u}(c) / \psi \to 1$ to get $\psi^{2} / \{ 4 c^{2} \mathsf{f}_{u}^{2}(c) + 4 \tau_{2}^{c} c \mathsf{f}_{u}(c) + \tau_{2}^{c} \} \to 1$, then find $\mathrm{efficiency} (\widehat{\beta}_{c}^{(1)}, \widetilde{\beta}) \to 1$.
\end{remark}

Figure \ref{efficiency loss against cut-offs of first and infinite step} also indicatess the efficiency of the first iteration estimator dominates the infinite step. This fact is consistent with Figure \ref{efficiency loss against iterations for three cut-offs} showing in general that the efficiency decreases as iteration increases regardless of cut-off values $c$. Figure \ref{efficiency loss against iterations for three cut-offs} reveals that the efficiency drops dramatically within the first few steps and converges to the fixed point upon through large iterations, also shown in Figure \ref{efficiency loss against cut-offs and iterations}. Intuition is that as iteration increases the algorithms throw away more observations and thus loss more efficiency until attaining the fixed point. It is not surprising that the critical value $c_{3} = 2.576 (1 \%)$ has the highest relative efficiency and the $c_{1} = 1.645 (10 \%)$ has the lowest while $c_{2} = 1.96 (5 \%)$ is in between.

Figure \ref{efficiency loss against cut-offs and iterations} (\ref{efficiency loss against cut-offs of first and infinite step} and \ref{efficiency loss against iterations for three cut-offs}) imply that Algorithms \ref{robustified two stage least squares} and \ref{split sample version} do not lose too much efficiency compared to the ordinary 2SLS if the cut-off value $c$ and iteration step $m$ are appropriately chosen under the null of no outliers, whereas these two procedures can produce a more robust estimate under alternatives where the data is contaminated.

\subsection{Adjusted inference on structural parameters $\beta$} \label{valid inference on structural parameters beta}
Assume that the density $\mathsf{f}_{u}$ is known, say standard normal $\mathsf{N}(0, 1)$, and choose the cut-off value $c$ and iteration step $m$, say $c = 1.96 (5 \%)$ and $m = 0$ or $\infty$, when carrying out Algorithms \ref{robustified two stage least squares} and \ref{split sample version}. The standard method of conducting inference based on ordinary 2SLS asymptotics\footnote{see ordinary 2SLS asymptotics in Lemma \ref{expansion of 2sls and its limiting distribution}.} are not reliable for three reasons. Firstly, variance $\sigma^{2}$ cannot be estimated consistently without bias correction factor $\varsigma_{c}^{2}$ introduced in (\ref{bias correction factor}, \S \ref{model})
Secondly, the standard method fails to take into account of the relative efficiency factor $\vartheta_{c}^{(m + 1)}$ appearing in the asymptotic variance of $\widehat{\beta}_{c}^{(m + 1)}$, see (\ref{asymptotic distribution for general step beta estimator}) and (\ref{efficiency between robust estimators with respect to the non-robust under the null}). Thirdly, to obtain the standard error of $\widehat{\beta}_{c}^{(m + 1)}$ the usual inferential procedure incorrectly divides the estimated asymptotic variance by the sub-sample size $\sum_{i = 1}^{n} v_{i, c}^{(m)}$ of all non-outlying observations retained, while the full sample size $n$ should have been used.

On implementing Algorithms \ref{robustified two stage least squares} and \ref{split sample version} to perform valid inference on the structural parameters $\beta$ given the cut-off value $c$ and iteration step $m$, we have to assume the explicit form of $\mathsf{f}_{u}$ in order to obtain the bias correction factor $\varsigma_{c}^{2}$ and relative efficiency factor $\vartheta_{c}^{(m + 1)}$. In practice, it is statistically difficult to estimate the density $\mathsf{f}_{u}$. This is because, without knowing $\varsigma_{c}^{2}$ as $\mathsf{f}_{u}$ is unknown, the standard estimate of $\sigma^{2}$ will be biased downward, which will subsequently produce upward biased standardised residuals for structural errors $u_{i}$ although residuals can be estimated consistently based on beta estimates. Then the estimated density of $\mathsf{f}_{u}$ using the biased standardised residuals will be inconsistent by whatever methods. Thus, assumption on the known density $\mathsf{f}_{u}$ is not avoidable to conduct valid inference on $\beta$ in this paper. Proving validity of bootsrap Kaji (2018) applied nonparametric bootstrap to obtain standard errors of beta estimates when conducting inference, which however does not require the assumed known form of the density $\mathsf{f}_{u}$.

Weak limits of beta estimates in the step $m + 1$ are provided by (\ref{asymptotic distribution for general step beta estimator}) such that
\begin{equation*}
n^{1/2} (\widehat{\beta}_{c}^{(m + 1)} - \beta) \overset{\mathsf{D}}{\to} \mathsf{N} \{ 0_{d_{x}}, (\vartheta_{c}^{(m + 1)})^{-1} \sigma^{2} M_{\tilde{x} \tilde{x}}^{-1} \},
\end{equation*}
where $\vartheta_{c}^{(m + 1)}$ is defined by (\ref{efficiency between robust estimators with respect to the non-robust under the null}) and the cut-off $c$ is given. By assuming on the form of $\mathsf{f}_{u}$, $\vartheta_{c}^{(m + 1)}$ is known as well as $\varsigma_{c}^{2}$ so $\sigma^{2}$ can be estimated consistently by $(\widehat{\sigma}_{c}^{(m + 1)})^{2}$, see Theorem \ref{tightness allows varying cut-off c for the iterated 2sls}. Corollary \ref{consistency of iterated location estimator in the first stage} implies we can consistently estimate $M_{\tilde{x} \tilde{x}} = \Pi^{\prime} \mathsf{E} z_{i} z_{i}^{\prime} \Pi$ by
\begin{equation} \label{robust estimation of moment M}
\widehat{M}_{\tilde{x} \tilde{x}, c}^{(m + 1)} = (\sum_{i = 1}^{n} v_{i, c}^{(m)})^{-1} (\widehat{\Pi}_{c}^{(m + 1) \prime} \sum_{i = 1}^{n} z_{i} z_{i}^{\prime} v_{i, c}^{(m)} \widehat{\Pi}_{c}^{(m + 1)}).
\end{equation}
Notice that $\sigma^{2}$ and $M_{\tilde{x} \tilde{x}}$ above are estimated using the robust estimator with the sub-sample after trimming all outliers detected, since the estimator of standard errors should be robust itself, otherwise the inference on $\beta$ might lead to incorrect results under the alternative hypothesis where data is contaminated. Then, the beta estimates are asymptotically approximated by
\begin{equation*}
n^{1/2} (\widehat{\beta}_{c}^{(m + 1)} - \beta) \overset{a}{\sim} \mathsf{N} \{ 0_{d_{x}}, (\vartheta_{c}^{(m + 1)})^{-1} (\widehat{\sigma}_{c}^{(m + 1)})^{2} (\widehat{M}_{\tilde{x} \tilde{x}, c}^{(m + 1)})^{-1} \}.
\end{equation*}
For conducting valid inference, the above estimated asymptotic variance can then be further divided by $n$ to obtain standard errors so that
\begin{equation*}
\widehat{\beta}_{c}^{(m + 1)} \overset{a}{\sim} \mathsf{N} \{ \beta, n^{-1} (\vartheta_{c}^{(m + 1)})^{-1} (\widehat{\sigma}_{c}^{(m + 1)})^{2} (\widehat{M}_{\tilde{x} \tilde{x}, c}^{(m + 1)})^{-1} \}.
\end{equation*}

The standard inferential procedure fails to take into consideration of $(\vartheta_{c}^{(m + 1)})^{-1}$, so mistakenly applies usual 2SLS asymptotics to get the asymptotic variance of $\widehat{\beta}_{c}^{(m + 1)}$ by $\mathrm{avar}_{\mathrm{2SLS}} = \sigma^{2} M_{\tilde{x} \tilde{x}}^{-1}$. Thus, to address this problem the relative efficiency factor $\vartheta_{c}^{(m + 1)}$ should be adjusted to $\mathrm{avar}_{\mathrm{2SLS}}$, then the normal weak limit is rearranged as
\begin{equation*}
n^{1/2} (\widehat{\beta}_{c}^{(m + 1)} - \beta) \overset{\mathsf{D}}{\to} \mathsf{N} \{ 0_{d_{x}}, (\vartheta_{c}^{(m + 1)})^{-1} \mathrm{avar}_{\mathrm{2SLS}} \}.
\end{equation*}

Furthermore, without considering the consistency factor $\varsigma_{c}^{-2}$ by the standard method $\sigma^{2}$ is inconsistently estimated by $\varsigma_{c}^{2} (\widehat{\sigma}_{c}^{(m + 1)})^{2}$. Then asymptotic variance of $\widehat{\beta}_{c}^{(m + 1)}$ is falsely estimated by $\widehat{\mathrm{avar}}_{\mathrm{2SLS}, c}^{(m + 1)} = \varsigma_{c}^{2} (\widehat{\sigma}_{c}^{(m + 1)})^{2} (\widehat{M}_{\tilde{x} \tilde{x}, c}^{(m + 1)})^{-1}$, which should be rectified by $\iota_{c}^{(m + 1)} = (\vartheta_{c}^{(m + 1)} \varsigma_{c}^{2})^{-1}$ so as to fix the inferential problem. Therefore, the correct asymtotic variance is the term $\widehat{\mathrm{avar}}_{\mathrm{2SLS}, c}^{(m + 1)}$ estimated by usual asymptotics but adjusted by $\iota_{c}^{(m + 1)}$, then error-free normal approximation can now be expressed as
\begin{equation*}
n^{1/2} (\widehat{\beta}_{c}^{(m + 1)} - \beta) \overset{a}{\sim} \mathsf{N} ( 0_{d_{x}}, \iota_{c}^{(m + 1)} \widehat{\mathrm{avar}}_{\mathrm{2SLS}, c}^{(m + 1)} ).
\end{equation*}

In addition, the standard method incorrectly estimates the variance of $\widehat{\beta}_{c}^{(m + 1)}$ by $\widehat{\mathrm{var}}_{\mathrm{2SLS}, c}^{(m + 1)}$ through dividing $\widehat{\mathrm{avar}}_{\mathrm{2SLS}, c}^{(m + 1)}$ by the sub-sample size $\sum_{i = 1}^{n} v_{i, c}^{(m)}$ of non-outlying observations rather than the full sample size $n$. Thus, the erroneous standard errors are computed as the square root of $\widehat{\mathrm{var}}_{\mathrm{2SLS}, c}^{(m + 1)} = \varsigma_{c}^{2} (\widehat{\sigma}_{c}^{(m + 1)})^{2} (\widehat{\Pi}_{c}^{(m + 1) \prime} \sum_{i = 1}^{n} z_{i} z_{i}^{\prime} v_{i, c}^{(m)} \widehat{\Pi}_{c}^{(m + 1)})^{-1}$. Therefore, correct standard errors should be based on $\widehat{\mathrm{var}}_{\mathrm{2SLS}, c}^{(m + 1)}$ by usual 2SLS asymptotics, however rectified by $n^{-1} \sum_{i = 1}^{n} v_{i, c}^{(m)}$ and $\iota_{c}^{(m + 1)}$. In summary, the valid inference should rely on the following rearranged normal approximation
\begin{equation*}
\widehat{\beta}_{c}^{(m + 1)} \overset{a}{\sim} \mathsf{N} \{ \beta, (n^{-1} \sum_{i = 1}^{n} v_{i, c}^{(m)}) \iota_{c}^{(m + 1)} \widehat{\mathrm{var}}_{\mathrm{2SLS}, c}^{(m + 1)} \}.
\end{equation*}

Consider the standard normal $\mathsf{N}(0, 1)$ for $\mathsf{f}_{u}$. Choose the cut-off value $c$ as $1.645 (10 \%)$, $1.96 (5 \%)$, $2.576 (1 \%)$ and set the iteration step $m$ as $0$, $\infty$. Notice that $n^{-1} \sum_{i = 1}^{n} v_{i, c}^{(m)}$ is the proportion of non-outlying observations retained, which has the probability limit $\psi$ under the null of no outliers due to Lemma \ref{n asymptotic expansions of empirical processes} and LLN. List in Table \ref{adjusting factors for variance estimates for asymptotic variance for standard errors} the probability $\psi$ within the cut-off, the bias correction factor $\varsigma_{c}^{2}$ for estimating $\sigma^{2}$, the relative efficiency factor $\vartheta_{c}^{(m + 1)}$ for asymptotic variance of $\widehat{\beta}_{c}^{(m + 1)}$, the adjusting factor $\iota_{c}^{(m + 1)}$ for estimating asymptotic variance, and the adjusting factor $\psi \iota_{c}^{(m + 1)}$ for estimating variance of $\widehat{\beta}_{c}^{(m + 1)}$ under the null of no outliers for large samples.

\begin{table}[ht]
\centering
\begin{tabular}{ccccccccc}
\toprule
$c$ & $\psi$ & $\varsigma_{c}^{2}$ & $\vartheta_{c}^{(1)}$ & $\vartheta_{c}^{(\ast)}$ & $\iota_{c}^{(1)}$ & $\iota_{c}^{(\ast)}$ & $\psi \iota_{c}^{(1)}$ & $\psi \iota_{c}^{(\ast)}$ \\
\midrule
$1.645$ & $0.900$ & $0.623$ & $0.767$ & $0.561$ & $2.093$ & $2.862$ & $1.884$ & $2.576$ \\
$1.960$ & $0.950$ & $0.759$ & $0.818$ & $0.721$ & $1.612$ & $1.828$ & $1.531$ & $1.737$ \\
$2.576$ & $0.990$ & $0.925$ & $0.927$ & $0.916$ & $1.167$ & $1.181$ & $1.155$ & $1.169$ \\
\bottomrule
\end{tabular}
\caption{Set $\mathsf{f}_{u}$ equal to the standard normal density and the cut-offs $c$ as $1.645$, $1.96$, $2.576$. The probability $\psi$ of the absolute value of $u_{i} / \sigma$ within the cut-off $c$, the bias correction factor $\varsigma_{c}^{2}$ for estimating $\sigma^{2}$, the relative efficiency factor $\vartheta_{c}^{(1)}$, $\vartheta_{c}^{(\ast)}$ for asymptotic variance of $\widehat{\beta}_{c}^{(1)}$, $\widehat{\beta}_{c}^{(\ast)}$, the adjusting factor $\iota_{c}^{(1)}$, $\iota_{c}^{(\ast)}$ for estimating asymptotic variance, and the adjusting factor $\psi \iota_{c}^{(1)}$, $\psi \iota_{c}^{(\ast)}$ for estimating variance of $\widehat{\beta}_{c}^{(1)}$, $\widehat{\beta}_{c}^{(\ast)}$ under the null of no outliers for large samples.}
\label{adjusting factors for variance estimates for asymptotic variance for standard errors}
\end{table}

It is fundamental to note that $\psi$, $\varsigma_{c}^{2}$, $\vartheta_{c}^{(1)}$, $\vartheta_{c}^{(\ast)}$ are between $0$ and $1$, so $\iota_{c}^{(1)} = (\vartheta_{c}^{(1)} \varsigma_{c}^{2})^{-1}$, $\iota_{c}^{(\ast)} = (\vartheta_{c}^{(\ast)} \varsigma_{c}^{2})^{-1}$ and $\psi \iota_{c}^{(1)}$, $\psi \iota_{c}^{(\ast)}$ are greater than $1$, since the trimming method underestimates the variance $\sigma^{2}$ and has the lower efficiency (larger asymptotic variance) than ordinary 2SLS such that the adjusted (estimated) asymptotic variance and standard errors of beta estimates are larger than those provided by usual asymptotics.

Consistent with Figures \ref{efficiency loss against cut-offs of first and infinite step} and \ref{efficiency loss against iterations for three cut-offs}, as the cut-off $c$ increases then $\psi$, $\varsigma_{c}^{2}$, $\vartheta_{c}^{(1)}$, $\vartheta_{c}^{(\ast)}$ increase and approach to $1$ so $\iota_{c}^{(1)}$, $\iota_{c}^{(\ast)}$ decrease and approach to $1$ along with the fact that $\psi \iota_{c}^{(1)}$, $\psi \iota_{c}^{(\ast)}$ approach to $\iota_{c}^{(1)}$, $\iota_{c}^{(\ast)}$. Moreover, $\vartheta_{c}^{(1)}$ in the first step dominates the infinite step iteration $\vartheta_{c}^{(\ast)}$ so that magnitudes $\iota_{c}^{(1)}$, $\psi \iota_{c}^{(1)}$ in the first step is smaller and closer to $1$ than that of the infinite step $\iota_{c}^{(\ast)}$, $\psi \iota_{c}^{(\ast)}$.

Nuisance parameters $\xi_{c}$, $\zeta_{c}^{+}$, $\zeta_{c}^{-}$\footnote{see $\xi_{c}$ in (\ref{conditional expectation of r given u = y}, \S \ref{model}) and $\zeta_{c}^{+}$, $\zeta_{c}^{-}$ in (\ref{plus and minus conditional mean}, \S \ref{model}).} measure the dependence between $x_{i}$ and $u_{i}$ (between $r_{i}$ and $u_{i}$) in the IVs setup. They are significantly important to evaluate the degree of endogeneity in the structural equation, but difficult to estimate in practice. However, performing valid inference on $\beta$ does not require estimating these nuisance parameters, since they all vanish asymptotically and thus do not contribute to standard errors. Therefore, Algorithms \ref{robustified two stage least squares} and \ref{split sample version} can be easily implemented not only to identify outliers and to obtain a robust estimator, but also to carry out valid inference.

\subsection{A new Hausman type test of outlier robustness} \label{a new Hausman type test of outlier robustness}
A frequent concern in empirical economics is whether a tiny set of outliers may have invalidated empirical results. The common practice is to carry out robustness checks by redoing the analyses with the sample after trimming all outliers detected by Algorithm \ref{iterated version of robust 2sls} and comparing results from the original ones with the full sample. Instead of heuristically checking the difference between robustified and ordinary 2SLS without knowing whether it is statistically significant, this paper formalizes the outlier robustness check as a new type of Durbin (1954)-Hausman (1978)-Wu (1973) test.

The test is based on trade-off between robustness and efficiency and enables to judge whether the least squares estimation is appropriate or the robust method should be preferred. As discussed in \S \ref{efficiency comparison}, the robust estimator produced by Algorithm \ref{iterated version of robust 2sls} is consistent both under the null and the alternative (although less efficient under the null), whereas ordinary 2SLS is efficient (and consistent) under the null, but inconsistent otherwise. The test statistics is looking on statistically significant difference between robustified and ordinary 2SLS. If the model is correctly specified, the Hausman type test statistics should be rather small under the null of no outliers, since two consistent methods should produce estimates that are very close. On the contrary when outliers have large influence on least squares estimation, the robust trimming method should be very different from the ordinary estimate, so 2SLS should be rejected and the robustified 2SLS should be preferred even if less efficient. The proposed test thus will evaluate whether the gain in robustness is more valuable than the corresponding loss in efficiency.

In practice, empirical researchers run the full sample 2SLS $\widetilde{\beta}$ and compare to the robustified 2SLS $\widehat{\beta}_{c}^{(m + 1)}$ on the sub-sample with all outliers removed by Algorithm \ref{robustified two stage least squares} or \ref{split sample version}. The new type of Hausman test can detect whether two estimates are significantly distinct by looking on the L2 norm of difference between $\widetilde{\beta}$ and $\widehat{\beta}_{c}^{(m + 1)}$. Thus, it is essential first to derive the stochastic expansion and weak limit of a sequence of processes $\mathbb{H}_{n}^{(m + 1)}(c) = n^{1/2} (\widehat{\beta}_{c}^{(m + 1)} - \widetilde{\beta})$ for $c \in [c_{+}, \infty)$ and $m \in [0, \infty)$.

\begin{theorem} \label{expansion and weak limit of processes of Hausman statistics}
Consider Algorithm \ref{robustified two stage least squares} or \ref{split sample version}. Suppose Assumption \ref{sufficient assumptions}$(ia, iia, iiia, iv)$ holds. For $c \in [c_{+}, \infty)$ and $m \in [0, \infty)$ denote the process $\mathbb{H}_{n}^{(m + 1)}(c) = n^{1/2} (\widehat{\beta}_{c}^{(m + 1)} - \widetilde{\beta})$. Then as $n \to \infty$ we have
\begin{equation*}
\mathbb{H}_{n}^{(m + 1)}(c) = (\varrho_{\beta \beta, c}^{(m + 1)} - 1) M_{\tilde{x} \tilde{x}, n}^{-1} \sum_{i = 1}^{n} \tilde{x}_{in} u_{i} + \varrho_{\beta \tilde{x} u, c}^{(m + 1)} M_{\tilde{x} \tilde{x}, n}^{-1} \sum_{i = 1}^{n} \tilde{x}_{in} u_{i} 1_{(|u_{i}| \le \sigma c)} + \mathrm{o}_{\mathsf{P}}(1),
\end{equation*}
where $\varrho_{\beta \beta, c}^{(m + 1)}$, $\varrho_{\beta \tilde{x} u, c}^{(m + 1)}$ are defined in Theorem \ref{stochastic expansion of the iterated 2sls in terms of initial estimators allows varying cut-off c}. Furthermore, $\mathbb{H}_{n}^{(m + 1)}$ weakly converges to a zero mean Gaussian process $\mathbb{H}^{(m + 1)}$ with variance given as
\begin{equation*}
\mathsf{Var} \{ \mathbb{H}^{(m + 1)}(c) \} = \{ (\varrho_{\beta \beta, c}^{(m + 1)} - 1)^{2} + 2 \tau_{2}^{c} (\varrho_{\beta \beta, c}^{(m + 1)} - 1) \varrho_{\beta \tilde{x} u, c}^{(m + 1)} + \tau_{2}^{c} (\varrho_{\beta \tilde{x} u, c}^{(m + 1)})^{2} \} \sigma^{2} M_{\tilde{x} \tilde{x}}^{-1}.
\end{equation*}
\end{theorem}

The proof of the above theorem follows from stochastic expansions of $\widehat{\beta}_{c}^{(m + 1)}$ in Theorem \ref{weak convergence for the m+1 step beta estimator for robustified 2sls}, \ref{weak convergence for the m+1 step beta estimator for split sample iterated 2sls} and $\widetilde{\beta}$ in Lemma \ref{expansion of 2sls and its limiting distribution}. The next corollary establishes a new type of Hausman test for outlier robustness checks, which is based on pointwise convergence of $\mathbb{H}_{n}^{(m + 1)}(c) = n^{1/2} (\widehat{\beta}_{c}^{(m + 1)} - \widetilde{\beta})$ implied by its weak limit in Theorem \ref{expansion and weak limit of processes of Hausman statistics}.

\begin{corollary} \label{pointwise convergence of Hausman test statistics}
Consider Algorithm \ref{robustified two stage least squares} or \ref{split sample version}. Suppose Assumption \ref{sufficient assumptions}$(ia, iia, iiia, iv)$ holds. For $c \in [c_{+}, \infty)$ and $m \in [0, \infty)$ then as $n \to \infty$ we have
\begin{equation*}
n^{1/2} (\widehat{\beta}_{c}^{(m + 1)} - \widetilde{\beta}) \overset{\mathsf{D}}{\to} \mathsf{N}\{ 0_{d_{x}}, \mathsf{avar}(\widehat{\beta}_{c}^{(m + 1)} - \widetilde{\beta}) \},
\end{equation*}
where $\mathsf{avar}(\widehat{\beta}_{c}^{(m + 1)} - \widetilde{\beta}) = \{ (\varrho_{\beta \beta, c}^{(m + 1)} - 1)^{2} + 2 \tau_{2}^{c} (\varrho_{\beta \beta, c}^{(m + 1)} - 1) \varrho_{\beta \tilde{x} u, c}^{(m + 1)} + \tau_{2}^{c} (\varrho_{\beta \tilde{x} u, c}^{(m + 1)})^{2} \} \sigma^{2} M_{\tilde{x} \tilde{x}}^{-1}$. Then the new type of Hausman test statistics has the weak limit
\begin{equation*}
H_{n, c}^{(m + 1)} = n (\widehat{\beta}_{c}^{(m + 1)} - \widetilde{\beta})^{\prime} \mathsf{avar}(\widehat{\beta}_{c}^{(m + 1)} - \widetilde{\beta})^{-1} (\widehat{\beta}_{c}^{(m + 1)} - \widetilde{\beta}) \overset{\mathsf{D}}{\to} \chi^{2}_{d_{x}}.
\end{equation*}
\end{corollary}

Hausman (1978) argued that, when two estimators are correlated (one is always consistent but inefficient under the null, the other efficient but not consistent under the alternative), the asymptotic variance of their difference is given by the difference of their respective asymptotic variances. However, his argument requires additional regularity conditions that might not be satisfied in our context. The next remark is to show that the argument does work so $\mathsf{avar}(\widehat{\beta}_{c}^{(m + 1)} - \widetilde{\beta}) = \mathsf{avar}(\widehat{\beta}_{c}^{(m + 1)}) - \mathsf{avar}(\widetilde{\beta})$ when $\mathsf{f}_{u} \sim \mathsf{N}(0, 1)$.

\begin{remark} \label{asymptotic variance of difference equals the difference of respective variances}
Rearrange the expression of $\mathsf{avar}(\widehat{\beta}_{c}^{(m + 1)} - \widetilde{\beta})$ in Corollary \ref{pointwise convergence of Hausman test statistics} to achieve
\begin{equation*}
\{ (\varrho_{\beta \beta, c}^{(m + 1)})^{2} - 2 \varrho_{\beta \beta, c}^{(m + 1)} + 1 + 2 \tau_{2}^{c} \varrho_{\beta \beta, c}^{(m + 1)} \varrho_{\beta \tilde{x} u, c}^{(m + 1)} - 2 \tau_{2}^{c} \varrho_{\beta \tilde{x} u, c}^{(m + 1)} + \tau_{2}^{c} (\varrho_{\beta \tilde{x} u, c}^{(m + 1)})^{2} \} \sigma^{2} M_{\tilde{x} \tilde{x}}^{-1}.
\end{equation*}
Assume $\mathsf{f}_{u} \sim \mathsf{N}(0, 1)$, then $\tau_{2}^{c} = \psi - 2 c \mathsf{f}_{u}(c)$. Recall from Theorem \ref{stochastic expansion of the iterated 2sls in terms of initial estimators allows varying cut-off c} that
\begin{equation*}
\varrho_{\beta \beta, c}^{(m + 1)} = \{ \frac{2c \mathsf{f}_{u}(c)}{\psi} \}^{m + 1}, \qquad \varrho_{\beta \tilde{x} u, c}^{(m + 1)} = \frac{\psi^{m + 1} - \{ 2 c \mathsf{f}_{u}(c) \}^{m + 1}}{\psi^{m + 1} \{ \psi - 2 c \mathsf{f}_{u}(c) \}}.
\end{equation*}
Apply these terms to attain $- 2 \varrho_{\beta \beta, c}^{(m + 1)} + 1 - 2 \tau_{2}^{c} \varrho_{\beta \tilde{x} u, c}^{(m + 1)} = -1$. Further notice that $\mathsf{avar}(\widehat{\beta}_{c}^{(m + 1)}) = \{ (\varrho_{\beta \beta, c}^{(m + 1)})^{2} + 2 \tau_{2}^{c} \varrho_{\beta \beta, c}^{(m + 1)} \varrho_{\beta \tilde{x} u, c}^{(m + 1)} + \tau_{2}^{c} (\varrho_{\beta \tilde{x} u, c}^{(m + 1)})^{2} \} \sigma^{2} M_{\tilde{x} \tilde{x}}^{-1}$ and $\mathsf{avar}(\widetilde{\beta}) = \sigma^{2} M_{\tilde{x} \tilde{x}}^{-1}$, see (\ref{asymptotic distribution for general step beta estimator}) and Lemma \ref{expansion of 2sls and its limiting distribution}. Thus, we finally shows
\begin{align*}
\mathsf{avar}(\widehat{\beta}_{c}^{(m + 1)} - \widetilde{\beta}) & = \{ (\varrho_{\beta \beta, c}^{(m + 1)} - 1)^{2} + 2 \tau_{2}^{c} (\varrho_{\beta \beta, c}^{(m + 1)} - 1) \varrho_{\beta \tilde{x} u, c}^{(m + 1)} + \tau_{2}^{c} (\varrho_{\beta \tilde{x} u, c}^{(m + 1)})^{2} \} \sigma^{2} M_{\tilde{x} \tilde{x}}^{-1} \\
& = \{ (\varrho_{\beta \beta, c}^{(m + 1)})^{2} + 2 \tau_{2}^{c} \varrho_{\beta \beta, c}^{(m + 1)} \varrho_{\beta \tilde{x} u, c}^{(m + 1)} + \tau_{2}^{c} (\varrho_{\beta \tilde{x} u, c}^{(m + 1)})^{2} \} \sigma^{2} M_{\tilde{x} \tilde{x}}^{-1} - \sigma^{2} M_{\tilde{x} \tilde{x}}^{-1} \\
& = \mathsf{avar}(\widehat{\beta}_{c}^{(m + 1)}) - \mathsf{avar}(\widetilde{\beta}).
\end{align*}
\end{remark}

Remark \ref{asymptotic variance of difference equals the difference of respective variances} demonstrates that, in practical applications, it is equivalent either to estimate $\mathsf{avar}(\widehat{\beta}_{c}^{(m + 1)} - \widetilde{\beta})$ or $\mathsf{avar}(\widehat{\beta}_{c}^{(m + 1)}) - \mathsf{avar}(\widetilde{\beta})$ so as to obtain test statistics of outlier robustness developed in Corollary \ref{pointwise convergence of Hausman test statistics}. If $u_{i} / \sigma \sim \mathsf{f}_{u}$ follows $\mathsf{N}(0, 1)$ according to which the cut-off value $c$ is chosen, terms $\tau_{2}^{c}$, $\varrho_{\beta \beta, c}^{(m + 1)}$, $\varrho_{\beta \tilde{x} u, c}^{(m + 1)}$ are known for any iteration step $m$, so the only item left to estimate in the asymptotic variance is $\sigma^{2} M_{\tilde{x} \tilde{x}}^{-1}$.

As discussed in \S \ref{valid inference on structural parameters beta}, the estimator of asymptotic variance should be robust itself, so $\sigma^{2} M_{\tilde{x} \tilde{x}}^{-1}$ should be estimated by $(\widehat{\sigma}_{c}^{(m + 1)})^{2}$, $\widehat{M}_{\tilde{x} \tilde{x}, c}^{(m + 1)}$, see (\ref{updated 2sls variance}), (\ref{robust estimation of moment M}), based on the sub-sample with all outliers removed instead of using the full sample. Otherwise the test might lose power and lead to incorrect results under the alternative, where estimates using the full sample are inconsistent though they are consistent (efficient) under the null of no outliers. Notice that to obtain test statistics many empirical researchers wrongly estimate $\mathsf{avar}(\widetilde{\beta})$ by using the full sample estimates of $\sigma^{2}$, $M_{\tilde{x} \tilde{x}}$, which might produce $\widehat{H}_{n, c}^{(m + 1)} \le 0$ as $\widehat{\mathsf{avar}}(\widehat{\beta}_{c}^{(m + 1)}) - \widehat{\mathsf{avar}}(\widetilde{\beta}) \le 0$ and make the test meaningless under the alternative, see discussion in detail in our empirical application to Acemoglu et al. (2019) in \S \ref{empirical application to Acemoglu et al. (2019)}.
Thus, in practice it is important to avoid this false procedure by applying the robust estimator of asymptotic variance of difference between ordinary and robustified 2SLS.

Once achieving the estimator $\widehat{\mathsf{avar}}(\widehat{\beta}_{c}^{(m + 1)} - \widetilde{\beta})$ of asymptotic variance by
\begin{equation} \label{estimated asymptotic variance}
\{ (\varrho_{\beta \beta, c}^{(m + 1)} - 1)^{2} + 2 \tau_{2}^{c} (\varrho_{\beta \beta, c}^{(m + 1)} - 1) \varrho_{\beta \tilde{x} u, c}^{(m + 1)} + \tau_{2}^{c} (\varrho_{\beta \tilde{x} u, c}^{(m + 1)})^{2} \} (\widehat{\sigma}_{c}^{(m + 1)})^{2} (\widehat{M}_{\tilde{x} \tilde{x}, c}^{(m + 1)})^{-1},
\end{equation}
which is consistent under the null and robust under the alternative, then we have
\begin{equation*}
n^{1/2} (\widehat{\beta}_{c}^{(m + 1)} - \widetilde{\beta}) \overset{a}{\sim} \mathsf{N}\{ 0_{d_{x}}, \widehat{\mathsf{avar}}(\widehat{\beta}_{c}^{(m + 1)} - \widetilde{\beta}) \},
\end{equation*}
and
\begin{equation*}
\widehat{H}_{n, c}^{(m + 1)} = n (\widehat{\beta}_{c}^{(m + 1)} - \widetilde{\beta})^{\prime} \widehat{\mathsf{avar}}(\widehat{\beta}_{c}^{(m + 1)} - \widetilde{\beta})^{-1} (\widehat{\beta}_{c}^{(m + 1)} - \widetilde{\beta}) \overset{a}{\sim} \chi^{2}_{d_{x}}.
\end{equation*}
Thus, a new type of Hausman test for outlier robustness has now been established, and it can either be performed as the one-sided test with the chi-squared limit using estimates of the whole parameter vector or as the two-sided test with the normal limit using the corresponding estimates of individual scalar parameters.

Note that the asymptotic variance in (\ref{estimated asymptotic variance}) is invertible, since the first scalar term in the expression $(\varrho_{\beta \beta, c}^{(m + 1)} - 1)^{2} + 2 \tau_{2}^{c} (\varrho_{\beta \beta, c}^{(m + 1)} - 1) \varrho_{\beta \tilde{x} u, c}^{(m + 1)} + \tau_{2}^{c} (\varrho_{\beta \tilde{x} u, c}^{(m + 1)})^{2} \neq 0$ and its inverse is
\begin{equation*}
\{ (\varrho_{\beta \beta, c}^{(m + 1)} - 1)^{2} + 2 \tau_{2}^{c} (\varrho_{\beta \beta, c}^{(m + 1)} - 1) \varrho_{\beta \tilde{x} u, c}^{(m + 1)} + \tau_{2}^{c} (\varrho_{\beta \tilde{x} u, c}^{(m + 1)})^{2} \}^{-1} (\widehat{\sigma}_{c}^{(m + 1)})^{-2} \widehat{M}_{\tilde{x} \tilde{x}, c}^{(m + 1)}.
\end{equation*}
We therefore avoid the rank deficiency problem described and addressed through generalized inverses in Hausman and Taylor (1981) and Holly (1982).


On implementing Algorithm \ref{robustified two stage least squares} or \ref{split sample version}, we are particularly interested in the trimming estimator just updated from the original 2SLS and in the fixed point estimator iterated upon through infinite steps. The next corollary provides two special cases of the test for outlier robustness frequently used in empirics, checking whether $\widetilde{\beta}$ is distinct from $\widehat{\beta}_{c}^{(1)}$ when $m = 0$ or from $\widehat{\beta}_{c}^{(\ast)}$ when $m = \infty$.

\begin{corollary} \label{Hausman test when m = 1 and m = infinite}
Consider Algorithm \ref{robustified two stage least squares} or \ref{split sample version}. Suppose Assumption \ref{sufficient assumptions}$(ia, iia, iiia, iv)$ holds. For $c \in [c_{+}, \infty)$ and a large $n$ then we have for $m = 0$
\begin{equation*}
n^{1/2} (\widehat{\beta}_{c}^{(1)} - \widetilde{\beta}) \overset{a}{\sim} \mathsf{N}\{ 0_{d_{x}}, \widehat{\mathsf{avar}}(\widehat{\beta}_{c}^{(1)} - \widetilde{\beta}) \},
\end{equation*}
and
\begin{equation*}
\widehat{H}_{n, c}^{(1)} = n (\widehat{\beta}_{c}^{(1)} - \widetilde{\beta})^{\prime} \widehat{\mathsf{avar}}(\widehat{\beta}_{c}^{(1)} - \widetilde{\beta})^{-1} (\widehat{\beta}_{c}^{(1)} - \widetilde{\beta}) \overset{a}{\sim} \chi^{2}_{d_{x}},
\end{equation*}
where
\begin{equation*}
\widehat{\mathsf{avar}}(\widehat{\beta}_{c}^{(1)} - \widetilde{\beta}) = \frac{\{2 c \mathsf{f}_{u}(c) - \psi \}^{2} + 2 \tau_{2}^{c} \{2 c \mathsf{f}_{u}(c) - \psi \} + \tau_{2}^{c}}{\psi^{2}} (\widehat{\sigma}_{c}^{(1)})^{2} (\widehat{M}_{\tilde{x} \tilde{x}, c}^{(1)})^{-1}.
\end{equation*}
In addition, we have for $m = \infty$
\begin{equation*}
n^{1/2} (\widehat{\beta}_{c}^{(\ast)} - \widetilde{\beta}) \overset{a}{\sim} \mathsf{N}\{ 0_{d_{x}}, \widehat{\mathsf{avar}}(\widehat{\beta}_{c}^{(\ast)} - \widetilde{\beta}) \},
\end{equation*}
and
\begin{equation*}
\widehat{H}_{n, c}^{(\ast)} = n (\widehat{\beta}_{c}^{(\ast)} - \widetilde{\beta})^{\prime} \widehat{\mathsf{avar}}(\widehat{\beta}_{c}^{(\ast)} - \widetilde{\beta})^{-1} (\widehat{\beta}_{c}^{(\ast)} - \widetilde{\beta}) \overset{a}{\sim} \chi^{2}_{d_{x}},
\end{equation*}
where
\begin{equation*}
\widehat{\mathsf{avar}}(\widehat{\beta}_{c}^{(\ast)} - \widetilde{\beta})\footnote{If $\mathsf{f}_{u} \sim \mathsf{N}(0, 1)$, then $\tau_{2}^{c} = \psi - 2 c \mathsf{f}_{u}(c)$ so $\widehat{\mathsf{avar}}(\widehat{\beta}_{c}^{(1)} - \widetilde{\beta})$ and $\widehat{\mathsf{avar}}(\widehat{\beta}_{c}^{(\ast)} - \widetilde{\beta})$ can be further simplified.} = \frac{ \{ 2 c \mathsf{f}_{u}(c) - \psi \}^{2} + 2 \tau_{2}^{c} \{ 2 c \mathsf{f}_{u}(c) - \psi \} + \tau_{2}^{c} }{\{ 2 c \mathsf{f}_{u}(c) - \psi \}^{2}} (\widehat{\sigma}_{c}^{(\ast)})^{2} (\widehat{M}_{\tilde{x} \tilde{x}, c}^{(\ast)})^{-1}.
\end{equation*}
\end{corollary}

Without knowing the asymptotic theory derived in this section, some empirical researchers simply use the similar form of the Hausman test statistics but incorrectly replace the standard error of the difference of two estimates $\widetilde{\beta}$ and $\widehat{\beta}_{c}^{(m + 1)}$ by the standard error of the marginal distribution of the baseline estimate $\widetilde{\beta}$. They subsequently compare this test statistics with the critical value either drawn from the standard normal for the two-sided test or from the chi-square with $d_{x}$ degree of freedom for the one-sided test. The above test procedure is commonly referred to as the heuristic method. Since the heuristic test statistics is constructed in a mistaken way, its size does not converge to the nominal level under the null and it has low statistical power under the alternative. Thus, even the heuristic method is informative for preliminary analysis, it is not consistent in theory and not reliable in practice.

Kaji (2018) applies the non-parametric bootstrap by randomly sampling observations $i$ from the data $\{(y_{i}, x_{i}, z_{i})\}_{i = 1}^{n}$ with replacement to draw the distribution of the $L_{1}$ norm of the difference between $\widehat{\beta}_{c}^{(m + 1)}$ and $\widetilde{\beta}$. To prove the validity of the bootstrap, he first characterises the weak convergence of a class of integrable empirical processes and then build the functional delta method for the corresponding quantile processes. Unlike the asymptotic theory in this paper, the bootstrap in Kaji (2018) does not require assuming the known form of the density $\mathsf{f}_{u}$. However, if the specification assumed on $\mathsf{f}_{u}$ is correct, our Hausman type test would have the higher power than the bootstrapping method under the alternative where outliers in fact present. This could be because under the alternative the bootstrap has a certain probability to sample outlying observations that have large influence on $\beta$ estimation. For these bootstrapping iterations having sampled large outliers from the original data, the test statistics on difference between the robust estimate $\widehat{\beta}_{c}^{(m + 1)}$ and the baseline estimate $\widetilde{\beta}$ would then become rather large, so the distribution of the test statistics drawn from the bootstrap sample would be highly distorted and significantly distinct from what it should be under the null of no outliers.



\section{Simulation studies} \label{simulation studies}
In this section, we verify the established asymptotic theory and check their finite sample performance using Monte Carlo simulations. This paper uses a DGP where the structural equation contains an intercept and one endogenous regressor, which is instrumented by one excluded and informative instrument, such that we are in the just-identified case. The structural equation is given as
\begin{equation*}
	y_{i} = \beta_{1} x_{1i} + \beta_{2} x_{2i} + u_{i}, \quad i = 1, 2, \ldots, n.
\end{equation*}
The first stage projection is
\begin{equation*}
	\begin{pmatrix}
		x_{1i} \\
		x_{2i}
	\end{pmatrix}
	= 
	\begin{pmatrix}
		1 & 0 \\
		\pi_{1} & \pi_{2}	
	\end{pmatrix}
	\begin{pmatrix}
		x_{1i} \\
		z_{2i}
	\end{pmatrix}
	+
	\begin{pmatrix}
		0 \\
		r_{2i}
	\end{pmatrix}, \quad i = 1, 2, \ldots, n.
\end{equation*}
Note $x_{1i} = 1$ for all $i$, and we choose $\beta_{1} = 2$, $\beta_{2} = 4$ and $\pi_{1} = 0$, $\pi_{2} = 1$. Simulated data $\{ (y_{i}, x_{i}, z_{i}) \}_{i = 1}^{n}$ can then be generated by drawing random samples from the Normal
\begin{equation*}
	\begin{pmatrix}
		u_{i} \\
		r_{2i} \\
		z_{2i} \\
	\end{pmatrix}
	\overset{\mathsf{D}}{=}
	\mathsf{N}
	\begin{Bmatrix}
		\begin{pmatrix}
			0 \\
			0 \\
			0 \\
		\end{pmatrix},
		\begin{pmatrix}
			\sigma^{2} & \Omega_{2} & 0 \\
			\Omega_{2} & \Sigma_{2} & 0 \\
			0 & 0 & \sigma_{z_{2}}^{2} \\
		\end{pmatrix}
	\end{Bmatrix}, \quad i = 1, 2, \ldots, n.
\end{equation*}
We choose the variances of $u_{i}$, $r_{2i}$, $z_{2i}$ as $\sigma^{2} = 1$, $\Sigma_{2} = 1$, $\sigma_{z_{2}}^{2} = 1$ and covariance between $u_{i}$ and $r_{2i}$ as $\Omega_{2}$ = 0.75, which is the source of endogeneity. The sample size is varied from $50$ to $5000$ to assess the finite sample performance and to check the correctness of the asymptotic theory, i.e. $n \in \{ 50, 100, 200, 500, 1000, 5000 \}$. We set $W = 10000$ replications for our Monte Carlo experiments. 

We run Algorithm \ref{iterated version of robust 2sls} with the choices of the initial estimator as the full-sample 2SLS, the cut-off value $1.96$, and the iterations $m \in \{ 1, 5, \textrm{Fixed point} \}$. As a result, for each Monte Carlo replication $w = 1, 2, \ldots, W$, the algorithm produces the $\beta$-estimators $\widehat{\beta}_{1, c, w}^{(m)}$, $\widehat{\beta}_{2, c, w}^{(m)}$ and $\sigma^{2}$-estimators $(\widehat{\sigma}_{c, w}^{(m)})^{2}$ on the simulated data $\{ (y_{i}, x_{i}, z_{i}) \}_{i = 1, w}^{n}$.

\begin{figure}[h!]
	\centering
	\includegraphics[width=\textwidth, height=0.7\textheight, keepaspectratio]{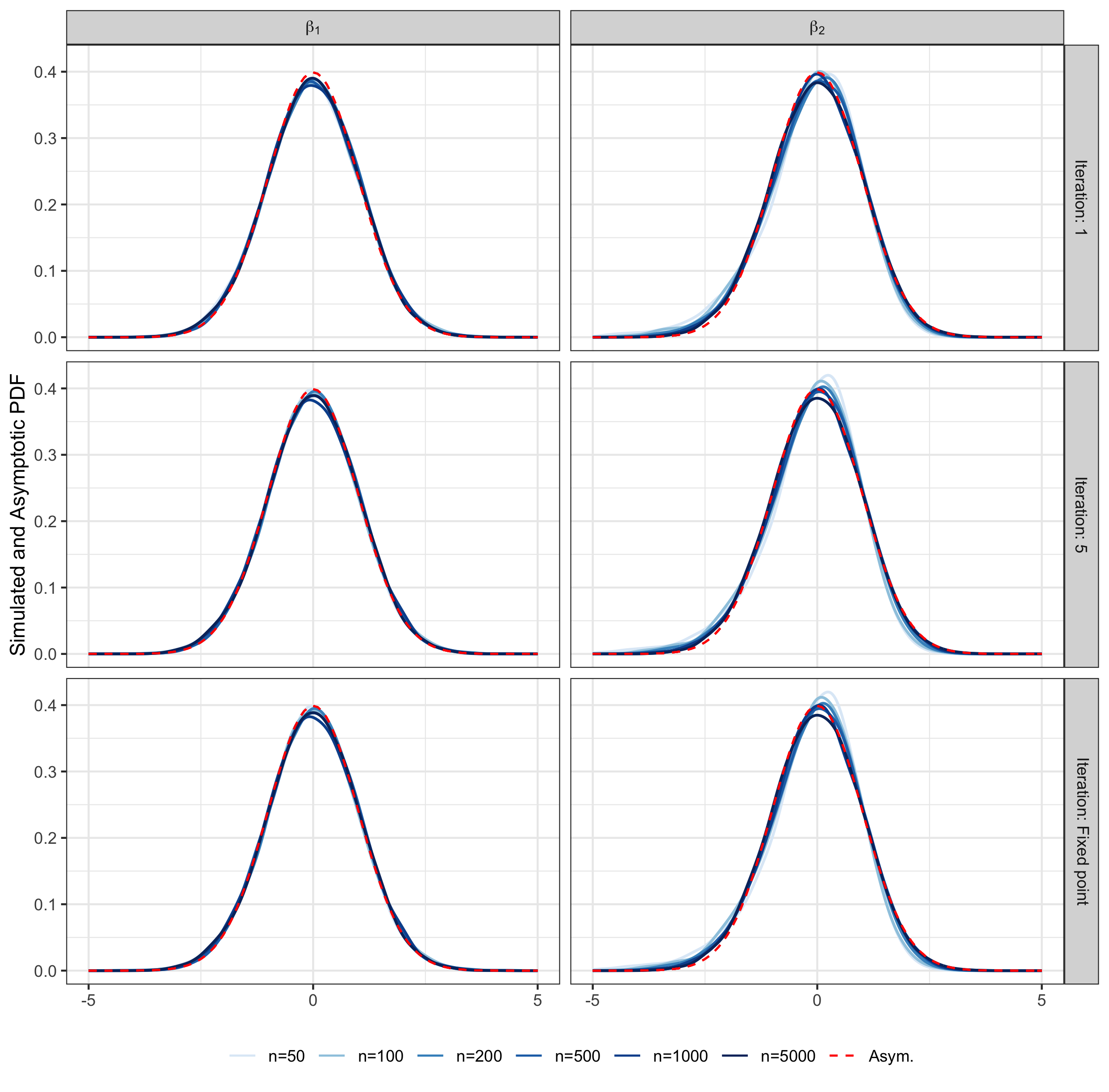}
	\caption{Comparison of the asymptotic distribution of the standardised $\beta$-estimators with the simulated distributions. Two settings are varied: The iteration step $m$ and the sample size $n$.} \label{asymptotic distribution of beta}
\end{figure}

\begin{figure}[h!]
	\centering
	\includegraphics[width=\textwidth, height=0.7\textheight, keepaspectratio]{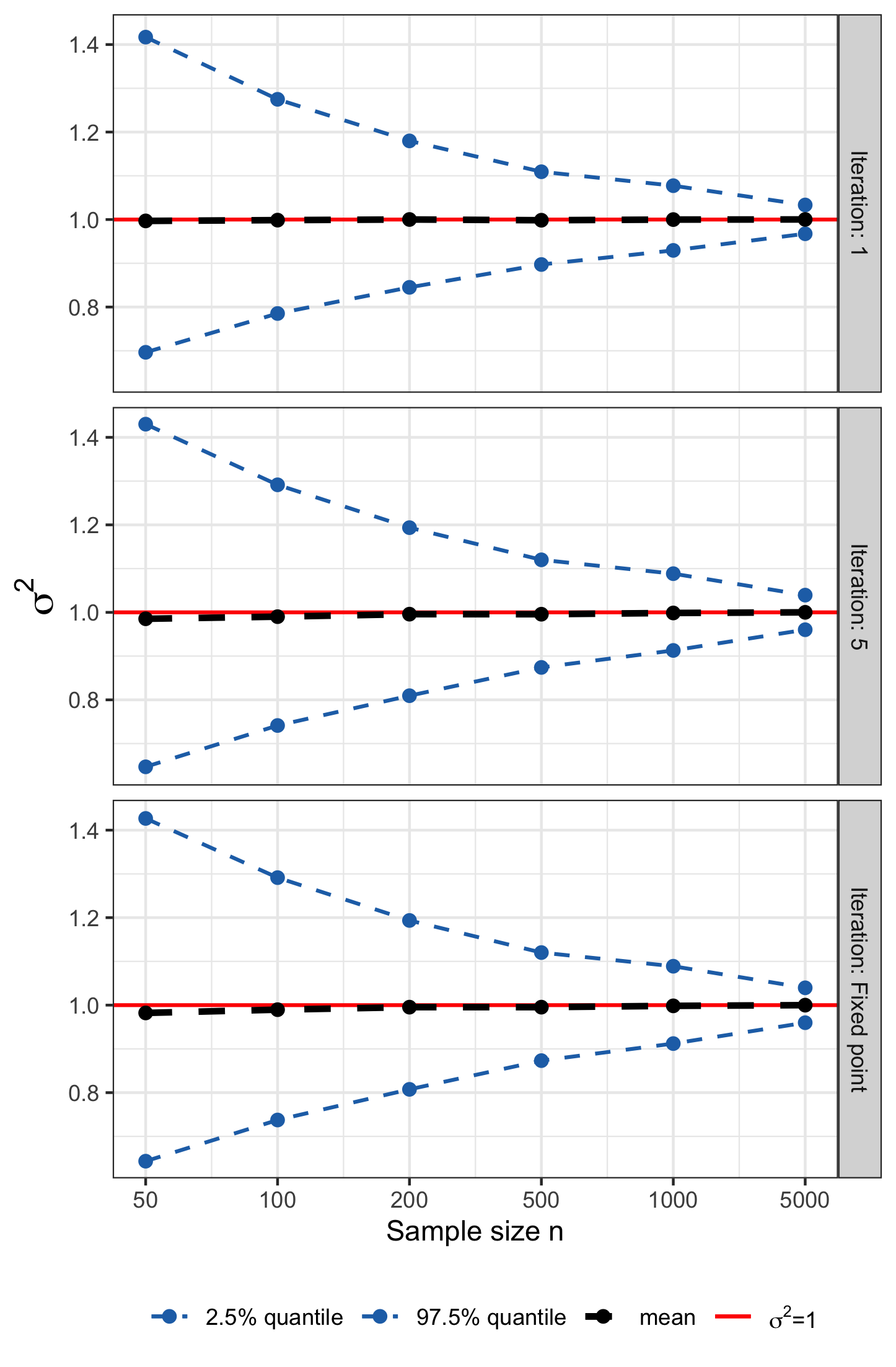}
	\caption{Comparison of the simulated $\sigma$-estimators with the population one. Two settings are varied: The iteration step $m$ and the sample size $n$.} \label{consistency of sigma}
\end{figure}

\subsection{Evaluation of asymptotic distribution of $\beta$-estimators} \label{evaluation of asymptotic distribution of beta}
Our asymptotic theory shows
\begin{equation*}
n^{1/2}	\begin{pmatrix}
		\widehat{\beta}_{1, c}^{(m)} - \beta_{1} \\
		\widehat{\beta}_{2, c}^{(m)} - \beta_{2}
		\end{pmatrix}
		\overset{\mathsf{D}}{\to}
\mathsf{N}	\begin{Bmatrix}
		\begin{pmatrix}
		0 \\
		0
		\end{pmatrix},
(\vartheta_{c}^{(m)})^{-1}
		\begin{pmatrix}
		1 & 0 \\
		0 & 1
		\end{pmatrix}
		\end{Bmatrix}.
\end{equation*}
The asymptotic variance is given by $(\vartheta_{c}^{(m)})^{-1} \sigma^{2} M_{\tilde{x} \tilde{x}}^{-1}$, where in our simulation specification $\sigma^{2} = 1$ and $M_{\tilde{x} \tilde{x}} = \mathsf{E} \tilde{x}_{i} \tilde{x}_{i}^{\prime} = \Pi^{\prime} \mathsf{E} z_{i} z_{i}^{\prime} \Pi = I_{2}$ as $z_{i}^{\prime} = (1, z_{2i})$ and $\Pi = I_{2}$. This asymptotic distribution is then evaluated by comparing the simulated distribution of $\{ n^{1/2} (\widehat{\beta}_{j, c, w}^{(m)} - \beta_{j}) / \sqrt{(\vartheta_{c}^{(m)})^{-1}} \}_{w = 1}^{W}$ for $j = 1, 2$ to $\mathsf{N}(0, 1)$ when the sample size $n$ increases. Figure \ref{asymptotic distribution of beta} plots kernel density estimates of the simulated distributions and the probability density function of the standard normal distribution as the asymptotic distribution. The graphs show that the standardised $\beta$-estimators converges to a standard normal distribution and that the approximation is very close for the sample size $n = 50$ and larger regardless of the iterations steps $m$. This indicates that the derived relative efficiency factor $\vartheta_{c}^{(m)}$ is consistent with the simulation studies.

\subsection{Evaluation of consistency of $\sigma^{2}$-estimators} \label{evaluation of consistency of sigma}
Our asymptotic theory shows $(\widehat{\sigma}_{c}^{(m)})^{2} \overset{\mathsf{P}}{\to} \sigma^{2}$. We assess the consistency by checking the difference between the simulated estimates $\{ (\widehat{\sigma}_{c, w}^{(m)})^{2} \}_{w = 1}^{W}$ and the population one $\sigma^{2} = 1$ when the sample size $n$ increases. Figure \ref{consistency of sigma} plots means, $97.5 \%$ quantiles, $2.5 \%$ quantiles of the simulated estimates and the true variance of structural errors. The graphs show that $\sigma^{2}$-estimators converges to the population one and that the deviation is already very small for the sample size $n = 50$ regardless of the iterations steps $m$. This indicates that the simulations lend support to the correctness of the derived bias correction factor $\varsigma_{c}^{2}$.

\section{Empirical application to Acemoglu et al. (2019)} \label{empirical application to Acemoglu et al. (2019)}
This paper now applies the robust procedure proposed in \S \ref{robust statistical algorithms to detect outliers} and its theory to recarry out outlier analysis in Acemoglu et al. (2019). Their paper tries to address the long-standing debate of whether democracy does cause economic growth through an empirical study. Acemoglu et al. (2019) find evidence that democracy has a significant and robust positive effect on GDP per capita and further suggest that democracy increases GDP by encouraging investment, increasing schooling, inducing economic reforms and openness, improving the provision of public goods, and reducing social unrest.

To investigate the causal effect of democracy on GDP per capita, Acemoglu et al. (2019) collect an annual panel that comprises 175 countries around the world from the year 1960 to 2010. The panel is unbalanced in the sense that not all observations of GDP and democracy are available for the entire sample. In order to reduce measurement error, they create a consolidated and dichotomous measure of democracy from several datasets and document robustness of their results to other measures of GDP and democracy.

Acemoglu et al. (2019) conduct three specifications that control for the rich dynamics of GDP and endogenous selection into democracy, which otherwise confound the effect of democracy on economic growth. Three specifications consist a dynamic panel model for log GDP per capita, a semi-parametric treatment effects model which does not rely on a parametric assumption for GDP dynamics, the same dynamic panel model but instrumenting democracy by regional waves of democratizations and reversals. We focus on the first and third specification (a dynamic panel model) in Acemoglu et al. (2019) where they carry out outlier analysis for OLS and 2SLS fixed effect estimates.

Their baseline preferred specification constructs a dynamic panel that models log GDP per capita as outcome variable, democracy as treatment variable, and four lags of log GDP per capita as other covariates to control GDP dynamics. They start with running the fixed effect regression implemented as OLS estimator by adding dummy variable regressors for each country and year. Although the fixed effect estimator suffers the Nickell (1981) bias of order one over time horizon resulting from the failure of strict exogenity in dynamic panel model, this bias should be small in their setup where the time dimension of panel is fairly large (each country is on average observed 38.8 times). This motivates their use of the OLS fixed effect estimator as a natural starting point even in dynamic panel models. Acemoglu et al. (2019) then apply Arellano and Bond (1991) GMM estimator that deal with Nickell bias, and further use an even more refined Hahn, Hausman, and Kuersteiner (2002) GMM estimator that address “too many instruments” problem faced by Arellano and Bond estimator. Consistent with their expectation, both type of GMM estimates are similar to the OLS fixed effec estimate such that the coefficient on democracy variable is estimated to be positive and highly significant.

Their third specification is identical to the dynamic panel model above, but they treat democracy as endogenous and exploit IVs strategy by using regional waves of democratization and transitions to nondemocracy as source of exogenous variation in democracy. The instrument regional waves of democratization is constructed as the jack-knifed average of democracy for countries in the same region with the same initial regime cell, which leaves out the own-country observation. Although there is no consensus on the factors that create such waves, the existing evidence suggests that they are not explained by regional economic trends. Thus, Acemoglu et al. (2019) argue that regional waves are valid as instruments such that their exclusion restriction assumption is plausible. They further show from the first stage regression that regional waves are informative such that rank conditions hold for identification. By applying 2SLS fixed effect and Hahn, Hausman, and Kuersteiner (2002) estimator, they reassure that the causal effect of democracy on log GDP per capita is significant and positive.

Three empirical strategies tackle different possibilities and find very similar results, which bolsters their confidence to demonstrate the positive causal effect of democracy. To defend against the critique that their empirical findings are driven by outliers, Acemoglu et al. (2019) explores the sensitivity of their results to outliers for their baseline OLS and 2SLS fixed effect regression.

For each observation in the model (\ref{structural equation}) and (\ref{first stage regression}), we now have the dependent variable, regressors, and instruments
\begin{align*}
y_{i, t} & :=  \log \mathrm{GDP}_{i, t}, \\
x_{i, t} & := (\mathrm{Democracy}_{i, t}, \log \mathrm{GDP}_{i, t - 1}, \dots, \log \mathrm{GDP}_{i, t - 4}, \\
& \quad \,\, 1_{i = 1}, \ldots, 1_{i = 175}, 1_{t = 1960}, \ldots, 1_{t = 2010})^{\prime}, \\
z_{i, t} & :=  (\mathrm{Wave}_{i, t - 1}, \ldots, \mathrm{Wave}_{i, t - 4}, \log \mathrm{GDP}_{i, t - 1}, \dots, \log \mathrm{GDP}_{i, t - 4}, \\
& \quad \,\,  1_{i = 1}, \ldots, 1_{i = 175}, 1_{t = 1960}, \ldots, 1_{t = 2010})^{\prime}.
\end{align*}
In regressors, only $\mathrm{Democracy}_{i, t}$ is treated as endogenous, while the rest of $x_{i, t}$ are assumed to be orthogonal to structural errors. Then the structural and first stage regression parameters are
\begin{align*}
\beta :=
\begin{pmatrix}
\beta_{0} \\
\beta_{1} \\
\vdots \\
\beta_{4} \\
\theta_{1} \\
\vdots \\
\theta_{175} \\
\phi_{1960} \\
\vdots \\
\phi_{2010}
\end{pmatrix}
,
\qquad
\Pi := \
\begin{pmatrix}
\Pi_{1} & 0 & \cdots & 0 & 0 & \cdots & 0 & 0 & \cdots & 0 \\
\vdots & \vdots & \ddots & \vdots & \vdots & \ddots & \vdots & \vdots & \ddots & \vdots \\
\Pi_{4} & 0 & \cdots & 0 & 0 & \cdots & 0 & 0 & \cdots & 0 \\
\pi_{1} & 1 & \cdots & 0 & 0 & \cdots & 0 & 0 & \cdots & 0 \\
\vdots & \vdots & \ddots & \vdots & \vdots & \ddots & \vdots & \vdots & \ddots & \vdots \\
\pi_{4} & 0 & \cdots & 1 & 0 & \cdots & 0 & 0 & \cdots & 0 \\
\vartheta_{1} & 0 & \cdots & 0 & 1 & \cdots & 0 & 0 & \cdots & 0 \\
\vdots & \vdots & \ddots & \vdots & \vdots & \ddots & \vdots & \vdots & \ddots & \vdots \\
\vartheta_{175} & 0 & \cdots & 0 & 0 & \cdots & 1 & 0 & \cdots & 0 \\
\varphi_{1960} & 0 & \cdots & 0 & 0 & \cdots & 0 & 1 & \cdots & 0 \\
\vdots & \vdots & \ddots & \vdots & \vdots & \ddots & \vdots & \vdots & \ddots & \vdots \\
\varphi_{2010} & 0 & \cdots & 0 & 0 & \cdots & 0 & 0 & \cdots & 1
\end{pmatrix}
.
\end{align*}
In the structural parameter vector $\beta$, we are particulary interested in $\beta_{0}$ on the effect of democracy as well as $\beta_{1}, \ldots, \beta_{4}$ indicating the dynamics of the GDP process\footnote{Aside from structural parameters $\beta_{0}, \beta_{1}, \ldots, \beta_{4}$, Acemoglu et al. (2019) report three more parameters: $(1)$ $\beta_{5} := \beta_{0} / (1 - \beta_{1} - \beta_{2} - \beta_{3} - \beta_{4})$ (long-run equilibrium effect of democracy); $(2)$ $\beta_{6} := e_{25}$, where $e_{j} = \beta_{0} + \beta_{1} e_{j - 1} + \beta_{2} e_{j - 2} + \beta_{3} e_{j - 3} + \beta_{4} e_{j - 4}$ and $e_{0} = e_{-1} = e_{-2} = e_{-3} = 0$ (effect of transition to democracy after 25 years); $(3)$ $\beta_{7} := \beta_{1} + \beta_{2} + \beta_{3} + \beta_{4}$ (persistence of the GDP process).}.

Implicitly assuming normal distribution for structural errors, i.e. $\mathsf{f}_{u} \overset{\mathsf{D}}{=} \mathsf{N}(0, 1)$, Acemoglu et al. (2019) choose the cut-off value $c = 1.96$ for $5 \%$ significance level such that outliers are identified among observations whose absolute values of standardized residuals of structural errors are beyond $1.96$. They then perform the same OLS and 2SLS fixed effect regression but exclude these outlying observations\footnote{This is the first robust method that Acemoglu et al. (2019) apply, and its results are reported in column $(2)$, Table A$8$ for OLS fixed effect regression; and in column $(2)$, Table A$13$ for 2SLS fixed effect regression in their paper. They also use other robust methods, such as Cook's distance, Li's (1985) procedure, and Huber M-estimator; and take into account the influence of outliers in both the first and second stage; see their reported results in other columns in Table A8, A13.}. Investigating the difference between the ordinary (full sample) and robust (sub-sample without outliers) estimates, they argue robustness of their results to outliers.

The above robust method is in fact Algorithm \ref{robustified two stage least squares} (Robustified 2SLS), which select the cut-off $c = 1.96$ $(5 \%)$ by assuming $\mathsf{f}_{u} \overset{\mathsf{D}}{=} \mathsf{N}(0, 1)$ and iterates only once by choosing $m = 0$. However, Acemoglu et al. (2019) do not take into consideration the case that observations are classified as outliers and removed just by chance, so consistency does not hold for the updated estimator of the structural error variance if without adjusting the bias correction factor (\ref{bias correction factor}), although estimation of structural parameters is still consistent. In addition, by not adjusting standard errors calculated from the usual asymptotoics that does not have the correct asymptotic variance of robust estimates, inference on regression coefficients is invalid. At last, without a formal statistical test we do not know whether the robust estimate is significantly distinct from the ordinary one, thus even when we end up with the robustness conclusion, our argument is not rooted in a correct statistical reasoning.


To solve all these problems, this paper applies Algorithm \ref{robustified two stage least squares} to re-explore outlier robustness analysis in Acemoglu et al. (2019). As in their paper, we first compute the full sample 2SLS $\widetilde{\beta}$, then choose the cut-off $c = 1.96$ to run a single iteration $m = 0$ of the algorithm to obtain $\widehat{\beta}_{1.96}^{(1)}$. Besides selecting $m = 0$, we continue to run the algorithm upon through infinite iterations $m = \infty$ until reaching the fixed point $\widehat{\beta}_{1.96}^{(\ast)}$. In adjusting the standard errors of $\widehat{\beta}_{1.96}^{(1)}$ and $\widehat{\beta}_{1.96}^{(\ast)}$ by the factors listed in Table \ref{adjusting factors for variance estimates for asymptotic variance for standard errors}, this paper performs valid inference on structural parameters $\beta$ and conducts the new Hausman type tests to check whether $\widetilde{\beta}$ is different from $\widehat{\beta}_{1.96}^{(1)}$ or $\widehat{\beta}_{1.96}^{(\ast)}$. Our re-examination of Acemoglu et al. (2019) finds that the effect of democracy estimated with the full sample $\widetilde{\beta}_{0}$ is distinct in magnitude from the robust estimates $\widehat{\beta}_{0, 1.96}^{(1)}$ or $\widehat{\beta}_{0, 1.96}^{(\ast)}$ produced by Algorithm \ref{robustified two stage least squares}, although we find little evidence for the view that democracy is a constraint on economic growth ($\beta_{0} < 0$).

Tables \ref{OLS fixed effect regression and its robust estimates and outlier robustness tests} and \ref{2SLS fixed effect regression and its robust estimates and outlier robustness tests} list results of the OLS and 2SLS fixed effect regressions together with formal statistical tests for the hypothesis that outliers have no effect on parameters. The first three columns show the OLS (2SLS) fixed effect and robust estimates in Table \ref{OLS fixed effect regression and its robust estimates and outlier robustness tests} (Table \ref{2SLS fixed effect regression and its robust estimates and outlier robustness tests}). Select the cut-off value $c = 1.96$ and set the iteration step $m = -1$, $m = 0$, $m = \infty$, then run Algorithm \ref{robustified two stage least squares} to produce three estimates $\widetilde{\beta}$, $\widehat{\beta}_{1.96}^{(1)}$, $\widehat{\beta}_{1.96}^{(\ast)}$. In Column $1$, $\widetilde{\beta}$ represents the full sample OLS (2SLS) estimate in Table \ref{OLS fixed effect regression and its robust estimates and outlier robustness tests} (Table \ref{2SLS fixed effect regression and its robust estimates and outlier robustness tests}). In Column $2$, we remove observations whose absolute values of standardized residuals are beyond the cut-off $1.96$ (in the second stage) and re-compute OLS (2SLS) to obtain the robust estimate $\widehat{\beta}_{1.96}^{(1)}$ in Table \ref{OLS fixed effect regression and its robust estimates and outlier robustness tests} (Table \ref{2SLS fixed effect regression and its robust estimates and outlier robustness tests}). In Column $3$, this procedure is iterated until the fixed point $\widehat{\beta}_{1.96}^{(\ast)}$ is attained. Values reported below the estimates in parentheses $()$ are standard errors\footnote{Acemoglu et al. (2019) also report standard errors robust against heteroskedasticity and serial correlation at the country level in their Table A$8$, A$13$.}. Note that in Columns $2$ and $3$ in Tables \ref{OLS fixed effect regression and its robust estimates and outlier robustness tests} and \ref{2SLS fixed effect regression and its robust estimates and outlier robustness tests} we report the adjusted standard errors\footnote{Acemoglu et al. (2019) do not take the adjustments to the standard errors of $\widehat{\beta}_{1.96}^{(1)}$ reported in Column $(2)$ in their Table A$8$, A$13$.} by applying the theory in \S \ref{valid inference on structural parameters beta}. In Table \ref{OLS fixed effect regression and its robust estimates and outlier robustness tests} (Table \ref{2SLS fixed effect regression and its robust estimates and outlier robustness tests}), ordinary standard errors are multiplied by $\sqrt{6044 / 6336}$ ($\sqrt{6019 / 6309}$) and $\sqrt{\iota_{1.96}^{(1)}}$ in Column $2$; and by $\sqrt{5213 / 6336}$ ($\sqrt{5202 / 6309}$) and $\sqrt{\iota_{1.96}^{(\ast)}}$ in Column $3$. Table \ref{adjusting factors for variance estimates for asymptotic variance for standard errors} in \S \ref{valid inference on structural parameters beta} shows $\iota_{1.96}^{(1)} = 1.612$ and $\iota_{1.96}^{(\ast)} = 1.828$, so in order to produce the reported standard errors in Table \ref{OLS fixed effect regression and its robust estimates and outlier robustness tests} (Table \ref{2SLS fixed effect regression and its robust estimates and outlier robustness tests}) we apply the adjusting factor $1.240$ ($1.240$) in Column $2$ and $1.226$ ($1.228$) in Column $3$.

In Tables \ref{OLS fixed effect regression and its robust estimates and outlier robustness tests} and \ref{2SLS fixed effect regression and its robust estimates and outlier robustness tests}, Columns $4 - 7$ present formal statistical tests for outlier robustness analysis. Columns $4 - 5$ give test results for the null hypothesis that the two parameters estimated by $\widetilde{\beta}$ and $\widehat{\beta}_{1.96}^{(1)}$ are identical. Columns $6 - 7$ present test results for the same hypothesis but comparing $\widetilde{\beta}$ with $\widehat{\beta}_{1.96}^{(\ast)}$. Columns $5$ and $7$ show the new Hausman type test statistics $\widehat{H}_{1.96}^{(1)}$ and $\widehat{H}_{1.96}^{(\ast)}$ that use the standard errors of the difference of two estimators $\widehat{\mathsf{avar}} (\widehat{\beta}_{1.96}^{(1)} - \widetilde{\beta})$ and $\widehat{\mathsf{avar}} (\widehat{\beta}_{1.96}^{(\ast)} - \widetilde{\beta})$\footnote{Consider $c = 1.96$ and $m = 0, \infty$ and note $\mathsf{avar}(\widehat{\beta}_{c}^{(m + 1)} - \widetilde{\beta}) = \mathsf{avar}(\widehat{\beta}_{c}^{(m + 1)}) - \mathsf{avar}(\widetilde{\beta})$. To construct the standard errors of the difference of two estimators, some empirical researchers estimate $\mathsf{avar}(\widetilde{\beta})$ using the full sample and estimate $\mathsf{avar}(\widehat{\beta}_{c}^{(m + 1)})$ using the sub-sample with all outliers removed. This produces the negative estimate of $\mathsf{avar}(\widehat{\beta}_{c}^{(m + 1)} - \widetilde{\beta})$ in this application as $\widehat{\mathsf{avar}}(\widehat{\beta}_{c}^{(m + 1)}) - \widehat{\mathsf{avar}}(\widetilde{\beta}) \le 0$ clearly shown in Tables \ref{OLS fixed effect regression and its robust estimates and outlier robustness tests} and \ref{2SLS fixed effect regression and its robust estimates and outlier robustness tests}, making the Hausman test meaningless. As suggested in \S \ref{a new Hausman type test of outlier robustness}, the estimator of $\mathsf{avar}(\widehat{\beta}_{c}^{(m + 1)} - \widetilde{\beta})$ should be robust itself to ensure the power performance under the alternative. In other words, we need to robustly estimate $\mathsf{avar}(\widetilde{\beta})$ using the sub-sample of all non-outlying observations.}, see Corollary \ref{Hausman test when m = 1 and m = infinite} in \S \ref{a new Hausman type test of outlier robustness} for construction of the Hausman test statistics. We first apply the two-sided type Hausman test with the normal limit for each individual scalar parameter $\beta_{0}$, $\beta_{1}$, $\ldots$, and $\beta_{4}$, then perform the one-sided test with the chi-squared limit for the whole parameter vector $(\beta_{0}, \beta_{1}, \ldots, \beta_{4})^{\prime}$. Columns $4$ and $6$ show the heuristic test statistics $\widehat{\mathrm{Heu}}_{1.96}^{(1)}$ and $\widehat{\mathrm{Heu}}_{1.96}^{(\ast)}$ for the same hypothesis, which however use the standard error of the marginal distribution of the baseline full sample estimate $\widetilde{\beta}$. P-values are reported in brackets $[]$ below test statistics.

\begin{table}[ht!]
\centering
\begin{tabular}{lcccccccc}
\toprule
& & \multicolumn{3}{c}{OLS Estimate} & \multicolumn{2}{c}{Test: $\widehat{\beta}_{1.96}^{(1)}$, $\widetilde{\beta}$} & \multicolumn{2}{c}{Test: $\widehat{\beta}_{1.96}^{(\ast)}$, $\widetilde{\beta}$}
\\
\cmidrule(lr){3-5} \cmidrule(lr){6-7} \cmidrule(lr){8-9}
& $\beta$ & $\widetilde{\beta}$ & $\widehat{\beta}_{1.96}^{(1)}$ & $\widehat{\beta}_{1.96}^{(\ast)}$ & $\widehat{\mathrm{Heu}}_{1.96}^{(1)}$ & $\widehat{H}_{1.96}^{(1)}$ & $\widehat{\mathrm{Heu}}_{1.96}^{(\ast)}$ & $\widehat{H}_{1.96}^{(\ast)}$ \\
\midrule
Democracy & $\beta_{0}$ & $\underset{(0.228)}{0.787}$ & $\underset{(0.186)}{0.556}$ & $\underset{(0.129)}{0.142}$ & $\underset{[0.312]}{-1.011}$ & $\underset{[0.004]}{-2.908}$ & $\underset{[0.005]}{-2.821}$ & $\underset{[0.000]}{-9.421}$ \\
GDP 1st lag & $\beta_{1}$ & $\underset{(0.013)}{1.238}$ & $\underset{(0.011)}{1.226}$ & $\underset{(0.009)}{1.264}$ & $\underset{[0.321]}{-0.992}$ & $\underset{[0.011]}{-2.537}$ & $\underset{[0.037]}{2.088}$ & $\underset{[0.000]}{5.501}$  \\
GDP 2nd lag & $\beta_{2}$ & $\underset{(0.020)}{-0.207}$ & $\underset{(0.018)}{-0.198}$ & $\underset{(0.013)}{-0.235}$ & $\underset{[0.664]}{0.435}$ & $\underset{[0.254]}{1.142}$ & $\underset{[0.145]}{-1.456}$ & $\underset{[0.000]}{-4.081}$  \\
GDP 3rd lag & $\beta_{3}$ & $\underset{(0.019)}{-0.026}$ & $\underset{(0.016)}{-0.027}$ & $\underset{(0.012)}{-0.030}$ & $\underset{[0.970]}{-0.038}$ & $\underset{[0.917]}{-0.104}$ & $\underset{[0.843]}{-0.198}$ & $\underset{[0.563]}{-0.579}$  \\
GDP 4th lag & $\beta_{4}$ & $\underset{(0.012)}{-0.043}$ & $\underset{(0.010)}{-0.030}$ & $\underset{(0.008)}{-0.019}$ & $\underset{[0.277]}{1.088}$ & $\underset{[0.003]}{2.975}$ & $
\underset{[0.047]}{1.985}$ & $\underset{[0.000]}{5.741}$ \\
\midrule
Whole vector & & & & & $\underset{[0.003]}{18.347}$ & $\underset{[0.000]}{140.514}$ & $\underset{[0.000]}{62.262}$ & $\underset{[0.000]}{645.087}$ \\
Sample size & & $6336$ & $6044$ & $5213$ & & & & \\
\bottomrule
\end{tabular}
\caption{This table presents the OLS fixed effect estimates of the simple regression exploring the effects of democracy and four lags of log GDP per capita on log GDP per capita. Then, formal statistical tests are shown to achieve outlier robustness analysis. Columns $1-3$ represent OLS estimates with standard errors reported below in parentheses $()$. The cut-off is chosen as $c = 1.96$, then Algorithm \ref{robustified two stage least squares} is iterated at the step $m = -1$, $m = 0$, $m = \infty$ to produce three estimates $\widetilde{\beta}$, $\widehat{\beta}_{1.96}^{(1)}$, $\widehat{\beta}_{1.96}^{(\ast)}$. In Column $1$ we show the baseline full sample OLS. In Column $2$ we remove observations with a standardized residual estimated above 1.96 or below -1.96. In Column $3$ we iterate the above procedure until reaching the fixed point. In Columns $2-3$ standard errors of $\widehat{\beta}_{1.96}^{(1)}$ and $\widehat{\beta}_{1.96}^{(\ast)}$ are adjusted by the theory in \S \ref{valid inference on structural parameters beta}. Columns $4-5$ represent outlier robustness tests comparing $\widetilde{\beta}$ with $\widehat{\beta}_{1.96}^{(1)}$ and Columns $6-7$ comparing $\widetilde{\beta}$ with $\widehat{\beta}_{1.96}^{(\ast)}$. Columns $4$ and $6$ show heuristic test statistics $\widehat{\mathrm{Heu}}_{1.96}^{(1)}$ and $\widehat{\mathrm{Heu}}_{1.96}^{(\ast)}$ that use the standard errors of $\widetilde{\beta}$. Columns $5$ and $7$ show the new Hausman type test statistics $\widehat{H}_{1.96}^{(1)}$ and $\widehat{H}_{1.96}^{(\ast)}$ that use the standard errors of $\widehat{\beta}_{1.96}^{(1)} - \widetilde{\beta}$ and $\widehat{\beta}_{1.96}^{(\ast)} - \widetilde{\beta}$ developed in \S \ref{a new Hausman type test of outlier robustness}. P-values are reported in brackets $[]$ below test statistics.}
\label{OLS fixed effect regression and its robust estimates and outlier robustness tests}
\end{table}

\begin{table}[ht!]
\centering
\begin{tabular}{lcccccccc}
\toprule
& & \multicolumn{3}{c}{2SLS Estimate} & \multicolumn{2}{c}{Test: $\widehat{\beta}_{1.96}^{(1)}$, $\widetilde{\beta}$} & \multicolumn{2}{c}{Test: $\widehat{\beta}_{1.96}^{(\ast)}$, $\widetilde{\beta}$}
\\
\cmidrule(lr){3-5} \cmidrule(lr){6-7} \cmidrule(lr){8-9}
& $\beta$ & $\widetilde{\beta}$ & $\widehat{\beta}_{1.96}^{(1)}$ & $\widehat{\beta}_{1.96}^{(\ast)}$ & $\widehat{\mathrm{Heu}}_{1.96}^{(1)}$ & $\widehat{H}_{1.96}^{(1)}$ & $\widehat{\mathrm{Heu}}_{1.96}^{(\ast)}$ & $\widehat{H}_{1.96}^{(\ast)}$ \\
\midrule
Democracy & $\beta_{0}$ & $\underset{(0.607)}{1.149}$ & $\underset{(0.487)}{0.941}$ & $\underset{(0.332)}{0.605}$ & $\underset{[0.732]}{-0.343}$ & $\underset{[0.317]}{-1.000}$ & $\underset{[0.370]}{-0.896}$ & $\underset{[0.002]}{-3.099}$ \\
GDP 1st lag & $\beta_{1}$ & $\underset{(0.013)}{1.238}$ & $\underset{(0.012)}{1.224}$ & $\underset{(0.009)}{1.258}$ & $\underset{[0.270]}{-1.103}$ & $\underset{[0.005]}{-2.820}$ & $\underset{[0.126]}{1.531}$ & $\underset{[0.000]}{4.059}$  \\
GDP 2nd lag & $\beta_{2}$ & $\underset{(0.020)}{-0.205}$ & $\underset{(0.018)}{-0.194}$ & $\underset{(0.013)}{-0.225}$ & $\underset{[0.573]}{0.564}$ & $\underset{[0.138]}{1.483}$ & $\underset{[0.336]}{-0.961}$ & $\underset{[0.007]}{-2.702}$  \\
GDP 3rd lag & $\beta_{3}$ & $\underset{(0.019)}{-0.029}$ & $\underset{(0.017)}{-0.032}$ & $\underset{(0.013)}{-0.036}$ & $\underset{[0.887]}{-0.142}$ & $\underset{[0.697]}{-0.389}$ & $\underset{[0.727]}{-0.349}$ & $\underset{[0.307]}{-1.022}$  \\
GDP 4th lag & $\beta_{4}$ & $\underset{(0.012)}{-0.040}$ & $\underset{(0.010)}{-0.027}$ & $\underset{(0.008)}{-0.016}$ & $\underset{[0.247]}{1.158}$ & $\underset{[0.002]}{3.157}$ & $
\underset{[0.045]}{2.001}$ & $\underset{[0.000]}{5.797}$ \\
\midrule
Whole vector & & & & & $\underset{[0.003]}{18.215}$ & $\underset{[0.000]}{138.445}$ & $\underset{[0.000]}{56.027}$ & $\underset{[0.000]}{571.720}$ \\
Sample size & & $6309$ & $6019$ & $5202$ & & & & \\
\bottomrule
\end{tabular}
\caption{This table presents the 2SLS fixed effect estimates of the IVs regression exploring the causal effects of democracy and four lags of log GDP per capita on log GDP per capita. Then, formal statistical tests are shown to achieve outlier robustness analysis. Columns $1-3$ represent 2SLS estimates with standard errors reported below in parentheses $()$. The cut-off is chosen as $c = 1.96$, then Algorithm \ref{robustified two stage least squares} is iterated at the step $m = -1$, $m = 0$, $m = \infty$ to produce three estimates $\widetilde{\beta}$, $\widehat{\beta}_{1.96}^{(1)}$, $\widehat{\beta}_{1.96}^{(\ast)}$. In Column $1$ we show the baseline full sample 2SLS. In Column $2$ we remove observations with a standardized residual estimated above 1.96 or below -1.96 in the second stage. In Column $3$ we iterate the above procedure until reaching the fixed point. In Columns $2-3$ standard errors of $\widehat{\beta}_{1.96}^{(1)}$ and $\widehat{\beta}_{1.96}^{(\ast)}$ are adjusted by the theory in \S \ref{valid inference on structural parameters beta}. Columns $4-5$ represent outlier robustness tests comparing $\widetilde{\beta}$ with $\widehat{\beta}_{1.96}^{(1)}$ and Columns $6-7$ comparing $\widetilde{\beta}$ with $\widehat{\beta}_{1.96}^{(\ast)}$. Columns $4$ and $6$ show heuristic test statistics $\widehat{\mathrm{Heu}}_{1.96}^{(1)}$ and $\widehat{\mathrm{Heu}}_{1.96}^{(\ast)}$ that use the standard errors of $\widetilde{\beta}$. Columns $5$ and $7$ show the new Hausman type test statistics $\widehat{H}_{1.96}^{(1)}$ and $\widehat{H}_{1.96}^{(\ast)}$ that use the standard errors of $\widehat{\beta}_{1.96}^{(1)} - \widetilde{\beta}$ and $\widehat{\beta}_{1.96}^{(\ast)} - \widetilde{\beta}$ developed in \S \ref{a new Hausman type test of outlier robustness}. P-values are reported in brackets $[]$ below test statistics.}
\label{2SLS fixed effect regression and its robust estimates and outlier robustness tests}
\end{table}

We see from Tables \ref{OLS fixed effect regression and its robust estimates and outlier robustness tests} and \ref{2SLS fixed effect regression and its robust estimates and outlier robustness tests} that $\beta_{0}$ on democracy is estimated by 2SLS to have the larger value than OLS fixed effect regressions while the effect of GDP dynamics remains roughly unchanged, which is consistent with our expectation that the IVs strategy affects mostly the coefficient of the endogenous variable democracy. The standard errors of the estimates $\widetilde{\beta}$, $\widehat{\beta}_{1.96}^{(1)}$, $\widehat{\beta}_{1.96}^{(\ast)}$ decrease with iteration, suggesting that it is under the alternative hypothesis where the data is contaminated by outlying observations driving up the standard error of the full sample estimate $\widetilde{\beta}$. Otherwise under the null of no outliers standard errors of $\widehat{\beta}_{1.96}^{(1)}$, $\widehat{\beta}_{1.96}^{(\ast)}$ should become larger than $\widetilde{\beta}$, since by gaining robustness the estimates $\widehat{\beta}_{1.96}^{(1)}$, $\widehat{\beta}_{1.96}^{(\ast)}$ lose efficiency relative to $\widetilde{\beta}$, leading to the larger asymptotic variance. The number of observations used in the OLS and 2SLS estimates are $6336$ and $6309$. Even when the model contains no outliers, Algorithm \ref{robustified two stage least squares} throws away around $5 \%$ of the whole observations with a large probability for the cut-off value $1.96$ to produce $\widehat{\beta}_{1.96}^{(1)}$ and $\widehat{\beta}_{1.96}^{(\ast)}$. Thus, under the null of no outliers the expected number of observations retained by Algorithm \ref{robustified two stage least squares} should be around $6019$ and $5994$ that are sub-sample sizes of $95 \%$ of the full OLS and 2SLS sample. We notice that the observed number of observations retained by $\widehat{\beta}_{1.96}^{(\ast)}$ are much less than the expected number for both of the OLS and 2SLS fixed effect regressions, although the sub-sample size of $\widehat{\beta}_{1.96}^{(1)}$ for OLS and 2SLS are not far away from the expected number. This fact again indicates that the model does contain outliers.

\begin{figure}[t]
	\centering
	\includegraphics[width=\textwidth, keepaspectratio]{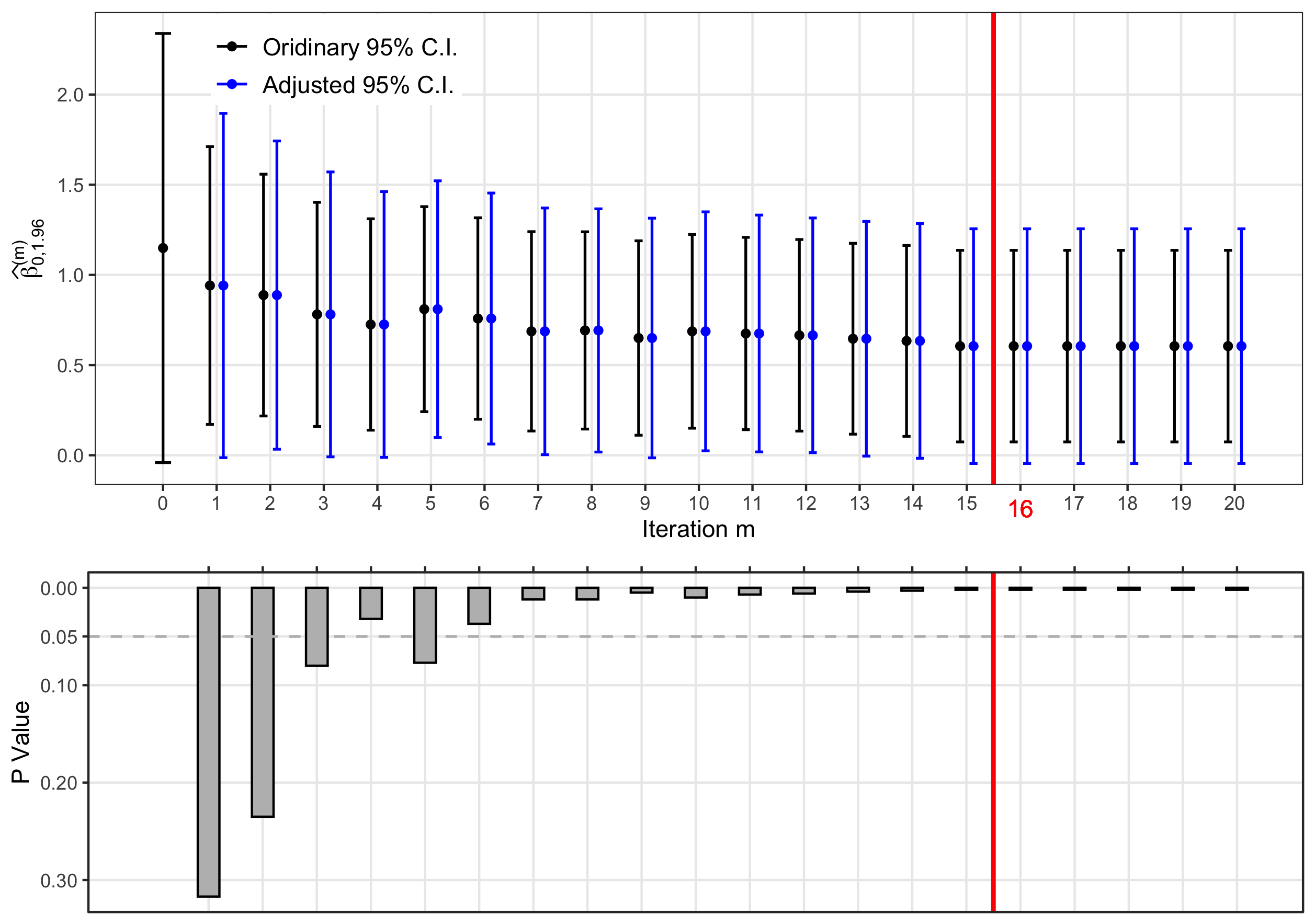}
	\caption{Estimated IV regression coefficients on democracy and the Hausman-type test results over iterations with the chosen cut-off value $1.96$. The upper panel reports point estimates, ordinary (black) and adjusted (blue) $95 \%$ confidence intervals. The lower panel reports p-values of the tests of comparing between the full-sample and the outlier removal estimators. The iterated procedure converges at the iteration step $16$.}
\label{IV regression on democracy over iterations with the cut-off 1.96}
\end{figure}

By only computing point estimates $\widetilde{\beta}$, $\widehat{\beta}_{1.96}^{(1)}$, $\widehat{\beta}_{1.96}^{(\ast)}$ for the OLS or 2SLS fixed effect regressions without applying formal statistical tests of outlier robustness, we cannot draw any conclusion on whether or not the empirical findings are driven by outliers, since the information contained in the standard errors also needs to be considered in order to evaluate the distance between ordinary and robust estimates. For either the simple regression in Table \ref{OLS fixed effect regression and its robust estimates and outlier robustness tests} or IVs specification in Table \ref{2SLS fixed effect regression and its robust estimates and outlier robustness tests}, both of the heuristic and Hausman tests reject the null hypothesis that the parameter vector $\beta$ estimated by $\widetilde{\beta}$ is identical to the robust estimate $\widehat{\beta}_{1.96}^{(1)}$ or $\widehat{\beta}_{1.96}^{(\ast)}$. It is obvious to note from Columns $4 - 7$ in Tables \ref{OLS fixed effect regression and its robust estimates and outlier robustness tests} and \ref{2SLS fixed effect regression and its robust estimates and outlier robustness tests} that the Hausman tests yield strictly the lower p-values than the heuristic tests, which is concordant with our theory in \S \ref{a new Hausman type test of outlier robustness} that the Hausman test has the better power performance than the heuristics. In the simple regressions, the Hausman tests reject the null of outlier robustness at the $5 \%$ significance level for most coefficients except $\beta_{2}$, $\beta_{3}$ when compare $\widetilde{\beta}$ with $\widehat{\beta}_{1.96}^{(1)}$; except $\beta_{3}$ when compare $\widetilde{\beta}$ with $\widehat{\beta}_{1.96}^{(\ast)}$. In the IVs regressions, the Hausman tests mostly give the same result but the identity of $\beta_{0}$ is not rejected by comparing $\widetilde{\beta}$ with $\widehat{\beta}_{1.96}^{(1)}$.

 \begin{table}[ht!]
\centering
\begin{tabular}{llcccc}
\toprule
& & Est. & Std. Err.(A) & $95 \%$ Confidence Interval(A) & P-value(A) \\
\midrule
\multirow{3}{*}{OLS} & $\widetilde{\beta}_{0}$ & $0.787$ & $0.228$ & $[0.339, 1.234]$ & $0.001$ \\
& $\widehat{\beta}_{0, 1.96}^{(1)}$ & $0.556$ & $0.150 (0.186)$ & $[0.262, 0.850] ([0.191, 0.920])$ & $0.000 (0.003)$ \\
& $\widehat{\beta}_{0, 1.96}^{(\ast)}$ & $0.142$ & $0.106 (0.129)$ & $[-0.064, 0.349] ([-0.111, 0.396])$ & $0.177 (0.271)$ \\
\midrule
\multirow{3}{*}{2SLS} & $\widetilde{\beta}_{0}$ & $1.149$ & $0.607$ & $[-0.041, 2.340]$ & $0.058$ \\
& $\widehat{\beta}_{0, 1.96}^{(1)}$ & $0.941$ & $0.393 (0.487)$ & $[0.170, 1.712] ([-0.014, 1.897])$ & $0.017 (0.054)$ \\
& $\widehat{\beta}_{0, 1.96}^{(\ast)}$ & $0.605$ & $0.271 (0.332)$ & $[0.074, 1.136] ([-0.046, 1.257])$ & $0.025 (0.069)$ \\
\bottomrule
\end{tabular}
\caption{This table concentrates on $\beta_{0}$ the effect of democracy on log GDP per capita. We present estimates of $\beta_{0}$, standard errors, $95 \%$ confidence intervals, and p-values of the $t$-test for the null hypothesis $\beta_{0} = 0$. To compare standard errors and inference with and without the adjustments, we first report the ordinary standard errors and inference, followed by in parentheses $()$ the ones adjusted by the theory in \S \ref{valid inference on structural parameters beta}. Rows $1-3$ represent OLS fixed effect regressions and Rows $4-6$ 2SLS fixed effect regressions. The cut-off is chosen as $c = 1.96$, then Algorithm \ref{robustified two stage least squares} is iterated at the step $m = -1$, $m = 0$, $m = \infty$ to produce three estimates $\widetilde{\beta}_{0}$, $\widehat{\beta}_{0, 1.96}^{(1)}$, $\widehat{\beta}_{0, 1.96}^{(\ast)}$ of $\beta_{0}$ in simple or IVs regressions. In Row $1$ we show the baseline full sample ordinary least squares, and instrument democracy in Row $4$. In Row $2$ we remove observations with a standardized residual estimated above 1.96 or below -1.96, and the same outlier selection is based on residuals in the second stage in Row $5$. In Rows $3$ and $6$ we iterate the above procedure until reaching the fixed point.}
\label{unadjusted and valid inference for democracy by different robust OLS and 2SLS methods}
\end{table}

Kaji (2018) also re-examines outlier robustness of the results found in Acemoglu et al. (2019). Without the weak limit of the Hausman type test statistics derived in our paper, he carries out the non-parametric bootstrap to draw the distribution of the $L_{1}$ norm of the difference between the baseline full sample estimate and the robust estimate. The bootstrap is implemented by randomly sampling countries $i$ with replacement. Each draw of country $i$ adds its time dimensional number of observations to the bootstrap sample. The bootstrap consists of $10000$ iterations. By comparing between the baseline OLS estimate $\widetilde{\beta}$ and outlier-removed estimate $\widehat{\beta}_{1.96}^{(1)}$, the bootstrap does not reject the null that two estimates are identical, indicating robustness to outliers of the empirical findings in Acemoglu et al. (2019). The reason why the null of outlier robustness is not rejected by Kaji (2018) could be the low power issue faced by the bootstrap under alternatives; see discussion in the last paragraph in \S \ref{a new Hausman type test of outlier robustness}. Note that Kaji (2018) only applies the tests for checking the difference for each coefficient $\beta_{0}, \beta_{1}, \ldots, \beta_{4}$, but not for the whole vector $\beta$.


Figure \ref{IV regression on democracy over iterations with the cut-off 1.96} plots the estimated IV regression results on democracy over iterations (with the cut-off value $1.96$). The upper panel reports the point estimates, the ordinary and adjusted $95 \%$ confidence intervals. The lower panel reports the p-values of the Hausman-type tests of comparing between the full-sample and outlier removal estimates. The red line shows that the iterated procedure converges at the iteration step $16$. We can find that the p-values of the tests become smaller (lower than the conventional level $ 0.05$) with increasing iterations. It indicates that we should reply on the trimmed estimates to conduct inference. The outlier removal estimates become closer to zero and the adjusted $95 \%$ confidence intervals contain zero for large iteration numbers.

Table \ref{unadjusted and valid inference for democracy by different robust OLS and 2SLS methods} focuses on $\beta_{0}$ the effect of democracy on log GDP per capita by summarizing regression results from Tables \ref{OLS fixed effect regression and its robust estimates and outlier robustness tests} and \ref{2SLS fixed effect regression and its robust estimates and outlier robustness tests}. We present estimates of $\beta_{0}$ and compare standard errors with and without the adjustments together with their corresponding inference including $95 \%$ confidence intervals and p-values of the $t$-test for the null hypothesis that $\beta_{0} = 0$. The adjusted standard errors and inference results are reported in parentheses $()$ after the ordinary ones.

We see from Table \ref{unadjusted and valid inference for democracy by different robust OLS and 2SLS methods} that the modified standard errors are uniformly larger than those without taking the adjustments, so the adjusted confidence bands are uniformly wider and p-values are uniformly larger; see the theory in \S \ref{valid inference on structural parameters beta}. Recall that the baseline full sample and robust estimates of $\beta_{0}$ are statistically different, indicated by the Hausman tests. In the OLS regressions, inference based on $\widehat{\beta}_{0, 1.96}^{(\ast)}$ shows that $\beta_{0}$ is not significantly distinct from $0$. In the 2SLS regressions, we cannot reject $\beta_{0} = 0$ at the $5 \%$ significance level by performing the adjusted inference either based on the baseline estimate $\widetilde{\beta}$ or the robust estimates $\widehat{\beta}_{0, 1.96}^{(1)}$ and $\widehat{\beta}_{0, 1.96}^{(\ast)}$. Thus, we cannot find the evidence robustly supporting the result $\beta_{0} > 0$ that the effect of democracy is strictly positive on economic growth.





\section{Weighted and marked empirical process} \label{weighted and marked empirical process}
Consider the weighted and marked empirical distribution function
\begin{equation} \label{one-sided weighted and marked empirical distribution function}
\widehat{\mathsf{F}}_{u, n}^{w, p}(a, b, c) = \frac{1}{n} \sum_{i=1}^{n} w_{in} u_{i}^{p} 1_{(u_{i} \le \sigma c +n^{-1/2}a c + z_{in}^{\prime} \Pi b + n^{-1/2}r_{i}^{\prime}b)},
\end{equation}
with $\mathcal{F}_{i - 1}$ adapted weights $w_{in}$ and $\mathcal{F}_{i}$ measurable marks $u_{i}^{p}$. The filtration $\mathcal{F}_{i - 1}$ is generated by $(z_{1}, \ldots, z_{i}, r_{1}, \ldots, r_{i - 1}, u_{1}, \ldots, u_{i - 1})$ so $z_{in} \in \mathcal{F}_{i - 1}$ while $u_{i}, r_{i} \in \mathcal{F}_{i}$ are independent of $\mathcal{F}_{i - 1}$. The observable data $\{(y_{i}, x_{i}, z_{i})\}_{i = 1}^{n}$ has i.i.d. structure in this paper, so $a \in\mathbb{R}$, $b \in \mathbb{R}^{d_{x}}$ represent normalized estimation errors $\widetilde{a} = n^{1/2} (\widetilde{\sigma} - \sigma)$, $\widetilde{b} = n^{1/2} (\widetilde{\beta} - \beta)$, and $c \in \mathbb{R}$ is the quantile, while $\sigma$, $\Pi$, $\Sigma$ are true parameters for the variance of the structural error, for the location coefficient in the first stage regression, and for the variance of the first stage error respectively. Let normalized instruments $z_{in} = n^{-1/2} z_{i}$ such that $M_{z z, n} = \sum_{i=1}^{n} z_{in} z_{in}^{\prime} = n^{-1} \sum_{i=1}^{n} z_{i} z_{i}^{\prime}$ converges in probability to $\mathsf{E} z_{i} z_{i}^{\prime} = M_{z z}$ by Law of Large Numbers when $\mathsf{E} |z_{i}|^{2} < \infty$. Our interest focuses on weights $w_{in}$ given as either of $1$, $n^{1/2} z_{in} = z_{i}$, $n z_{in}z_{in}^{\prime} = z_{i} z_{i}^{\prime}$ and $p$ as either of $0$, $1$, $2$. To form the empirical process, introduce the compensator
\begin{equation} \label{one-sided weighted and marked compensator}
\overline{\mathsf{F}}_{u, n}^{w, p}(a, b, c) = \frac{1}{n} \sum_{i=1}^{n} w_{in} \mathsf{E}_{i-1} u_{i}^{p} 1_{(u_{i} \le \sigma c +n^{-1/2}a c + z_{in}^{\prime} \Pi b + n^{-1/2}r_{i}^{\prime}b)},
\end{equation}
where $\mathsf{E}_{i-1}(\cdot) = \mathsf{E}( \cdot | \mathcal{F}_{i-1})$. Note that $\overline{\mathsf{F}}_{u, n}^{1, 0}(0,0,c) = \mathsf{F}_{u}(c) = \mathsf{P}(u_{i} \le \sigma c)$.

We embed these processes into the space $D[0,1]$ that are processes continuous from the right and with limits of left, where the space is endowed with the Skorokhod metric. We do this as follows. The indicator $1_{(u_i \le \sigma c)}$ and the distribution function
$\mathsf{F}_{u}(c)$ can be defined as 0 or 1 when $c$ takes the values $-\infty$ and $\infty$ respectively. We can then define quantiles $c_\psi=\mathsf{F}_{u}^{-1}(\psi)$ for $0 \le \psi \le 1$. Correspondingly we can continously extend the definition of the weighted and marked empirical distribution function and its compensator by chosing $\widehat{\mathsf{F}}_{u, n}^{w, p}(a, b, -\infty) =\overline{\mathsf{F}}_{u, n}^{w, p}(a, b, -\infty) = 0$ while we then have $\widehat{\mathsf{F}}_{u, n}^{w, p}(a, b, \infty) = n^{-1}\sum_{i=1}^{n} w_{in} u_{i}^{p}$ and $\overline{\mathsf{F}}_{u, n}^{w, p}(a, b, \infty) = n^{-1} \sum_{i=1}^{n} w_{in}\mathsf{E}_{i-1} u_{i}^{p}$. We now define the empirical process for $0 \le \psi \le 1$
\begin{equation} \label{one-sided weighted and marked empirical process}
\mathbb{F}_{u, n}^{w, p}(a, b, c_\psi) = n^{1/2} \{ \widehat{\mathsf{F}}_{u, n}^{w, p}(a, b, c_\psi) - \overline{\mathsf{F}}_{u, n}^{w, p}(a, b, c_\psi) \}.
\end{equation}

In the following we first present assumptions, and then results follow for the one-sided process and finally extend to the absolute case.

\subsection{Assumptions}
In this section, the density $\mathsf{f}_{u}$ is not necessarily symmetric. The conditions listed below are weaker than Assumption \ref{sufficient assumptions}.

\begin{assumption} \label{explicit assumptions}
Let $\mathcal{F}_{i - 1} = \sigma(z_{1}, \ldots, z_{i}, r_{1}, \ldots, r_{i - 1}, u_{1}, \ldots, u_{i - 1})$ be an increasing sequence of $\sigma$-fields so $u_{i-1}$, $r_{i - 1}$, $z_{i}$, $w_{in}$ are $\mathcal{F}_{i-1}$ measurable while $r_{i}$, $u_{i}$ are independent of $\mathcal{F}_{i - 1}$. Suppose $(u_{i}/\sigma, \Sigma^{-1/2} r_{i})$ have continuously differentiable joint, conditional, and marginal densities $\mathsf{f}_{u, r}(y, x) = \mathsf{f}_{u|r}(y|x) \mathsf{f}_{r}(x) = \mathsf{f}_{r|u}(x|y) \mathsf{f}_{u}(y)$ which are positive on $y \in \mathbb{R}$, $x \in \mathbb{R}^{d_{x}}$. Let $p, \eta, \kappa$ be given so $p \in \mathbb{N}_{0}$, $0 \le \kappa < \eta \le 1/4$. Choose $0 < \nu \le 1$, $s \in \mathbb{N}_{0}$ such that
\begin{equation} \label{moment condition for b, c uniform convergence}
2^{s - 1} > 1 + (1/4 - \eta)(1 + d_{x}).
\end{equation}
Suppose \\
$(i)$ the marginal density $\mathsf{f}_{u}(y)$ satisfies for $y \in \mathbb{R}$ \\
\indent $(a)$ moments: $\int_{-\infty}^{\infty} |y|^{2^{s}p / \nu} \mathsf{f}_{u}(y) dy < \infty$; \\
\indent $(b)$ boundedness: $\sup_{y \in \mathbb{R}} |y^{2^{s} p + 1} \mathsf{f}_{u}(y) + y^{2^{s} p + 2} \dot{\mathsf{f}}_{u}(y)| < \infty$; \\
\indent $(c)$ smoothness: a $C_{\mathrm{H}} \in \mathbb{N}$ exists so that for all $\epsilon > 0$
\begin{equation*}
\frac{\sup_{y \ge \epsilon} (1 + y^{2^{s}p})\mathsf{f}_{u}(y)}{\inf_{0 \le y \le \epsilon}(1 + y^{2^{s}p})\mathsf{f}_{u}(y)} \le C_{\mathrm{H}}, \qquad
\frac{\sup_{y \le -\epsilon} (1 + y^{2^{s}p})\mathsf{f}_{u}(y)}{\inf_{-\epsilon \le y \le 0}(1 + y^{2^{s}p})\mathsf{f}_{u}(y)} \le C_{\mathrm{H}};
\end{equation*}
$(ii)$ the marginal density $\mathsf{f}_{r}(x)$ satisfies for $x \in \mathbb{R}^{d_{x}}$ \\
\indent $(a)$ moments: $\int_{x \in \mathbb{R}^{d_{x}}} |x|^{4} \mathsf{f}_{r}(x) (dx) < \infty$; \\
$(iii)$ the joint and conditional densities $\mathsf{f}_{u, r}(y, x), \mathsf{f}_{u|r}(y|x)$ satisfy for $y \in \mathbb{R}, x \in \mathbb{R}^{d_{x}}$ \\
\indent $(a)$ boundedness: $\sup_{y \in \mathbb{R}, x \in \mathbb{R}^{d_{x}}} |(1 + y) y^{2^{s} p - 1} \mathsf{f}_{u|r}(y|x) + y^{2^{s} p} \dot{\mathsf{f}}_{u|r, y}(y|x)| < \infty$; \\
$(iv)$ the instruments $z_{i}$ satisfy \\
\indent $(a)$ $\max_{1 \le i \le n} |n^{1/2 - \kappa} z_{in}| = \mathrm{O}_{\mathsf{P}}(1)$; \\
$(v)$ the weights $w_{in}$ satisfy \\
\indent $(a)$ $n^{-1} \mathsf{E} \sum_{i = 1}^{n} |w_{in}|^{2^{s}} (1 + |n^{1/2} z_{in}|) = \mathrm{O}(1)$; \\
\indent $(b)$ $n^{-1} \sum_{i = 1}^{n} |w_{in}| (1 + |n^{1/2} z_{in}|^{2}) = \mathrm{O}_{\mathsf{P}}(1)$.
\end{assumption}

\begin{remark} \label{discussion and the relationship between sufficient and explicit assumptions}
Assumption \ref{sufficient assumptions}$(ia, iia, iiia, ivb, ivc)$ implies Assumption \ref{explicit assumptions} with $s \ge 2$ satisfying (\ref{moment condition for b, c uniform convergence}) when $w_{in}$ is either of $1$, $n^{1/2} z_{in} = z_{i}$, $n z_{in} z_{in}^{\prime} = z_{i} z_{i}^{\prime}$ and $p$ is either of $0$, $1$, $2$. Details are given in Lemma \ref{sufficient assumptions imply explicit assumptions} in the appendix.
\end{remark}

\subsection{Asymptotic results for empirical process}
We present three asymptotic results. The first theorem shows that the estimation error for the scale and regression parameters is negligible uniformly in the quantile.

\begin{theorem} \label{one-sided empirical process result for all additive and multiplicative shifts}
Let $c_{\psi} = \mathsf{F}_{u}^{-1}(\psi)$. Suppose Assumption \ref{explicit assumptions} holds with $\nu = 1$ and $s \ge 2$ satisfying (\ref{moment condition for b, c uniform convergence}). Then for any $B > 0$ and as $n \to \infty$
\begin{equation*}
\sup_{0 \le \psi \le 1} \sup_{|a|, |b| \le n^{1/4 - \eta}B} |\mathbb{F}_{u, n}^{w, p}(a, b, c_{\psi}) - \mathbb{F}_{u, n}^{w, p}(0, 0, c_{\psi})| = \mathrm{o}_{\mathsf{P}}(1).
\end{equation*}
\end{theorem}

The proof involves a chaining argument. For this, we apply an iterated martingale inequality, see Lemma \ref{iterated martingale inequality for asymptotic result}, to explore the tail behaviour of the maximum of a family of martingales. We can, however, split this argument into two parts. First, we keep $b$ fixed and consider variation in $a, c_{\psi}$. Then we keep $a$ fixed and consider variation in $b, c_{\psi}$. Combining these two arguments will finally finish the proof.


The second result provides a linearization of the compensator.

\begin{theorem} \label{linearization of one-sided empirical compensator}
Let $c_{\psi} = \mathsf{F}_{u}^{-1}(\psi)$. Suppose Assumption \ref{explicit assumptions}$(ia, ib, iia, iiia, vb)$ holds with $\nu = 1$, $s = 0$. Then for any $B > 0$ and as $n \to \infty$
\begin{equation*}
\sup_{0 \le \psi \le 1} \sup_{|a|, |b| \le n^{1/4 - \eta}B} |n^{1/2} \{ \overline{\mathsf{F}}_{u, n}^{w, p}(a, b, c_{\psi}) - \overline{\mathsf{F}}_{u, n}^{w, p}(0, 0, c_{\psi}) \} - \mathcal{B}_{\mathsf{F}_{u}, n}(a, b, c_{\psi})| = \mathrm{O}_{\mathsf{P}}(n^{-2\eta}),
\end{equation*}
where the bias term, with $\xi_{c_{\psi}} = \mathsf{E} (\Sigma^{-1/2} r_{i} | u_{i}/\sigma = c_{\psi})$, is defined as
\begin{equation*}
\mathcal{B}_{\mathsf{F}_{u}, n}(a, b, c_{\psi}) = \sigma^{p - 1} c_{\psi}^{p} \mathsf{f}_{u}(c_{\psi}) n^{-1/2} \sum_{i=1}^{n} w_{in} ( n^{-1/2} a c_{\psi} + n^{-1/2} \xi_{c_{\psi}}^{\prime} \Sigma^{1/2} b + z_{in}^{\prime} \Pi b ).
\end{equation*}
\end{theorem}

Since $x_{i}$ and $u_{i}$ are correlated so for $u_{i}$ and $r_{i}$, this dependent structure requires a more intricate analysis for the compensator. Similar as the proof for uniform convergence in Theorem \ref{one-sided empirical process result for all additive and multiplicative shifts}, we proceed by first considering variation in $a, c_{\psi}$ while $b = 0$. We then approximate the compensator by combining with the argument for variation in $b, c_{\psi}$ while $a = 0$.

Finally, we argue the empirical process $\mathbb{F}_{u, n}^{w, p}(0, 0, c_{\psi})$ is tight when viewed as a sequence in $n$ of processes on $D[0, 1]$. By Billingsley (1968, Theorem 13.2), two conditions need to be checked. First, it holds by construction that $\mathbb{F}_{u, n}^{w, p}(0, 0, c_{0}) = 0$ where $c_{0} = \mathsf{F}_{u}^{-1}(0) = -\infty$. Second, the next theorem shows that the modulus of continuity is small. The proof uses a dyadic argument and then apply an iterated martingale exponential inequality, see Lemma \ref{iterated martingale inequality for tightness}, to explore the tail probability of the maximum of a family of martingales.

\begin{theorem} \label{tightness result for one-sided empirical process}
Let $c_{\psi} = \mathsf{F}_{u}^{-1}(\psi)$. Suppose Assumption \ref{explicit assumptions}$(ia, va)$ holds with $0 < \nu < 1$, $s = 2$. Then for all $\epsilon > 0$
\begin{equation*}
\lim_{\phi \downarrow 0} \limsup_{n \to \infty} \mathsf{P} \{ \sup_{0 \le \psi \le \psi^{\dagger} \le 1 : \psi^{\dagger} - \psi \le \phi} |\mathbb{F}_{u, n}^{w, p}(0, 0, c_{\psi^{\dagger}}) - \mathbb{F}_{u, n}^{w, p}(0, 0, c_{\psi})| > \epsilon \} \to 0.
\end{equation*}
\end{theorem}

The empirical distribution functions $\widehat{\mathsf{F}}_{u, n}^{w, p}(a, b, c)$ and $\widehat{\Pi}_{c}^{(m + 1)}$ involve in the updated estimators $\widehat{\beta}_{c}^{(m + 1)}$, $(\widehat{\sigma}_{c}^{(m + 1)})^{2}$, see (\ref{updated 2sls location}), (\ref{updated 2sls variance}). While $n^{1/2}$-convergence results have been established for the empirical process $n^{1/2} \widehat{\mathsf{F}}_{u, n}^{w, p}(a, b, c)$ by Theorem \ref{one-sided empirical process result for all additive and multiplicative shifts}, \ref{linearization of one-sided empirical compensator}, \ref{tightness result for one-sided empirical process}, consistency for the estimator of $\Pi$ is then required to build up the asymptotic distribution theory for the iterated estimators of $\beta$, $\sigma^{2}$. A product moment appearing in $\widehat{\Pi}_{c}^{(m + 1)}$ but not included in $\widehat{\mathsf{F}}_{u, n}^{w, p}(a, b, c)$ is $n^{-1} \sum_{i = 1}^{n} z_{i} r_{i}^{\prime} v_{i, c}^{(m)}$, see (\ref{2sls indicator for non-outlying observations}), (\ref{location estimator for the first stage regression}). Thus, proving consistency of the $\Pi$ estimator, we need the following $n$-convergence theorem for a new class of empirical processes with the weight $n^{1/2} z_{in} = z_{i}$ and mark $r_{i}^{\prime}$ instead of mark $u_{i}$ in the previous case.

\begin{theorem} \label{n-convergence results for a particular one-sided process}
Let $c_{\psi} = \mathsf{F}_{u}^{-1}(\psi)$. Suppose Assumption \ref{explicit assumptions} holds with $\nu = 1$ and $s \ge 2$ satisfying (\ref{moment condition for b, c uniform convergence}). Then for any $B > 0$ and as $n \to \infty$
\begin{equation*}
\sup_{0 \le \psi \le 1} \sup_{|a|, |b| \le n^{1/4 - \eta}B} |n^{-1/2} \sum_{i = 1}^{n} z_{in} r_{i}^{\prime} \{ 1_{(u_{i} \le \sigma c_{\psi} +n^{-1/2}a c_{\psi} + z_{in}^{\prime} \Pi b + n^{-1/2}r_{i}^{\prime}b)} - 1_{(u_{i} \le \sigma c_{\psi})}\}| = \mathrm{o}_{\mathsf{P}}(1).
\end{equation*}
\end{theorem}

The proof first considers uniform convergence for the empirical process without weights $n^{1/2} z_{in} = z_{i}$ and marks $r_{i}$. Since we establish $n$-convergence result instead of with the normalization $n^{1/2}$, the H{\"o}lder inequality is then used to prove Theorem \ref{n-convergence results for a particular one-sided process} by adding the weights and marks into the process with only indicators.

\subsection{A result for the two-sided empirical process}
Estimators in Algorithm \ref{iterated version of robust 2sls}, \ref{robustified two stage least squares}, \ref{split sample version} involves indicators depending on the absolute value of residuals of structural errors. We therefore present some results for a class of two-sided weighted and marked empirical processes.

Define the weighted and marked absolute empirical distribution function
\begin{equation} \label{absolute empirical distribution}
\widehat{\mathsf{G}}_{u, n}^{w, p}(a, b, c) = \frac{1}{n} \sum_{i = 1}^{n} w_{in} u_{i}^{p} 1_{(|u_{i} - z_{in}^{\prime} \Pi b - n^{-1/2}r_{i}^{\prime}b| \le \sigma c + n^{-1/2} a c)}.
\end{equation}
We suppose $a$ such that $\sigma + n^{-1/2}a > 0$, in which case it suffices to consider $c \ge 0$. This restriction on $a$ is satisfied when choosing $a$ as $\widetilde{a} = n^{1/2} (\widetilde{\sigma} - \sigma)$ so that $\sigma + n^{-1/2} \widetilde{a} = \widetilde{\sigma} > 0$. Introduce the compensator of $\widehat{\mathsf{G}}_{u, n}^{w, p}(a, b, c)$
\begin{equation} \label{absolute compensator}
\overline{\mathsf{G}}_{u, n}^{w, p}(a, b, c) = \frac{1}{n} \sum_{i=1}^{n} w_{in} \mathsf{E}_{i-1} u_{i}^{p} 1_{(|u_{i} - z_{in}^{\prime} \Pi b - n^{-1/2}r_{i}^{\prime}b| \le \sigma c + n^{-1/2}a c)}.
\end{equation}
Note that $\overline{\mathsf{G}}_{u, n}^{1, 0}(0, 0, c) = \mathsf{G}_{u}(c) = \mathsf{P}(|u_{i}| \le \sigma c)$. Similar as defining the one-sided empirical process $\mathbb{F}_{u, n}^{w, p}(a, b, c_{\psi})$ on the space $D[0, 1]$ equipped with the Skorokhod metric, let $c_{\psi} = \mathsf{G}_{u}^{-1}(\psi)$ for $0 \le \psi \le 1$ so $c_{\psi} \ge 0$ and the absolute empirical process is
\begin{equation} \label{absolute empirical process}
\mathbb{G}_{u, n}^{w, p}(a, b, c_{\psi}) = n^{1/2} \{ \widehat{\mathsf{G}}_{u, n}^{w, p}(a, b, c_{\psi}) - \overline{\mathsf{G}}_{u, n}^{w, p}(a, b, c_{\psi}) \}.
\end{equation}

We can now derive asymptotic theory for the absolute empirical process from Theorems \ref{one-sided empirical process result for all additive and multiplicative shifts}, \ref{linearization of one-sided empirical compensator}, \ref{tightness result for one-sided empirical process}, \ref{n-convergence results for a particular one-sided process}. These results are presented under more restrictive Assumption \ref{sufficient assumptions}, where the distribution $\mathsf{f}_{u}$ of structural error is symmetric, see Remark \ref{discussion and the relationship between sufficient and explicit assumptions}. In this section, we only consider $w_{in}$ chosen as $1$, $n^{1/2} z_{in} = z_{i}$, $nz_{in} z_{in}^{\prime} = z_{i} z_{i}^{\prime}$ and $p$ as $0$, $1$, $2$.

The first theorem states the absolute processes with estimation errors converge to the one without estimation errors uniformly in the quantile.

\begin{theorem} \label{absolute empirical process result for all additive and multiplicative shifts}
Let $c_{\psi} = \mathsf{G}_{u}^{-1}(\psi)$. Suppose Assumption \ref{sufficient assumptions}$(ia, iia, iiia, ivb, ivc)$ holds. Then for any $B > 0$ and as $n \to \infty$
\begin{equation*}
\sup_{0 \le \psi \le 1} \sup_{|a|, |b| \le n^{1/4 - \eta} B} |\mathbb{G}_{u, n}^{w, p} ( a, b, c_{\psi}) - \mathbb{G}_{u, n}^{w, p} ( 0, 0, c_{\psi})| = \mathrm{o}_{\mathsf{P}}(1).
\end{equation*}
\end{theorem}

Then we provide the first-order approximation to the absolute compensator.

\begin{theorem} \label{linearization of absolute empirical compensator}
Let $c_{\psi} = \mathsf{G}_{u}^{-1}(\psi)$. Suppose Assumption \ref{sufficient assumptions}$(ia, iia, iiia, ivc)$ holds. Then for any $B > 0$ and as $n \to \infty$
\begin{equation*}
\sup_{0 \le \psi \le 1} \sup_{|a|, |b| \le n^{1/4 - \eta}B} |n^{1/2} \{ \overline{\mathsf{G}}_{u, n}^{w, p}(a, b, c_{\psi}) - \overline{\mathsf{G}}_{u, n}^{w, p}(0, 0, c_{\psi}) \} - \mathcal{B}_{\mathsf{G}_{u}, n}(a, b, c_{\psi})| = \mathrm{O}_{\mathsf{P}}(n^{-2\eta}),
\end{equation*}
where the bias term, with $\xi_{c_{\psi}} = \mathsf{E} (\Sigma^{-1/2} r_{i} | u_{i}/\sigma = c_{\psi})$, is defined as
\begin{align*}
\mathcal{B}_{\mathsf{G}_{u}, n}(a, b, c_{\psi}) & = \sigma^{p - 1} c_{\psi}^{p} \mathsf{f}_{u}(c_{\psi}) n^{-1/2} \sum_{i=1}^{n} w_{in} [ \{1 + (-1)^{p} \} n^{-1/2} a c_{\psi}  \\
& \quad \,\, + n^{-1/2} \{ \xi_{c_{\psi}} - (-1)^{p} \xi_{-c_{\psi}} \}^{\prime} \Sigma^{1/2} b + \{ 1 - (-1)^{p} \} z_{in}^{\prime} \Pi b ].
\end{align*}
\end{theorem}

The next theorem extends the tightness result to the case of empirical processes with absolute indicators.

\begin{theorem} \label{tightness result for absolute empirical process}
Let $c_{\psi} = \mathsf{G}_{u}^{-1}(\psi)$. Suppose Assumption \ref{sufficient assumptions}$(ia, ivc)$ holds. Then for all $\epsilon > 0$
\begin{equation*}
\lim_{\phi \downarrow 0} \limsup_{n \to \infty} \mathsf{P} \{ \sup_{0 \le \psi \le \psi^{\dagger} \le 1 : \psi^{\dagger} - \psi \le \phi} |\mathbb{G}_{u, n}^{w, p}(0, 0, c_{\psi^{\dagger}}) - \mathbb{G}_{u, n}^{w, p}(0, 0, c_{\psi})| > \epsilon \} \to 0.
\end{equation*}
\end{theorem}

While the above theorems deal with $n^{1/2}$-convergence results, the final result considers $n$-convergence for a new type of absolute processes used to prove consistency of the estimator for $\Pi$ in the first stage regression (\ref{first stage regression}).

\begin{theorem} \label{n-convergence results for a particular absolute process}
Let $c_{\psi} = \mathsf{G}_{u}^{-1}(\psi)$. Suppose Assumption \ref{sufficient assumptions}$(ia, iia, iiia, ivb, ivc)$ holds. Then for any $B > 0$ and as $n \to \infty$
\begin{equation*}
\sup_{0 \le \psi \le 1} \sup_{|a|, |b| \le n^{1/4 - \eta}B} |n^{-1/2} \sum_{i = 1}^{n} z_{in} r_{i}^{\prime} \{ 1_{(|u_{i} - z_{in}^{\prime} \Pi b - n^{-1/2}r_{i}^{\prime}b| \le \sigma c_{\psi} + n^{-1/2}a c_{\psi})} - 1_{(|u_{i}| \le \sigma c_{\psi})}\}| = \mathrm{o}_{\mathsf{P}}(1).
\end{equation*}
\end{theorem}

The proofs of empirical processes are given in Appendix \ref{proofs of the empirical process results}, but first we discuss a metric on $\mathbb{R}$ and provide some useful inequalities applied repeatedly in the chaining argument in Appendix \ref{a metric on R and some inequalities}. Finally, the proofs of the main results follow in Appendix \ref{proofs of the main results}.

\section{Conclusion and discussion}
Since most empirical analyses are affected by atypical observations, when many applied economists apply IVs regressions they first run ordinary 2SLS and use the resulting residuals to find non-outlying observations. They then re-run 2SLS on this subset, subsequently iterating this procedure until they obtain robust results.

In this paper we analyze this simple robust algorithm asymptotically, and provide the consistent estimation and valid inferential procedures given the cut-off value. We concentrate on cross-sectional i.i.d. data in this paper, though our analysis can be extended to time series, whether stationary, deterministic trend, or unit root. The results are all derived under the null hypothesis that there are no outliers in the model. Since in this case there is still a positive probability to find outliers using such algorithms, we then propose the concept of gauge - the expected retention rate of falsely discovered outliers - to evaluate their performance.

Future work should focus on building asymptotic theory for the gauge, to give empirical researchers an indirect way of choosing the cut-off. Furthermore, some specification tests could be devised to check whether outliers truly exist with respect to the reference model (\ref{structural equation}, \ref{first stage regression}), as in Jiao and Pretis (2018).

It is well known that the first order asymptotic approximation is fragile in some small finite sample situations. Therefore, it would be of interest to carry out simulation studies to evaluate the finite sample performance of the results in this paper. Likewise it would be of interest to extend these results to situations where outliers are actually present in the data generating process. Possible scenarios include single outliers, clusters of outliers, level shifts, symmetric or non-symmetric outliers. In such situations, we would analyze the potency, which is the retention rate for relevant outliers. 


\appendix
\section{A metric on $\mathbb{R}$ and some inequalities} \label{a metric on R and some inequalities}
The asymptotic theory uses a chaining argument. This involves a partitioning of the quantile axis using a metric, which is presented first. Then follows some preliminary inequalities including an iterated exponential martingale inequality.

For $x, y \in \mathbb{R}$ define the function
\begin{equation} \label{Jxy}
J_{i, p}(x,y) = (\frac{u_{i}}{\sigma})^{p}\{1_{(u_{i} / \sigma \le y)} - 1_{(u_{i} / \sigma \le x)}\}.
\end{equation}
Our interest focus on $J_{i,p}(x,y)$ of order $2^{s}$ with $s \in \mathbb{N}$. For $y, z \in \mathbb{R}$ note $z^{2^{s}p}$ is non-negative since $2^{s}p$ is even for $p \in \mathbb{N}_{0}$ and $s \in \mathbb{N}$, so introduce a positive and increasing function
\begin{equation}
\mathsf{H}_{s}(y) = \int_{-\infty}^{y} (1 + z^{2^{s}p})\mathsf{f}_{u}(z) dz.
\end{equation}
The derivative of this function is $\dot{\mathsf{H}}_{s}(y) = (1 + y^{2^{s}p})\mathsf{f}_{u}(y)$. Then, denote the constant
\begin{equation}
H_{s} = \mathsf{H}_{s}(\infty) = \int_{-\infty}^{\infty} (1 + z^{2^{s}p})\mathsf{f}_{u}(z) dz,
\end{equation}
which is assumed to be finite. Selection of the specific $s \in \mathbb{N}$ will be more clear in proofs of the empirical process results. The intuition of $\mathsf{H}_{s}(y)$ is obtained through setting $p = 0$ so that $\mathsf{H}_{s}(y) = 2\mathsf{F}_{u}(y)$, $\dot{\mathsf{H}}_{s}(y) = 2\mathsf{f}_{u}(y)$ and $H_{s} = 2$. Therefore, $\mathsf{H}_{s}(y)$ is the generalization of the distribution $\mathsf{F}_{u}(y) \sim u_{i} / \sigma$. Define the metric $\mathsf{d}$ on $x, y \in \mathbb{R}$ as
\begin{equation} \label{metric d0}
\mathsf{d}(x, y) = |\mathsf{H}_{s}(x) - \mathsf{H}_{s}(y)|.
\end{equation}
Notice the metric $\mathsf{d}$ is totally bounded on $\mathbb{R}$ since $\mathsf{d}(-\infty, \infty) = H_{s} < \infty$, and it is a special case of the pseudo metric proposed by Koul and Ossiander (1994). For $x \le y$ and $1 \le q \le s$,
\begin{equation} \label{inequality Jxy}
0 \le \mathsf{E}_{i - 1} J_{i,p}(x,y)^{2^{q}} = \mathsf{E} J_{i,p}(x,y)^{2^{q}} < \mathsf{H}_{s}(y) - \mathsf{H}_{s}(x) = \mathsf{d}(x, y),
\end{equation}
as $|z|^{q} < 1 + |z|^{s}$ for $0 \le q \le s$.

In the context of chaining, partition the range of $\mathsf{H}_{s}(c)$ into $K$ intervals of equal size $H_{s}/K$. In other words, use the metric $\mathsf{d}$ to partition $\mathbb{R}$ into $K$ intervals by endpoints
\begin{equation} \label{partition}
-\infty = c_{0} < c_{1} < \cdots < c_{K-1} < c_{K} = \infty,
\end{equation}
with $c_{-k} = c_{0}$ for $k \in \mathbb{N}$ so that $\mathsf{d}(c_{k - 1}, c_{k}) = H_{s} / K$ for $1 \le k \le K$.



We first prove an inequality for the difference of two indicators, which will be used to control variation within the chaining interval.

\begin{lemma} \label{inequality for the difference of two indicators}
For any $e, e_{0}$ on the $u$-axis and $d \ge |e - e_{0}|$, we have an inequality
\begin{equation*}
|1_{(u \le e)} - 1_{(u \le e_{0})}| \le 1_{(e_{0} - d \le u \le e_{0} + d)}.
\end{equation*}
\end{lemma}

\begin{proof}[\textnormal{\textbf{Proof of Lemma \ref{inequality for the difference of two indicators}}}]
From $|e - e_{0}| \le d$ we find $e_{0} - d \le e \le e_{0} + d$. Suppose $e \ge e_{0}$, then using the monotonicity of the indicator function in $e$ and $e_{0}$, we have
\begin{equation*}
|1_{(u \le e)} - 1_{(u \le e_{0})}| = 1_{(e_{0} \le u \le e)} \le 1_{(e_{0} - d \le u \le e_{0} + d)}.
\end{equation*}
The similar argument applies for $e < e_{0}$, so we also have
\begin{equation*}
|1_{(u \le e)} - 1_{(u \le e_{0})}| = 1_{(e \le u \le e_{0})} \le 1_{(e_{0} - d \le u \le e_{0} + d)}. \qedhere
\end{equation*}
\end{proof}

The next lemma is the bias correction for differences in expectations involving two dependent random variables, which will be used to linearize the compensator.

\begin{lemma} \label{lemma for compensator}
Let $(Y, X)$ be a random vector with dimension $1 + d_{x}$ with the joint density $\mathsf{m}(y, x)$ for $y \in \mathbb{R}, x \in \mathbb{R}^{d_{x}}$. Suppose the joint density decomposes into the conditional and marginal densities as $\mathsf{m}(y, x) = \mathsf{m}_{Y|X}(y|x) \mathsf{m}_{X}(x) = \mathsf{m}_{X|Y}(x|y) \mathsf{m}_{Y}(y)$ and its derivative for $y$ exists so $\dot{\mathsf{m}}_{y}(y, x) = \dot{\mathsf{m}}_{Y|X, y}(y|x) \mathsf{m}_{X}(x)$. Let $b_{1} \in \mathbb{R}$, $b_{2} \in \mathbb{R}^{d_{x}}$, and $p \in \mathbb{N}_{0}$. Assume $\mathsf{E} |X|^{2} < \infty$, then
\begin{align*}
& \sup_{c \in \mathbb{R}} |\mathsf{E} Y^{p} \{ 1_{(Y \le c + b_{1} + X^{\prime} b_{2})} - 1_{(Y \le c)} \} - c^{p} \mathsf{m}_{Y}(c) \{ b_{1} + \mathsf{E}(X^{\prime} b_{2}|Y=c) \} | \\
\le & (\frac{1}{2} b_{1}^{2} + |b_{1}| |b_{2}| \mathsf{E}|X| + \frac{1}{2} |b_{2}|^{2} \mathsf{E} |X|^{2}) \sup_{y \in \mathbb{R}, x \in \mathbb{R}^{d_{x}}} |p y^{p - 1} \mathsf{m}_{Y|X}(y|x) + y^{p} \dot{\mathsf{m}}_{Y|X, y}(y|x)|.
\end{align*}
\end{lemma}

\begin{proof}[\textnormal{\textbf{Proof of Lemma \ref{lemma for compensator}}}]
Let $\mathcal{E} = \mathsf{E} Y^{p} \{ 1_{(Y \le c + b_{1} + X^{\prime} b_{2})} - 1_{(Y \le c)} \}$, and write the expectation as an integral
\begin{equation*}
\mathcal{E} = \iint_{y \in \mathbb{R}, x \in \mathbb{R}^{d_{x}}} y^{p} 1_{(c \le y \le c + b_{1} + x^{\prime} b_{2})} \mathsf{m}(y, x) dy (dx) = \int_{x \in \mathbb{R}^{d_{x}}} \int_{c}^{c + b_{1} + x^{\prime} b_{2}} y^{p} \mathsf{m}(y, x) dy (dx).
\end{equation*}
Apply Taylor expansion at point $c$ to the inner integral to obtain
\begin{equation*}
\int_{c}^{c + b_{1} + x^{\prime} b_{2}} y^{p} \mathsf{m}(y, x) dy = (b_{1} + x^{\prime} b_{2}) c^{p} \mathsf{m}(c, x) + \frac{1}{2} (b_{1} + x^{\prime} b_{2})^{2} \{ p \tilde{c}^{p - 1} \mathsf{m}(\tilde{c}, x) + \tilde{c}^{p} \dot{\mathsf{m}}_{y}(\tilde{c}, x) \},
\end{equation*}
where $|\tilde{c} - c| \le |b_{1} + x^{\prime} b_{2}|$. Then $\mathcal{E}$ is approximated by the first order term $\mathcal{E}_{1}$ with the remainder term $\mathcal{E}_{2}$ left so $\mathcal{E} = \mathcal{E}_{1} + \mathcal{E}_{2}$ where
\begin{align*}
\mathcal{E}_{1} & = \int_{\mathbb{R}^{d_{x}}} (b_{1} + x^{\prime} b_{2}) c^{p} \mathsf{m}(c, x) (dx), \\
\mathcal{E}_{2} & = \frac{1}{2} \int_{\mathbb{R}^{d_{x}}} (b_{1} + x^{\prime} b_{2})^{2} \{ p \tilde{c}^{p - 1} \mathsf{m}(\tilde{c}, x) + \tilde{c}^{p} \dot{\mathsf{m}}_{y}(\tilde{c}, x) \} (dx).
\end{align*}
Analyze the first order term $\mathcal{E}_{1}$ to get
\begin{equation*}
\mathcal{E}_{1} = c^{p} \mathsf{m}_{Y}(c) b_{1} + c^{p} \mathsf{m}_{Y}(c) \int_{\mathbb{R}^{d_{x}}} x^{\prime} b_{2} \mathsf{m}_{X|Y}(x|c) (dx) = c^{p} \mathsf{m}_{Y}(c) \{ b_{1} + \mathsf{E}(X^{\prime} b_{2}|Y=c) \}.
\end{equation*}
For the remainder term $\mathcal{E}_{2} = \mathcal{E} - \mathcal{E}_{1}$, apply the triangle inequality so
\begin{equation*}
|\mathcal{E}_{2}| \le \frac{1}{2} \int_{\mathbb{R}^{d_{x}}} \{ b_{1}^{2} + 2 |b_{1} x^{\prime} b_{2}| + (x^{\prime} b_{2})^{2} \} |p \tilde{c}^{p - 1} \mathsf{m}_{Y|X}(\tilde{c}|x) + \tilde{c}^{p} \dot{\mathsf{m}}_{Y|X, y}(\tilde{c}|x)| \mathsf{m}_{X}(x) (dx).
\end{equation*}
As $|b_{1} x^{\prime} b_{2}| \le |b_{1}| |b_{2}| |x|$ and $(x^{\prime} b_{2})^{2} \le |x|^{2} |b_{2}|^{2}$, so
\begin{equation*}
|\mathcal{E}_{2}| \le (\frac{1}{2} b_{1}^{2} + |b_{1}| |b_{2}| \mathsf{E}|X| + \frac{1}{2} |b_{2}|^{2} \mathsf{E} |X|^{2}) \sup_{y \in \mathbb{R}, x \in \mathbb{R}^{d_{x}}} |p y^{p - 1} \mathsf{m}_{Y|X}(y|x) + y^{p} \dot{\mathsf{m}}_{Y|X, y}(y|x)|.  \qedhere
\end{equation*}
\end{proof}

The next inequality provides a tightness type result for the metric $\mathsf{d}$, which will be used to prove the tightness property of empirical process $\mathbb{F}_{u, n}^{w, p}(0, 0, c_{\psi})$.

\begin{lemma} \label{tightness inequality for d}
(Johansen and Nielsen, 2016a, Lemma B.3). Let $c_{\psi} = \mathsf{F}_{u}^{-1}(\psi)$. Suppose Assumption \ref{explicit assumptions}$(ia)$ holds with $0 < \nu < 1$, $s \in \mathbb{N}_0$. Then there exist $C_{\nu}$, $0 < \phi_{0} < 1$ so that for $0 \le \phi \le \phi_{0}$ it follows $\sup_{0 \le \psi \le 1 - \phi} \mathsf{d} (c_{\psi}, c_{\psi + \phi}) \le C_{\nu} \phi^{1 - \nu}$.
\end{lemma}

The chaining argument involves the tail behaviour of the maximum of a family of martingales which can be controlled using the iterated exponential martingale inequality from Johansen and Nielsen (2016a). It builds on an exponential martingale inequality derived by Bercu and Touati (2008, Theorem 2.1) which relaxes the need of boundedness for martingale difference series but required in Freedman (1975) inequality. The following are two special cases of iterated exponential martingale inequalities, where the number of elements in the martingale family is increasing and where it is fixed. The first inequality will be used to demonstrate the uniform convergence result for empirical processes, while the second is for handling the tightness property.

\begin{lemma} \label{iterated martingale inequality for asymptotic result}
(Johansen and Nielsen, 2016a, Theorem 5.2) For $l$ so $1 \le l \le L$, let $z_{l, i}$ be $\mathcal{F}_{i}$ adapted satisfying $\mathsf{E} z_{l, i}^{2^{s}} < \infty$ for some $s \in \mathbb{N}$. Let $D_{q} = \max_{1 \le l \le L} \sum_{i=1}^{n} \mathsf{E}_{i - 1} z_{l, i}^{2^{q}}$ for $1 \le q \le s$. Suppose, for some $\varsigma \ge 0$, $\lambda > 0$, that $L = \mathrm{O}(n^{\lambda})$ and $\mathsf{E} D_{q} = \mathrm{O}(n^{\varsigma})$ for $1 \le q \le s$. If $\upsilon > 0$ is chosen such that \\
\indent $(i)$ $\varsigma < 2 \upsilon$; \\
\indent $(ii)$ $\varsigma + \lambda < \upsilon 2^{s}$; \\
then, we have for all $\kappa > 0$
\begin{equation*}
\lim_{n \to \infty} \mathsf{P} \{ \max_{1 \le l \le L} |\sum_{i=1}^{n} {(z_{l, i} - \mathsf{E}_{i - 1} z_{l, i})}| > \kappa n^{\upsilon} \} = 0.
\end{equation*}
\end{lemma}

\begin{lemma} \label{iterated martingale inequality for tightness}
(Johansen and Nielsen, 2016a, Theorem 5.3) For $l$ so $1 \le l \le L$, let $z_{l, i}$ be $\mathcal{F}_{i}$ adapted satisfying $\mathsf{E} z_{l, i}^{4} < \infty$. Let $D_{q} = \max_{1 \le l \le L} \sum_{i=1}^{n} \mathsf{E}_{i - 1} z_{l, i}^{2^{q}}$ and suppose $\mathsf{E} D_{q} \le D n$ for $q = 1, 2$ and some $D > 0$. Then for all $\theta, \kappa > 0$
\begin{equation*}
\mathsf{P} \{ \max_{1 \le l \le L} |\sum_{i=1}^{n} {(z_{l, i} - \mathsf{E}_{i - 1} z_{l, i})}| > \kappa n^{1/2} \} \le \frac{(L + 1) \theta^{3} D}{\kappa n} + \frac{\theta D}{\kappa} + 4 L \exp (- \frac{\kappa \theta}{14}).
\end{equation*}
\end{lemma}

\section{Proofs of the empirical process results} \label{proofs of the empirical process results}
We first prove uniform $n^{1/2}$-convergence and tightness for the one-sided empirical processes, then linearize the compensator term and show $n$-convergence for a specific type of the process. Finally, extend all results to the processes considering absolute residuals.

\subsection{Uniform $n^{1/2}$-convergence of empirical process}
We first consider Theorem \ref{uniform convergence of one-sided empirical process with variation in a, b but fixed c} which allows variation in $a, b$ but $c$ fixed. Then use the argument in Johansen and Nielsen (2016b, Theorem 4.1) to extend the result to Theorem \ref{uniform convergence of one-sided empirical process with variation in b, c but fixed a} where we have variation in $b, c$ but $a = 0$, also see Berenguer-Rico, Johansen, and Nielsen (2019a, Theorem 3.1), while use Jiao and Nielsen (2017, Theorem 5) to prove Theorem \ref{uniform convergence of one-sided empirical process with variation in a, c but fixed b} showing convergence uniformly in $a, c$ but $b = 0$. Finally, by combining Theorem \ref{uniform convergence of one-sided empirical process with variation in b, c but fixed a}, \ref{uniform convergence of one-sided empirical process with variation in a, c but fixed b} the uniform convergence in $a, b, c$ follows.

\begin{theorem} \label{uniform convergence of one-sided empirical process with variation in a, b but fixed c}
Let $c_{\psi} = \mathsf{F}_{u}^{-1}(\psi)$. Suppose Assumption \ref{explicit assumptions}$(ia, iia, iiia, v)$ holds with $\nu = 1$ and $s \ge 1$ such that $2^{s - 1} > 1/2 + (1/4 - \eta)(2 + d_{x}) $. Then for any $0 \le \psi \le 1$, $B > 0$ and as $n \to \infty$
\begin{equation*}
\sup_{|a|, |b| \le n^{1/4 - \eta} B} |\mathbb{F}_{u, n}^{w, p}(a, b, c_{\psi}) - \mathbb{F}_{u, n}^{w, p}(0, 0, c_{\psi})| = \mathrm{o}_{\mathsf{P}}(1).
\end{equation*}
\end{theorem}

\begin{proof}[\textnormal{\textbf{Proof of Theorem \ref{uniform convergence of one-sided empirical process with variation in a, b but fixed c}}}]
Without loss of generality let $\sigma = 1$, $\Sigma^{1/2} = I_{d_{x}}$. Denote $R_{n}(a, b, c_{\psi}) = \mathbb{F}_{u, n}^{w, p}(a, b, c_{\psi}) - \mathbb{F}_{u, n}^{w, p}(0, 0, c_{\psi})$ so we want to show $\mathcal{R}_{\psi, n} = \mathrm{o}_{\mathsf{P}}(1)$ for any $0 \le \psi \le 1$ as $n \to \infty$ where $\mathcal{R}_{\psi, n} = \sup_{|a|, |b| \le n^{1/4 - \eta} B} |R_{n}(a, b, c_{\psi})|$. Throughout denote $C > 0$ as usual a constant which may have different values in different expressions.

\emph{$1$. Construct $a, b$-balls}. For $\delta, n > 0$, cover the set $|a|, |b| \le n^{1/4 - \eta} B$ with balls of radius $\delta$ and centers $a_{j}, b_{m}$. The numbers of $a, b$-balls are $J = n^{1/4 - \eta} B / \delta = \mathrm{O}(n^{1/4 - \eta} / \delta)$ and $M = (n^{1/4 - \eta} B / \delta)^{d_{x}} = \mathrm{O}(n^{(1/4 - \eta) d_{x}} / \delta)$ respectively. Thus, for any $a, b$ there exist $a_{j}, b_{m}$ so that $|a - a_{j}| \le \delta$ and $|b - b_{m}| \le \delta$.

\emph{$2$. Apply chaining}. Write $R_{n}(a, b, c_{\psi}) = R_{n}(a_{j}, b_{m}, c_{\psi}) + \{ R_{n}(a, b, c_{\psi}) - R_{n}(a_{j}, b_{m}, c_{\psi}) \}$ where $R_{n}(a_{j}, b_{m}, c_{\psi})$ is a discrete point term while $R_{n}(a, b, c_{\psi}) - R_{n}(a_{j}, b_{m}, c_{\psi})$ is a local oscillation term. The triangle inequality gives $\mathcal{R}_{\psi, n} \le \mathcal{R}_{\psi, n, 1} + \mathcal{R}_{\psi, n, 2}$ where
\begin{align*}
\mathcal{R}_{\psi, n, 1} & = \max_{1 \le j \le J, 1 \le m \le M} |R_{n}(a_{j}, b_{m}, c_{\psi})|, \\
\mathcal{R}_{\psi, n, 2} & = \max_{1 \le j \le J, 1 \le m \le M} \sup_{|a - a_{j}| \le \delta, |b - b_{m}| \le \delta} |R_{n}(a, b, c_{\psi}) - R_{n}(a_{j}, b_{m}, c_{\psi})|.
\end{align*}
It then suffices to show that $\mathcal{R}_{\psi, n, 1}, \mathcal{R}_{\psi, n, 2}$ are $\mathrm{o}_{\mathsf{P}}(1)$.

\emph{$3$. The discrete point term $\mathcal{R}_{\psi, n, 1}$ is $\mathrm{o}_{\mathsf{P}}(1)$}. Recall the notation $J_{i, p}$ in (\ref{Jxy}), then write $R_{n}(a_{j}, b_{m}, c_{\psi}) = n^{-1/2} \sum_{i = 1}^{n} (z_{l, i} - \mathsf{E}_{i - 1} z_{l, i})$ where
\begin{equation*}
z_{l, i} = w_{in} J_{i, p}(c_{\psi}, c_{\psi} + n^{-1/2} a_{j} c_{\psi} + z_{in}^{\prime} \Pi b_{m} + n^{-1/2} r_{i}^{\prime} b_{m}).
\end{equation*}
Use Lemma \ref{iterated martingale inequality for asymptotic result} with $\upsilon = 1/2$, index $l = (j, m)$ and $L = J M$ to prove $\mathcal{R}_{\psi, n, 1} = \mathrm{o}_{\mathsf{P}}(1)$, and we need to verify its conditions.

\emph{The moment $\mathsf{E} z_{l, i}^{2^{s}} < \infty$ for $s \in \mathbb{N}$}. Bounding the difference of indicator functions by unity and using independence of $u_{i}$ with $\mathcal{F}_{i - 1}$ gives $\mathsf{E} z_{l, i}^{2^{s}} \le \mathsf{E} |w_{in}|^{2^{s}} u_{i}^{2^{s} p} = \mathsf{E} |w_{in}|^{2^{s}} \mathsf{E} u_{i}^{2^{s} p}$. This is finite by Assumption \ref{explicit assumptions}$(ia, va)$.

\emph{The parameter $\lambda$}. The set of indices $l$ has the size $L = J M$. Since $J = \mathrm{O}(n^{1/4 - \eta} / \delta)$ and $M = \mathrm{O}(n^{(1/4 - \eta) d_{x}} / \delta)$ while $\delta$ is now treated as a constant then $L = \mathrm{O}(n^{\lambda})$ where $\lambda = (1/4 - \eta)(1 + d_{x})$.

\emph{The parameter $\varsigma$}. Consider $1 \le q \le s$. Notice we have
\begin{equation*}
|1_{(u_{i} \le c_{\psi} + n^{-1/2} a_{j} c_{\psi} + z_{in}^{\prime} \Pi b_{m} + n^{-1/2} r_{i}^{\prime} b_{m})} - 1_{(u_{i} \le c_{\psi})}| \le 1_{(\underline{c}_{i \psi j m, r_{i}} \le u_{i} \le \overline{c}_{i \psi j m, r_{i}})},
\end{equation*}
where we denote
\begin{align*}
\underline{c}_{i \psi j m, r_{i}} & = c_{\psi} - n^{-1/2} |a_{j}| |c_{\psi}| - |z_{in}| |\Pi| |b_{m}| - n^{-1/2} |r_{i}| |b_{m}|, \\
\overline{c}_{i \psi j m, r_{i}} & = c_{\psi} + n^{-1/2} |a_{j}| |c_{\psi}| + |z_{in}| |\Pi| |b_{m}| + n^{-1/2} |r_{i}| |b_{m}|.
\end{align*}
The above follows from Lemma \ref{inequality for the difference of two indicators} by $e = c_{\psi} + n^{-1/2} a_{j} c_{\psi} + z_{in}^{\prime} \Pi b_{m} + n^{-1/2} r_{i}^{\prime} b_{m}$, $e_{0} = c_{\psi}$, $d = n^{-1/2} |a_{j}| |c_{\psi}| + |z_{in}| |\Pi| |b_{m}| + n^{-1/2} |r_{i}| |b_{m}| \ge |n^{-1/2} a_{j} c_{\psi} + z_{in}^{\prime} \Pi b_{m} + n^{-1/2} r_{i}^{\prime} b_{m}| = |e - e_{0}|$. Let $\mathcal{E}_{i} = \mathsf{E}_{i - 1} J_{i, p}^{2^{q}}(c_{\psi}, c_{\psi} + n^{-1/2} a_{j} c_{\psi} + z_{in}^{\prime} \Pi b_{m} + n^{-1/2} r_{i}^{\prime} b_{m})$, then by applying the above inequality concerning the difference of indicator functions we have
\begin{equation*}
\mathcal{E}_{i} \le \mathsf{E}_{i - 1} J_{i, p}^{2^{q}}(\underline{c}_{i \psi j m, r_{i}}, \overline{c}_{i \psi j m, r_{i}})
= \iint_{y \in \mathbb{R}, x \in \mathbb{R}^{d_{x}}} y^{2^{q} p} 1_{(\underline{c}_{i \psi j m, x} \le y \le \overline{c}_{i \psi j m, x})} \mathsf{f}_{u, r}(y, x) dy (dx).
\end{equation*}
Decompose the joint density into the conditional and marginal ones, then apply the mean value theorem to the inner integral to get
\begin{align*}
\mathcal{E}_{i} & \le \int_{x \in \mathbb{R}^{d_{x}}} \{ \int_{\underline{c}_{i \psi j m, x}}^{\overline{c}_{i \psi j m, x}} (1 + y^{2^{s} p}) \mathsf{f}_{u|r}(y|x) dy \} \mathsf{f}_{r}(x) (dx)  \\
& = \int_{x \in \mathbb{R}^{d_{x}}} (\overline{c}_{i \psi j m, x} - \underline{c}_{i \psi j m, x}) (1 + \tilde{c}^{2^{s} p}) \mathsf{f}_{u|r}(\tilde{c}|x) \mathsf{f}_{r}(x) (dx),
\end{align*}
for an intermediate point $\tilde{c}$ such that $\underline{c}_{i \psi j m, x} \le \tilde{c} \le \overline{c}_{i \psi j m, x}$. Insert $\underline{c}_{i \psi j m, x}, \overline{c}_{i \psi j m, x}$ to get
\begin{equation*}
\mathcal{E}_{i} \le 2 \sup_{y \in \mathbb{R}, x \in \mathbb{R}^{d_{x}}} (1 + y^{2^{s} p}) \mathsf{f}_{u|r}(y|x) \int_{x \in \mathbb{R}^{d_{x}}} n^{-1/2} (|a_{j}| |c_{\psi}| + |n^{1/2} z_{in}| |\Pi| |b_{m}| + |x| |b_{m}|) \mathsf{f}_{r}(x) (dx).
\end{equation*}
Note $|a_{j}|, |b_{m}| \le n^{1/4 - \eta} B$ for each $1 \le j \le J$, $1 \le m \le M$ while $c_{\psi}$, $\Pi$ are fixed. Furthermore, since Assumption \ref{explicit assumptions}$(iia, iiia)$ that $\int_{x \in \mathbb{R}^{d_{x}}} |x| \mathsf{f}_{r}(x) (dx) < \infty$ and that $\sup_{y \in \mathbb{R}, x \in \mathbb{R}^{d_{x}}} (1 + y^{2^{s} p}) \mathsf{f}_{u|r}(y|x) < \infty$, then there exists $C > 0$ so uniformly in $j, m$
\begin{equation*}
\mathcal{E}_{i} \le C n^{-1/4 - \eta} \int_{x \in \mathbb{R}^{d_{x}}} (1 + |n^{1/2} z_{in}| + |x|) \mathsf{f}_{r}(x) (dx) \le C n^{-1/4 - \eta} (1 + |n^{1/2} z_{in}|).
\end{equation*}
Insert the bound of $\mathcal{E}_{i}$ to get
\begin{equation*}
D_{q} = \max_{1 \le l \le L} \sum_{i = 1}^{n} \mathsf{E}_{i - 1} z_{l, i}^{2^{q}} = \max_{1 \le j \le J, 1 \le m \le M} \sum_{i = 1}^{n} |w_{in}|^{2^{q}} \mathcal{E}_{i} \le C n^{-1/4 - \eta} \sum_{i = 1}^{n} |w_{in}|^{2^{q}} (1 + |n^{1/2} z_{in}|).
\end{equation*}
Since $|w_{in}|^{2^{q}} \le (1 + |w_{in}|^{2^{s}})$ and $\mathsf{E} \sum_{i = 1}^{n} (1 + |w_{in}|^{2^{s}}) (1 + |n^{1/2} z_{in}|) = \mathrm{O}(n)$ by Assumption \ref{explicit assumptions}$(va)$, then we have $\mathsf{E} D_{q} = \mathrm{O}(n^{\varsigma})$ where $\varsigma = 3/4 - \eta$.

\emph{Condition $(i)$} is that $\varsigma < 2 \upsilon$. This holds since $\eta > 0$ so $\varsigma = 3/4 - \eta < 1 = 2 \upsilon$.

\emph{Condition $(ii)$} is that $\varsigma + \lambda < \upsilon 2^{s}$. We have
\begin{equation*}
\varsigma + \lambda = 3/4 - \eta + (1/4 - \eta)(1 + d_{x}) = 1/2 + (1/4 - \eta)(2 + d_{x}).
\end{equation*}
We choose $s$ so that $\varsigma + \lambda < \upsilon 2^{s} = 2^{s - 1}$.

\emph{$4$. Decompose the oscillation term $\mathcal{R}_{\psi, n, 2}$}. As the terms $\mathbb{F}_{u, n}^{w, p}(0, 0, c_{\psi})$ in $R_{n}(a, b, c_{\psi})$ and $R_{n}(a_{j}, b_{m}, c_{\psi})$ cancel out, then we have
\begin{equation*}
R_{n}(a, b, c_{\psi}) - R_{n}(a_{j}, b_{m}, c_{\psi}) = \mathbb{F}_{u, n}^{w, p}(a, b, c_{\psi}) - \mathbb{F}_{u, n}^{w, p}(a_{j}, b_{m}, c_{\psi}).
\end{equation*}
Let $s_{i, p}(a, a_{j}, b, b_{m}, c_{\psi})$ denote the term
\begin{equation*}
w_{in} u_{i}^{p} \{ 1_{(u_{i} \le c_{\psi} + n^{-1/2} a c_{\psi} + z_{in}^{\prime} \Pi b + n^{-1/2} r_{i}^{\prime} b)} - 1_{(u_{i} \le c_{\psi} + n^{-1/2} a_{j} c_{\psi} + z_{in}^{\prime} \Pi b_{m} + n^{-1/2} r_{i}^{\prime} b_{m})} \},
\end{equation*}
it then follows
\begin{equation*}
R_{n}(a, b, c_{\psi}) - R_{n}(a_{j}, b_{m}, c_{\psi}) = n^{-1/2} \sum_{i = 1}^{n} \{ s_{i, p}(a, a_{j}, b, b_{m}, c_{\psi}) - \mathsf{E}_{i - 1} s_{i, p}(a, a_{j}, b, b_{m}, c_{\psi}) \}.
\end{equation*}
Choose $e = c_{\psi} + n^{-1/2} a c_{\psi} + z_{in}^{\prime} \Pi b + n^{-1/2} r_{i}^{\prime} b$ and $e_{0} = c_{\psi} + n^{-1/2} a_{j} c_{\psi} + z_{in}^{\prime} \Pi b_{m} + n^{-1/2} r_{i}^{\prime} b_{m}$, then find $|e - e_{0}| = |n^{-1/2} (a - a_{j}) c_{\psi} + z_{in}^{\prime} \Pi (b - b_{m}) + n^{-1/2} r_{i}^{\prime} (b - b_{m})|$. For any $j, m$ it holds $|a - a_{j}| \le \delta$, $|b - b_{m}| \le \delta$, so $|e - e_{0}| \le n^{-1/2} \delta (|c_{\psi}| + |n^{1/2} z_{in}| |\Pi| + |r_{i}|) = d$. Thus, apply Lemma \ref{inequality for the difference of two indicators} to bound the difference of two indicators so that
\begin{equation*}
|1_{(u_{i} \le c_{\psi} + n^{-1/2} a c_{\psi} + z_{in}^{\prime} \Pi b + n^{-1/2} r_{i}^{\prime} b)} - 1_{(u_{i} \le c_{\psi} + n^{-1/2} a_{j} c_{\psi} + z_{in}^{\prime} \Pi b_{m} + n^{-1/2} r_{i}^{\prime} b_{m})}| \le 1_{(\underline{\underline{c}}_{i \psi j m, r_{i}} \le u_{i} \le \overline{\overline{c}}_{i \psi j m, r_{i}})},
\end{equation*}
where we deonote
\begin{align*}
\underline{\underline{c}}_{i \psi j m, r_{i}} & = c_{\psi} + n^{-1/2} a_{j} c_{\psi} + z_{in}^{\prime} \Pi b_{m} + n^{-1/2} r_{i}^{\prime} b_{m} - n^{-1/2} \delta (|c_{\psi}| + |n^{1/2} z_{in}| |\Pi| + |r_{i}|), \\
\overline{\overline{c}}_{i \psi j m, r_{i}} & = c_{\psi} + n^{-1/2} a_{j} c_{\psi} + z_{in}^{\prime} \Pi b_{m} + n^{-1/2} r_{i}^{\prime} b_{m} + n^{-1/2} \delta (|c_{\psi}| + |n^{1/2} z_{in}| |\Pi| + |r_{i}|).
\end{align*}
Apply the above inequality to get $|s_{i, p}(a, a_{j}, b, b_{m}, c_{\psi})| \le |w_{in}| |J_{i, p}(\underline{\underline{c}}_{i \psi j m, r_{i}}, \overline{\overline{c}}_{i \psi j m, r_{i}})|$, then it follows by the triangle and Jensen's inequality that
\begin{align*}
& |R_{n}(a, b, c_{\psi}) - R_{n}(a_{j}, b_{m}, c_{\psi})| \\
\le & n^{-1/2} \sum_{i = 1}^{n} |w_{in}| \{ |J_{i, p}(\underline{\underline{c}}_{i \psi j m, r_{i}}, \overline{\overline{c}}_{i \psi j m, r_{i}})| + \mathsf{E}_{i - 1} |J_{i, p}(\underline{\underline{c}}_{i \psi j m, r_{i}}, \overline{\overline{c}}_{i \psi j m, r_{i}})| \}.
\end{align*}
Decompose the oscillation term $\mathcal{R}_{\psi, n, 2}$ into a martingale and two compensators such that $\mathcal{R}_{\psi, n, 2} \le \widetilde{\mathcal{R}}_{\psi, n, 2} + 2 \overline{\mathcal{R}}_{\psi, n, 2}$, where
\begin{align*}
\widetilde{\mathcal{R}}_{\psi, n, 2} & = \max_{1 \le j \le J, 1 \le m \le M} n^{-1/2} \sum_{i = 1}^{n} |w_{in}| \{ |J_{i, p}(\underline{\underline{c}}_{i \psi j m, r_{i}}, \overline{\overline{c}}_{i \psi j m, r_{i}})| - \mathsf{E}_{i - 1} |J_{i, p}(\underline{\underline{c}}_{i \psi j m, r_{i}}, \overline{\overline{c}}_{i \psi j m, r_{i}})| \} \\
\overline{\mathcal{R}}_{\psi, n, 2} & = \max_{1 \le j \le J, 1 \le m \le M} n^{-1/2} \sum_{i = 1}^{n} |w_{in}| \mathsf{E}_{i - 1} |J_{i, p}(\underline{\underline{c}}_{i \psi j m, r_{i}}, \overline{\overline{c}}_{i \psi j m, r_{i}})|.
\end{align*}
Thus, it suffices to prove that $\widetilde{\mathcal{R}}_{\psi, n, 2}$ and $\overline{\mathcal{R}}_{\psi, n, 2}$ are $\mathrm{o}_{\mathsf{P}}(1)$.

\emph{$5$. Bounding the compensator $\overline{\mathcal{R}}_{\psi, n, 2}$}. Let $\overline{\mathcal{E}}_{i} = \mathsf{E}_{i - 1} |J_{i, p}(\underline{\underline{c}}_{i \psi j m, r_{i}}, \overline{\overline{c}}_{i \psi j m, r_{i}})|$, then it is clear to see $\overline{\mathcal{E}}_{i} \le \mathcal{H}_{i} = \mathsf{E}_{i - 1} (1 + u_{i}^{2^{s} p}) 1_{(\underline{\underline{c}}_{i \psi j m, r_{i}} \le u_{i} \le \overline{\overline{c}}_{i \psi j m, r_{i}})}$ for $s \ge 1$, and thus our aim now is to bound $\mathcal{H}_{i}$. Write $\mathcal{H}_{i}$ as an integral and decompose the joint density so that
\begin{align*}
\mathcal{H}_{i} & = \iint_{y \in \mathbb{R}, x \in \mathbb{R}^{d_{x}}} (1 + y^{2^{s} p}) 1_{(\underline{\underline{c}}_{i \psi j m, x} \le y \le \overline{\overline{c}}_{i \psi j m, x})}  \mathsf{f}_{u, r}(y, x) dy (dx) \\
& = \int_{x \in \mathbb{R}^{d_{x}}} \{ \int_{\underline{\underline{c}}_{i \psi j m, x}}^{\overline{\overline{c}}_{i \psi j m, x}} (1 + y^{2^{s} p}) \mathsf{f}_{u|r}(y|x) dy \} \mathsf{f}_{r}(x) (dx).
\end{align*}
Apply the mean value theorem to the inner integral to get
\begin{equation*}
\mathcal{H}_{i} = \int_{x \in \mathbb{R}^{d_{x}}} (\overline{\overline{c}}_{i \psi j m, x} - \underline{\underline{c}}_{i \psi j m, x}) (1 + \tilde{\tilde{c}}^{2^{s} p}) \mathsf{f}_{u|r}(\tilde{\tilde{c}}|x) \mathsf{f}_{r}(x) (dx),
\end{equation*}
where $\tilde{\tilde{c}}$ is an intermediate point so that $\underline{\underline{c}}_{i \psi j m, x} \le \tilde{\tilde{c}} \le \overline{\overline{c}}_{i \psi j m, x}$. Insert $\underline{\underline{c}}_{i \psi j m, x}$, $\overline{\overline{c}}_{i \psi j m, x}$ into $\mathcal{H}_{i}$ and bound the function $(1 + y^{2^{s} p}) \mathsf{f}_{u|r}(y|x)$, then it follows
\begin{equation*}
\mathcal{H}_{i} \le 2 n^{-1/2} \delta \sup_{y \in \mathbb{R}, x \in \mathbb{R}^{d_{x}}} (1 + y^{2^{s} p}) \mathsf{f}_{u|r}(y|x) \int_{x \in \mathbb{R}^{d_{x}}} (|c_{\psi}| + |n^{1/2} z_{in}| |\Pi| + |x|) \mathsf{f}_{r}(x) (dx).
\end{equation*}
Due to Assumption \ref{explicit assumptions}$(iia, iiia)$ and that $c_{\psi}$, $\Pi$ are fixed, it holds uniformly in $1 \le j \le J$, $1 \le m \le M$ that $\mathcal{H}_{i} \le C n^{-1/2} \delta (1 + |n^{1/2} z_{in}|)$. Therefore, recall $\overline{\mathcal{E}}_{i} \le \mathcal{H}_{i}$ then we find
\begin{equation*}
\overline{\mathcal{R}}_{\psi, n, 2} = \max_{1 \le j \le J, 1 \le m \le M} n^{-1/2} \sum_{i = 1}^{n} |w_{in}| \overline{\mathcal{E}}_{i} \le C n^{-1/2} \delta n^{-1/2} \sum_{i = 1}^{n} |w_{in}| (1 + |n^{1/2} z_{in}|).
\end{equation*}
Assumption \ref{explicit assumptions}$(vb)$ that $n^{-1} \sum_{i = 1}^{n} |w_{in}| (1 + |n^{1/2} z_{in}|) = \mathrm{O}_{\mathsf{P}}(1)$ shows $\overline{\mathcal{R}}_{\psi, n, 2} = \delta \mathrm{O}_{\mathsf{P}}(1)$, where the $\mathrm{O}_{\mathsf{P}}(1)$-term does not depend on $\delta$. Choose $\delta$ sufficiently small then it follows $\overline{\mathcal{R}}_{\psi, n, 2} = \mathrm{o}_{\mathsf{P}}(1)$.

\emph{$6$. The martingale $\widetilde{\mathcal{R}}_{\psi, n, 2}$ is $\mathrm{o}_{\mathsf{P}}(1)$}. Note that once $\delta$ is selected in item $5$, it is then treated as constant so does not play any particular role in other items. To show $\widetilde{\mathcal{R}}_{\psi, n, 2} = \mathrm{o}_{\mathsf{P}}(1)$, use Lemma \ref{iterated martingale inequality for asymptotic result} with $\upsilon = 1/2$ and $z_{l, i} = |w_{in}| |J_{i, p}(\underline{\underline{c}}_{i \psi j m, r_{i}}, \overline{\overline{c}}_{i \psi j m, r_{i}})|$ where the index $l = (j, m)$ has the size $L = JM$. Two conditions need to shown.

\emph{The moment $\mathsf{E} z_{l, i}^{2^{s}} < \infty$ for $s \in \mathbb{N}$}. This can be shown by the same argument used in the item $3$.

\emph{The parameter $\lambda$}. The size $L = \mathrm{O}(n^{\lambda})$ where $\lambda = (1/4 - \eta)(1 + d_{x})$, since $L = JM$ and $J = \mathrm{O}(n^{1/4 - \eta})$, $M = \mathrm{O}(n^{(1/4 - \eta) d_{x}})$.

\emph{The parameter $\varsigma$}. Consider $1 \le q \le s$ and let $\widetilde{\mathcal{E}}_{i} = \mathsf{E}_{i - 1} J_{i, p}^{2^{q}}(\underline{\underline{c}}_{i \psi j m, r_{i}}, \overline{\overline{c}}_{i \psi j m, r_{i}})$. Item $5$ shows the bound of $\mathcal{H}_{i}$, so we find $\widetilde{\mathcal{E}}_{i} \le \mathcal{H}_{i} \le C n^{-1/2} (1 + |n^{1/2} z_{in}|)$ uniformly in $j, m$. Insert the bound of $\widetilde{\mathcal{E}}_{i}$ to get
\begin{equation*}
D_{q} = \max_{1 \le l \le L} \sum_{i = 1}^{n} \mathsf{E}_{i - 1} z_{l, i}^{2^{q}} = \max_{1 \le j \le J, 1 \le m \le M} \sum_{i = 1}^{n} |w_{in}|^{2^{q}} \widetilde{\mathcal{E}}_{i} \le C n^{-1/2} \sum_{i = 1}^{n} |w_{in}|^{2^{q}} (1 + |n^{1/2} z_{in}|).
\end{equation*}
Due to $|w_{in}|^{2^{q}} \le (1 + |w_{in}|^{2^{s}})$ and Assumption \ref{explicit assumptions}$(va)$, we have $\mathsf{E} D_{q} = \mathrm{O}(n^{\varsigma})$ where $\varsigma = 1/2$.

\emph{Condition $(i)$} holds, since $\varsigma = 1/2 < 1 = 2 \upsilon$.

\emph{Condition $(ii)$} holds, since
\begin{equation*}
\varsigma + \lambda = 1/2 + (1/4 - \eta)(1 + d_{x}) \le 1/2 + (1/4 - \eta)(2 + d_{x}) < 2^{s - 1} = \upsilon 2^{s}.
\end{equation*}
Notice that the first inequality above is due to $\eta \le 1/4$ and the number of moments $s$ is selected in the item $3$ such that $2^{s - 1} \ge 1/2 + (1/4 - \eta)(2 + d_{x})$.
\end{proof}

\begin{theorem} \label{uniform convergence of one-sided empirical process with variation in b, c but fixed a}
Let $c_{\psi} = \mathsf{F}_{u}^{-1}(\psi)$. Suppose Assumption \ref{explicit assumptions}$(i, iia, iiia, iva, va)$ holds with $\nu = 1$ and $s \ge 2$ satisfying (\ref{moment condition for b, c uniform convergence}). Then for any $B > 0$ and as $n \to \infty$
\begin{equation*}
\sup_{0 \le \psi \le 1} \sup_{|b| \le n^{1/4 - \eta} B} |\mathbb{F}_{u, n}^{w, p}(0, b, c_{\psi}) - \mathbb{F}_{u, n}^{w, p}(0, 0, c_{\psi})| = \mathrm{o}_{\mathsf{P}}(1).
\end{equation*}
\end{theorem}

\begin{proof}[\textnormal{\textbf{Proof of Theorem \ref{uniform convergence of one-sided empirical process with variation in b, c but fixed a}}}]
Set $a = 0$ and then apply the argument by Johansen and Nielsen (2016a, Theorem 4.1) and Berenguer-Rico, Johansen, and Nielsen (2019a, Theorem 3.1) to extend pointwise convergence shown in Theorem \ref{uniform convergence of one-sided empirical process with variation in a, b but fixed c} to convergence uniformly in $\psi$. It further requires chaining for the $c_{\psi}$-axis, so for $\delta, n > 0$ use the metric $\mathsf{d}$ to partition the quantile axis as laid out in (\ref{partition}) with $K = \mathrm{int} (H_{s} n^{1/2} / \delta)$ where $s \ge 2$ such that (\ref{moment condition for b, c uniform convergence}) $2^{s - 1} > 1 + (1/4 - \eta) (1 + d_{x})$ holds. Notice $s \ge 2$ satisfying $2^{s - 1} > 1 + (1/4 - \eta) (1 + d_{x})$ implies $s \ge 1$ and $2^{s - 1} > 1/2 + (1/4 - \eta) (2 + d_{x})$, so we need the higher moment condition on $s$ for proving the uniform convergence.
\end{proof}

\begin{theorem} \label{uniform convergence of one-sided empirical process with variation in a, c but fixed b}
Let $c_{\psi} = \mathsf{F}_{u}^{-1}(\psi)$. Suppose Assumption \ref{explicit assumptions}$(ia, ib, v)$ holds with $\nu = 1$, $s = 2$. Then for any $B > 0$ and as $n \to \infty$
\begin{equation*}
\sup_{0 \le \psi \le 1} \sup_{|a| \le n^{1/4 - \eta} B} |\mathbb{F}_{u, n}^{w, p}(a, 0, c_{\psi}) - \mathbb{F}_{u, n}^{w, p}(0, 0, c_{\psi})| = \mathrm{o}_{\mathsf{P}}(1).
\end{equation*}
\end{theorem}

\begin{proof}[\textnormal{\textbf{Proof of Theorem \ref{uniform convergence of one-sided empirical process with variation in a, c but fixed b}}}]
For $\delta, n > 0$ use the metric $\mathsf{d}$ to partition the $c_{\psi}$-axis as laid out in (\ref{partition}) with $K = \mathrm{int} (H_{s} n^{1/2} / \delta)$ where $s = 2$, then argue as the proof of Theorem 5 in Jiao and Nielsen (2017) to show convergence uniformly in $\psi$, $a$.
\end{proof}

\begin{proof}[\textnormal{\textbf{Proof of Theorem \ref{one-sided empirical process result for all additive and multiplicative shifts}}}]
The term of interest is $\mathcal{W} = \mathbb{F}_{u, n}^{w, p}(a, b, c_{\psi}) - \mathbb{F}_{u, n}^{w, p}(0, 0, c_{\psi})$. Denote $c_{\psi^{\dagger}} = c_{\psi}(1 + n^{-1/2} a / \sigma)$. Notice that $\mathbb{F}_{u, n}^{w, p}(a, b, c_{\psi}) = \mathbb{F}_{u, n}^{w, p}(0, b, c_{\psi^{\dagger}})$ so it follows that $\mathcal{W} = \mathbb{F}_{u, n}^{w, p}(0, b, c_{\psi^{\dagger}}) - \mathbb{F}_{u, n}^{w, p}(0, 0, c_{\psi})$. Add and subtract $\mathbb{F}_{u, n}^{w, p}(a, 0, c_{\psi}) = \mathbb{F}_{u, n}^{w, p}(0, 0, c_{\psi^{\dagger}})$ and apply the triangle inequality to get
\begin{equation*}
|\mathcal{W}| \le |\mathbb{F}_{u, n}^{w, p}(0, b, c_{\psi^{\dagger}}) - \mathbb{F}_{u, n}^{w, p}(0, 0, c_{\psi^{\dagger}})| + |\mathbb{F}_{u, n}^{w, p}(a, 0, c_{\psi}) - \mathbb{F}_{u, n}^{w, p}(0, 0, c_{\psi})|.
\end{equation*}
Thus, the problem reduces to showing
\begin{align}
\sup_{0 \le \psi^{\dagger} \le 1} \sup_{|b| \le n^{1/4 - \eta}B} |\mathbb{F}_{u, n}^{w, p}(0, b, c_{\psi^{\dagger}}) - \mathbb{F}_{u, n}^{w, p}(0, 0, c_{\psi^{\dagger}})| & = \mathrm{o}_{\mathsf{P}}(1), \label{one-sided empirical process additive term} \\
\sup_{0 \le \psi \le 1} \sup_{|a| \le n^{1/4 - \eta}B} |\mathbb{F}_{u, n}^{w, p}(a, 0, c_{\psi}) - \mathbb{F}_{u, n}^{w, p}(0, 0, c_{\psi})| & = \mathrm{o}_{\mathsf{P}}(1). \label{one-sided empirical process multiplicative term}
\end{align}
Then (\ref{one-sided empirical process additive term}) was considered by Theorem \ref{uniform convergence of one-sided empirical process with variation in b, c but fixed a} using Assumption \ref{explicit assumptions}$(i, iia, iiia, iva, va)$ with $\nu = 1$ and $s \ge 2$ such that (\ref{moment condition for b, c uniform convergence}) holds. Further, (\ref{one-sided empirical process multiplicative term}) follows by Theorem \ref{uniform convergence of one-sided empirical process with variation in a, c but fixed b}, which requires Assumption \ref{explicit assumptions}$(ia, ib, v)$ with $\nu = 1$, $s = 2$.
\end{proof}

\subsection{Linearization of compensator}
Theorem \ref{linearization of one-sided empirical compensator} is proved by further considering two separate cases where variation in $a, c_{\psi}$ but $b = 0$ and where in $b, c_{\psi}$ but $a = 0$. The next theorem shows the bias correction term in the first situation where $b$ is set to $0$ and only $a, c_{\psi}$ vary. The proof has the similar spirit as linearization of compensator in Jiao and Nielsen (2017, Theorem 8).

\begin{theorem} \label{linearization of one-sided empirical compensator when b is 0}
Let $c_{\psi} = \mathsf{F}_{u}^{-1}(\psi)$. Suppose Assumption \ref{explicit assumptions}$(ia, ib, vb)$ holds with $\nu = 1$, $s = 0$. Then for any $B > 0$ and as $n \to \infty$
\begin{equation*}
\sup_{0 \le \psi \le 1} \sup_{|a| \le n^{1/4 - \eta}B} |n^{1/2} \{ \overline{\mathsf{F}}_{u, n}^{w, p}(a, 0, c_{\psi}) - \overline{\mathsf{F}}_{u, n}^{w, p}(0, 0, c_{\psi}) \} - \mathcal{B}_{\mathsf{F}_{u}, n}(a, 0, c_{\psi})| = \mathrm{O}_{\mathsf{P}}(n^{-2 \eta}),
\end{equation*}
where the bias term is defined as
\begin{equation*}
\mathcal{B}_{\mathsf{F}_{u}, n}(a, 0, c_{\psi}) = \sigma^{p - 1} c_{\psi}^{p} \mathsf{f}_{u}(c_{\psi}) n^{-1/2} \sum_{i=1}^{n} w_{in} n^{-1/2} a c_{\psi}.
\end{equation*}
\end{theorem}

\begin{proof}[\textnormal{\textbf{Proof of Theorem \ref{linearization of one-sided empirical compensator when b is 0}}}]
See the proof of Theorem 8 in Jiao and Nielsen (2017) which is the generalized version of the above theorem, so follow their argument but set $b = 0$ in their context.
\end{proof}

We state next theorem to handle variation in $b, c_{\psi}$ while $a$ is set to $0$.

\begin{theorem} \label{linearization of one-sided empirical compensator when a is 0}
Let $c_{\psi} = \mathsf{F}_{u}^{-1}(\psi)$. Suppose Assumption \ref{explicit assumptions}$(ia, iia, iiia, vb)$ holds with $\nu = 1$, $s = 0$. Then for any $B > 0$ and as $n \to \infty$
\begin{equation*}
\sup_{0 \le \psi \le 1} \sup_{|b| \le n^{1/4 - \eta}B} |n^{1/2} \{ \overline{\mathsf{F}}_{u, n}^{w, p}(0, b, c_{\psi}) - \overline{\mathsf{F}}_{u, n}^{w, p}(0, 0, c_{\psi}) \} - \mathcal{B}_{\mathsf{F}_{u}, n}(0, b, c_{\psi})| = \mathrm{O}_{\mathsf{P}}(n^{-2\eta}),
\end{equation*}
where the bias term, with $\xi_{c_{\psi}} = \mathsf{E} (\Sigma^{-1/2} r_{i} | u_{i}/\sigma = c_{\psi})$, is defined as
\begin{equation*}
\mathcal{B}_{\mathsf{F}_{u}, n}(0, b, c_{\psi}) = \sigma^{p - 1} c_{\psi}^{p} \mathsf{f}_{u}(c_{\psi}) n^{-1/2} \sum_{i=1}^{n} w_{in} ( n^{-1/2} \xi_{c_{\psi}}^{\prime} \Sigma^{1/2} b + z_{in}^{\prime} \Pi b ).
\end{equation*}
\end{theorem}

\begin{proof}[\textnormal{\textbf{Proof of Theorem \ref{linearization of one-sided empirical compensator when a is 0}}}]
The object of interest is
\begin{equation*}
D_{n}(0, b, c_{\psi}) = n^{1/2} \{ \overline{\mathsf{F}}_{u, n}^{w, p}(0, b, c_{\psi}) - \overline{\mathsf{F}}_{u, n}^{w, p}(0, 0, c_{\psi}) \} - \mathcal{B}_{\mathsf{F}_{u}, n}(0, b, c_{\psi}),
\end{equation*}
where $\overline{\mathsf{F}}_{u, n}^{w, p}(0, b, c_{\psi})$ is well-defined due to Assumption \ref{explicit assumptions}$(ia)$ as the indicator function in expectation is bounded by $1$. Let $g_{i}^{0, b, c_{\psi}} = 1_{(u_{i} \le \sigma c_{\psi} + z_{in}^{\prime} \Pi b + n^{-1/2} r_{i}^{\prime} b )} - 1_{(u_{i} \le \sigma c_{\psi})}$. Then denote $h_{i}(0, b, c_{\psi}) = ( n^{-1/2} \xi_{c_{\psi}}^{\prime} \Sigma^{1/2} b + z_{in}^{\prime} \Pi b )/\sigma$ and $s(c_{\psi}) = c_{\psi}^{p} \mathsf{f}_{u}(c_{\psi})$. We define that $S_{i}(0, b, c_{\psi}) = \mathsf{E}_{i - 1} u_{i}^{p} g_{i}^{0, b, c_{\psi}} - \sigma^{p} s(c_{\psi}) h_{i}(0, b, c_{\psi})$ such that $D_{n}(0, b, c_{\psi})$ is expressed as $n^{-1/2} \sum_{i=1}^{n} w_{in} S_{i}(0, b, c_{\psi})$. To apply Lemma \ref{lemma for compensator}, let $Y = u_{i}/\sigma$ and $X = \Sigma^{-1/2} r_{i}$ so their joint density is $\mathsf{f}_{u, r}(y, x)$ for $y \in \mathbb{R}$ and $x \in \mathbb{R}^{d_{x}}$ while $Y, X$ are independent of $\mathcal{F}_{i - 1}$ and $z_{i} \in \mathcal{F}_{i - 1}$. Also let $b_{1} = z_{in}^{\prime} \Pi b / \sigma \in \mathbb{R}$, $b_{2} = n^{-1/2} \Sigma^{1/2} b \in \mathbb{R}^{d_{x}}$, and $c = c_{\psi}$. By Assumption \ref{explicit assumptions}$(iii)$ that the joint density can be decomposed into conditional and marginal ones, the lemma gives
\begin{align*}
|S_{i}(0, b, c_{\psi})| & \le \sigma^{p} \{ \frac{1}{2} (z_{in}^{\prime} \Pi b / \sigma)^{2} + |z_{in}^{\prime} \Pi b / \sigma| |n^{-1/2} \Sigma^{1/2} b| \mathsf{E}|\Sigma^{-1/2} r_{i}| \\
& \quad \,\, + \frac{1}{2} |n^{-1/2} \Sigma^{1/2} b|^{2} \mathsf{E}|\Sigma^{-1/2} r_{i}|^{2} \} \sup_{y \in \mathbb{R}, x \in \mathbb{R}^{d_{x}}} |p y^{p - 1} \mathsf{f}_{u|r}(y|x) + y^{p} \dot{\mathsf{f}}_{u|r, y}(y|x)|,
\end{align*}
uniformly in $\psi$. Note $|z_{in}^{\prime} \Pi b| \le |z_{in}| |\Pi| |b|$ and $|\Sigma^{1/2} b| \le |\Sigma^{1/2}| |b|$ due to the consistency property of the spectral norm. As $|b| \le n^{1/4 - \eta} B$ and Assumption \ref{explicit assumptions}$(iia, iiia)$ with $s = 0$ that $\mathsf{E}|\Sigma^{-1/2} r_{i}|, \mathsf{E}|\Sigma^{-1/2} r_{i}|^{2} < \infty$ and $\sup_{y \in \mathbb{R}, x \in \mathbb{R}^{d_{x}}} |p y^{p - 1} \mathsf{f}_{u|r}(y|x) + y^{p} \dot{\mathsf{f}}_{u|r, y}(y|x)| < \infty$, we have $|S_{i}(0, b, c_{\psi})| = \mathrm{O}(n^{-1/2 - 2 \eta}) (1 + |n^{1/2} z_{in}|^{2})$. Then the triangle inequality gives
\begin{equation*}
|D_{n}(0, b, c_{\psi})| \le n^{-1/2} \sum_{i=1}^{n} |w_{in}| |S_{i}(0, b, c_{\psi})| = \mathrm{O}(n^{-2\eta}) n^{-1} \sum_{i=1}^{n} |w_{in}| (1 + |n^{1/2} z_{in}|^{2}).
\end{equation*}
By Assumption \ref{explicit assumptions}$(vb)$, the term $|D_{n}(0, b, c_{\psi})|$ has order $\mathrm{O}_{\mathsf{P}}(n^{-2\eta})$ uniformly in $\psi$, $b$.
\end{proof}

Then we combine terms $\mathcal{B}_{\mathsf{F}_{u}, n}(a, 0, c_{\psi})$ and $\mathcal{B}_{\mathsf{F}_{u}, n}(0, b, c_{\psi})$ to correct the bias in the case where variation in $a, b, c_{\psi}$ is jointly considered.

\begin{proof}[\textnormal{\textbf{Proof of Theorem \ref{linearization of one-sided empirical compensator}}}]
The interest is
\begin{equation*}
D_{n}(a, b, c_{\psi}) = n^{1/2} \{ \overline{\mathsf{F}}_{u, n}^{w, p}(a, b, c_{\psi}) - \overline{\mathsf{F}}_{u, n}^{w, p}(0, 0, c_{\psi}) \} - \mathcal{B}_{\mathsf{F}_{u}, n}(a, b, c_{\psi}),
\end{equation*}
where $\overline{\mathsf{F}}_{u, n}^{w, p}(a, b, c_{\psi})$ is well-defined by Assumption \ref{explicit assumptions}$(ia)$. Let $c_{\psi^{\dagger}} = c_{\psi}(1 + n^{-1/2} a / \sigma)$. Note $\overline{\mathsf{F}}_{u, n}^{w, p}(a, b, c_{\psi}) = \overline{\mathsf{F}}_{u, n}^{w, p}(0, b, c_{\psi^{\dagger}})$ and $\mathcal{B}_{\mathsf{F}_{u}, n}(a, b, c_{\psi}) = \mathcal{B}_{\mathsf{F}_{u}, n}(0, b, c_{\psi}) + \mathcal{B}_{\mathsf{F}_{u}, n}(a, 0, c_{\psi})$. Add and subtract $n^{1/2} \overline{\mathsf{F}}_{u, n}^{w, p}(a, 0, c_{\psi}) = n^{1/2} \overline{\mathsf{F}}_{u, n}^{w, p}(0, 0, c_{\psi^{\dagger}})$ and $\mathcal{B}_{\mathsf{F}_{u}, n}(0, b, c_{\psi^{\dagger}})$ and apply the trianlge inequality to obtain
\begin{equation*}
|D_{n}(a, b, c_{\psi})| \le |D_{n}(0, b, c_{\psi^{\dagger}})| + |D_{n}(a, 0, c_{\psi})| + |\mathcal{B}_{\mathsf{F}_{u}, n}(0, b, c_{\psi^{\dagger}}) - \mathcal{B}_{\mathsf{F}_{u}, n}(0, b, c_{\psi})|.
\end{equation*}
Since $c_{\psi}$ varies in $\mathbb{R}$ and $|a| \le n^{1/4 - \eta} B$, then $c_{\psi^{\dagger}}$ also moves on the whole real line so variation in $\psi, a$ transfers to variation in $\psi^{\dagger}$. Furthermore, $c_{\psi^{\dagger}}$ converges to $c_{\psi}$ as $n \to \infty$. Theorem \ref{linearization of one-sided empirical compensator when a is 0} shows $|D_{n}(0, b, c_{\psi^{\dagger}})| = \mathrm{O}_{\mathsf{P}}(n^{-2\eta})$ uniformly in $\psi^{\dagger}, b$ by Assumption \ref{explicit assumptions}$(ia, iia, iiia, vb)$. While due to Assumption \ref{explicit assumptions}$(ia, ib, vb)$ Theorem \ref{linearization of one-sided empirical compensator when b is 0} demonstrates $|D_{n}(a, 0, c_{\psi})| = \mathrm{O}_{\mathsf{P}}(n^{-2\eta})$ uniformly in $\psi, a$. Notice
\begin{equation*}
\xi_{y} = \mathsf{E} (\Sigma^{-1/2} r_{i}|\frac{u_{i}}{\sigma} = y) = \int_{\mathbb{R}^{d_{x}}} x \mathsf{f}_{r|u}(x|y) (dx),
\end{equation*}
then $\xi_{y}$ is continuous uniformly in $y$ due to the property of integral and Assumption \ref{explicit assumptions}$(iii)$ that the conditional density $\mathsf{f}_{r|u}(x|y)$ is continuous uniformly in $x \in \mathbb{R}^{d_{x}}, y \in \mathbb{R}$. Also note $y^{p}$, $\mathsf{f}_{u}(y)$ are uniformly continuous in $y \in \mathbb{R}$ by Assumption \ref{explicit assumptions}$(i)$. Since $c_{\psi^{\dagger}}$ only involves in $\mathcal{B}_{\mathsf{F}_{u}, n}(0, b, c_{\psi^{\dagger}})$ through $c_{\psi^{\dagger}}^{p}$, $\mathsf{f}_{u}(c_{\psi^{\dagger}})$, $\xi_{c_{\psi^{\dagger}}}$, then $\mathcal{B}_{\mathsf{F}_{u}, n}(0, b, c_{\psi^{\dagger}})$ converges to $\mathcal{B}_{\mathsf{F}_{u}, n}(0, b, c_{\psi})$ uniformly in all its arguments as $c_{\psi^{\dagger}} \to c_{\psi}$, and subsequently the third term in the inequality of $|D_{n}(a, b, c_{\psi})|$ vanishes.
\end{proof}

\subsection{Tightness of empirical process}
The proof has the similar sprit as the asymptotic equi-continuity argument in Johansen and Nielsen (2016a, Theorem 4.4).
\begin{proof}[\textnormal{\textbf{Proof of Theorem \ref{tightness result for one-sided empirical process}}}]
Let $R(c_{\psi}, c_{\psi^{\dagger}}) = \mathbb{F}_{u, n}^{w, p}(0, 0, c_{\psi^{\dagger}}) - \mathbb{F}_{u, n}^{w, p}(0, 0, c_{\psi})$. The aim is to bound the probability $\mathcal{P} = \mathsf{P}\{ \mathcal{R} > \epsilon \}$ where $\mathcal{R} = \sup_{0 \le \psi \le \psi^{\dagger} \le 1 : \psi^{\dagger} - \psi \le \phi} |R(c_{\psi}, c_{\psi^{\dagger}})|$. Throughout, denote $C > 0$ as usual a constant not depending on $\epsilon$, $n$, $\phi$, which may have different values in different expressions. Assume $\sigma = 1$ without loss of generality and use the metric $\mathsf{d}$ built on the distance function $\mathsf{H}_{s}(x)$.

\emph{$1$. Coefficients $\epsilon$, $\phi$, $s$}. Let $s = 2$. Note $0 < \epsilon < 1$ and $0 < \phi < 1$. Take $\epsilon$, $n$ as given and choose $\phi$ such that $\phi^{(1 - \nu) / 4} \le \epsilon^{2}$ for some $0 < \nu < 1$. A dyadic argument will be used, so given $\epsilon$, $n$, $\phi$ we will select numbers $\overline{m}$, $\underline{m}$ and derive a bound to the probability $\mathcal{P}$ not depending on $\overline{m}$, $\underline{m}$.

\emph{$2$. Fine grid}. Choose $\overline{m}$ so $2^{- \overline{m}} H_{s} \le n^{-1/2} \epsilon \phi^{(1 - \nu) / 4} \le 2^{1 - \overline{m}} H_{s}$ where $H_{s} < \infty$ due to Assumption \ref{explicit assumptions}$(ia)$.

\emph{$3$. Coarse grid}. Choose $\underline{m}$ so $2^{-\underline{m} - 1} H_{s} \le C \phi^{1 - \nu} \le 2^{-\underline{m}} H_{s}$. For large $n$, $\overline{m} > \underline{m}$.

\emph{$4$. Partition the support}. For each of $m = \underline{m}, \ldots, \overline{m}$, use the metric $\mathsf{d}$ to partition $\mathbb{R}$ as laid out in (\ref{partition}) with $K_{m} = 2^{m}$, since $H_{s} < \infty$ by Assumption \ref{explicit assumptions}$(ia)$.

\emph{$5$. Assign $c_{\psi}$ and $c_{\psi^{\dagger}}$ to the partitioned support}. For each of $m = \underline{m}, \ldots, \overline{m}$, there exist $k_{m} \le k_{m}^{\dagger}$ and grid points $c_{k_{m} - 1, m}$, $c_{k_{m}, m}$ and $c_{k_{m}^{\dagger} - 1, m}$, $c_{k_{m}^{\dagger}, m}$ so $c_{k_{m} - 1, m} < c_{\psi} \le c_{k_{m}, m}$ and $c_{k_{m}^{\dagger} - 1, m} < c_{\psi^{\dagger}} \le c_{k_{m}^{\dagger}, m}$. Let $\underline{c}_{m} = c_{k_{m} - 1, m}$, $\overline{c}_{m} = c_{k_{m}, m}$ and $\underline{c}_{m}^{\dagger} = c_{k_{m}^{\dagger} - 1, m}$, $\overline{c}_{m}^{\dagger} = c_{k_{m}^{\dagger}, m}$ so $\underline{c}_{m} < c_{\psi} \le \overline{c}_{m}$ and $\underline{c}_{m}^{\dagger} < c_{\psi^{\dagger}} \le \overline{c}_{m}^{\dagger}$. Then $\overline{c}_{m - 1} = c_{k_{m - 1}, m - 1}$ equals either $\overline{c}_{m} = c_{k_{m}, m}$ or $c_{k_{m} + 1, m}$ so $\mathsf{d}(\overline{c}_{m }, \overline{c}_{m - 1})$ is either $0$ or $2^{-m} H_{s}$. Since the bound $C \phi^{1 - \nu} \le 2^{-\underline{m}} H_{s}$ in item $3$ and Lemma \ref{tightness inequality for d} using Assumption \ref{explicit assumptions}$(ia)$ with some $0 < \nu < 1$, there exist $C_{\nu}, \phi_{0} > 0$ so for $0 \le \phi \le \phi_{0}$
\begin{equation*}
\sup_{0 \le \psi \le \psi^{\dagger} \le 1 : \psi^{\dagger} - \psi \le \phi} \mathsf{d}(c_{\psi}, c_{\psi^{\dagger}}) \le C_{\nu} \phi^{1 - \nu} \le 2^{-\underline{m}} H_{s}.
\end{equation*}
Thus there is at most one $\underline{m}$-grid point in the interval $(c_{\psi}, c_{\psi^{\dagger}})$.

\emph{$6$. Apply chaining}. Relate $c_{\psi}$ to the nearest right fine grid point $\overline{c}_{\overline{m}}$ and $c_{\psi^{\dagger}}$ to the nearest left fine grid point $\underline{c}_{\overline{m}}^{\dagger}$. Then split the interval $(c_{\psi}, c_{\psi^{\dagger}})$ into the three intervals $(c_{\psi}, \overline{c}_{\overline{m}})$, $(\overline{c}_{\overline{m}}, \underline{c}_{\overline{m}}^{\dagger})$, $(\underline{c}_{\overline{m}}^{\dagger}, c_{\psi^{\dagger}})$. If $c_{\psi}$, $c_{\psi^{\dagger}}$ are in the neighbouring $\overline{m}$-interval then $\overline{c}_{\overline{m}} = \underline{c}_{\overline{m}}^{\dagger}$, and if they are in the same $\overline{m}$-interval then $\overline{c}_{\overline{m}} > \underline{c}_{\overline{m}}^{\dagger}$ so $\underline{c}_{\overline{m}} = \underline{c}_{\overline{m}}^{\dagger}$ and $\overline{c}_{\overline{m}} = \overline{c}_{\overline{m}}^{\dagger}$. Thus
\begin{equation*}
R(c_{\psi}, c_{\psi^{\dagger}}) = R(c_{\psi}, \overline{c}_{\overline{m}}) + R(\underline{c}_{\overline{m}}^{\dagger}, c_{\psi^{\dagger}}) - 1_{(\overline{c}_{\overline{m}} > \underline{c}_{\overline{m}}^{\dagger})} R(\underline{c}_{\overline{m}}, \overline{c}_{\overline{m}}) + 1_{(\overline{c}_{\overline{m}} < \underline{c}_{\overline{m}}^{\dagger})} R(\overline{c}_{\overline{m}}, \underline{c}_{\overline{m}}^{\dagger}).
\end{equation*}
An iterative argument can be made for the fourth term. Since $\overline{c}_{\overline{m}} < \underline{c}_{\overline{m}}^{\dagger}$, the coarser $(\overline{m} - 1)$-grid points satisfy $\overline{c}_{\overline{m}} \le \overline{c}_{\overline{m} - 1} \le \underline{c}_{\overline{m} - 1}^{\dagger} \le \underline{c}_{\overline{m}}^{\dagger}$ so that
\begin{equation*}
R(\overline{c}_{\overline{m}}, \underline{c}_{\overline{m}}^{\dagger}) = R(\overline{c}_{\overline{m}}, \overline{c}_{\overline{m} - 1}) + R(\overline{c}_{\overline{m} - 1}, \underline{c}_{\overline{m} - 1}^{\dagger}) + R(\underline{c}_{\overline{m} - 1}^{\dagger}, \underline{c}_{\overline{m}}^{\dagger}).
\end{equation*}
If $\overline{c}_{\overline{m} - 1} = \underline{c}_{\overline{m} - 1}^{\dagger}$, then $R(\overline{c}_{\overline{m} - 1}, \underline{c}_{\overline{m} - 1}^{\dagger}) = 0$ and iteration stops. In this case for $m < \overline{m} - 1$ the $m$-grid points cross over as $\overline{c}_{m} \ge \overline{c}_{\overline{m} - 1} = \underline{c}_{\overline{m} - 1}^{\dagger} \ge \underline{c}_{m}^{\dagger}$. If $\overline{c}_{\overline{m} - 1} < \underline{c}_{\overline{m} - 1}^{\dagger}$, the argument can be made again for $R(\overline{c}_{\overline{m} - 1}, \underline{c}_{\overline{m} - 1}^{\dagger})$. In the $m$-th step, iteration continues if $\overline{c}_{m} < \underline{c}_{m}^{\dagger}$, and noting that if there are no other $m$-grid points between $\overline{c}_{m}$, $\underline{c}_{m}^{\dagger}$, the contribution from the $(m - 1)$-th step is zero. Item $5$ shows there is at most one $\underline{m}$-grid point in the interval $(c_{\psi}, c_{\psi^{\dagger}})$, so the $\underline{m}$-th step either gives a zero contribution or the grid points have crossed over at an earlier stage. Therefore, the fourth term satisfies
\begin{equation*}
1_{(\overline{c}_{\overline{m}} < \underline{c}_{\overline{m}}^{\dagger})} R(\overline{c}_{\overline{m}}, \underline{c}_{\overline{m}}^{\dagger}) = \sum_{m = \underline{m} + 1}^{\overline{m}} 1_{(\overline{c}_{m} < \underline{c}_{m}^{\dagger})} \{ R(\overline{c}_{m}, \overline{c}_{m - 1}) + R(\underline{c}_{m - 1}^{\dagger}, \underline{c}_{m}^{\dagger}) \}.
\end{equation*}
Then apply the triangle inequality to get
\begin{align*}
|R(c_{\psi}, c_{\psi^{\dagger}})| & \le |R(c_{\psi}, \overline{c}_{\overline{m}})| + |R(\underline{c}_{\overline{m}}^{\dagger}, c_{\psi^{\dagger}})| + |1_{(\overline{c}_{\overline{m}} > \underline{c}_{\overline{m}}^{\dagger})} R(\underline{c}_{\overline{m}}, \overline{c}_{\overline{m}})| \\
& \quad + |\sum_{m = \underline{m} + 1}^{\overline{m}} 1_{(\overline{c}_{m} < \underline{c}_{m}^{\dagger})} R(\overline{c}_{m}, \overline{c}_{m - 1})| + |\sum_{m = \underline{m} + 1}^{\overline{m}} 1_{(\overline{c}_{m} < \underline{c}_{m}^{\dagger})} R(\underline{c}_{m - 1}^{\dagger}, \underline{c}_{m}^{\dagger})|.
\end{align*}
Thus, it is argued that $\mathcal{R} \le \mathcal{R}_{1} + \mathcal{R}_{2} + \mathcal{R}_{3} + \mathcal{R}_{4} + \mathcal{R}_{5}$ where
\begin{align*}
\mathcal{R}_{1} & = \sup_{0 \le \psi \le \psi^{\dagger} \le 1 : \psi^{\dagger} - \psi \le \phi} |R(c_{\psi}, \overline{c}_{\overline{m}})|, \\
\mathcal{R}_{2} & = \sup_{0 \le \psi \le \psi^{\dagger} \le 1 : \psi^{\dagger} - \psi \le \phi} |R(\underline{c}_{\overline{m}}^{\dagger}, c_{\psi^{\dagger}})|, \\
\mathcal{R}_{3} & = \sup_{0 \le \psi \le \psi^{\dagger} \le 1 : \psi^{\dagger} - \psi \le \phi} |1_{(\overline{c}_{\overline{m}} > \underline{c}_{\overline{m}}^{\dagger})} R(\underline{c}_{\overline{m}}, \overline{c}_{\overline{m}})|, \\
\mathcal{R}_{4} & = \sup_{0 \le \psi \le \psi^{\dagger} \le 1 : \psi^{\dagger} - \psi \le \phi} |\sum_{m = \underline{m} + 1}^{\overline{m}} 1_{(\overline{c}_{m} < \underline{c}_{m}^{\dagger})} R(\overline{c}_{m}, \overline{c}_{m - 1})|, \\
\mathcal{R}_{5} & = \sup_{0 \le \psi \le \psi^{\dagger} \le 1 : \psi^{\dagger} - \psi \le \phi} |\sum_{m = \underline{m} + 1}^{\overline{m}} 1_{(\overline{c}_{m} < \underline{c}_{m}^{\dagger})} R(\underline{c}_{m - 1}^{\dagger}, \underline{c}_{m}^{\dagger})|.
\end{align*}
By Boole's inequality, it follows that $\mathcal{P} \le \mathcal{P}_{1} + \mathcal{P}_{2} + \mathcal{P}_{3} + \mathcal{P}_{4} + \mathcal{P}_{5}$ where $\mathcal{P}_{1} = \mathsf{P}(\mathcal{R}_{1} > \epsilon / 5)$, $\mathcal{P}_{2} = \mathsf{P}(\mathcal{R}_{2} > \epsilon / 5)$, $\mathcal{P}_{3} = \mathsf{P}(\mathcal{R}_{3} > \epsilon / 5)$, $\mathcal{P}_{4} = \mathsf{P}(\mathcal{R}_{4} > \epsilon / 5)$, $\mathcal{P}_{5} = \mathsf{P}(\mathcal{R}_{5} > \epsilon / 5)$. It then suffices to find bounds to $\mathcal{P}_{1}$, $\mathcal{P}_{2}$, $\mathcal{P}_{3}$, $\mathcal{P}_{4}$, $\mathcal{P}_{5}$.

\emph{$7$. Decompose the term $\mathcal{R}_{1}$}. Since $\underline{c}_{\overline{m}} < c_{\psi} \le \overline{c}_{\overline{m}}$ and $\underline{c}_{\overline{m}} = c_{k_{\overline{m}} - 1, \overline{m}}$, $\overline{c}_{\overline{m}} = c_{k_{\overline{m}}, \overline{m}}$, we have $|J_{i, p}(c_{\psi}, \overline{c}_{\overline{m}})| \le |J_{i, p}(\underline{c}_{\overline{m}}, \overline{c}_{\overline{m}})| = |J_{i, p}(c_{k_{\overline{m}} - 1, \overline{m}}, c_{k_{\overline{m}}, \overline{m}})|$. Then by the triangle inequality, it follows
\begin{equation*}
\mathcal{R}_{1} \le \max_{1 \le k_{\overline{m}} \le K_{\overline{m}}} n^{-1/2} \sum_{i = 1}^{n} |w_{in}| \{ |J_{i, p}(c_{k_{\overline{m}} - 1, \overline{m}}, c_{k_{\overline{m}}, \overline{m}})| + \mathsf{E}_{i - 1} |J_{i, p}(c_{k_{\overline{m}} - 1, \overline{m}}, c_{k_{\overline{m}}, \overline{m}})| \}.
\end{equation*}
Thus a martingale decomposition gives $\mathcal{R}_{1} \le \widetilde{\mathcal{R}}_{1} + 2 \overline{\mathcal{R}}_{1}$, where
\begin{align*}
\widetilde{\mathcal{R}}_{1} & = \max_{1 \le k_{\overline{m}} \le K_{\overline{m}}} |n^{-1/2} \sum_{i = 1}^{n} |w_{in}| \{ |J_{i, p}(c_{k_{\overline{m}} - 1, \overline{m}}, c_{k_{\overline{m}}, \overline{m}})| - \mathsf{E}_{i - 1} |J_{i, p}(c_{k_{\overline{m}} - 1, \overline{m}}, c_{k_{\overline{m}}, \overline{m}})| \}|, \\
\overline{\mathcal{R}}_{1} & = \max_{1 \le k_{\overline{m}} \le K_{\overline{m}}} n^{-1/2} \sum_{i = 1}^{n} |w_{in}| \mathsf{E}_{i - 1} |J_{i, p}(c_{k_{\overline{m}} - 1, \overline{m}}, c_{k_{\overline{m}}, \overline{m}})|.
\end{align*}
By Boole's inequality, it shows that $\mathcal{P}_{1} \le \widetilde{\mathcal{P}}_{1} + 2 \overline{\mathcal{P}}_{1}$ where we have $\widetilde{\mathcal{P}}_{1} = \mathsf{P}(\widetilde{\mathcal{R}}_{1} > \epsilon / 15)$, $\overline{\mathcal{P}}_{1} = \mathsf{P}(\overline{\mathcal{R}}_{1} > \epsilon / 15)$. It suffices to find bounds to $\widetilde{\mathcal{P}}_{1}$, $\overline{\mathcal{P}}_{1}$.

\emph{$8$. Bounding the probability $\widetilde{\mathcal{P}}_{1}$}. Let $z_{l, i} = |w_{in}| |J_{i, p}(c_{k_{\overline{m}} - 1, \overline{m}}, c_{k_{\overline{m}}, \overline{m}})|$ and write $\widetilde{\mathcal{R}}_{1}$ as the maximum of $|n^{-1/2} \sum_{i=1}^{n} (z_{l, i} - \mathsf{E}_{i - 1} z_{l, i})|$ over index $l = k_{\overline{m}}$ with $L = K_{\overline{m}} = 2^{\overline{m}}$. The difference of indicators can be bounded by one, so it holds for the moment condition that $\mathsf{E} z_{l, i}^{4} \le \mathsf{E} |w_{in}|^{4} u_{i}^{4p} = \mathsf{E} |w_{in}|^{4} \mathsf{E} u_{i}^{4p} < \infty$. This is because $u_{i}$ is independent of $w_{in} \in \mathcal{F}_{i - 1}$ and due to Assumption \ref{explicit assumptions}$(ia, va)$ that $\mathsf{E} u_{i}^{4p}, \mathsf{E} |w_{in}|^{4} < \infty$. Consider $1 \le q \le s$ and $s = 2$. The inequality of (\ref{inequality Jxy}) shows uniformly in $l = k_{\overline{m}}$ that
\begin{equation*}
\mathsf{E}_{i - 1} z_{l, i}^{2^{q}} = |w_{in}|^{2^{q}} \mathsf{E}_{i - 1} J_{i, p}^{2^{q}}(c_{k_{\overline{m}} - 1, \overline{m}}, c_{k_{\overline{m}}, \overline{m}}) < (1 + |w_{in}|^{2^{s}}) \mathsf{d}(c_{k_{\overline{m}} - 1, \overline{m}}, c_{k_{\overline{m}}, \overline{m}}).
\end{equation*}
As $\mathsf{d}(c_{k_{\overline{m}} - 1, \overline{m}}, c_{k_{\overline{m}}, \overline{m}}) = H_{s} / K_{\overline{m}}$ and Assumption \ref{explicit assumptions}$(va)$ $n^{-1} \mathsf{E} \sum_{i = 1}^{n} (1 + |w_{in}|^{2^{s}}) = \mathrm{O}(1)$, then we have $D = 2^{- \overline{m}} H_{s}$ such that
\begin{equation*}
\mathsf{E} D_{q} = \mathsf{E} \max_{1 \le l \le L} \sum_{i=1}^{n} \mathsf{E}_{i - 1} z_{l, i}^{2^{q}} < 2^{- \overline{m}} H_{s} \mathsf{E} \sum_{i = 1}^{n} (1 + |w_{in}|^{2^{s}}) \le D n.
\end{equation*}
Apply Lemma \ref{iterated martingale inequality for tightness} with $\kappa = \epsilon / 15$ to get
\begin{equation*}
\widetilde{\mathcal{P}}_{1} \le \frac{15 (2^{\overline{m}} + 1) \theta^{3} H_{s}}{2^{\overline{m}} \epsilon n} + \frac{15 \theta H_{s}}{2^{\overline{m}} \epsilon} + 2^{2} 2^{\overline{m}} \exp(- \frac{\epsilon \theta}{210}) \le \frac{C \theta^{3}}{\epsilon n} + \frac{C \theta}{2^{\overline{m}} \epsilon} + C 2^{\overline{m}} \exp(- \frac{\epsilon \theta}{210}).
\end{equation*}
Choose $\theta = 210 \epsilon^{-1} (\log 2^{2 \overline{m}} + \log \phi^{-1})$. Use the inequality $(x + y)^{3} \le C (x^{3} + y^{3})$ and then $\epsilon^{-2} \le \phi^{- (1 - \nu) / 4}$ and $n^{-1} \epsilon^{2} \le C 2^{- 2 \overline{m}} \phi^{- (1 - \nu) / 2}$ implied by bounds in item $1$, $2$ to obtain
\begin{equation*}
\widetilde{\mathcal{P}}_{1, 1} = \frac{C \theta^{3}}{\epsilon n} \le \frac{C}{\epsilon^{4} n} (\overline{m}^{3} + \log^{3} \frac{1}{\phi}) \le C 2^{- 2 \overline{m}} \phi^{- (1 - \nu) / 2} \phi^{- 3 (1 - \nu) / 4} (\overline{m}^{3} + \log^{3} \frac{1}{\phi}).
\end{equation*}
Since $2^{- \overline{m} / 2} \le 2^{- \underline{m} / 2} \le C \phi^{(1 - \nu) / 2}$ implied by the bound in item $3$, and the functions $m^{3} 2^{- m / 2}$ and $\phi^{(1 - \nu) / 2} \log^{3} \phi^{-1}$ are bounded for $m \ge 1$ and $0 < \phi < 1$, we have
\begin{equation*}
\widetilde{\mathcal{P}}_{1, 1} \le C (2^{- \overline{m} / 2} \phi^{- (1 - \nu) / 2})^{3} (\overline{m}^{3} 2^{- \overline{m} / 2} + 2^{- \overline{m} / 2} \phi^{- (1 - \nu) / 2} \phi^{(1 - \nu) / 2} \log^{3} \frac{1}{\phi}) \phi^{(1 - \nu) / 4} \le C \phi^{(1 - \nu) / 4}.
\end{equation*}
By $\epsilon^{-2} \le \phi^{- (1 - \nu) / 4}$ given in item $1$, it holds
\begin{equation*}
\widetilde{\mathcal{P}}_{1, 2} = \frac{C \theta}{2^{\overline{m}} \epsilon} \le \frac{C}{2^{\overline{m}} \epsilon^{2}} (\overline{m} + \log \frac{1}{\phi}) \le C 2^{- \overline{m}} \phi^{- (1 - \nu) / 4} (\overline{m} + \log \frac{1}{\phi}).
\end{equation*}
Rewrite this bound and argue as for $\widetilde{\mathcal{P}}_{1, 1}$ to get
\begin{equation*}
\widetilde{\mathcal{P}}_{1, 2} \le C 2^{- \overline{m} / 2} \phi^{- (1 - \nu) / 2} (\overline{m} 2^{- \overline{m} / 2} + 2^{- \overline{m} / 2} \phi^{- (1 - \nu) / 2} \phi^{(1 - \nu) / 2} \log \frac{1}{\phi}) \phi^{(1 - \nu) / 4} \le C \phi^{(1 - \nu) / 4}.
\end{equation*}
For $0 < \phi < 1$ and $0 < \nu < 1$ it holds $\phi \le \phi^{(1 - \nu) / 4}$ so that
\begin{equation*}
\widetilde{\mathcal{P}}_{1, 3} = C 2^{\overline{m}} \exp(- \frac{\epsilon \theta}{210}) = C 2^{- \overline{m}} \phi \le C \phi^{(1 - \nu) / 4}.
\end{equation*}
In summary, $\widetilde{\mathcal{P}}_{1} \le \widetilde{\mathcal{P}}_{1, 1} + \widetilde{\mathcal{P}}_{1, 2} + \widetilde{\mathcal{P}}_{1, 3} \le C \phi^{(1 - \nu) / 4}$.

\emph{$9$. Bounding the probability $\overline{\mathcal{P}}_{1}$}. The inequality of (\ref{inequality Jxy}) shows that uniformly in $k_{\overline{m}}$ we have $\mathsf{E}_{i - 1} |J_{i, p}(c_{k_{\overline{m}} - 1, \overline{m}}, c_{k_{\overline{m}}, \overline{m}})| < \mathsf{d}(c_{k_{\overline{m}} - 1, \overline{m}}, c_{k_{\overline{m}}, \overline{m}}) = H_{s} / K_{\overline{m}}$. Then by $K_{\overline{m}} = 2^{\overline{m}}$ and Assumption \ref{explicit assumptions}$(va)$ it follows $\mathsf{E} \overline{\mathcal{R}}_{1} < n^{-1/2} \mathsf{E} \sum_{i = 1}^{n} |w_{in}| H_{s} / K_{\overline{m}} \le C n^{1/2} 2^{- \overline{m}} H_{s}$. By Markov inequality and the bound $2^{- \overline{m}} H_{s} \le n^{-1/2} \epsilon \phi^{(1 - \nu) / 4}$ in item $2$, we get
\begin{equation*}
\overline{\mathcal{P}}_{1} \le \frac{15 \mathsf{E} \overline{\mathcal{R}}_{1}}{\epsilon} \le \frac{C n^{1/2} 2^{- \overline{m}} H_{s}}{\epsilon} \le C \phi^{(1 - \nu) / 4}.
\end{equation*}

\emph{$10$. Bounding the probability $\mathcal{P}_{1}$}. Combine results in item $7 - 9$ to get $\mathcal{P}_{1} \le C \phi^{(1 - \nu) / 4}$.

\emph{$11$. Bounding the probability $\mathcal{P}_{2}$}. This is similar as to show $\mathcal{P}_{1} \le C \phi^{(1 - \nu) / 4}$. Thus the same argument can be made to demonstrate $\mathcal{P}_{2} \le C \phi^{(1 - \nu) / 4}$ through item $7 - 10$.

\emph{$12$. Bounding the probability $\mathcal{P}_{3}$}. Notice $\mathcal{R}_{3} \le \mathcal{R}_{1}$ so $\mathcal{P}_{3} \le \mathcal{P}_{1} \le C \phi^{(1 - \nu) / 4}$.

\emph{$13$. Decompose the term $\mathcal{R}_{4}$}. From partition in item $4, 5$, $\overline{c}_{m - 1} = c_{k_{m - 1}, m - 1}$ equals either $\overline{c}_{m} = c_{k_{m}, m}$ or $c_{k_{m} + 1, m}$, then $\overline{c}_{m}$, $\overline{c}_{m - 1}$ are at most $1$ step apart in the $m$-grid so that $R(\overline{c}_{m}, \overline{c}_{m - 1})$ is either $0$ or $R(c_{k_{m}, m}, c_{k_{m} + 1, m})$. Then we get
\begin{equation*}
|R(\overline{c}_{m}, \overline{c}_{m - 1})| \le |R(c_{k_{m}, m}, c_{k_{m} + 1, m})| \le \max_{1 \le k_{m} \le K_{m}} |R(c_{k_{m} - 1, m}, c_{k_{m}, m})|.
\end{equation*}
Note that the right-hand side on the last inequality does not depend on $\psi$, $\psi^{\dagger}$ so
\begin{equation*}
\mathcal{R}_{4} \le \sum_{m = \underline{m} + 1}^{\overline{m}} \max_{1 \le k_{m} \le K_{m}} |R(c_{k_{m} - 1, m}, c_{k_{m}, m})|.
\end{equation*}

\emph{$14$. Bounding the probability $\mathcal{P}_{4}$}. Notice first that
\begin{equation*}
\sum_{m = \underline{m} + 1}^{\overline{m}} 2^{(\underline{m} - m) / 4} < \sum_{j = 1}^{\infty} 2^{- j / 4} = \frac{1}{2^{1/4} - 1} < 6,
\end{equation*}
so $\sum_{m = \underline{m} + 1}^{\overline{m}} 2^{(\underline{m} - m) / 4} \epsilon / 30 < \epsilon /5$. Then $\mathcal{R}_{4} > \epsilon / 5$ implies
\begin{equation*}
\sum_{m = \underline{m} + 1}^{\overline{m}} \max_{1 \le k_{m} \le K_{m}} |R(c_{k_{m} - 1, m}, c_{k_{m}, m})| > \sum_{m = \underline{m} + 1}^{\overline{m}} \frac{2^{(\underline{m} - m) / 4} \epsilon}{30}.
\end{equation*}
It therefore holds
\begin{equation*}
\mathcal{P}_{4} \le \mathsf{P} [ \bigcup_{m = \underline{m} + 1}^{\overline{m}} \{ \max_{1 \le k_{m} \le K_{m}} |R(c_{k_{m} - 1, m}, c_{k_{m}, m})| > \frac{2^{(\underline{m} - m) / 4} \epsilon}{30} \} ].
\end{equation*}
Then by Boole's inequality, we have
\begin{equation*}
\mathcal{P}_{4} \le \sum_{m = \underline{m} + 1}^{\overline{m}} \mathsf{P} \{ \max_{1 \le k_{m} \le K_{m}} |R(c_{k_{m} - 1, m}, c_{k_{m}, m})| > \frac{2^{(\underline{m} - m) / 4} \epsilon}{30} \}.
\end{equation*}
Let $z_{l, i} = w_{in} J_{i, p}(c_{k_{m} - 1, m}, c_{k_{m}, m})$ and write $R(c_{k_{m} - 1, m}, c_{k_{m}, m})$ as $n^{-1/2} \sum_{i = 1}^{n} (z_{l, i} - \mathsf{E}_{i - 1} z_{l, i})$, where $l$ represents the index $k_{m}$ with $L = K_{m} = 2^{m}$. The moment condition $\mathsf{E} z_{l, i}^{4} < \infty$ holds due to Assumption \ref{explicit assumptions}$(ia, va)$. Consider $1 \le q \le s$ and $s = 2$. The inequality of (\ref{inequality Jxy}) shows uniformly in $l = k_{m}$
\begin{equation*}
\mathsf{E}_{i - 1} z_{l, i}^{2^{q}} = |w_{in}|^{2^{q}} \mathsf{E}_{i - 1} J_{i, p}^{2^{q}}(c_{k_{m} - 1, m}, c_{k_{m}, m}) < (1 + |w_{in}|^{2^{s}}) \mathsf{d}(c_{k_{m} - 1, m}, c_{k_{m}, m}).
\end{equation*}
Since $\mathsf{d}(c_{k_{m} - 1, m}, c_{k_{m}, m}) = H_{s} / K_{m}$ and Assumption \ref{explicit assumptions}$(va)$, we have $D = 2^{-m} H_{s}$ so that
\begin{equation*}
\mathsf{E} D_{q} = \mathsf{E} \max_{1 \le l \le L} \sum_{i = 1}^{n} \mathsf{E}_{i - 1} z_{l, i}^{2^{q}} < 2^{-m} H_{s} \mathsf{E} \sum_{i = 1}^{n} (1 + |w_{in}|^{2^{s}}) \le D n.
\end{equation*}
Apply Lemma \ref{iterated martingale inequality for tightness} with $\kappa = 2^{(\underline{m} - m) / 4} \epsilon / 30$ to get
\begin{equation*}
\mathcal{P}_{4} \le \sum_{m = \underline{m} + 1}^{\overline{m}} \{ \frac{C \theta_{m}^{3}}{2^{(\underline{m} - m) / 4} \epsilon n} + \frac{C \theta_{m}}{2^{m} 2^{(\underline{m} - m) / 4} \epsilon} + C 2^{m} \exp (- \frac{2^{(\underline{m} - m) / 4} \epsilon \theta_{m}}{420}) \}.
\end{equation*}
Choose $\theta_{m} = 420 \epsilon^{-1} 2^{(m - \underline{m}) / 4} (\log 4^{m - \underline{m}} + \log \phi^{-1})$. For the first term above, use the inequality $(x + y)^{3} \le C (x^{3} + y^{3})$  to get
\begin{equation*}
\mathcal{P}_{4, 1} = \sum_{m = \underline{m} + 1}^{\overline{m}} \frac{C \theta_{m}^{3}}{2^{(\underline{m} - m) / 4} \epsilon n} \le \sum_{m = \underline{m} + 1}^{\overline{m}} \frac{C}{2^{\underline{m} - m} \epsilon^{4} n} \{ (m - \underline{m})^{3} + \log^{3} \frac{1}{\phi} \}.
\end{equation*}
Since $\epsilon^{-2} \le \phi^{- (1 - \nu) / 4}$ and $n^{-1} \epsilon^{2} \le C 2^{-2 \overline{m}} \phi^{- (1 - \nu) / 2}$ implied by bounds in item $1$, $2$, then
\begin{equation*}
\mathcal{P}_{4, 1} \le \sum_{m = \underline{m} + 1}^{\overline{m}} C 2^{- 2 \overline{m} - \underline{m} + m} \phi^{-5 (1 - \nu) / 4} \{ (m - \underline{m})^{3} + \log^{3} \frac{1}{\phi} \}.
\end{equation*}
Rewrite $2^{- 2 \overline{m} - \underline{m} + m} = 2^{3 (m - \overline{m}) / 2 - (m - \underline{m}) / 2 - \overline{m} / 2 - 3 \underline{m} / 2}$ so
\begin{equation*}
\mathcal{P}_{4, 1} \le \sum_{m = \underline{m} + 1}^{\overline{m}} C (2^{- \underline{m} / 2} \phi^{-(1 - \nu) / 2})^{3} \frac{2^{3 (m - \overline{m}) / 2}}{2^{(m - \underline{m}) / 2}} \{ \frac{(m - \underline{m})^{3}}{2^{\overline{m} / 2}} + \frac{\phi^{(1 - \nu) / 2}}{2^{\overline{m} / 2} \phi^{(1 - \nu) / 2}} \log^{3} \frac{1}{\phi} \} \phi^{(1 - \nu) / 4}.
\end{equation*}
Using $2^{- \overline{m} / 2} \le 2^{- \underline{m} / 2} \le C \phi^{(1 - \nu) / 2}$ from item $3$ and that geometric sums are finite, argue as for $\widetilde{\mathcal{P}}_{1, 1}$ in item $8$ to obtain $\mathcal{P}_{4, 1} \le C \phi^{(1 - \nu) / 4}$. Then the second term satisfies
\begin{equation*}
\mathcal{P}_{4, 2} = \sum_{m = \underline{m} + 1}^{\overline{m}} \frac{C \theta_{m}}{2^{m} 2^{(\underline{m} - m) / 4} \epsilon} \le \sum_{m = \underline{m} + 1}^{\overline{m}} \frac{C}{2^{(\underline{m} + m) / 2} \epsilon^{2}} \{ (m - \underline{m}) + \log \frac{1}{\phi} \}.
\end{equation*}
Use $\epsilon^{-2} \le \phi^{- (1 - \nu) / 4}$ from item $1$ to get
\begin{equation*}
\mathcal{P}_{4, 2} \le \sum_{m = \underline{m} + 1}^{\overline{m}} C 2^{- (\underline{m} + m) / 2} \phi^{-(1 - \nu) / 4} \{ (m - \underline{m}) + \log \frac{1}{\phi} \}.
\end{equation*}
Rewrite $2^{- (\underline{m} + m) / 2} = 2^{-(m - \underline{m}) / 2} 2^{-\underline{m}}$ so
\begin{equation*}
\mathcal{P}_{4, 2} \le \sum_{m = \underline{m} + 1}^{\overline{m}} C 2^{-\underline{m} / 2} \phi^{-(1 - \nu) / 2} 2^{-(m - \underline{m}) / 2} ( \frac{m - \underline{m}}{2^{\underline{m} / 2}} + \frac{\phi^{(1 - \nu) / 2}}{2^{\underline{m} / 2} \phi^{(1 - \nu) / 2}} \log \frac{1}{\phi} ) \phi^{(1 - \nu) / 4}.
\end{equation*}
Since $2^{- \underline{m} / 2} \le C \phi^{(1 - \nu) / 2}$ from item $3$, argue as for $\mathcal{P}_{4, 1}$ to show $\mathcal{P}_{4, 2} \le C \phi^{(1 - \nu) / 4}$. Then the third term satisfies
\begin{equation*}
\mathcal{P}_{4, 3} = \sum_{m = \underline{m} + 1}^{\overline{m}} C 2^{m} \exp (- \frac{2^{(\underline{m} - m) / 4} \epsilon \theta_{m}}{420}) = \sum_{m = \underline{m} + 1}^{\overline{m}} C 2^{- (m - \underline{m})} 2^{\underline{m}} \phi.
\end{equation*}
As $2^{\underline{m}} \le C \phi^{\nu - 1}$ implied by the bound in item $3$, $\mathcal{P}_{4, 3} \le \sum_{m = \underline{m} + 1}^{\overline{m}} C 2^{- (m - \underline{m})} \phi^{\nu} \le C \phi^{\nu}$. In summary, $\mathcal{P}_{4} \le \mathcal{P}_{4, 1} + \mathcal{P}_{4, 2} + \mathcal{P}_{4, 3} \le C \phi^{(1 - \nu) / 4} + C \phi^{\nu}$.

\emph{$15$. Bounding the probability $\mathcal{P}_{5}$}. Apply the same argument as for $\mathcal{P}_{4}$ through item $13$, $14$ to show $\mathcal{P}_{5} \le C \phi^{(1 - \nu) / 4} + C \phi^{\nu}$.

\emph{$16$. Tightness}. Combine the above result to obtain $\mathcal{P} \le C \phi^{(1 - \nu) / 4} + C \phi^{\nu}$. For a given $0 < \epsilon < 1$ and a large $n$, the only constraint to $\phi \in (0, 1)$ is $\phi^{(1 - \nu) / 4} \le \epsilon^{2}$. Thus we have $\mathcal{P} \to 0$ as $\phi \downarrow 0$.
\end{proof}

\subsection{Uniform $n$-convergence of empirical process}
While the above is to show $n^{1/2}$-convergence result, we then prove $n$-convergence for a particular type of empirical process with the weight $n^{1/2} z_{in} = z_{i}$ and mark $r_{i}$. We first demonstrate uniform convergence for the indicator process, then incorporate the weights and marks by H{\"o}lder inequality.

\begin{theorem} \label{n-convergence results for a particular one-sided process only with indicators}
Let $c_{\psi} = \mathsf{F}_{u}^{-1}(\psi)$. Suppose Assumption \ref{explicit assumptions} holds with $\nu = 1$ and $s \ge 2$ satisfying (\ref{moment condition for b, c uniform convergence}). Then for any $B > 0$ and as $n \to \infty$
\begin{equation*}
\sup_{0 \le \psi \le 1} \sup_{|a|, |b| \le n^{1/4 - \eta} B} \frac{1}{n} \sum_{i = 1}^{n} |1_{(u_{i} \le \sigma c_{\psi} +n^{-1/2} a c_{\psi} + z_{in}^{\prime} \Pi b + n^{-1/2}r_{i}^{\prime}b)} - 1_{(u_{i} \le \sigma c_{\psi})}| = \mathrm{O}_{\mathsf{P}}(n^{-1/4 - \eta}).
\end{equation*}
\end{theorem}

\begin{proof}[\textnormal{\textbf{Proof of Theorem \ref{n-convergence results for a particular one-sided process only with indicators}}}]
Denote $\mathcal{R}_{n}^{\prime} = \sup_{0 \le \psi \le 1} \sup_{|a|, |b| \le n^{1/4 - \eta} B} R_{n}^{\prime}(a, b, c_{\psi})$, where $g_{i}^{a, b, c_{\psi}} = 1_{(u_{i} \le \sigma c_{\psi} +n^{-1/2} a c_{\psi} + z_{in}^{\prime} \Pi b + n^{-1/2}r_{i}^{\prime}b)} - 1_{(u_{i} \le \sigma c_{\psi})}$ and $R_{n}^{\prime}(a, b, c_{\psi}) = n^{-1} \sum_{i = 1}^{n} |g_{i}^{a, b, c_{\psi}}|$, then we show $\mathcal{R}_{n}^{\prime} = \mathrm{O}_{\mathsf{P}}(n^{-1/4 - \eta})$ as $n \to \infty$. We construct the martingale $R_{n, 1}^{\prime}(a, b, c_{\psi})$ by adding and subtracting a compensator $R_{n, 2}^{\prime}(a, b, c_{\psi})$ to $R_{n}^{\prime}(a, b, c_{\psi})$, therefore we have $R_{n}^{\prime}(a, b, c_{\psi}) = R_{n, 1}^{\prime}(a, b, c_{\psi}) + R_{n, 2}^{\prime}(a, b, c_{\psi})$ where
\begin{align*}
R_{n, 1}^{\prime}(a, b, c_{\psi}) = \frac{1}{n} \sum_{i = 1}^{n} \{ |g_{i}^{a, b, c_{\psi}}| - \mathsf{E}_{i - 1} |g_{i}^{a, b, c_{\psi}}| \}, \qquad R_{n, 2}^{\prime}(a, b, c_{\psi}) = \frac{1}{n} \sum_{i = 1}^{n} \mathsf{E}_{i - 1} |g_{i}^{a, b, c_{\psi}}|.
\end{align*}
Triangle inequality gives $\mathcal{R}_{n}^{\prime} \le \mathcal{R}_{n, 1}^{\prime} + \mathcal{R}_{n, 2}^{\prime}$ where
\begin{equation*}
\mathcal{R}_{n, 1}^{\prime} = \sup_{0 \le \psi \le 1} \sup_{|a|, |b| \le n^{1/4 - \eta} B} |R_{n, 1}^{\prime}(a, b, c_{\psi})|, \qquad \mathcal{R}_{n, 2}^{\prime} = \sup_{0 \le \psi \le 1} \sup_{|a|, |b| \le n^{1/4 - \eta} B} R_{n, 2}^{\prime}(a, b, c_{\psi}).
\end{equation*}
It then suffices to prove $\mathcal{R}_{n, 1}^{\prime} = \mathrm{o}_{\mathsf{P}}(n^{-1/2})$ and $\mathcal{R}_{n, 2}^{\prime} = \mathrm{O}_{\mathsf{P}}(n^{-1/4 - \eta})$. Set $w_{in} = 1$, $p = 0$ and adjust the rate $n^{-1}$ instead of $n^{-1/2}$ to the empirical process, then use the argument in Theorem \ref{one-sided empirical process result for all additive and multiplicative shifts} to show $\mathcal{R}_{n, 1}^{\prime} = \mathrm{o}_{\mathsf{P}}(n^{-1/2})$, while $\mathcal{R}_{n, 2}^{\prime} = \mathrm{O}_{\mathsf{P}}(n^{-1/4 - \eta})$ can be demonstrated by the reasoning of Theorem \ref{linearization of one-sided empirical compensator}.
\end{proof}

\begin{proof}[\textnormal{\textbf{Proof of Theorem \ref{n-convergence results for a particular one-sided process}}}]
Let $g_{i}^{a, b, c_{\psi}} = 1_{(u_{i} \le \sigma c_{\psi} +n^{-1/2} a c_{\psi} + z_{in}^{\prime} \Pi b + n^{-1/2}r_{i}^{\prime}b)} - 1_{(u_{i} \le \sigma c_{\psi})}$ and $R_{n}^{z, r \prime}(a, b, c_{\psi}) = n^{-1/2} \sum_{i = 1}^{n} z_{in} r_{i}^{\prime} g_{i}^{a, b, c_{\psi}}$, then we want to show $\mathcal{R}_{n}^{z, r \prime} = \mathrm{o}_{\mathsf{P}}(1)$ as $n \to \infty$ where $\mathcal{R}_{n}^{z, r \prime} = \sup_{0 \le \psi \le 1} \sup_{|a|, |b| \le n^{1/4 - \eta} B} |R_{n}^{z, r \prime}(a, b, c_{\psi})|$. By the triangle inequality we have $|R_{n}^{z, r \prime}(a, b, c_{\psi})| \le n^{-1} \sum_{i = 1}^{n} |n^{1/2} z_{in}| |r_{i}| |g_{i}^{a, b, c_{\psi}}|$. Then apply H{\"o}lder inequality to get
\begin{equation*}
|R_{n}^{z, r \prime}(a, b, c_{\psi})| \le (n^{-1} \sum_{i = 1}^{n} |n^{1/2} z_{in}|^{2} |r_{i}|^{2})^{1/2} (n^{-1} \sum_{i = 1}^{n} |g_{i}^{a, b, c_{\psi}}|^{2})^{1/2}.
\end{equation*}
Notice $|g_{i}^{a, b, c_{\psi}}|^{2} = |g_{i}^{a, b, c_{\psi}}|$ and let $\mathcal{R}_{n}^{\prime} = \sup_{0 \le \psi \le 1} \sup_{|a|, |b| \le n^{1/4 - \eta} B} R_{n}^{\prime}(a, b, c_{\psi})$ where $R_{n}^{\prime}(a, b, c_{\psi}) = n^{-1} \sum_{i = 1}^{n} |g_{i}^{a, b, c_{\psi}}|$ so $\mathcal{R}_{n}^{z, r \prime} \le (n^{-1} \sum_{i = 1}^{n} |n^{1/2} z_{in}|^{2} |r_{i}|^{2})^{1/2} (\mathcal{R}_{n}^{\prime})^{1/2}$. Theorem \ref{n-convergence results for a particular one-sided process only with indicators} shows $\mathcal{R}_{n}^{\prime} = \mathrm{O}_{\mathsf{P}}(n^{-1/4 - \eta})$ so $(\mathcal{R}_{n}^{\prime})^{1/2} = \mathrm{O}_{\mathsf{P}}(n^{-1/8 - \eta / 2}) = \mathrm{o}_{\mathsf{P}}(1)$. Then it suffices to demonstrate $n^{-1} \sum_{i = 1}^{n} |n^{1/2} z_{in}|^{2} |r_{i}|^{2} = \mathrm{O}_{\mathsf{P}}(1)$. This is shown by Markov inequality using independence of $z_{i}$ and $r_{i}$ and Assumption \ref{explicit assumptions}$(iia, va)$ that $\mathsf{E} |n^{1/2} z_{in}|^{2}, \mathsf{E} |r_{i}|^{2} < \infty$.
\end{proof}

\subsection{Results for empirical processes of absolute residuals}
Empirical processes with absolute indicators can be expressed by the difference of two one-sided processes, thus we prove Theorem \ref{absolute empirical process result for all additive and multiplicative shifts}, \ref{linearization of absolute empirical compensator}, \ref{tightness result for absolute empirical process}, \ref{n-convergence results for a particular absolute process} by using the one-sided process results in Theorem \ref{one-sided empirical process result for all additive and multiplicative shifts}, \ref{linearization of one-sided empirical compensator}, \ref{tightness result for one-sided empirical process}, \ref{n-convergence results for a particular one-sided process}. The absolute empirical process results are given under more restrictive Assumption \ref{sufficient assumptions}, so we first investigate the lemma concerning the relationship between Assumption \ref{sufficient assumptions} and \ref{explicit assumptions}.

\begin{lemma} \label{sufficient assumptions imply explicit assumptions}
Suppose $w_{in}$ is either of $1$, $n^{1/2} z_{in} = z_{i}$, $n z_{in} z_{in}^{\prime} = z_{i} z_{i}^{\prime}$ and $p$ is either of $0$, $1$, $2$. Then Assumption \ref{sufficient assumptions}$(ia, iia, iiia, ivb, ivc)$ implies Assumption \ref{explicit assumptions} with $s \ge 2$ satisfying (\ref{moment condition for b, c uniform convergence}).
\end{lemma}

\begin{proof}[\textnormal{\textbf{Proof of Theorem \ref{sufficient assumptions imply explicit assumptions}}}]
Assumption \ref{sufficient assumptions}$(ia)$ shows Assumption \ref{explicit assumptions}$(ia, ic)$, while Assumption \ref{explicit assumptions}$(ib)$ further needs continuous differentiability of the marginal density $\mathsf{f}_{u}$, see discussion in Johansen and Nielsen (2016a, Remark 4.1c). Assumption \ref{sufficient assumptions}$(iia, ivb)$ is the same as Assumption \ref{explicit assumptions}$(iia, iva)$, and Assumption \ref{sufficient assumptions}$(iiia)$ deduces Assumption \ref{explicit assumptions}$(iiia)$ with continuous differentiability of the conditional density $\mathsf{f}_{u|r}$. At last Assumption \ref{sufficient assumptions}$(ivc)$ implies Assumption \ref{explicit assumptions}$(va)$ and $(vb)$ by Markov inequality.
\end{proof}

\begin{proof}[\textnormal{\textbf{Proof of Theorem \ref{absolute empirical process result for all additive and multiplicative shifts}}}]
The term of interest is $\mathcal{G} = \mathbb{G}_{u, n}^{w, p}(a, b, c_{\psi}) - \mathbb{G}_{u, n}^{w, p}(0, 0, c_{\psi})$. Our focus is on the absolute quantile $c_{\psi} = \mathsf{G}_{u}^{-1}(\psi) > 0$ rather than the one-sided quantile $c_{\psi^{\ast}} = \mathsf{F}_{u}^{-1}(\psi^{\ast}) \in \mathbb{R}$. Note $|u_{i}| / \sigma \sim \mathsf{G}_{u}$ and $u_{i} / \sigma \sim \mathsf{F}_{u}$. Since
\begin{align*}
& 1_{(|u_{i} - z_{in}^{\prime} \Pi b - n^{-1/2} r_{i}^{\prime} b| \le \sigma c_{\psi} + n^{-1/2} a c_{\psi})} \\
= & 1_{(u_{i} \le \sigma c_{\psi} + n^{-1/2} a c_{\psi} + z_{in}^{\prime} \Pi b + n^{-1/2} r_{i}^{\prime} b)} - 1_{(u_{i} < - \sigma c_{\psi} - n^{-1/2} a c_{\psi} + z_{in}^{\prime} \Pi b + n^{-1/2} r_{i}^{\prime} b)}
\end{align*}
and by (\ref{one-sided weighted and marked empirical process}), (\ref{absolute empirical process}), we have $\mathbb{G}_{u, n}^{w, p}(a, b, c_{\psi}) = \mathbb{F}_{u, n}^{w, p}(a, b, c_{\psi}) - \lim_{c_{\psi}^{\dagger} \downarrow c_{\psi}} \mathbb{F}_{u, n}^{w, p}(a, b, -c_{\psi}^{\dagger})$ for any $c_{\psi} > 0$. By this and the triangle inequality, then for any $c_{\psi} = \mathsf{G}_{u}^{-1}(\psi) > 0$,
\begin{equation*}
|\mathcal{G}| \le |\mathbb{F}_{u, n}^{w, p}(a, b, c_{\psi}) - \mathbb{F}_{u, n}^{w, p}(0,0,c_{\psi})| + \lim_{c_{\psi}^{\dagger} \downarrow c_{\psi} } |\mathbb{F}_{u, n}^{w, p}(a, b, -c_{\psi}^{\dagger}) - \mathbb{F}_{u, n}^{w, p}(0,0, -c_{\psi}^{\dagger})|.
\end{equation*}
These vanish uniformly in $\psi, a, b$ by Theorem \ref{one-sided empirical process result for all additive and multiplicative shifts} using Assumption \ref{explicit assumptions} with $\nu = 1$ and $s \ge 2$ such that (\ref{moment condition for b, c uniform convergence}) holds. Lemmma \ref{sufficient assumptions imply explicit assumptions} shows Assumption \ref{sufficient assumptions}$(ia, iia, iiia, ivb, ivc)$ suffices.
\end{proof}

\begin{proof}[\textnormal{\textbf{Proof of Theorem \ref{linearization of absolute empirical compensator}}}]
Argue as in the proof of Theorem \ref{absolute empirical process result for all additive and multiplicative shifts} but using Theorem \ref{linearization of one-sided empirical compensator} instead of Theorem \ref{one-sided empirical process result for all additive and multiplicative shifts}. Due to symmetry of $\mathsf{f}_{u}$, approximate the compensator by
\begin{align*}
\mathcal{B}_{\mathsf{G}_{u}, n}(a, b, c_{\psi}) & = \sigma^{p - 1} c_{\psi}^{p} \mathsf{f}_{u}(c_{\psi}) n^{-1/2} \sum_{i=1}^{n} w_{in} [ \{1 + (-1)^{p} \} n^{-1/2} a c_{\psi}  \\
& \quad \,\, + n^{-1/2} \{ \xi_{c_{\psi}} - (-1)^{p} \xi_{-c_{\psi}} \}^{\prime} \Sigma^{1/2} b + \{ 1 - (-1)^{p} \} z_{in}^{\prime} \Pi b ]. \qedhere
\end{align*}
\end{proof}

\begin{proof}[\textnormal{\textbf{Proof of Theorem \ref{tightness result for absolute empirical process}}}]
The same argument as Theorem \ref{absolute empirical process result for all additive and multiplicative shifts} applies while the tightness result in Theorem \ref{tightness result for one-sided empirical process} is used in this proof instead.
\end{proof}

\begin{proof}[\textnormal{\textbf{Proof of Theorem \ref{n-convergence results for a particular absolute process}}}]
Argue along the lines of the proof in Theorem \ref{absolute empirical process result for all additive and multiplicative shifts}, however replace Theorem \ref{one-sided empirical process result for all additive and multiplicative shifts} by Theorem \ref{n-convergence results for a particular one-sided process}.
\end{proof}

\section{Proofs of the main results} \label{proofs of the main results}
We first present some preliminary results for stochastic expansions of two stage least squares statistics. Asymptotics are shown for Algorithm \ref{iterated version of robust 2sls}, then weak convergence for Algorithm \ref{robustified two stage least squares} and \ref{split sample version}. Finally, a new type of Hausman test is established for outlier robustness checks.  

\subsection{Auxillary results for product moments}
The robust estimators in Algorithm \ref{iterated version of robust 2sls} are two stage least squares estimators for selected observations. Outliers are detected through the absolute value of estimated structural errors $u_{i}$ adjusted by its scale $\sigma$ in (\ref{structural equation}). Introduce the indicators in the empirical processes in \S 4 as
\begin{equation} \label{general stochastic indicators}
v_{i}^{a,b,c} = 1_{( |u_{i} - z_{in}^{\prime} \Pi b - n^{-1/2} r_{i}^{\prime} b| \le \sigma c + n^{-1/2} a c )},
\end{equation}
so given the cut-off $c$ the indicators (\ref{2sls indicator for non-outlying observations}) for selecting non-outlying observations have the relationship
\begin{equation} \label{relationship between indicators in the estimators and the empirical process}
v_{i, c}^{(m)} = 1_{(|y_{i} - x_{i}^{\prime} \widehat{\beta}_{c}^{(m)}| \le \widehat{\sigma}_{c}^{(m)} c )} = 1_{( |u_{i} - z_{in}^{\prime} \Pi \widehat{b}_{c}^{(m)} - n^{-1/2} r_{i}^{\prime} \widehat{b}_{c}^{(m)}| \le \sigma c + n^{-1/2} \widehat{a}_{c}^{(m)} c )} = v_{i}^{\widehat{a}_{c}^{(m)}, \widehat{b}_{c}^{(m)}, c},
\end{equation}
where $\widehat{b}_{c}^{(m)} = n^{1/2} (\widehat{\beta}_{c}^{(m)} - \beta)$, $\widehat{a}_{c}^{(m)} = n^{1/2} (\widehat{\sigma}_{c}^{(m)} - \sigma)$ are the estimation errors for $\beta, \sigma$. Thus in Algorithm \ref{iterated version of robust 2sls} the least squares statistics in the estimators (\ref{location estimator for the first stage regression}), (\ref{updated 2sls location}), (\ref{updated 2sls variance}) for $\Pi, \beta, \sigma$ are expressed by the weighted and marked empirical processes analyzed in section \S \ref{weighted and marked empirical process}. The following results describe the asymptotic behaviour of the corresponding product moments.

\begin{lemma} \label{n half asymptotic expansions of empirical processes}
Suppose Assumption \ref{sufficient assumptions}$(ia, iia, iiia, ivb, ivc)$ holds. Then we have
\begin{align*}
n^{-1/2} \sum_{i = 1}^{n} v_{i}^{a, b, c} & = n^{-1/2} \sum_{i=1}^{n} 1_{(|u_{i}| \le \sigma c)} + \mathsf{f}_{u}(c) \frac{2ac + \zeta_{c}^{- \prime} \Sigma^{1/2} b}{\sigma}  + R_{v}(a, b, c), \\
n^{-1/2} \sum_{i = 1}^{n} u_{i}^{2} v_{i}^{a, b, c} & = n^{-1/2} \sum_{i=1}^{n} u_{i}^{2} 1_{(|u_{i}| \le \sigma c)} + \sigma^{2} c^{2} \mathsf{f}_{u}(c) \frac{2ac + \zeta_{c}^{- \prime} \Sigma^{1/2} b}{\sigma} + R_{v u u}(a, b, c), \\
\sum_{i = 1}^{n} z_{in} u_{i} v_{i}^{a, b, c} & = \sum_{i=1}^{n} z_{in} u_{i} 1_{(|u_{i}| \le \sigma c)} + c \mathsf{f}_{u}(c) \sum_{i = 1}^{n} z_{in} (n^{-1/2} \zeta_{c}^{+ \prime} \Sigma^{1/2}  + 2 z_{in}^{\prime} \Pi ) b \\
& \quad \,\, + R_{v z u}(a, b, c), \\
n^{1/2} \sum_{i = 1}^{n} z_{in} z_{in}^{\prime} v_{i}^{a, b, c} & = n^{1/2} \sum_{i=1}^{n} z_{in} z_{in}^{\prime} 1_{(|u_{i}| \le \sigma c)} + \mathsf{f}_{u}(c) M_{z z, n} \frac{2ac + \zeta_{c}^{- \prime} \Sigma^{1/2} b}{\sigma} + R_{v z z}(a, b, c).
\end{align*}
Let $R(a, b, c) = |R_{v}(a, b, c)| + |R_{v u u}(a, b, c)| + |R_{v z u}(a, b, c)| + |R_{v z z}(a, b, c)|$. Then for any $B > 0$ and as $n \to \infty$
\begin{equation*}
\sup_{0 < c < \infty} \sup_{|a|, |b| \le n^{1/4 - \eta} B} |R(a, b, c)| = \mathrm{o}_{\mathsf{P}}(1).
\end{equation*}
\end{lemma}

\begin{proof}[\textnormal{\textbf{Proof of Lemma \ref{n half asymptotic expansions of empirical processes}}}]
The terms of interest are
\begin{equation*}
\mathcal{M}_{n} = n^{-1/2} \sum_{i=1}^{n} w_{in} u_{i}^{p} v_{i}^{a, b, c}, \quad v_{i}^{a,b,c} = 1_{( |u_{i} - z_{in}^{\prime} \Pi b - n^{-1/2} r_{i}^{\prime} b| \le \sigma c + n^{-1/2} a c )}.
\end{equation*}

\emph{$1$. Decompose $\mathcal{M}_{n}$}. Write $\mathcal{M}_{n} = \mathcal{M}_{n, 1} + \mathcal{M}_{n, 2} + \mathcal{M}_{n, 3}$, where
\begin{align*}
\mathcal{M}_{n, 1} & = n^{-1/2} \sum_{i=1}^{n} w_{in} u_{i}^{p} 1_{(|u_{i}| \le \sigma c)}, \quad \mathcal{M}_{n, 2} = n^{-1/2} \sum_{i=1}^{n} w_{in} \mathsf{E}_{i - 1} u_{i}^{p} \{ v_{i}^{a, b, c} - 1_{(|u_{i}| \le \sigma c)} \}, \\
\mathcal{M}_{n, 3} & = n^{-1/2} \sum_{i=1}^{n} w_{in} u_{i}^{p} \{ v_{i}^{a, b, c} - 1_{(|u_{i}| \le \sigma c)} \} - n^{-1/2} \sum_{i=1}^{n} w_{in} \mathsf{E}_{i - 1} u_{i}^{p} \{ v_{i}^{a, b, c} - 1_{(|u_{i}| \le \sigma c)} \}.
\end{align*}
Therefore, the first term in stochastic expansion is $\mathcal{M}_{n, 1}$. We will linearize $\mathcal{M}_{n, 2}$ to obtain the second term, and argue that $\mathcal{M}_{n, 3}$ is small in probability.

\emph{$2$. Linearize $\mathcal{M}_{n, 2}$}. Note $\mathcal{M}_{n, 2} = n^{1/2} \{ \overline{\mathsf{G}}_{u, n}^{w, p}(a, b, c) - \overline{\mathsf{G}}_{u, n}^{w, p}(0, 0, c) \}$, see (\ref{absolute compensator}). Apply Theorem \ref{linearization of absolute empirical compensator} by Assumption \ref{sufficient assumptions}$(ia, iia, iiia, ivc)$ to get $\mathcal{M}_{n, 2} = \mathcal{B}_{\mathsf{G}_{u}, n}(a, b, c) + \mathrm{O}_{\mathsf{P}}(n^{-2 \eta})$. Then $\mathcal{B}_{\mathsf{G}_{u}, n}(a, b, c)$ reduces as desired. Note $0 < \eta \le 1/4$ so $\mathcal{M}_{n, 2} = \mathcal{B}_{\mathsf{G}_{u}, n}(a, b, c) + \mathrm{o}_{\mathsf{P}}(1)$ uniformly in $0 < c < \infty$ and $|a|, |b| \le n^{1/4 - \eta}B$.

\emph{$3$. Bounding $\mathcal{M}_{n, 3}$}. Note $\mathcal{M}_{n, 3} = \mathbb{G}_{u, n}^{w, p}(a, b, c) - \mathbb{G}_{u, n}^{w, p}(0, 0, c)$, see (\ref{absolute empirical process}). By Assumption \ref{sufficient assumptions}$(ia, iia, iiia, ivb, ivc)$, Theorem \ref{absolute empirical process result for all additive and multiplicative shifts} shows $\mathcal{M}_{n, 3} = \mathrm{o}_{\mathsf{P}}(1)$ uniformly in $a, b, c$.
\end{proof}

While the above results are asymptotic expansions with $n^{1/2}$ rate for two stage least squares statistics, stochastic approximation with rate $n$ are also required to demonstrate the main asymptotic results for the estimators and gauge. Thus we establish some $n$-convergence expansions in the following.

\begin{lemma} \label{n asymptotic expansions of empirical processes}
Suppose Assumption \ref{sufficient assumptions}$(ia, iia, iiia, iv)$ holds. Then we have
\begin{align*}
n^{-1} \sum_{i=1}^{n} v_{i}^{a,b,c} & = n^{-1} \sum_{i = 1}^{n} 1_{(|u_{i}| \le \sigma c)} + R_{v}^{\prime}(a, b, c), \\
\sum_{i=1}^{n} z_{in} z_{in}^{\prime} v_{i}^{a, b, c} & = \sum_{i=1}^{n} z_{in} z_{in}^{\prime} 1_{(|u_{i}| \le \sigma c)}  + R_{vzz}^{\prime}(a, b, c),
\end{align*}
where for any $B > 0$ and as $n \to \infty$
\begin{equation*}
\sup_{0 < c < \infty} \sup_{|a|, |b| \le n^{1/4 - \eta} B} |R_{v}^{\prime}(a, b, c)| + |R_{vzz}^{\prime}(a, b, c)| = \mathrm{o}_{\mathsf{P}}(1).
\end{equation*}
\end{lemma}

\begin{proof}[\textnormal{\textbf{Proof of Lemma \ref{n asymptotic expansions of empirical processes}}}]
Adjust the first and fourth items in Lemma \ref{n half asymptotic expansions of empirical processes} by the rate $n^{-1/2}$ and then observe the second terms in the expansions. We find $\sup_{c \in \mathbb{R}_{+}} c \mathsf{f}_{u}(c) < \infty$, $\sup_{c \in \mathbb{R}_{+}} |\zeta_{c}^{-}| \mathsf{f}_{u}(c) < \infty$ by Assumption \ref{sufficient assumptions}$(ia)$, and $n^{-1/2} |a|, n^{-1/2} |b| \le n^{-1/4 - \eta} B$, while $M_{z z, n} \overset{\mathsf{P}}{\to} M_{z z}$ by Assumption \ref{sufficient assumptions}$(iva)$ and $M_{z z}$, $\Sigma$, $\sigma$ are treated as constant. Notice that $\zeta_{c}^{-}$ can be expressed as the function of $c$, for instance under normality (\ref{normality for scaled error vector}) it holds $\zeta_{c}^{-} = 2 \Omega c$. Indeed the second terms in expansions are $\mathrm{O}_{\mathsf{P}}(n^{-1/4 - \eta})$ whereas the third terms are $\mathrm{o}_{\mathsf{P}}(n^{-1/2})$ so combine them to show the remainder terms are $\mathrm{o}_{\mathsf{P}}(1)$.
\end{proof}

\begin{lemma} \label{n asymptotic expansion for product moment in the first stage}
Suppose Assumption \ref{sufficient assumptions}$(ia, iia, iiia, ivb, ivc)$ holds. Then we have
\begin{equation*}
n^{-1/2} \sum_{i = 1}^{n} z_{in} r_{i}^{\prime} v_{i}^{a, b, c} = n^{-1/2} \sum_{i = 1}^{n} z_{in} r_{i}^{\prime} 1_{(|u_{i}| \le \sigma c)} + R_{vzr}^{\prime}(a, b, c),
\end{equation*}
where for any $B > 0$ and as $n \to \infty$
\begin{equation*}
\sup_{0 < c < \infty} \sup_{|a|, |b| \le n^{1/4 - \eta} B} |R_{vzr}^{\prime}(a, b, c)| = \mathrm{o}_{\mathsf{P}}(1).
\end{equation*}
\end{lemma}

\begin{proof}[\textnormal{\textbf{Proof of Lemma \ref{n asymptotic expansion for product moment in the first stage}}}]
Decompose the term of interest as
\begin{equation*}
n^{-1/2} \sum_{i = 1}^{n} z_{in} r_{i}^{\prime} v_{i}^{a, b, c} = n^{-1/2} \sum_{i = 1}^{n} z_{in} r_{i}^{\prime} 1_{(|u_{i}| \le \sigma c)} + n^{-1/2} \sum_{i = 1}^{n} z_{in} r_{i}^{\prime} \{ v_{i}^{a, b, c} - 1_{(|u_{i}| \le \sigma c)} \}.
\end{equation*}
Then Theorem \ref{n-convergence results for a particular absolute process} shows the second term in the above is $\mathrm{o}_{\mathsf{P}}(1)$ uniformly in $a, b, c$.
\end{proof}

\subsection{Results for Algorithm \ref{iterated version of robust 2sls}} \label{results for iterated 2sls}
Since the indicators $v_{i, c}^{(m)}$ in 2SLS statistics equals to $v_{i}^{\widehat{a}_{c}^{(m)}, \widehat{b}_{c}^{(m)}, c}$, see (\ref{relationship between indicators in the estimators and the empirical process}), all processes $M_{n}(\widehat{a}, \widehat{b}, c)$ we considered depend on estimation errors $\widehat{b} = n^{1/2} (\widehat{\beta} - \beta)$, $\widehat{a} = n^{1/2} (\widehat{\sigma} - \sigma)$. If $\widehat{a}, \widehat{b}$ are bounded by a compact set in probability, then $M_{n}(\widehat{a}, \widehat{b}, c)$ can be studied by analyzing the behaviour of $M_{n}(a, b, c)$ uniformly in $a, b \in \Theta$ for a compact set $\Theta$. This is due to the lemma below.

\begin{lemma} \label{lemma to treat randomness of estimation errors}
If for any $\epsilon > 0$ there exists a compact set $\Theta$ so $\mathsf{P}(\widehat{a}, \widehat{b} \in \Theta^{c}) < \epsilon$ as $n \to \infty$. Then we have $\mathsf{P}\{ M_{n}(\widehat{a}, \widehat{b}, c) > \epsilon \} \le \mathsf{P}\{ \sup_{a, b \in \Theta} |M_{n}(a, b, c)| > \epsilon \} + \epsilon$.
\end{lemma}

\begin{proof}[\textnormal{\textbf{Proof of Lemma \ref{lemma to treat randomness of estimation errors}}}]
Let $\mathcal{A} = \{ |M_{n}(\widehat{a}, \widehat{b}, c)| > \epsilon \}, \mathcal{B} = (\widehat{a}, \widehat{b} \in \Theta)$. Then Boole's inequality shows $\mathsf{P}(\mathcal{A}) \le \mathsf{P}(\mathcal{A} \cap \mathcal{B}) + \mathsf{P}(\mathcal{B}^{c})$ as $\mathcal{A} = \mathcal{A} \cap (\mathcal{B} \cup \mathcal{B}^{c}) = (\mathcal{A} \cap \mathcal{B}) \cup (\mathcal{A} \cap \mathcal{B}^{c})$. The probability $\mathsf{P}(\mathcal{A} \cap \mathcal{B})$ is bounded by $\mathsf{P}\{ \sup_{a, b \in \Theta} |M_{n}(a, b, c)| > \epsilon \}$ while $\mathsf{P}(\mathcal{B}^{c}) < \epsilon$ by assumption.
\end{proof}

First demonstrate consistency of $\widehat{\Pi}_{c}^{(m + 1)}$ uniformly in $c \in [c_{+}, \infty)$ for a finite number $c_{+} > 0$ given that $\widehat{b}_{c}^{(m)} = n^{1/2} (\widehat{\beta}_{c}^{(m)} - \beta)$, $\widehat{a}_{c}^{(m)} = n^{1/2} (\widehat{\sigma}_{c}^{(m)} - \sigma)$ are bounded in probability for any $m \in [0, \infty)$. The lemma below explores the uniform limiting behaviour of least squares statistics in the first stage regression (\ref{first stage regression}).

\begin{lemma} \label{n uniform convergence of first stage least squares statistics}
Let $c_{\psi} = \mathsf{G}_{u}^{-1}(\psi)$. Suppose Assumption \ref{sufficient assumptions}$(iia, ivc)$ holds. Then as $n \to \infty$
\begin{equation*}
\sup_{0 \le \psi \le 1} |n^{-1/2} \sum_{i = 1}^{n} z_{in} r_{i}^{\prime} 1_{(|u_{i}| \le \sigma c_{\psi})}| = \mathrm{o}_{\mathsf{P}}(1).
\end{equation*}
\end{lemma}

\begin{proof}[\textnormal{\textbf{Proof of Lemma \ref{n uniform convergence of first stage least squares statistics}}}]
First observe that $\psi$ is varying in the compact set $[0, 1]$. Moreover $z_{i} r_{i}^{\prime} 1_{(|u_{i}| \le \sigma c_{\psi})}$ is continuous in $\psi \in [0, 1]$ with probability $1$ since the only discontinuous point for $c_{\psi}$ is at the realized value of $u_{i}/\sigma$ while $u_{i}/\sigma$ is assumed to have the continuous density $\mathsf{f}_{u}$ by Assumption \ref{sufficient assumptions}. Then $|z_{i} r_{i}^{\prime} 1_{(|u_{i}| \le \sigma c_{\psi})}| \le |z_{i}| |r_{i}|$ and $\mathsf{E} |z_{i}| |r_{i}| < \infty$ since $z_{i}$ is independent of $r_{i}, u_{i}$ by Assumption \ref{sufficient assumptions} and $\mathsf{E} |z_{i}|, \mathsf{E} |r_{i}| < \infty$ by Assumption \ref{sufficient assumptions}$(iia, ivc)$. Therefore apply Uniform Law of Large Numbers, see Theorem 2 in Jennrich (1969), to show
\begin{equation*}
n^{-1} \sum_{i = 1}^{n} z_{i} r_{i}^{\prime} 1_{(|u_{i}| \le \sigma c_{\psi})} \overset{\mathsf{P}}{\to} \mathsf{E} z_{i} r_{i}^{\prime} 1_{(|u_{i}| \le \sigma c_{\psi})} = \mathsf{E} z_{i} \mathsf{E} r_{i}^{\prime} 1_{(|u_{i}| \le \sigma c_{\psi})} = 0_{d_{z} \times d_{x}}
\end{equation*}
uniformly in $\psi$. The last equality follows $\mathsf{E} z_{i} = 0_{d_{z}}$ and $\mathsf{E} |r_{i}^{\prime} 1_{(|u_{i}| \le \sigma c_{\psi})}| \le \mathsf{E} |r_{i}| < \infty$.
\end{proof}

\begin{remark} \label{assumption on expectation of r and indicator}
Notice in general that $\mathsf{E} r_{i}^{\prime} 1_{(|u_{i}| \le \sigma c_{\psi})} \neq 0_{d_{x}}^{\prime}$ except when $\psi = 1$ and $c_{\psi} = \infty$ so $\mathsf{E} r_{i}^{\prime} 1_{(|u_{i}| \le \sigma c_{\psi})} = \mathsf{E} r_{i}^{\prime} = 0_{d_{x}}^{\prime}$. Write the expectation as an integral
\begin{equation*}
\mathsf{E} r_{i}^{\prime} 1_{(|u_{i}| \le \sigma c_{\psi})} = \int_{x \in \mathbb{R}^{d_{x}}} \int_{-c_{\psi}}^{c_{\psi}} x^{\prime} \mathsf{f}_{u, r}(y, x) dy (dx) \Sigma^{1/2}.
\end{equation*}
The above expected value equals $0_{d_{x}}^{\prime}$ if for any $y \in \mathbb{R}, x \in \mathbb{R}^{d_{x}}$ the density $\mathsf{f}_{u, r}$ satisfies $\mathsf{f}_{u, r}(y, x) = \mathsf{f}_{u, r}(y, -x)$ or $\mathsf{f}_{u, r}(y, x) = - \mathsf{f}_{u, r}(-y, x)$ unless $\psi = 1$. However, symmetricity does not hold for the joint density in its second argument, for instance, under the normality assumption (\ref{normality for scaled error vector}) when $\Omega \neq 0_{d_{x}}$, while obviously $\mathsf{f}_{u, r}$ is not an odd function.
\end{remark}

\begin{proof}[\textnormal{\textbf{Proof of Theorem \ref{consistency of location estimator in the first stage}}}]
For any $m \in [0, \infty)$ the $m + 1$ step estimator for $\Pi$ defined in (\ref{location estimator for the first stage regression}) satisfies
\begin{equation} \label{expression for location estimator in the first stage regression}
\widehat{\Pi}_{c}^{(m + 1)} - \Pi = (\sum_{i = 1}^{n} z_{in} z_{in}^{\prime} v_{i, c}^{(m)})^{-1} (n^{-1/2} \sum_{i = 1}^{n} z_{in} r_{i}^{\prime} v_{i, c}^{(m)}).
\end{equation}
First notice $v_{i, c}^{(m)} = v_{i}^{\widehat{a}_{c}^{(m)}, \widehat{b}_{c}^{(m)}, c}$, see (\ref{relationship between indicators in the estimators and the empirical process}). Suppose $|\widehat{b}_{c}^{(m)}| + |\widehat{a}_{c}^{(m)}| = \mathrm{O}_{\mathsf{P}}(1)$ and by Assumption \ref{sufficient assumptions}$(ia, iia, iiia, iv)$, Lemma \ref{n asymptotic expansions of empirical processes} and Lemma \ref{n asymptotic expansion for product moment in the first stage} show asymptotic expansions for least squares statistics for non-outlying observations. Lemma \ref{lemma to treat randomness of estimation errors} allows to replace deterministic terms $b, a$ by stochastic estimation errors $\widehat{b}_{c}^{(m)}, \widehat{a}_{c}^{(m)}$. Substitute these expansions into (\ref{expression for location estimator in the first stage regression}) to obtain
\begin{equation*}
\widehat{\Pi}_{c}^{(m + 1)} - \Pi = (\sum_{i = 1}^{n} z_{in} z_{in}^{\prime} 1_{(|u_{i}| \le \sigma c)})^{-1} (n^{-1/2} \sum_{i = 1}^{n} z_{in} r_{i}^{\prime} 1_{(|u_{i}| \le \sigma c)}) + R_{\Pi}^{\prime}(\widehat{a}_{c}^{(m)}, \widehat{b}_{c}^{(m)}, c),
\end{equation*}
where $R_{\Pi}^{\prime}(a, b, c)$ vanishes uniformly in $c_{+} \le c < \infty$ and $|a|, |b| \le B$. Check the moment condition $\mathsf{E} |z_{i} z_{i}^{\prime} 1_{(|u_{i}| \le \sigma c)}| \le \mathsf{E} |z_{i}|^{2} < \infty$ by Assumption \ref{sufficient assumptions}$(ivc)$, Law of Large Numbers shows for any $c \in [c_{+}, \infty)$
\begin{equation*}
n^{-1} \sum_{i = 1}^{n} z_{i} z_{i}^{\prime} 1_{(|u_{i}| \le \sigma c)} \overset{\mathsf{P}}{\to} \mathsf{E} z_{i} z_{i}^{\prime} 1_{(|u_{i}| \le \sigma c)} = \mathsf{E} z_{i} z_{i}^{\prime} \mathsf{E} 1_{(|u_{i}| \le \sigma c)} = M_{z z} \psi,
\end{equation*}
where the first equality is because $z_{i}$ is independent of $u_{i}$. Theorem \ref{tightness result for absolute empirical process} demonstrates tightness of the above sequence so $\sum_{i = 1}^{n} z_{in} z_{in}^{\prime} 1_{(|u_{i}| \le \sigma c)}$ converges in probability to $M_{z z} \psi$ uniformly in $c \in [c_{+}, \infty)$. A key to this is that $c$ is bounded away from $0$ and $M_{z z} > 0$ by Assumption \ref{sufficient assumptions}$(iva)$ so $\psi, M_{z z} \psi > 0$. Thus further by Lemma \ref{n uniform convergence of first stage least squares statistics} due to Assumption \ref{sufficient assumptions}$(iia, ivc)$ and by Slutsky's theorem uniform consistency then follows.
\end{proof}

Then we prove the theorem about stochastic expansion of the updated estimator in Algorithm \ref{iterated version of robust 2sls}.

\begin{proof}[\textnormal{\textbf{Proof of Theorem \ref{1-step stochastic expansion allows varying cut-off c for the iterated 2sls}}}]
First notice that by substituting the $\Pi$ estimator (\ref{location estimator for the first stage regression}) the inverse term in the updated $\beta$ estimator (\ref{updated 2sls location}) is expressed as
\begin{align}
\widehat{\Pi}_{c}^{(m + 1) \prime} \sum_{i=1}^{n} z_{i}z_{i}^{\prime} v_{i, c}^{(m)} \widehat{\Pi}_{c}^{(m + 1)} & = \widehat{\Pi}_{c}^{(m + 1) \prime} (\sum_{i=1}^{n} z_{i}z_{i}^{\prime} v_{i, c}^{(m)}) (\sum_{i = 1}^{n} z_{i}z_{i}^{\prime} v_{i, c}^{(m)})^{-1} (\sum_{i = 1}^{n} z_{i}x_{i}^{\prime} v_{i, c}^{(m)}) \nonumber \\
& = \widehat{\Pi}_{c}^{(m + 1) \prime} \sum_{i = 1}^{n} z_{i}x_{i}^{\prime} v_{i, c}^{(m)}. \label{two expressions for denomiator of beta estimator}
\end{align}
Apply (\ref{two expressions for denomiator of beta estimator}) and (\ref{structural equation}) into the $m + 1$ step $\beta$ estimator (\ref{updated 2sls location}) so for any $m \in [0, \infty)$
\begin{equation*}
\widehat{\beta}_{c}^{(m+1)} = \beta + (\widehat{\Pi}_{c}^{(m + 1) \prime} \sum_{i = 1}^{n} z_{i}x_{i}^{\prime} v_{i, c}^{(m)})^{-1} (\widehat{\Pi}_{c}^{(m + 1) \prime} \sum_{i=1}^{n} z_{i} u_{i} v_{i, c}^{(m)}).
\end{equation*}
Use (\ref{two expressions for denomiator of beta estimator}) again and adjust the above equation by the rate $n^{1/2}$ to get
\begin{equation} \label{expression for beta in the m+1 step after manipulations}
n^{1/2} (\widehat{\beta}_{c}^{(m+1)} - \beta) = (\widehat{\Pi}_{c}^{(m + 1) \prime} \sum_{i=1}^{n} z_{in} z_{in}^{\prime} v_{i, c}^{(m)} \widehat{\Pi}_{c}^{(m + 1)})^{-1} (\widehat{\Pi}_{c}^{(m + 1) \prime} \sum_{i=1}^{n} z_{in} u_{i} v_{i, c}^{(m)}).
\end{equation}
Suppose $|\widehat{b}_{c}^{(m)}| + |\widehat{a}_{c}^{(m)}| = \mathrm{O}_{\mathsf{P}}(1)$, Theorem \ref{consistency of location estimator in the first stage} by Assumption \ref{sufficient assumptions}$(ia, iia, iiia, iv)$ shows $\widehat{\Pi}_{c}^{(m + 1)} \overset{\mathsf{P}}{\to} \Pi$ uniformly in $c \in [c_{+}, \infty)$. Apply the same argument in the proof of Theorem \ref{consistency of location estimator in the first stage}, so by Assumption \ref{sufficient assumptions}$(ia, iia, iiia, iv)$ it allows to substitute expansions in Lemma \ref{n half asymptotic expansions of empirical processes} and \ref{n asymptotic expansions of empirical processes}  into (\ref{expression for beta in the m+1 step after manipulations}). Then by Slutsky's theorem we have
\begin{align*}
\widehat{b}_{c}^{(m + 1)} & = (\sum_{i = 1}^{n} \tilde{x}_{in} \tilde{x}_{in}^{\prime} 1_{(|u_{i}| \le \sigma c)})^{-1} \{ c \mathsf{f}_{u}(c) \sum_{i = 1}^{n} \tilde{x}_{in} (n^{-1/2} \zeta_{c}^{+ \prime} \Sigma^{1/2} + 2 \tilde{x}_{in}^{\prime}) \widehat{b}_{c}^{(m)} \\
& \quad \,\, + \sum_{i = 1}^{n} \tilde{x}_{in} u_{i} 1_{(|u_{i}| \le \sigma c)} \} + R_{\beta}(\widehat{a}_{c}^{(m)}, \widehat{b}_{c}^{(m)}, c),
\end{align*}
where $R_{\beta}(a, b, c)$ vanishes uniformly in $c_{+} \le c < \infty$ and $|a|, |b| \le B$. Explore the inverse term in the above equation
\begin{equation*}
\sum_{i = 1}^{n} \tilde{x}_{in} \tilde{x}_{in}^{\prime} 1_{(|u_{i}| \le \sigma c)} = \sum_{i = 1}^{n} \tilde{x}_{in} \tilde{x}_{in}^{\prime} \psi + \sum_{i = 1}^{n} \tilde{x}_{in} \tilde{x}_{in}^{\prime} (1_{(|u_{i}| \le \sigma c)} - \psi) = M_{\tilde{x} \tilde{x}, n} \psi + \mathrm{o}_{\mathsf{P}}(1),
\end{equation*}
where the second term converges in probability to zero uniformly in $c$ by Law of Large Numbers and tightness property shown in Theorem \ref{tightness result for absolute empirical process}, see the similar argument in the proof of Theorem \ref{consistency of location estimator in the first stage}. Note $n^{-1/2} \sum_{i = 1}^{n} \tilde{x}_{in} \overset{\mathsf{P}}{\to} \mathsf{E} \tilde{x}_{i} = \Pi^{\prime} \mathsf{E} z_{i} = 0_{d_{x}}$ as $\mathsf{E} z_{i} = 0_{d_{z}}$. Then explore the first term in $\widehat{b}_{c}^{(m + 1)}$ it follows
\begin{equation*}
(M_{\tilde{x} \tilde{x}, n} \psi)^{-1} c \mathsf{f}_{u}(c) n^{-1/2} \sum_{i = 1}^{n} \tilde{x}_{in} \zeta_{c}^{+ \prime} \Sigma^{1/2} \widehat{b}_{c}^{(m)} \overset{\mathsf{P}}{\to} (M_{\tilde{x} \tilde{x}} \psi)^{-1} c \mathsf{f}_{u}(c) 0_{d_{x}} \zeta_{c}^{+ \prime} \Sigma^{1/2} \widehat{b}_{c}^{(m)} = 0_{d_{x}},
\end{equation*}
uniformly in $c_{+} \le c < \infty$. Uniformity requires $\sup_{c \in \mathbb{R}_{+}} c |\zeta_{c}^{+}| \mathsf{f}_{u}(c) < \infty$ by Assumption \ref{sufficient assumptions}$(ia)$ and $\widehat{b}_{c}^{(m)} = \mathrm{O}_{\mathsf{P}}(1)$. Also notice $\zeta_{c}^{+} = 0_{d_{x}}$ if $\mathsf{E} x_{i} u_{i} = 0_{d_{x}}$ and then $\mathsf{E} r_{i} u_{i} = 0_{d_{x}}$ or under normality assumption (\ref{normality for scaled error vector}) for $(u_{i} / \sigma, \Sigma^{-1/2} r_{i})$ so the above term also vanishes in these two situations. Then we obtain
\begin{equation*}
\widehat{b}_{c}^{(m + 1)} = \frac{2c \mathsf{f}_{u}(c)}{\psi} \widehat{b}_{c}^{(m)} + (M_{\tilde{x} \tilde{x}, n} \psi)^{-1} \sum_{i = 1}^{n} \tilde{x}_{in} u_{i} 1_{(|u_{i}| \le \sigma c)} + R_{\beta}(\widehat{a}_{c}^{(m)}, \widehat{b}_{c}^{(m)}, c),
\end{equation*}
where $R_{\beta}(a, b, c)$ vanishes uniformly. A key to above derivation is that $c$ is bounded away from $0$ and $M_{\tilde{x} \tilde{x}} > 0$ by Assumption \ref{sufficient assumptions}$(iva)$ so $\psi, M_{\tilde{x} \tilde{x}} \psi > 0$.

The expansion for $\sigma$ is established in Lemma B.4 of Jiao and Kurle (2026).
\end{proof}

\begin{remark} \label{consistency of location estimator}
 For any $m \in [0, \infty)$ Algorithm \ref{iterated version of robust 2sls} applies indicators $v_{i, c}^{(m)}$ in (\ref{2sls indicator for non-outlying observations}) to choose non-outlying observations for both the first stage regression (\ref{first stage regression}) and structural equation (\ref{structural equation}), then updated estimators $\widehat{\Pi}_{c}^{(m + 1)}$, $\widehat{\beta}_{c}^{(m+1)}$, $(\widehat{\sigma}_{c}^{(m+1)})^{2}$ are obtained based on the same sub-sample selected, see (\ref{location estimator for the first stage regression}), (\ref{updated 2sls location}), (\ref{updated 2sls variance}). The equality (\ref{two expressions for denomiator of beta estimator}) comes from the fact that the sub-sample chosen for non-outlying observations is the same for (\ref{first stage regression}) and (\ref{structural equation}), while it is significant for proving consistency of $\beta$ estimator. Otherwise without (\ref{two expressions for denomiator of beta estimator}), we can only get
\begin{align*}
\widehat{\beta}_{c}^{(m + 1)} & = (\widehat{\Pi}_{c}^{(m + 1) \prime} \sum_{i=1}^{n} z_{i}z_{i}^{\prime} v_{i, c}^{(m)} \widehat{\Pi}_{c}^{(m + 1)})^{-1} (\widehat{\Pi}_{c}^{(m + 1) \prime} \sum_{i = 1}^{n} z_{i}x_{i}^{\prime} v_{i, c}^{(m)}) \beta \\
& \quad \,\, + (\widehat{\Pi}_{c}^{(m + 1) \prime} \sum_{i=1}^{n} z_{i}z_{i}^{\prime} v_{i, c}^{(m)} \widehat{\Pi}_{c}^{(m + 1)})^{-1} (\widehat{\Pi}_{c}^{(m + 1) \prime} \sum_{i=1}^{n} z_{i} u_{i} v_{i, c}^{(m)}),
\end{align*}
and the first product term before $\beta$ does not reduce to the identity matrix so form the bias term, whereas the second product moment is treated by the same argument in the proof of Theorem \ref{1-step stochastic expansion allows varying cut-off c for the iterated 2sls}. Thus the estimator for $\beta$ is inconsistent if two different sub-samples are used for running the first stage and structural equation.
\end{remark}

The next proof is to establish the tightness result.

\begin{proof}[\textnormal{\textbf{Proof of Theorem \ref{tightness allows varying cut-off c for the iterated 2sls}}}]
Due to Assumption \ref{sufficient assumptions}$(ia, iia, iiia, iv)$, Theorem \ref{1-step stochastic expansion allows varying cut-off c for the iterated 2sls} demonstrates for any $m \in [0, \infty)$
\begin{equation} \label{autoregressive equation for theta with 0 remainder term allows the varying cut-off c}
\widehat{\theta}_{c}^{(m + 1)} = \Gamma_{c} \widehat{\theta}_{c}^{(m)} + \Lambda_{c} + R_{\theta}(\widehat{\theta}_{c}^{(m)}, c),
\end{equation}
where the remainder term satisfies $\sup_{c_{+} \le c < \infty} \sup_{|\theta| \le B} |R_{\theta}(\theta, c)| = \mathrm{o}_{\mathsf{P}}(1)$ and
\begin{align}
\widehat{\theta}_{c}^{(m)} & =
                                          \begin{pmatrix}
                                          \widehat{b}_{c}^{(m)} \\
                                          \widehat{a}_{c}^{(m)}
                                          \end{pmatrix}
=
\begin{Bmatrix}
n^{1/2} (\widehat{\beta}_{c}^{(m)} - \beta) \\
n^{1/2} (\widehat{\sigma}_{c}^{(m)} - \sigma)
\end{Bmatrix}
,
\Gamma_{c} =
                         \begin{Bmatrix}
                         \frac{2c \mathsf{f}_{u}(c)}{\psi} I_{d_{x}} & 0_{d_{x}} \\
                         \frac{\mathsf{f}_{u}(c)}{\tau_{2}^{c}} ( \frac{c^{2} - \varsigma_{c}^{2}}{2} \zeta_{c}^{-} - \frac{2 c}{\psi} \omega_{c} )^{\prime} \Sigma^{1/2}  & \frac{c(c^{2} - \varsigma_{c}^{2})\mathsf{f}_{u}(c)}{\tau_{2}^{c}}
                         \end{Bmatrix} \label{estimation error theta and autoregressive coefficient with the same c}, \\
\Lambda_{c} & =
               \begin{Bmatrix}
               (M_{\tilde{x} \tilde{x}, n} \psi)^{-1} & 0_{d_{x}} \\
               - \frac{1}{\psi \tau_{2}^{c}} \omega_{c}^{\prime} \Sigma^{1/2} M_{\tilde{x} \tilde{x}, n}^{-1} & \frac{\sigma}{2 \tau_{2}^{c}}
               \end{Bmatrix}
\sum_{i=1}^{n}
                            \begin{Bmatrix}
                            \tilde{x}_{in} u_{i} \\
                            n^{-1/2} (\frac{u_{i}^{2}}{\sigma^{2}} - \varsigma_{c}^{2})
                            \end{Bmatrix}
1_{(|u_{i}| \le \sigma c)}. \label{kernel with the same c}
\end{align}

Apply the autoregressive equation (\ref{autoregressive equation for theta with 0 remainder term allows the varying cut-off c}) recursively to obtain the representation
\begin{equation} \label{recursive representation for theta with 0 remainder term allows the varying cut-off c}
\widehat{\theta}_{c}^{(m + 1)} = \Gamma_{c}^{m + 1} \widehat{\theta}_{c}^{(0)} + \sum_{l = 0}^{m} \Gamma_{c}^{l} \{ \Lambda_{c} + R_{\theta}(\widehat{\theta}_{c}^{(m - l)}, c) \}.
\end{equation}
Note the spectral norm of a matrix is compatible with Euclidean norm of a vector so $|Mx| \le | M | |x|$, and use the triangle inequality to get
\begin{equation*}
|\widehat{\theta}_{c}^{(m + 1)}| \le | \Gamma_{c}^{m + 1} | |\widehat{\theta}_{c}^{(0)}| + \{ | \Lambda_{c} | + \max_{0 \le l \le m} |R_{\theta}(\widehat{\theta}_{c}^{(l)}, c)| \} \sum_{l = 0}^{m} | \Gamma_{c}^{l} |.
\end{equation*}
Notice $\Gamma_{c}$ has $d_{x}$ number of eigenvalues $2c \mathsf{f}_{u}(c)/\psi$ and one eigenvalue $c(c^{2} - \varsigma_{c}^{2})\mathsf{f}_{u}(c)/\tau_{2}^{c}$, and (\ref{stationary condition for autoregressive coefficients}) in Remark \ref{coefficients in the 1-step stochastic expansion of iterated estimator} shows both of them are strictly bounded by one, so we have
\begin{equation} \label{spectral radius of autoregressive coefficient is bounded by 1}
\sup_{c_{+} \le c < \infty} \max | \mathrm{eigen}(\Gamma_{c}) | < 1.
\end{equation}
The fact, that the spectral radius of autoregressive coefficient matrix $\Gamma_{c}$ is bounded by one as shown in (\ref{spectral radius of autoregressive coefficient is bounded by 1}), is essential to build tightness and fixed point results for the iterative system (\ref{autoregressive equation for theta with 0 remainder term allows the varying cut-off c}). Then because of (\ref{spectral radius of autoregressive coefficient is bounded by 1}) two equalities follow
\begin{equation} \label{finite sum of geometric matrix sequence}
\sum_{l = 0}^{m} \Gamma_{c}^{l} = (I_{d_{x} + 1} - \Gamma_{c})^{-1} (I_{d_{x} + 1} - \Gamma_{c}^{m+1}) = (I_{d_{x} + 1} - \Gamma_{c}^{m+1}) (I_{d_{x} + 1} - \Gamma_{c})^{-1}.
\end{equation}
When $m \to \infty$, see 5.4$(3b)$ in L{\" u}tkepohl (1996), it holds
\begin{equation} \label{infinite sum of geometric matrix sequence}
\Gamma_{c}^{m + 1} \to 0_{(d_{x} + 1) \times (d_{x} + 1)}, \qquad \sum_{l = 0}^{\infty} \Gamma_{c}^{l} = (I_{d_{x} + 1} - \Gamma_{c})^{-1}.
\end{equation}
See Varga (2000, Theorem 3.4), Gelfand's formula gives
\begin{equation} \label{Gelfand formula}
\lim_{m \to \infty} | \Gamma_{c}^{m} |^{1/m} = \max | \mathrm{eigen}(\Gamma_{c}) |.
\end{equation}
Then (\ref{Gelfand formula}) implies for some $\omega$ such that $\sup_{c_{+} \le c < \infty} \max | \mathrm{eigen}(\Gamma_{c}) | < \omega < 1$ there exists $m_{0} > 0$ so for all $m > m_{0}$ we have
\begin{equation} \label{one inequality by Gelfand's formula for large m}
\sup_{c_{+} \le c < \infty} | \Gamma_{c}^{m} | < \omega^{m} < 1.
\end{equation}
This with the equality in (\ref{infinite sum of geometric matrix sequence}) implies for some $1 < B_{0} < \infty$
\begin{equation} \label{two inequalities by Gelfand's formula for all m}
\sup_{0 \le m < \infty} \sup_{c_{+} \le c < \infty} | \Gamma_{c}^{m} | < B_{0}, \quad \sup_{c_{+} \le c < \infty} | (I_{d_{x} + 1} - \Gamma_{c})^{-1} |  \le \sum_{l = 0}^{\infty} \sup_{c_{+} \le c < \infty} | \Gamma_{c}^{l} | < B_{0}.
\end{equation}
Use (\ref{two inequalities by Gelfand's formula for all m}) to show for all $m \in [0, \infty)$
\begin{equation} \label{inequality for recursive equation after Gelfand formula}
|\widehat{\theta}_{c}^{(m+1)}| < B_{0} \{ |\widehat{\theta}_{c}^{(0)}| + |\Lambda_{c}| + \max_{0 \le l \le m} |R_{\theta}(\widehat{\theta}_{c}^{(l)}, c)| \}.
\end{equation}
Assumption \ref{sufficient assumptions}$(v)$ with $\eta = 1/4$ guarantees boundedness of $\widehat{\theta}_{c}^{(0)}$, while the kernel $\Lambda_{c}$ is tight by Theorem \ref{tightness result for absolute empirical process} using Assumption \ref{sufficient assumptions}$(ia, ivc)$. Thus, for all $\epsilon, \delta > 0$ there exist $n_{0}, \Theta_{0} > 0$ so that for all $n > n_{0}$ the set
\begin{equation} \label{large probability set for tightness and fixed point}
\mathcal{A}_{n} = \{ B_{0} \sup_{c_{+} \le c < \infty} (|\widehat{\theta}_{c}^{(0)}| + |\Lambda_{c}|) \le \Theta_{0}/3, B_{0} \sup_{c_{+} \le c < \infty} \sup_{|\theta| \le \Theta_{0}} |R_{\theta}(\theta, c)| < \delta / 2 \}
\end{equation}
has probability larger than $1 - \epsilon$.

Mathematical induction over $m$ is used to show $\sup_{0 \le m < \infty} \sup_{c_{+} \le c < \infty} |\widehat{\theta}_{c}^{(m)}| \le \Theta_{0}$ on the set $\mathcal{A}_{n}$. For $m = 0$ as induction starts, $\sup_{c_{+} \le c < \infty} |\widehat{\theta}_{c}^{(0)}| \le B_{0}^{-1}\Theta_{0}/3 < \Theta_{0}$ holds since $B_{0} > 1$. The induction assumption is that $\sup_{0 \le l \le m} \sup_{c_{+} \le c < \infty} |\widehat{\theta}_{c}^{(l)}| \le \Theta_{0}$. This implies $B_{0} \max_{0 \le l \le m} |R_{\theta}(\widehat{\theta}_{c}^{(l)}, c)| < \delta / 2$, and then the bound in (\ref{inequality for recursive equation after Gelfand formula}) becomes $\sup_{c_{+} \le c < \infty} |\widehat{\theta}_{c}^{(m+1)}| < 2 \Theta_{0}/3 + \delta / 2 < \Theta_{0}$ so that $\sup_{0 \le l \le m+1} \sup_{c_{+} \le c < \infty} |\widehat{\theta}_{c}^{(l)}| \le \Theta_{0}$.
\end{proof}

\begin{proof}[\textnormal{\textbf{Proof of Corollary \ref{consistency of iterated location estimator in the first stage}}}]
Tightness of initial estimators for structural parameters $\beta$, $\sigma^{2}$ in (\ref{structural equation}) implicitly assumes consistency of the initial estimator for location parameter $\Pi$ in the first stage regression (\ref{first stage regression}). Combine Theorem \ref{tightness allows varying cut-off c for the iterated 2sls} and \ref{consistency of location estimator in the first stage}, then consistency of $\widehat{\Pi}_{c}^{(m)}$ follows uniformly in $m \in [0, \infty)$ and $c \in [c_{+}, \infty)$.
\end{proof}

With tightness results in Theorem \ref{tightness allows varying cut-off c for the iterated 2sls}, one-step expansion in Theorem \ref{1-step stochastic expansion allows varying cut-off c for the iterated 2sls} can be applied recursively to obtain Theorem \ref{stochastic expansion of the iterated 2sls in terms of initial estimators allows varying cut-off c}.

\begin{proof}[\textnormal{\textbf{Proof of Theorem \ref{stochastic expansion of the iterated 2sls in terms of initial estimators allows varying cut-off c}}}]
By Assumption \ref{sufficient assumptions}, Theorem \ref{1-step stochastic expansion allows varying cut-off c for the iterated 2sls} shows the one-step expansion (\ref{autoregressive equation for theta with 0 remainder term allows the varying cut-off c}). Then apply it recursively to obtain (\ref{recursive representation for theta with 0 remainder term allows the varying cut-off c}) so for any $m \in [0, \infty)$
\begin{equation*}
\widehat{\theta}_{c}^{(m + 1)} = \Gamma_{c}^{m + 1} \widehat{\theta}_{c}^{(0)} + \sum_{l = 0}^{m} \Gamma_{c}^{l} \Lambda_{c} + \sum_{l = 0}^{m} \Gamma_{c}^{l} R_{\theta}(\widehat{\theta}_{c}^{(m - l)}, c),
\end{equation*}
where $\sup_{c_{+} \le c < \infty} \sup_{|\theta| \le B} |R_{\theta}(\theta, c)| = \mathrm{o}_{\mathsf{P}}(1)$. As the spectral radius of $\Gamma_{c}$ is bounded by one, see (\ref{spectral radius of autoregressive coefficient is bounded by 1}), then (\ref{two inequalities by Gelfand's formula for all m}) shows for $1 < B_{0} < \infty$
\begin{equation*}
\sup_{0 \le m < \infty} \sup_{c_{+} \le c < \infty} | \Gamma_{c}^{m} | < B_{0}, \quad \sup_{c_{+} \le c < \infty} |\sum_{l = 0}^{m} \Gamma_{c}^{l}| \le \sup_{c_{+} \le c < \infty} \sum_{l = 0}^{m} |\Gamma_{c}^{l}| \le \sum_{l = 0}^{\infty} \sup_{c_{+} \le c < \infty} |\Gamma_{c}^{l}| < B_{0}.
\end{equation*}
Further with tightness $\sup_{0 \le m < \infty} \sup_{c_{+} \le c < \infty} |\widehat{\theta}_{c}^{(m)}| = \mathrm{O}_{\mathsf{P}}(1)$ shown in Theorem \ref{tightness allows varying cut-off c for the iterated 2sls}, the third term in the representation of $\widehat{\theta}_{c}^{(m + 1)}$ vanishes, then apply the second equality in (\ref{finite sum of geometric matrix sequence}) so uniformly in $c \in [c_{+}, \infty)$ we have for any $m \in [0, \infty)$
\begin{equation} \label{m+1 step stochastic expansion in terms of initial estimators, kernels, and small remainder terms}
\widehat{\theta}_{c}^{(m + 1)} = \Gamma_{c}^{m + 1} \widehat{\theta}_{c}^{(0)} + (I_{d_{x} + 1} - \Gamma_{c}^{m + 1}) (I_{d_{x} + 1} - \Gamma_{c})^{-1} \Lambda_{c} + \mathrm{o}_{\mathsf{P}}(1).
\end{equation}
See $\Gamma_{c}$ in (\ref{estimation error theta and autoregressive coefficient with the same c}), by matrix manipulations we get
\begin{align*}
\Gamma_c^{m+1}
&=
\begin{Bmatrix}
\varrho_{\beta\beta,c}^{(m+1)}I_{d_x}
&
0_{d_x}
\\[6pt]
\displaystyle
\frac{\mathsf f_u(c)\varrho_{\sigma\beta,c}^{(m+1)}}
{\tau_2^c}
(
\frac{c^2-\varsigma_c^2}{2}\zeta_c^-
-\frac{2c}{\psi}\omega_c
)^{\prime}\Sigma^{1/2}
&
\varrho_{\sigma\sigma,c}^{(m+1)}
\end{Bmatrix},
\\[10pt]
(I_{d_x+1}-\Gamma_c^{m+1})(I_{d_x+1}-\Gamma_c)^{-1}
&=
\begin{Bmatrix}
\psi\varrho_{\beta\tilde{x}u,c}^{(m+1)}I_{d_x}
&
0_{d_x}
\\[6pt]
\displaystyle
\psi\varrho_{\sigma\tilde{x}u,c}^{(m+1)}
(
\frac{c^2-\varsigma_c^2}{2}\zeta_c^-
-\frac{2c}{\psi}\omega_c
)^{\prime}\Sigma^{1/2}
&
\tau_2^c\varrho_{\sigma uu,c}^{(m+1)}
\end{Bmatrix}.
\end{align*}
see $\varrho_{\beta \beta, c}^{(m + 1)}$, $\varrho_{\beta \tilde{x} u, c}^{(m + 1)}$, $\varrho_{\sigma \sigma, c}^{(m + 1)}$, $\varrho_{\sigma u u, c}^{(m + 1)}$, $\varrho_{\sigma \beta, c}^{(m + 1)}$, $\varrho_{\sigma \tilde{x} u, c}^{(m + 1)}$ in Theorem \ref{stochastic expansion of the iterated 2sls in terms of initial estimators allows varying cut-off c}. Substitute these and $\widehat{\theta}_{c}^{(0)}$, $\Lambda_{c}$ into (\ref{m+1 step stochastic expansion in terms of initial estimators, kernels, and small remainder terms}), see $\widehat{\theta}_{c}^{(0)}$ in (\ref{estimation error theta and autoregressive coefficient with the same c}) and $\Lambda_{c}$ in (\ref{kernel with the same c}), then the expression for $\widehat{\theta}_{c}^{(m + 1)}$ is established.
\end{proof}

The next step is to prove the fixed point theorem when iteration step becomes sufficiently large.

\begin{proof}[\textnormal{\textbf{Proof of Theorem \ref{fixed point allows varying cut-off c for the iterated 2sls}}}]
Since the spectral radius of $\Gamma_{c}$ is strictly smaller than one, see (\ref{spectral radius of autoregressive coefficient is bounded by 1}), when $m \to \infty$ we can apply (\ref{infinite sum of geometric matrix sequence}) to the recursive representation (\ref{m+1 step stochastic expansion in terms of initial estimators, kernels, and small remainder terms}) shown in Theorem \ref{stochastic expansion of the iterated 2sls in terms of initial estimators allows varying cut-off c} by Assumption \ref{sufficient assumptions}, then as $n \to \infty$ and uniformly in $c \in [c_{+}, \infty)$ we have the fixed point
\begin{equation} \label{fixed point for theta with the varying cut-off c}
\widehat{\theta}_{c}^{(\ast)} = \widehat{\theta}_{c}^{(\infty)} = (I_{d_{x} + 1} - \Gamma_{c})^{-1} \Lambda_{c}.
\end{equation}
See $\Gamma_{c}$ in (\ref{estimation error theta and autoregressive coefficient with the same c}), we get
\begin{equation*}
(I_{d_{x} + 1} - \Gamma_{c})^{-1} =
\begin{bmatrix}
\frac{\psi}{\psi - 2 c \mathsf{f}_{u}(c)} I_{d_{x}} & 0_{d_{x}} \\
\frac{\psi \mathsf{f}_{u}(c)}{ \{ \psi - 2 c \mathsf{f}_{u}(c) \} \{ \tau_{2}^{c} - c (c^{2} - \varsigma_{c}^{2}) \mathsf{f}_{u}(c) \}} (\frac{c^2-\varsigma_c^2}{2}\zeta_c^- -\frac{2c}{\psi}\omega_c)^{\prime}\Sigma^{1/2} & \frac{\tau_{2}^{c}}{\tau_{2}^{c} - c (c^{2} - \varsigma_{c}^{2}) \mathsf{f}_{u}(c)}
\end{bmatrix}
.
\end{equation*}
Then substitute this and $\Lambda_{c}$ into (\ref{fixed point for theta with the varying cut-off c}), see $\Lambda_{c}$ in (\ref{kernel with the same c}), so the expression for the fixed point $\widehat{\theta}_{c}^{(\ast)}$ is demonstrated.

Substitute (\ref{recursive representation for theta with 0 remainder term allows the varying cut-off c}) and (\ref{fixed point for theta with the varying cut-off c}) into the deviation $\widehat{\Delta}_{c}^{(m+1)} = \widehat{\theta}_{c}^{(m+1)} - \widehat{\theta}_{c}^{(\ast)}$, then apply the second equality in (\ref{infinite sum of geometric matrix sequence}) so uniformly in $c \in [c_{+}, \infty)$ we have for any $m \in [0, \infty)$
\begin{equation*}
\widehat{\Delta}_{c}^{(m+1)} = \Gamma_{c}^{m+1} \{ \widehat{\theta}_{c}^{(0)} - (I_{d_{x} + 1} - \Gamma_{c})^{-1}\Lambda_{c} \} + \sum_{l = 0}^{m} \Gamma_{c}^{l} R_{\theta}(\widehat{\theta}_{c}^{(m-l)}, c).
\end{equation*}
To bound $\widehat{\Delta}_{c}^{(m+1)}$, use the triangle inequality and $|M x| \le | M | |x|$ to get
\begin{equation*}
|\widehat{\Delta}_{c}^{(m+1)}| \le | \Gamma_{c}^{m+1} | \{ |\widehat{\theta}_{c}^{(0)}| + | (I_{d_{x} + 1} - \Gamma_{c})^{-1} | |\Lambda_{c}| \} + \max_{0 \le l \le m} |R_{\theta}(\widehat{\theta}_{c}^{(l)}, c)| \sum_{l = 0}^{m} | \Gamma_{c}^{l} |.
\end{equation*}
Now consider the large $m$ such that $m > m_{0}$ then (\ref{one inequality by Gelfand's formula for large m}) can be used, so along with the second inequality in (\ref{two inequalities by Gelfand's formula for all m}) it follows
\begin{equation*}
|\widehat{\Delta}_{c}^{(m+1)}| < \omega^{m+1} ( |\widehat{\theta}_{c}^{(0)}| + B_{0} |\Lambda_{c}| ) + B_{0} \max_{0 \le l \le m} |R_{\theta}(\widehat{\theta}_{c}^{(l)}, c)|.
\end{equation*}
On the set $\mathcal{A}_{n}$ as in (\ref{large probability set for tightness and fixed point}), Theorem \ref{tightness allows varying cut-off c for the iterated 2sls} shows $\sup_{0 \le m < \infty} \sup_{c_{+} \le c < \infty} |\widehat{\theta}_{c}^{(m)}| \le \Theta_{0}$, so
\begin{equation*}
|\widehat{\Delta}_{c}^{(m+1)}| < \omega^{m+1}(B_{0}^{-1}\Theta_{0}/3 + \Theta_{0}/3) + \delta / 2 < \omega^{m+1}\Theta_{0} + \delta / 2,
\end{equation*}
where the second inequality follows by $1 < B_{0} < \infty$. As $0 < \omega < 1$, $\omega^{m+1}$ declines exponentially so $m_{0}$ can be chosen sufficiently large that for $m > m_{0}$ then $\omega^{m+1}\Theta_{0} < \delta / 2$. Thus for $n > n_{0}$, $m > m_{0}$ we have $\mathsf{P} (\sup_{c_{+} \le c < \infty} |\widehat{\Delta}_{c}^{(m+1)}| < \delta) > 1 - \epsilon$ .
\end{proof}

\subsection{Results for Algorithm \ref{robustified two stage least squares} and \ref{split sample version}} \label{results for robustified and saturated 2sls}
We first present a lemma to describe the limiting distribution of kernels.

\begin{lemma} \label{limiting distribution of kernels}
Suppose Assumption \ref{sufficient assumptions}$(ia, iva, ivc)$ holds. Then as $n \to \infty$ we have for any $c \in [c_{+}, \infty)$
\begin{equation*}
\sum_{i = 1}^{n}
\begin{pmatrix}
\tilde{x}_{in} u_{i} \\
\tilde{x}_{in} u_{i} 1_{(|u_{i}| \le \sigma c)}
\end{pmatrix}
\overset{\mathsf{D}}{\to}
\mathsf{N}
\begin{Bmatrix}
\begin{pmatrix}
0_{d_{x}} \\
0_{d_{x}}
\end{pmatrix}
,
\sigma^{2} \tau_{2}^{c}
\begin{pmatrix}
\frac{1}{\tau_{2}^{c}} M_{\tilde{x} \tilde{x}} & M_{\tilde{x} \tilde{x}} \\
M_{\tilde{x} \tilde{x}} & M_{\tilde{x} \tilde{x}}
\end{pmatrix}
\end{Bmatrix}
.
\end{equation*}
\end{lemma}

\begin{proof}[\textnormal{\textbf{Proof of Lemma \ref{limiting distribution of kernels}}}]
As $\tilde{x}_{i} = \Pi^{\prime} z_{i}$ and $z_{i}$ is independent of $u_{i}$ by Assumption \ref{sufficient assumptions}. Then $\mathsf{E} \tilde{x}_{i} u_{i} = 0_{d_{x}}$ and $\mathsf{Var} (\tilde{x}_{i} u_{i}) = \mathsf{E} \tilde{x}_{i} \tilde{x}_{i}^{\prime} u_{i}^{2} = \mathsf{E} u_{i}^{2} \mathsf{E} \tilde{x}_{i} \tilde{x}_{i}^{\prime} = \sigma^{2} M_{\tilde{x} \tilde{x}}$. Furthermore, we have $\mathsf{E} \tilde{x}_{i} u_{i} 1_{(|u_{i}| \le \sigma c)} = \mathsf{E} \tilde{x}_{i} \mathsf{E} u_{i} 1_{(|u_{i}| \le \sigma c)} = 0_{d_{x}}$ and
\begin{equation*}
\mathsf{Var} (\tilde{x}_{i} u_{i} 1_{(|u_{i}| \le \sigma c)}) = \mathsf{E} \tilde{x}_{i} \tilde{x}_{i}^{\prime} u_{i}^{2} 1_{(|u_{i}| \le \sigma c)} = \mathsf{E} u_{i}^{2} 1_{(|u_{i}| \le \sigma c)} \mathsf{E} \tilde{x}_{i} \tilde{x}_{i}^{\prime} = \sigma^{2} \tau_{2}^{c} M_{\tilde{x} \tilde{x}}.
\end{equation*}
Finally we calculate $\mathsf{Cov} (\tilde{x}_{i} u_{i}, \tilde{x}_{i} u_{i} 1_{(|u_{i}| \le \sigma c)}) = \mathsf{E} \tilde{x}_{i} \tilde{x}_{i}^{\prime} u_{i}^{2} 1_{(|u_{i}| \le \sigma c)} = \sigma^{2} \tau_{2}^{c} M_{\tilde{x} \tilde{x}}$. Note that $\tilde{x}_{in} = n^{-1/2} \tilde{x}_{i}$, so apply CLT to obtain for any $c \in [c_{+}, \infty)$
\begin{equation*}
n^{-1/2} \sum_{i = 1}^{n}
\begin{pmatrix}
\tilde{x}_{i} u_{i} \\
\tilde{x}_{i} u_{i} 1_{(|u_{i}| \le \sigma c)}
\end{pmatrix}
\overset{\mathsf{D}}{\to}
\mathsf{N}
\begin{Bmatrix}
\begin{pmatrix}
0_{d_{x}} \\
0_{d_{x}}
\end{pmatrix}
,
\begin{pmatrix}
\sigma^{2} M_{\tilde{x} \tilde{x}} & \sigma^{2}  \tau_{2}^{c} M_{\tilde{x} \tilde{x}} \\
\sigma^{2}  \tau_{2}^{c} M_{\tilde{x} \tilde{x}} & \sigma^{2}  \tau_{2}^{c} M_{\tilde{x} \tilde{x}}
\end{pmatrix}
\end{Bmatrix}
. \qedhere
\end{equation*}
\end{proof}

The next step is to provide the expansion and limiting distribution of the full sample two stage least squares estimator.

\begin{proof}[\textnormal{\textbf{Proof of Lemma \ref{expansion of 2sls and its limiting distribution}}}]
Substitute (\ref{first stage regression}) into the location estimator for $\Pi$ to obtain
\begin{equation*}
\widetilde{\Pi} = \Pi + M_{z z, n}^{-1} n^{-1/2} \sum_{i = 1}^{n} z_{in} r_{i}^{\prime} \overset{\mathsf{P}}{\to} \Pi + M_{z z}^{-1} \mathsf{E} z_{i} r_{i}^{\prime} = \Pi + M_{z z}^{-1} 0_{d_{z} \times d_{x}} = \Pi,
\end{equation*}
where the probability limit follows by Assumption \ref{sufficient assumptions}$(iva)$ that $M_{z z, n} \overset{\mathsf{P}}{\to} M_{z z} > 0$, LLN, Slutsky's theorem, and $\mathsf{E} z_{i} r_{i}^{\prime} = 0_{d_{z} \times d_{x}}$. Notice LLN requires moment conditions $\mathsf{E} |r_{i}|, \mathsf{E} |z_{i}| < \infty$ by Assumption \ref{sufficient assumptions}$(iia, ivc)$. We have $\widetilde{\Pi}^{\prime} \sum_{i = 1}^{n} z_{i} z_{i}^{\prime} \widetilde{\Pi} = \widetilde{\Pi}^{\prime} \sum_{i = 1}^{n} z_{i} x_{i}^{\prime}$, see (\ref{two expressions for denomiator of beta estimator}), then apply this and structural equation (\ref{structural equation}) to get
\begin{equation*}
n^{1/2} (\widetilde{\beta} - \beta) = (\widetilde{\Pi}^{\prime} M_{z z, n} \widetilde{\Pi})^{-1} (\widetilde{\Pi}^{\prime} \sum_{i = 1}^{n} z_{in} u_{i}) = M_{\tilde{x} \tilde{x}, n}^{-1} \sum_{i = 1}^{n} \tilde{x}_{in} u_{i} + \mathrm{o}_{\mathsf{P}}(1),
\end{equation*}
where the second equality is due to consistency of $\widetilde{\Pi}$. Thus the limiting distribution follows by Lemma \ref{limiting distribution of kernels} and Assumption \ref{sufficient assumptions}$(iva)$ that $M_{\tilde{x} \tilde{x}, n} \overset{\mathsf{P}}{\to} M_{\tilde{x} \tilde{x}} > 0$.
\end{proof}

Then weak convergence results are proved for Algorithm \ref{robustified two stage least squares}, \ref{split sample version}.

\begin{proof}[\textnormal{\textbf{Proof of Theorem \ref{weak convergence for the m+1 step beta estimator for robustified 2sls}}}]
Algorithm \ref{robustified two stage least squares} chooses the full sample two stage least squares as the initial estimator $\widehat{\beta}_{c}^{(0)}$ in Algorithm \ref{iterated version of robust 2sls}. Thus substitute the asymptotic expansion of $n^{1/2} (\widehat{\beta}_{c}^{(0)} - \beta) = n^{1/2} (\widetilde{\beta} - \beta)$ in Lemma \ref{expansion of 2sls and its limiting distribution} into the recursive equation in Theorem \ref{stochastic expansion of the iterated 2sls in terms of initial estimators allows varying cut-off c}, then uniformly in $c \in [c_{+}, \infty)$ we have for any $m \in [0, \infty)$
\begin{equation*}
\mathbb{G}_{n}^{(m + 1)}(c) =
\begin{pmatrix}
\varrho_{\beta \beta, c}^{(m + 1)} M_{\tilde{x} \tilde{x}, n}^{-1} \\
\varrho_{\beta \tilde{x} u, c}^{(m + 1)} M_{\tilde{x} \tilde{x}, n}^{-1}
\end{pmatrix}^{\prime}
\sum_{i = 1}^{n}
\begin{pmatrix}
\tilde{x}_{in} u_{i} \\
\tilde{x}_{in} u_{i} 1_{(|u_{i}| \le \sigma c)}
\end{pmatrix}
+ \mathrm{o}_{\mathsf{P}}(1),
\end{equation*}
where expressions of $\varrho_{\beta \beta, c}^{(m + 1)}$, $\varrho_{\beta \tilde{x} u, c}^{(m + 1)}$ are shown in Theorem \ref{stochastic expansion of the iterated 2sls in terms of initial estimators allows varying cut-off c}. Then $M_{\tilde{x} \tilde{x}, n} \overset{\mathsf{P}}{\to} M_{\tilde{x} \tilde{x}} > 0$, Lemma \ref{limiting distribution of kernels}, and Slutsky's theorem show for any $c \in [c_{+}, \infty)$
\begin{equation*}
\mathbb{G}_{n}^{(m + 1)}(c) \overset{\mathsf{D}}{\to}
\begin{pmatrix}
\varrho_{\beta \beta, c}^{(m + 1)} M_{\tilde{x} \tilde{x}}^{-1} \\
\varrho_{\beta \tilde{x} u, c}^{(m + 1)} M_{\tilde{x} \tilde{x}}^{-1}
\end{pmatrix}^{\prime}
\mathsf{N}
\begin{Bmatrix}
\begin{pmatrix}
0_{d_{x}} \\
0_{d_{x}}
\end{pmatrix}
,
\sigma^{2} \tau_{2}^{c}
\begin{pmatrix}
\frac{1}{\tau_{2}^{c}} M_{\tilde{x} \tilde{x}} & M_{\tilde{x} \tilde{x}} \\
M_{\tilde{x} \tilde{x}} & M_{\tilde{x} \tilde{x}}
\end{pmatrix}
\end{Bmatrix}.
\end{equation*}
Since a transformation of multivariate normal is still normal, it follows
\begin{equation*}
\mathbb{G}_{n}^{(m + 1)}(c) \overset{\mathsf{D}}{\to} \mathsf{N} [0_{d_{x}}, \{ (\varrho_{\beta \beta, c}^{(m + 1)})^{2} + 2 \tau_{2}^{c} \varrho_{\beta \beta, c}^{(m + 1)} \varrho_{\beta \tilde{x} u, c}^{(m + 1)} + \tau_{2}^{c} (\varrho_{\beta \tilde{x} u, c}^{(m + 1)})^{2} \} \sigma^{2} M_{\tilde{x} \tilde{x}}^{-1}].
\end{equation*}
Theorem \ref{tightness allows varying cut-off c for the iterated 2sls} demonstrates the process $\mathbb{G}_{n}^{(m + 1)}$ is tight for any $m \in [0, \infty)$ thus
\begin{equation*}
\mathbb{G}_{n}^{(m + 1)} \leadsto \mathbb{G}^{(m + 1)},
\end{equation*}
where the weak limit $\mathbb{G}^{(m + 1)}$ is a zero mean Gaussian process with the variance
\begin{equation*}
\mathsf{Var} \{ \mathbb{G}^{(m + 1)}(c) \} = \{ (\varrho_{\beta \beta, c}^{(m + 1)})^{2} + 2 \tau_{2}^{c} \varrho_{\beta \beta, c}^{(m + 1)} \varrho_{\beta \tilde{x} u, c}^{(m + 1)} + \tau_{2}^{c} (\varrho_{\beta \tilde{x} u, c}^{(m + 1)})^{2} \} \sigma^{2} M_{\tilde{x} \tilde{x}}^{-1}. \qedhere
\end{equation*}
\end{proof}

\begin{proof}[\textnormal{\textbf{Proof of Corollary \ref{weak convergence for the first step beta estimator for robustified 2sls}}}]
This is a special case of Theorem \ref{weak convergence for the m+1 step beta estimator for robustified 2sls} when $m = 0$ such that $\varrho_{\beta \beta, c}^{(1)} = 2 c \mathsf{f}_{u}(c) / \psi$, $\varrho_{\beta \tilde{x} u, c}^{(1)} = \psi^{-1}$.
\end{proof}

\begin{proof}[\textnormal{\textbf{Proof of Theorem \ref{weak convergence for the m+1 step beta estimator for split sample iterated 2sls}}}]
The first step updated estimator for $\Pi$ is expressed as
\begin{align*}
\widehat{\Pi}_{c}^{(1)} - \Pi & = (\sum_{j = 1, 2} \frac{n_{3 - j}}{n} \sum_{i \in \mathcal{I}_{3 - j}} z_{i n_{3 - j}} z_{i n_{3 - j}}^{\prime} 1_{(|y_{i} - x_{i}^{\prime} \widehat{\beta}_{j}| \le \widehat{\sigma}_{j} c)})^{-1} \\
& \quad \,\, (\sum_{j = 1, 2} \frac{n_{3 - j}}{n} n_{3 - j}^{-1/2} \sum_{i \in \mathcal{I}_{3 - j}} z_{i n_{3 - j}} r_{i}^{\prime} 1_{(|y_{i} - x_{i}^{\prime} \widehat{\beta}_{j}| \le \widehat{\sigma}_{j} c)}).
\end{align*}
Assumption \ref{sufficient assumptions}$(ia, iia, iva, ivc)$ holds for each of sub-sample $\mathcal{I}_{j}$, so two stage least squares $\widehat{\beta}_{j}, \widehat{\sigma}_{j}^{2}$ are tight for $j = 1, 2$, see Lemma \ref{expansion of 2sls and its limiting distribution}. Argue as Theorem \ref{consistency of location estimator in the first stage} to obtain
\begin{equation*}
\widehat{\Pi}_{c}^{(1)} - \Pi = (\sum_{i = 1}^{n} z_{in} z_{in}^{\prime} 1_{(|u_{i}| \le \sigma c)})^{-1} (n^{-1/2} \sum_{i = 1}^{n} z_{in} r_{i}^{\prime} 1_{(|u_{i}| \le \sigma c)}) + \mathrm{o}_{\mathsf{P}}(1),
\end{equation*}
uniformly in $c \in [c_{+}, \infty)$. Therefore uniform consistency of $\widehat{\Pi}_{c}^{(1)}$ immediately follows, see proof of Theorem \ref{consistency of location estimator in the first stage}.

The updated estimator for $\beta$ is expressed as
\begin{align*}
n^{1/2} (\widehat{\beta}_{c}^{(1)} - \beta) & = (\widehat{\Pi}_{c}^{(1) \prime} \sum_{j = 1, 2} \frac{n_{3 - j}}{n} \sum_{i \in \mathcal{I}_{3 - j}} z_{i n_{3 - j}} z_{i n_{3 - j}}^{\prime} 1_{(|y_{i} - x_{i}^{\prime} \widehat{\beta}_{j}| \le \widehat{\sigma}_{j} c)} \widehat{\Pi}_{c}^{(1)})^{-1} \\
& \quad \,\, (\widehat{\Pi}_{c}^{(1) \prime} \sum_{j = 1, 2} (\frac{n_{3 - j}}{n})^{1/2} \sum_{i \in \mathcal{I}_{3 - j}} z_{i n_{3 - j}} u_{i} 1_{(|y_{i} - x_{i}^{\prime} \widehat{\beta}_{j}| \le \widehat{\sigma}_{j} c)}).
\end{align*}
Argue along the lines of Theorem \ref{1-step stochastic expansion allows varying cut-off c for the iterated 2sls}, then it follows
\begin{align*}
n^{1/2} (\widehat{\beta}_{c}^{(1)} - \beta) & = (M_{\tilde{x} \tilde{x}, n} \psi)^{-1} \{ 2 c \mathsf{f}_{u}(c) \sum_{j = 1, 2} (\frac{n_{3 - j}}{n})^{1/2} (\frac{n_{3 - j}}{n_{j}})^{1/2} M_{\tilde{x} \tilde{x}, n_{3 - j}}^{\mathcal{I}_{3 - j}} n_{j}^{1/2} (\widehat{\beta}_{j} - \beta) \\
& \quad \,\, + \sum_{i = 1}^{n} \tilde{x}_{i n} u_{i} 1_{(|u_{i}| \le \sigma c)} \} + \mathrm{o}_{\mathsf{P}}(1),
\end{align*}
uniformly in $c \in [c_{+}, \infty)$ and where we denote $M_{\tilde{x} \tilde{x}, n_{j}}^{\mathcal{I}_{j}} = \sum_{i \in \mathcal{I}_{j}} \tilde{x}_{i n_{j}} \tilde{x}_{i n_{j}}^{\prime}$ for $j = 1, 2$. Again Assumption \ref{sufficient assumptions}$(ia, iia, iva, ivc)$ holds for each sub-sample $\mathcal{I}_{j}$, so Lemma \ref{expansion of 2sls and its limiting distribution} shows
$n_{j}^{1/2} (\widehat{\beta}_{j} - \beta) = (M_{\tilde{x} \tilde{x}, n_{j}}^{\mathcal{I}_{j}})^{-1} \sum_{i \in \mathcal{I}_{j}} \tilde{x}_{i n_{j}} u_{i} + \mathrm{o}_{\mathsf{P}}(1)$ and notice $M_{\tilde{x} \tilde{x}, n_{j}}^{\mathcal{I}_{j}} \overset{\mathsf{P}}{\to} M_{\tilde{x} \tilde{x}}$ as $n \to \infty$ such that $n_{j} \to \infty$ for each $j = 1, 2$. Thus we have
\begin{equation*}
n^{1/2} (\widehat{\beta}_{c}^{(1)} - \beta) = \frac{2 c \mathsf{f}_{u}(c)}{\psi} M_{\tilde{x} \tilde{x}, n}^{-1} \sum_{j = 1, 2} \frac{n_{3 - j}}{n_{j}} \sum_{i \in \mathcal{I}_{j}} \tilde{x}_{in} u_{i} + (M_{\tilde{x} \tilde{x}, n} \psi)^{-1} \sum_{i = 1}^{n} \tilde{x}_{in} u_{i} 1_{(|u_{i}| \le \sigma c)} + \mathrm{o}_{\mathsf{P}}(1),
\end{equation*}
and if assume $n_{1} = n_{2} = n/2$ it further holds
\begin{equation*}
n^{1/2} (\widehat{\beta}_{c}^{(1)} - \beta) = \frac{2 c \mathsf{f}_{u}(c)}{\psi} M_{\tilde{x} \tilde{x}, n}^{-1} \sum_{i = 1}^{n} \tilde{x}_{in} u_{i} + (M_{\tilde{x} \tilde{x}, n} \psi)^{-1} \sum_{i = 1}^{n} \tilde{x}_{in} u_{i} 1_{(|u_{i}| \le \sigma c)} + \mathrm{o}_{\mathsf{P}}(1).
\end{equation*}

We find the expansion of $n^{1/2} (\widehat{\beta}_{c}^{(1)} - \beta)$ for Algorithm \ref{split sample version} is the same as Algorithm \ref{robustified two stage least squares}, see Corollary \ref{weak convergence for the first step beta estimator for robustified 2sls}, therefore two algorithms have the identical asymptotics even they apply different initial estimates when running Algorithm \ref{iterated version of robust 2sls}.
\end{proof}

Finally we establish weak convergence for the fixed point of Algorithm \ref{iterated version of robust 2sls}, \ref{robustified two stage least squares}, \ref{split sample version}.

\begin{proof}[\textnormal{\textbf{Proof of Theorem \ref{weak convergence for the fixed point of beta estimator}}}]
Consider Algorithm \ref{iterated version of robust 2sls}, \ref{robustified two stage least squares}, \ref{split sample version}. Let $m \to \infty$ then Theorem \ref{fixed point allows varying cut-off c for the iterated 2sls} shows the fixed point for $c \in [c_{+}, \infty)$
\begin{equation*}
\mathbb{G}_{n}^{(\ast)}(c) = \frac{1}{\psi - 2c \mathsf{f}_{u}(c)} M_{\tilde{x} \tilde{x}, n}^{-1} \sum_{i = 1}^{n} \tilde{x}_{in} u_{i} 1_{(|u_{i}| \le \sigma c)}.
\end{equation*}
As $n \to \infty$ then $M_{\tilde{x} \tilde{x}, n} \overset{\mathsf{P}}{\to} M_{\tilde{x} \tilde{x}} > 0$ and Lemma \ref{limiting distribution of kernels} demonstrate finite dimensional convergence
\begin{equation*}
\mathbb{G}_{n}^{(\ast)}(c) \overset{\mathsf{D}}{\to} \mathsf{N}[ 0_{d_{x}}, \frac{\tau_{2}^{c}}{\{ \psi - 2c \mathsf{f}_{u}(c) \}^{2}} \sigma^{2} M_{\tilde{x} \tilde{x}}^{-1} ].
\end{equation*}
Tightness shown in Theorem \ref{tightness allows varying cut-off c for the iterated 2sls} further establishes weak convergence
\begin{equation*}
\mathbb{G}_{n}^{(\ast)} \leadsto \mathbb{G}^{(\ast)},
\end{equation*}
where $\mathbb{G}^{(\ast)}$ is a zero mean Gaussian process with variance
\begin{equation*}
\mathsf{Var}\{ \mathbb{G}^{(\ast)}(c) \} = \frac{\tau_{2}^{c}}{\{ \psi - 2c \mathsf{f}_{u}(c) \}^{2}} \sigma^{2} M_{\tilde{x} \tilde{x}}^{-1}. \qedhere
\end{equation*}
\end{proof}

\subsection{Results for a new type of Hausman test}
We first prove the stochastic expansion and weak limit of the difference processes between the robustified and full sample 2SLS.

\begin{proof}[\textnormal{\textbf{Proof of Theorem \ref{expansion and weak limit of processes of Hausman statistics}}}]
Notice that $\mathbb{H}_{n}^{(m + 1)}(c)$ can be rewritten as
\begin{equation*}
\mathbb{H}_{n}^{(m + 1)}(c) = n^{1/2} (\widehat{\beta}_{c}^{(m + 1)} - \widetilde{\beta}) = n^{1/2} (\widehat{\beta}_{c}^{(m + 1)} - \beta) - n^{1/2} (\widetilde{\beta} - \beta).
\end{equation*}
Substitute expansions of $n^{1/2} (\widehat{\beta}_{c}^{(m + 1)} - \beta)$ in Theorem \ref{weak convergence for the m+1 step beta estimator for robustified 2sls}, \ref{weak convergence for the m+1 step beta estimator for split sample iterated 2sls} and $n^{1/2} (\widetilde{\beta} - \beta)$ in Lemma \ref{expansion of 2sls and its limiting distribution} into the above term, then the expansion is attained for $\mathbb{H}_{n}^{(m + 1)}(c)$ for any $c \in [c_{+}, \infty)$ and $m \in [0, \infty)$. Argue along Proof of Theorem \ref{weak convergence for the m+1 step beta estimator for robustified 2sls} but replace $\varrho_{\beta \beta, c}^{(m + 1)}$ by $\varrho_{\beta \beta, c}^{(m + 1)} - 1$ to obtain the weak Gaussian limit $\mathbb{H}^{(m + 1)}$ of processes $\mathbb{H}_{n}^{(m + 1)}$.
\end{proof}

The next two proofs establish corollaries on a new Hausman type test for performing outlier robustness checks.

\begin{proof}[\textnormal{\textbf{Proof of Corollary \ref{pointwise convergence of Hausman test statistics}}}]
Gaussian weak limit in Theorem \ref{expansion and weak limit of processes of Hausman statistics} immediately implies the finite dimensional convergence result as normal distribution. The limiting $\chi^{2}$ distribution then follows for the proposed test statistics with the degree of freedom $d_{x}$ as the dimension of the structural parameter $\beta$.
\end{proof}

\begin{proof}[\textnormal{\textbf{Proof of Corollary \ref{Hausman test when m = 1 and m = infinite}}}]
Set $m = 0, \infty$ and substitute the terms $\varrho_{\beta \beta, c}^{(1)} = 2 c \mathsf{f}_{u}(c) / \psi$, $\varrho_{\beta \tilde{x} u, c}^{(1)} = \psi^{-1}$, $\varrho_{\beta \beta, c}^{(\infty)} = 0$, $\varrho_{\beta \tilde{x} u, c}^{(\infty)} = \{ \psi - 2 c \mathsf{f}_{u}(c) \}^{-1}$ into $\widehat{\mathsf{avar}}(\widehat{\beta}_{c}^{(m + 1)} - \widetilde{\beta})$ given by (\ref{estimated asymptotic variance}), then our proposed tests follow.
\end{proof}


\newpage
\bibliographystyle{plain}


\end{document}